\documentclass[aap]{imsart}

\usepackage{amsmath, amsfonts, amsthm, verbatim,
	amssymb,dsfont,color, accents}

\usepackage[utf8x]{inputenc}
\usepackage{mathtools}
\usepackage[round]{natbib}
\usepackage{tikz}
\usetikzlibrary{arrows.meta,arrows,decorations.pathreplacing}
\RequirePackage[colorlinks,citecolor=blue,urlcolor=blue]{hyperref}
\usepackage{tabularx,booktabs}

\startlocaldefs

\newcommand{\R}{\mathbb{R}}

\newcommand{\N}{\mathbb{N}}
\newcommand{\E}{\mathbb{E}}
\newcommand{\F}{\mathbb{F}}

\renewcommand{\P}{\mathbb{P}}
\renewcommand{\L}{\mathbb{L}}
\newcommand{\C}{\mathbb{C}}
\newcommand{\Q}{\mathbb{Q}}
\renewcommand{\S}{\mathbb{S}}

\newcommand{\bone}{\mathbf 1}

\theoremstyle{plain}
\newtheorem{theorem}{Theorem}[section]
\newtheorem{lemma}[theorem]{Lemma}
\newtheorem{proposition}[theorem]{Proposition}
\newtheorem{corollary}[theorem]{Corollary}
\newtheorem{condition}{Condition}
\newtheorem{Assumption}{Assumption}
\makeatletter
\newcommand{\settheoremtag}[1]{% \settheoremtag{<tag>}
	\let\oldtheAssumption\theAssumption% Store \thetheorem
	\renewcommand{\theAssumption}{#1}% Redefine it to a fixed value
	\g@addto@macro\endAssumption{% At \end{theorem}, ...
		\global\let\theAssumption\oldtheAssumption}% ...restore \thetheorem
}
\makeatother
\theoremstyle{remark}
\newtheorem{remark}[theorem]{Remark}

\newcommand{\calc}{{\cal C}}

\newcommand{\calu}{{\cal U}}
\newcommand{\calf}{{\cal F}}

\newcommand{\cale}{{\cal E}}

\newcommand{\call}{{\cal L}}

\newcommand{\calt}{{\cal T}}

\newcommand{\calg}{{\cal G}}

\newcommand{\caly}{{\cal Y}}

\newcommand{\calv}{{\cal V}}

\newcommand{\calz}{{\cal Z}}
\newcommand{\pf}{{\mathfrak p}}
\newcommand{\qf}{{\mathfrak q}}

\newcommand{\limd}{\stackrel{d}{\longrightarrow}}
\newcommand{\limst}{\stackrel{\mathcal{L}-s}{\longrightarrow}}
\newcommand{\limp}{\stackrel{\mathbb{P}}{\longrightarrow}}

\newcommand{\al}{{\alpha}}
\newcommand{\la}{{\lambda}}
\newcommand{\La}{{\Lambda}}
\newcommand{\eps}{{\epsilon}}
\newcommand{\ga}{{\gamma}}
\newcommand{\Ga}{{\Gamma}}
\newcommand{\vp}{{\varphi}}
\newcommand{\si}{{\sigma}}

\newcommand{\om}{{\omega}}
\newcommand{\Om}{{\Omega}}

\newcommand{\ov}{\overline}
\newcommand{\un}{\underline}
\newcommand{\wh}{\widehat}
\newcommand{\wt}{\widetilde}

\newcommand{\mm}{\mathrm{m}}

\newcommand{\cc}{\mathrm{c}}

\newcommand{\bj}{{\boldsymbol{j}}}
\newcommand{\bs}{{\boldsymbol{s}}}
\newcommand{\bz}{{\boldsymbol{z}}}
\newcommand{\by}{{\boldsymbol{y}}}
\newcommand{\bY}{{\boldsymbol{Y}}}

\newcommand{\uc}{\mathrm{uc}}

\newcommand{\Leb}{\mathrm{Leb}}
\newcommand{\Den}{\Delta_n}
\newcommand{\avar}{\mathrm{AVar}}

\newcommand{\bthm}{\begin{theorem}}
	\newcommand{\ethm}{\end{theorem}}

\newcommand{\bcor}{\begin{corollary}}
	\newcommand{\ecor}{\end{corollary}}

\newcommand{\blem}{\begin{lemma}}
	\newcommand{\elem}{\end{lemma}}

\newcommand{\bprop}{\begin{proposition}}
	\newcommand{\eprop}{\end{proposition}}

\newcommand{\bcond}{\begin{condition}}
	\newcommand{\econd}{\end{condition}}

\newcommand{\bdf}{\begin{definition}}
	\newcommand{\edf}{\end{definition}}

\newcommand{\bex}{\begin{example}}
	\newcommand{\eex}{\end{example}}

\newcommand{\brem}{\begin{remark}}
	\newcommand{\erem}{\end{remark}}

\newcommand{\bpr}{\begin{proof}}
	\newcommand{\epr}{\end{proof}}

\newcommand{\benu}{\begin{enumerate}}
	\newcommand{\eenu}{\end{enumerate}}

\newcommand{\beq}{\begin{equation}}
	\newcommand{\eeq}{\end{equation}}

\newcommand{\bit}{\begin{itemize}}
	\newcommand{\eit}{\end{itemize}}

\newcommand{\bass}{\begin{Assumption}}
	\newcommand{\eass}{\end{Assumption}}

\numberwithin{equation}{section}

\DeclareMathOperator\Log{Log}

\allowdisplaybreaks[4]
\endlocaldefs

\begin{document}
	
	\begin{frontmatter}
	\title{Asymptotic Expansions for High-Frequency Option Data}
		%\title{A sample article title with some additional note\thanksref{t1}}
		\runtitle{Asymptotic Expansions for High-Frequency Option Data}
		%\thankstext{T1}{A sample additional note to the title.}
		
		\begin{aug}
			\author[A]{\fnms{Carsten H.} \snm{Chong}\ead[label=e1]{carstenchong@ust.hk}}
			\and
			\author[B]{\fnms{Viktor} \snm{Todorov}\ead[label=e2]{v-todorov@kellogg.northwestern.edu}}
			
			\address[A]{Department of Information Systems, Business Statistics and Operations Management
				The Hong Kong University of Science and Technology,
				\printead{e1}}
			
			\address[B]{Department of Finance, Northwestern University,
				\printead{e2}}
		\end{aug}

			\begin{abstract}
			 We derive a nonparametric higher-order asymptotic expansion for small-time changes of conditional characteristic functions of  It\^o semimartingale increments. The asymptotics setup is of joint type: both the length of the time interval of the increment of the underlying process and the time gap between evaluating the conditional characteristic function are shrinking. The spot semimartingale characteristics of the underlying process as well as their spot semimartingale characteristics appear as leading terms in the derived asymptotic expansions. The analysis applies to a general class of It\^o semimartingales that includes in particular L\'{e}vy-driven SDEs and time-changed L\'{e}vy processes. The asymptotic expansion results are subsequently used to construct a test for whether volatility jumps are of finite or infinite variation. In an application to high-frequency data of options written on the S\&P 500 index, we find evidence for infinite variation volatility jumps.
			  %of direct use for constructing nonparametric estimates pertaining to the stochastic volatility dynamics of an asset from high-frequency data of options written on the underlying asset.    
			%Deep It\^o semimartingales, that our analysis applies to, are a general class of It\^o semimartingales whose spot characteristics are themselves It\^o semimartingales, and the same is true for several layers of the semimartingale characteristics dynamics. 
		\end{abstract}
		
		
		
	\end{frontmatter}

	\section{Introduction}
	
	Trading in options have increased a lot over the past decade with stock option volume even exceeding underlying shares volume in 2020. The improved liquidity in options markets makes the use of option data sampled at high frequencies practically feasible. According to a result by \cite{CM01}, the conditional characteristic function of a future price increment can be constructed nonparametrically from a portfolio of options with different strikes. More precisely, if $x$ denotes the logarithmic price of an asset and $O_{t,T}(k)$ denotes the price at time $t$ of an European style out-of-the-money option expiring at $t+T$ and with log-strike $k$, that is, 
	\begin{equation}\label{eq:opt} 
		O_{t,T}(k)=\begin{cases} \E_t[(e^k- e^{x_{t+T}})\vee 0] &\text{if } k\leq x_t, \\ \E_t[( e^{x_{t+T}}-e^k)\vee 0] &\text{if } k> x_t, \end{cases}
	\end{equation}
	then
	\begin{equation}\label{eq:spanning}
		\E_t[e^{iu(x_{t+T}-x_t)/\sqrt{T}}] = 1-\left(\frac{u^2}{T}+i\frac{u}{\sqrt{T}}\right)e^{-x_t}\int_{\mathbb{R}}e^{(iu/\sqrt{T}-1)(k-x_t)}O_{t,T}(k)dk
	\end{equation}
	for all $u\in\R$. In \eqref{eq:opt} and \eqref{eq:spanning}, $\E_t $ denotes conditional expectation with respect to the risk-neutral probability measure  given  information up to time $t$. For simplicity, we have assumed in \eqref{eq:spanning} that both the dividend yield and the risk-free interest rate are equal to  zero.
	
	Since option portfolios at a single time $t$ can be used to obtain $\E_t[e^{iu(x_{t+T}-x_t)/\sqrt{T}}]$, high-frequency option data gives access to increments of the conditional characteristic function of the form
	\begin{equation}\label{eq:exp} 
		\mathbb{E}_t[e^{iu(x_{t+T}-x_t)/\sqrt{T}}]- \mathbb{E}_{t-\Delta}[e^{iu(x_{t+T}-x_{t-\Delta})/\sqrt{T-\Delta}}],~~~u\in\mathbb{R}~\textrm{and}~\Delta>0.
	\end{equation}  
	
	Motivated by this, the goal of this paper is to develop higher-order asymptotic expansions of \eqref{eq:exp} for some It\^o semimartingale process $x$, when both $\Delta\downarrow 0$ and $T\downarrow 0$ simultaneously. 
Such expansions relate to several strands of work in the literature. First,  asymptotic expansions to first order for $	\mathbb{E}_t[e^{iu(x_{t+T}-x_t)/\sqrt{T}}]$ have been derived by \cite{jacod2014efficient,jacod2018limit}. Second, \cite{T21} extends the work of \cite{jacod2014efficient,jacod2018limit} by deriving higher-order expansions for characteristic functions of It\^o semimartingales. The generalization of this expansion result to the case with rough volatility and jump intensity is provided by \cite{CT22}. Finally, there is a lot of earlier work on asymptotic expansions of conditional expectations of various transforms of It\^o semimartingale increments over short time intervals and/or of their conditional distribution, see e.g., \cite{ruschendorf2002expansion}, \cite{figueroa2009small}, \cite{MN11}, \cite{figueroa2012small}, \cite{bentata2012short}, among others. 
	
The main difference of the current paper from all of the above-cited work is that our interest is in deriving an asymptotic expansion result for an \emph{increment}   of the conditional characteristic function. Since $\Delta$ in \eqref{eq:exp} is typically much smaller than $T$ (e.g., several minutes versus several days), just characterizing the first few leading terms of $\mathbb{E}_t[e^{iu(x_{t+T}-x_t)/\sqrt{T}}]$ as $T$ shrinks is not sufficient in general for deriving the asymptotic behavior of \eqref{eq:exp}. Instead, in order to  derive a general expansion result for \eqref{eq:exp} in terms of $\Delta$ and $T$, with only mild restrictions on the relative speed with which $\Delta$ and $T$ shrink to zero, we first need to expand $\mathbb{E}_t[e^{iu(x_{t+T}-x_t)/\sqrt{T}}]$ up to a residual error term of asymptotic order $O_p(T^k)$ (where $k$ is connected to our assumptions on $x$ and, in general, can be arbitrarily large) and then further analyze its behavior as $t$ varies. To the best of our knowledge, such double asymptotics have not been investigated thus far.

The asymptotic expansion results of the paper are derived under a so-called deep It\^o semimartingale assumption for the process $x$. Deep It\^o semimartingales are It\^o semimartingales for which the spot semimartingale characteristics obey It\^o semimartingale dynamics, the spot semimartingale characteristics of the spot semimartingale characteristics obey It\^o semimartingale dynamics, etc., up to order $k$, for some integer $k$. The class of deep It\^o semimartingales is rather large. In particular, it covers the class of L\'{e}vy-driven SDE-s and time-changed L\'{e}vy processes, commonly used in applied work. Importantly, our results are derived for general specifications of the jump part of the process with arbitrary dependence between price and diffusive volatility jumps and without restricting the level of jump activity. We do rule out, however, processes with rough volatility paths. Asymptotic expansions of $\mathbb{E}_t[e^{iu(x_{t+T}-x_t)/\sqrt{T}}]$ derived in \cite{CT22} for such processes are quite different from that for an It\^o semimartingale, and the same will carry over to asymptotic expansions for increments of the conditional characteristic functions. We thus leave extensions of our results to a  rough setting to future work.
	
The theoretical expansions of this paper are of direct use for the analysis of volatility dynamics using high-frequency option data. In fact, they can be used to infer from the option data any feature of the volatility dynamics that is identifiable if one had access to observations of the true diffusive volatility over a fixed time span. Examples of such quantities include diffusive volatility of volatility as well as covariation between price and diffusive volatility (also known as leverage effect). These quantities can be estimated from high-frequency data of the underlying price process alone, see e.g.,  \cite{vetter2015estimation}, \cite{sanfelici2015high}, \cite{clinet2021estimation}, \cite{li2022volatility} and \cite{toscano2022volatility} for the volatility of volatility and \cite{WM14}, \cite{AFLWY17}, \cite{KX17} and \cite{C19} for the leverage effect. The rates of convergence of estimators constructed from the underlying price data are known, however, to be slow which is a natural consequence of the fact that volatility is not directly observable from the underlying price. By taking advantage of the option data, we can effectively treat volatility as directly observable. This makes estimation of such quantities from option data easier and significantly more precise. In  \cite{CT23_b}, we develop and implement on real data such volatility of volatility and leverage effect estimators using the asymptotic expansions derived in the current paper.        
	
As another application of our expansion results, for which in contrast to the previous ones there is no return-based equivalent known in the literature, we use high-frequency option data to construct a test for whether volatility has infinite variation jumps or not. A lot of parametric volatility models, considered in earlier work, differ in their ability to accommodate infinite variation jumps. For example, the popular exponentially affine class of asset pricing models of \cite{DPS00} and \cite{duffie2003affine} can have volatility jumps of finite variation only. The proposed test is nonparametric in nature and is based on the empirical characteristic function of high-frequency increments of option-based volatility estimates. In the absence of infinite variation volatility jumps, this empirical characteristic function is asymptotically determined by the diffusive volatility of volatility and potential option observation errors. Both of them have distinct scaling properties and can be eliminated effectively by taking (second-order) differences. The result is an asymptotically mixed normal random variable under the null hypothesis of no infinite variation volatility jumps, allowing for size control. If infinite variation volatility jumps are present, these  scale differently than noise and the Gaussian part of volatility and therefore survive the above-described differencing procedure. As a result, under the alternative hypothesis that volatility has  infinite variation  jumps, our test will have an asymptotic power of $1$.

	%Since $T\downarrow 0$, the leading component of $\mathbb{E}_t[e^{iu(x_{t+T}-x_t)/\sqrt{T}}]$ is the conditional characteristic function of a process whose spot semimartingale characteristics are those of $x$ ``frozen'' at their value at time $t$. Higher-order terms in the asymptotic expansion of $\mathbb{E}_t[e^{iu(x_{t+T}-x_t)/\sqrt{T}}]$ are determined by the dynamics of the spot semimartingale characteristics of $x$. 

	The rest of the paper is organized as follows. In Section~\ref{sec:model}, we introduce the model and assumptions needed for our main theorem, by first giving a precise definition of deep It\^o semimartingales in Section~\ref{sec:deep} and then stating and discussing our hypotheses for the dynamics of $x$ and its coefficients in Section~\ref{sec:ass}. In Section~\ref{sec:expansion-2}, we present our main expansion result, Theorem~\ref{thm:incr-L}, where we consider the conditional characteristic function $\E_t[e^{iu(x_{t+T}-x_t)/\sqrt{T}}]$ of a normalized change of an It\^o semimartingale over a short interval of length $T$ and develop an asymptotic expansion of its increments as  $t$ varies. We then use this result in Section~\ref{sec:app} to obtain in Theorem~\ref{thm:incr} a high-frequency expansion of increments of the volatility estimators introduced by \cite{T19}. Important special cases and examples are considered in Corollaries~\ref{cor:incr} and \ref{cor:finact}. These theoretical expansion results are used in Section~\ref{sec:test} to construct a test for whether volatility jumps are of finite or infinite variation. This test is implemented on simulations and on real data in Section~\ref{sec:num}. The proofs are given in Sections~\ref{sec:proofs} and \ref{sec:proof51}, with some technical proofs deferred to the appendix.
	
	\section{Model and Assumptions}\label{sec:model}
	
	On a filtered probability space $(\Om,\calf,\F=(\calf_t)_{t\geq0},\P)$, consider an It\^o semimartingale process of the form
	\begin{equation}\label{eq:x0}  
		x_t=x_0+\int_0^t\al_s ds +\int_0^t \si_s dW_s + \text{jumps}, %+\int_0^t\int_\R \ga(s,z) (\mu-\nu)(ds,dz) + \int_0^t\int_\R \Ga(s,z)\mu(ds,dz),
	\end{equation}
	where $W$ is a Brownian motion, $\al$  and $\si$ are the drift and  volatility of $x$, respectively, and the jumps of $x$ will be specified later. In financial applications, $x$  typically models
	the logarithmic  price of an asset.  % while  $\P$ denotes the true statistical probability measure. 
	% and $\mu$ is a Poisson random measure with compensator $\nu(dt,dz) = dt \otimes \la(dz)$ (and $\la$ is a L\'evy measure). The processes $\al$ and $\si$ denote the drift and volatility of the asset price, respectively, while $\ga$ and $\Ga$ model the jumps of $x$. In principle, it suffices to assume that $\ga(t,z)=\delta(t,z)\bone_{\{\lvert \delta(t,z)\rvert\leq 1\}}$ and $\Ga(t,z)=\delta(t,z)\bone_{\{\lvert \delta(t,z)\rvert> 1\}}$ for a third process $\delta$, but we  adopt the form \eqref{eq:x} in order to simplify the statement of our assumptions below. 
	One motivation for this paper is the estimation of the dynamical properties of volatility from high-frequency option data. Under the assumption that
	\begin{equation}\label{eq:si0} \begin{split}
			\si_t&=\si_0+\int_0^t\al^\si_s ds + \int_0^t \si^\si_s dW_s + \int_0^t \ov\si^\si_s d\ov W_s +\text{jumps},%+\int_0^t\int_\R \ga^\si(s,z) (\mu-\nu)(ds,dz) \\
			%&\quad+ \int_0^t\int_\R \Ga^\si(s,z)\mu(ds,dz),
		\end{split} 
	\end{equation}
	where $\al^\si$, $\si^\si$ and $\ov\si^\si$ %and $\ga^\si$ and $\Ga^\si$ 
	are the drift and diffusive % and jump 
	coefficients of $\si$, respectively, and $\ov W$ is a Brownian motion independent of $W$,
	the expansions developed in this paper combined with high-frequency option data essentially allows us to use the results of \cite{AJ14}  to  
	develop feasible inference for volatility dynamics, \emph{as if volatility was directly observable at high frequency}. Examples of applications include, but are not limited to, the estimation of  volatility of volatility  and the leverage effect (which we demonstrate in our companion paper \cite{CT23_b}), testing for volatility jumps or testing for rough volatility.

	\subsection{Deep It\^o Semimartingales}\label{sec:deep}
	
	In order to formulate our main expansion result, we need the price process to have additional structural properties. % under $\Q$. %To begin with, we suppose that the big jumps in $x$ (under $\Q$) satisfy $\int_0^t \int_\R \lvert \Ga(s,z)\rvert ds\la(dz)<\infty$ a.s.\ for all $t$, in which case we can absorb $\Ga$ into $\ga$ upon changing the drift $\al$ and assume
	%\begin{equation}\label{eq:Ga} 
	%	\Ga(t,z)\equiv 0
	%\end{equation}
	%in \eqref{eq:x}. In addition, 
	More precisely, we assume that $x$ is what we call a \emph{deep It\^o semimartingale with $N$ hidden layers}, where $N$ is an integer. Informally speaking, this means that $x$ is an It\^o semimartingale  whose drift $\al$, volatility $\si$, and jump coefficients are again It\^o semimartingales  (layer 1), which in turn have coefficients that are  It\^o semimartingales (layer 2) etc., repeated in this way until the $N$th layer. For example, it is a common hypothesis in high-frequency financial econometrics that one has \eqref{eq:x0} and \eqref{eq:si0}. If, in addition, $\al$ and the jump coefficients of $x$ (which will be introduced more rigorously below) are also It\^o semimartingales, then $x$ is a deep It\^o semimartingale with one hidden layer. A  second hidden layer, which includes the It\^o semimartingale assumption for the coefficients of $\si$, has also been partially assumed in previous work, for example, in the estimation of  volatility of volatility based on return data (see \cite{li2022volatility} and \cite{vetter2015estimation}) or of the leverage effect both from return and option data (see \cite{AFLWY17}, \cite{WM14} and  \cite{T21}). Let us note that most asset pricing models satisfy the deep  It\^o semimartingale assumption, in particular, if $\si_t$ (or $\si^2_t$ or $\log \si_t$) is modeled as the solution to a stochastic differential equation satisfying mild regularity conditions. A notable exception, however, are rough volatility models \citep{gatheral2018volatility}. To extend the results of this paper   to  a rough setting, the deep semimartingale assumption should be replaced by one where each coefficient in each layer follows a fractional-type process.  %If the jump part of $x$ and of any of its coefficients in any layer can be expressed as a compensated sum of jumps, we say that $x$ is a \emph{special} deep It\^o semimartingale.
	
	%Equipped with this informal definition of deep It\^o semimartingales, the reader can safely jump to Section~\ref{sec:expansion-2}  and continue reading from there  until the end of Section~\ref{sec:emp} (taking Assumption~\ref{ass:main} for granted), before coming back to the formal definition, which we now give. To this end,
	In order to state the formal definition of deep semimartingales, we
	fix two integers $d,d'\geq2$ (which may depend on $N$), a $d$-dimensional standard $\F$-Brownian motion $\mathbb{W}=(W^{(1)},\dots,W^{(d)})^\top$    such that $W^{(1)}=W$ and $W^{(2)}=\ov W$, and an $\F$-Poisson random measure $\pf$ on $[0,\infty)\times\R^{d'}\times\R$ (independent of $\mathbb{W}$) with intensity measure $\qf(dt,dz,dv)=dt F(dz) dv$ for some $d'$-dimensional measure $F$. By choosing $d$ and $d'$ sufficiently large (e.g., as the total number of processes appearing in all layers of $x$), the model allows for arbitrary dependence between the diffusive parts and between the jump parts. Starting from $\pf$, we construct an integer-valued random measure $\mu$ by setting 
	\begin{equation}\label{eq:mu} 
		\mu(dt,dz) =\int_{\R} \bone_{\{0\leq v\leq \la(t-,z)\}} \pf(dt,dz,dv),
	\end{equation}
	which has compensator $\nu(dt,dz)=\nu_{t-}(dz) dt =\la(t-,z) F(dz)dt$. In \eqref{eq:mu}, $\la(t,z)$ is a nonnegative intensity process which we assume is an It\^o semimartingale (for fixed $z$) of the form
	\begin{equation}\label{eq:la} \begin{split}
			\la(t,z)&=\la(0,z)+\int_0^t\al^\la(s,z)ds+\sum_{i=1}^d 	\int_0^t\si^{\la,(i)}(s,z) dW^{(i)}_s \\
			&\quad+ \iiint_0^t  \ga^\la(s,z,z',v') (\pf-\qf)(ds,dz',dv')%+ \iiint_0^t \Ga^\la(s,z,z',v') \pf(ds,dz',dv')
			,\end{split}
	\end{equation}
	where we used the notation $\iiint_s^t = \int_s^t\int_{\R^{d'}}\int_{\R}$. In what follows, we also use the notation $\iint_s^t = \int_s^t \int_{\R^{d'}}$.
	%In general, $x$ and each of its coefficients can have two jump components, a small jump part and a large jump part. Therefore, $x$ and any of its coefficients will themselves have $d+3$ coefficients (a drift part, $d$ diffusive parts and a small and a large jump part). In total, we therefore need to model the dynamics of $d'\leq\sum_{i=0}^{N} (d+3)^i$ processes simultaneously (the actual number is smaller than the bound because $x$, for instance, only has one diffusive component according to \eqref{eq:x0}). To achieve this, we consider an integer-valued random measure $\mu$ on $[0,\infty)\times \R^{d'}$ with $\F$-compensator $\nu(dt,dz)=\la(t,z) dt F(dz)$, where $F$ is a Lévy measure on $\R^{d'}$ and $\la$ is a predictable mapping. %We assume that $\un z = (z_\bj)_{\bj\in\mathcal{J}}$, where  $\mathcal{J}  =\bigcup_{i=0}^{N} \{1,\dots,d+3\}^i$ and by convention $\{1,\dots,d+3\}^0=\{\emptyset\}$. The idea is that $z_\emptyset$ denotes the jumps of $x$, $z_\bj$ for $\bj \in \{1,\dots,d+3\}$ denotes the jumps of the drift, the diffusive and the small jump and large jump coefficient of $x$, respectively. Similarly, $z_\bj$ for $\bj\in\{(1,j_2):j_2\in\{1,\dots,d+3\}\}$ respectively stands for the jumps of the drift, the diffusive and the small jump and large jump coefficient of $\al$, etc. 
	Having introduced $\mu$ and writing   $\wh\mu=\mu-\nu$, we specify the jumps of $x$ and $\si$ from \eqref{eq:x0} and \eqref{eq:si0} as follows:
	\begin{equation}\label{eq:x}   
		x_t=x_0+\int_0^t\al_s ds +\int_0^t \si_s dW_s +\iint_0^t \ga(s,z) \wh\mu(ds,dz) + \iint_0^t \Ga(s,z)\mu(ds,dz),
	\end{equation}
	and
	\begin{equation}\label{eq:si} \begin{split}
			\si_t&=\si_0+\int_0^t\al^\si_s ds + \int_0^t \si^\si_s dW_s + \int_0^t \ov\si^\si_s d\ov W_s +\iint_0^t  \ga^\si(s,z) \wh\mu(ds,dz) \\
			&\quad+ \iint_0^t \Ga^\si(s,z)\mu(ds,dz),
		\end{split} 
	\end{equation}
	where $\ga$ and $\ga^\si$ stand for the small jumps and $\Ga$ and $\Ga^\si$ for the big jumps of $x$ and $\si$, respectively.
	
	Next, we specify the $N$ different layers of $x$, and to do so in a concise way, we use matrix notation and define
	\begin{align*} 
		\vartheta(s,z)&= (\si_s,0,\dots,0,\al_s,\ga(s,z),\Ga(s,z)) \in \R^{1\times (d+3)},	 \\
		Y(ds,dz)&=(dW^{(1)}_s\delta_0(dz),\dots,dW^{(d)}_s \delta_0(dz), ds \delta_0(dz),\wh \mu(ds,dz),\mu(ds,dz))^\top, 
	\end{align*}
	where $\delta_0$ is the Dirac measure at $0$. Note that $Y$ is a measure with $d+3$ components and we can   compactly write
	\begin{equation}\label{eq:xshort} 
		x_t=x_0 +\iint_0^t \vartheta(s,z) Y(ds,dz). 
	\end{equation} 
	For every $i=1,\dots,N$, we now  define processes $\{\vartheta(t,z_1,\dots,z_i):t\geq0, z_1,\dots,z_i\in\R^{d'}\}$ with values in $\R^{1\times (d+3)\times\cdots\times (d+3)}=\R^{1\times (d+3)^{\times i}}$ recursively by setting
	\begin{equation}\label{eq:beta} 
		\vartheta(t,z_1,\dots,z_i) = \vartheta(0,z_1,\dots,z_i)+\iint_0^t \vartheta(s,z_1,\dots,z_i,z_{i+1}) Y(ds,dz_{i+1}),
	\end{equation}
	where for $A\in \R^{k_1\times\dots\times k_i}$ and $v\in\R^{k_i}$, the product $Av\in\R^{k_1\times\dots\times k_{i-1}}$ is given by
	\[ (Av)_{j_1,\dots,j_{i-1}} = \sum_{j_i=1}^{k_i} A_{j_1,\dots,j_i}v_{j_i},\quad j_\ell=1,\dots,k_\ell,\quad \ell=1,\dots,i-1. \]
	Using the notations 
	$$ \bs_i= ( s_1,\dots,s_i),\quad\bz_i= (  z_1,\dots,z_i),
	\quad \bY(d\bs_i,d\bz_i)=Y(ds_i,dz_i)\cdots Y(ds_1,dz_1),$$
	we  say that $x$ is a \emph{deep It\^o semimartingale with $N$ hidden layers} if there are $\calf_0\otimes \text{Borel}$-measurable random variables $\vartheta(0,\bz_i)\in\R^{1\times (d+3)^{\times i}}$ (for $i=1,\dots,N$) and $\la(0,z)$, a predictable $\R^{1\times (d+3)^{\times (N+1)}}$-valued process $(t,\bz_{N+1})\mapsto \vartheta(t,\bz_{N+1})$ and predictable real-valued processes $(t,z)\mapsto\al^\la(t,z)$, $(t,z)\mapsto\si^{\la,(i)}(t,z)$ and $(t,z,z',v')\mapsto\ga^\la(t,z,z',v')$ %and $(t,z,z',v')\mapsto\Ga^\la(t,z,z',v')$ 
	such that we have \eqref{eq:xshort} and \eqref{eq:beta} as well as \eqref{eq:mu} and \eqref{eq:la} (which includes the requirement that the stochastic integrals in \eqref{eq:beta} and \eqref{eq:la} be well defined and $\la$ from \eqref{eq:la} be nonnegative). %and a sequence $(T_n)_{n\in\N}$ of stopping times with $T_n \uparrow \infty$ and  deterministic nonnegative measurable functions $J_n(z)$ such that $\int_\R (J_n(z))^2\la(dz)<\infty$ for each $n$ and

	In any dimension of the tensor/array $\vartheta(t,\bz_i)$, the first $d$ entries correspond to the diffusive coefficients while the entries $d+1$, $d+2$ and $d+3$ correspond to the drift,   the small jumps and the big jumps, respectively. For example, if $d=2$ (i.e., we only have two independent Brownian motions, $W$ and $\ov W$), then $\vartheta(t,\bz_5)_{12345}$ is the big jump part (evaluated at $z_5$) of the small jump part (evaluated at $z_4$) of the drift of $\ov\si^\si$ (which is the second diffusive coefficient of $\si$, which in turn is the first diffusive coefficient of $x$). If $\bj_i=(j_1,\dots,j_i)$ and $j_\ell\in\{1,\dots,d+1\}$ for some $\ell\in\{1,\dots,i\}$, then only the value of $\vartheta(t,\bz_i)_{\bj_i}$ at $z_\ell=0$ matters, because  the spatial variable $z_\ell$ only matters for the jump parts but not for drift or volatility. For instance, in the above example,  there is no loss of generality when assuming $\vartheta(t,\bz_5)_{12345}=\vartheta(t,0,0,0,z_4,z_5)_{12345}$. 
	
	As a second example, assume that $x$ and $\si$ are given by \eqref{eq:x} and \eqref{eq:si}, respectively. Then 
	\begin{equation}\label{eq:x-1} 
		\begin{aligned}
			\vartheta(t,z)_1&=\si_t, & \vartheta(t,z)_2&=\cdots=\vartheta(t,z)_{d}=0, & \vartheta(t,z)_{d+1}&=\al_t,\\
			\vartheta(t,z)_{d+2}&=\ga(t,z),&  \vartheta(t,z)_{d+3}&=\Ga(t,z),\quad &&
		\end{aligned}
	\end{equation}
	and
	\begin{equation}\label{eq:si-1} 
		\begin{aligned}
			\vartheta(t,z_1,z_2)_{21} &= \si^\si_t,~ \vartheta(t,z_1,z_2)_{22}=\ov \si^\si_t,~
			\vartheta(t,z_1,z_2)_{23}=\cdots=\vartheta(t,z_1,z_2)_{2,d}=0,\\
			\vartheta(t,z_1,z_2)_{2,d+1}	&=\al^\si_t,~ \vartheta(t,z_1,z_2)_{2,d+2}=\ga^\si(t,z_2),~ \vartheta(t,z_1,z_2)_{2,d+3}=\Ga^\si(t,z_2).
		\end{aligned} 
	\end{equation}
	%The reason why the first subscript in the last line is always $2$ is because we specify the characteristics of $\si$ (the second component of $\vartheta(t,z_1)$) here. Also note that as part of our definition, if $x$ is a deep It\^o semimartingale, then $\Ga(t,z)\equiv\Ga^\si(t,z)=0$. So to be very precise,  $x$ should be referred to as a deep \emph{special} It\^o semimartingale.
	Let us introduce two additional abbreviations to be used later:
	\begin{equation}\label{eq:notation} \begin{aligned}
			\si^\ga(t,z)&=\vartheta(t,z,0)_{d+2,1} &&\text{(the diffusive part of $\ga$ with respect to $W$)},\\
			\ga^\ga(t,z_1,z_2)&=\vartheta(t,z_1,z_2)_{d+2,d+2}&&\text{(the small jump part of $\ga$)}.
		\end{aligned}
	\end{equation}
	
	Finally, if $x$ is a deep It\^o semimartingale with $N$ hidden layers, we say that $x$ is \emph{special} if for any $i=1,\dots,N+1$, we have $\vartheta(t,\bz_i)_{j_1,\dots,j_i}=0$ as soon as $j_\ell=d+3$ for some $\ell\in\{1,\dots,i\}$. In other words, $x$ is special if its coefficients in all layers can be represented without a   $\Ga$-component. In this case, if we let
	\begin{equation}\label{eq:notation-2} 
		\begin{split}
			\theta(t,\bz_i)&= (\vartheta(t,\bz_i)_{j_1,\dots,j_i})_{j_1,\dots,j_i=1}^{d+2}\in \R^{1\times (d+2)^{\times i}},\qquad i=1,\dots,N+1,	 \\
			y(ds,dz)&=(dW^{(1)}_s\delta_0(dz),\dots,dW^{(d)}_s \delta_0(dz),ds \delta_0(dz), \wh \mu(ds,dz))^\top, \\
			\by(d\bs_i,d\bz_i)&=y(ds_i,dz_i)\cdots y(ds_1,dz_1),
		\end{split}
	\end{equation}
	we have
	\begin{equation}\label{eq:xshort-2} 
		x_t=x_0 + \iint_0^t \theta(s,z)y(ds,dz),\quad \theta(t,\bz_i) = \vartheta(0,\bz_i)+\iint_0^t \theta(s,\bz_{i+1}) y(ds,dz_{i+1}),\!\!
	\end{equation}
	for all $i=1,\dots,N$.
	
	\begin{remark}\label{rem:prm}
		The reader may wonder why we did not simply use $\pf$ (or even a Poisson random measure on $(0,\infty)\times \R$) to represent $x$ and its ingredients. It is true that every It\^o semimartingale can be represented as in \eqref{eq:x} with a Poisson random measure instead of $\mu$. However, in this case, the corresponding jump coefficients $\ga(t,z)$ and $\Ga(t,z)$ also have to be adjusted and may no longer be It\^o semimartingales in $t$ for fixed $z$. For example, if we used $\pf$ instead of $\mu$, then the new small jump coefficient would be $\ga(t,z,v)=\ga(t,z)\bone_{\{v\leq \la(t,z)\}}$. But it is evident that for fixed $z$ and $v$ the process $t\mapsto \ga(t,z,v)$ is typically not a semimartingale (as $\ga$ may jump infinitely often between $0$ and $\ga(t,z)$ on a bounded interval).
	\end{remark}
	
	\subsection{Assumptions for Main Theorem}\label{sec:ass}
	We are now in the position to state the  structural assumptions we need for our main expansion theorem. We let $\calg_t$ denote the $\si$-field generated by the increments of $\mathbb{W}$ and $\pf$ \emph{after} time $t$.  We also write 			let $x_+ = x\vee 0$.
	
	%\settheoremtag{A$_0$}
	\begin{Assumption}\label{ass:main} %Under the risk-neutral probability measure $\Q$, 
		The logarithmic price process $x$ is a special deep It\^o semimartingale with $N\geq3$ hidden layers given by \eqref{eq:xshort-2} such that there exist a localizing sequence $(T_n)_{n\in\N}$ of stopping times, an exponent $r\in[1,2)$, deterministic nonnegative measurable functions $J_n(z)$, $j_n(z,v)$ and $\mathcal{J}_n(\bz_{N+1})$, and for all $0<t<t'<t+1$, an $\F$-predictable process $(s,z)\mapsto \la_{t,t'}(s,z)$ defined for $s\in [t,t+1]$ and $z\in\R^{d'}$ with the following properties:
		\begin{enumerate}
			\item The functions $J_n(z)$ and $j_n(z,v)$ are real-valued and   $\int_{\R^{d'}\times\R} j_n(z,v) F(dz)dv<\infty$ for each $n\in\N$. Moreover, for all $t<T_n$, $t'\in[t,t+1]$, $z,z'\in\R^{d'}$ and $v'\in\R$, 
			\begin{equation}\label{eq:bound} \begin{split}
			&	\lvert\la(t,z)\rvert+	 \lvert \al^\la(t,z)\rvert+\sum_{i=1}^{d} \lvert \si^{\la,(i)}(t,z)\rvert \leq J_n(z),\\
			&	 \lvert \ga^\la(t,z,z',v')\rvert^2%+\lvert \Ga^\la(t,z,z',v')\rvert 
				\leq J_n(z)^2j_n(z',v')\end{split}
			\end{equation}
			and
			\begin{equation}\label{eq:cont}\begin{split}
					%\E[(\la(t',z)-\la(t,z))^2\wedge 1]^{1/2}&\leq  \lvert t'-t\rvert^{1/2}J_n(z),\\
					\sum_{i=1}^d   \E[(\si^{\la,(i)}(t',z)-\si^{\la,(i)}(t,z))^2\wedge1]^{1/2}& \leq \lvert t'-t\rvert^{1/2}J_n(z),	 \\
					\E[(\ga^\la(t',z,z',v')-\ga^\la(t,z,z',v'))^2\wedge1]^{1/2}	&\leq \lvert t'-t\rvert^{1/2}J_n(z)	j_n(z',v')^{1/2}.
			\end{split}\end{equation}
			%		for some strictly positive locally bounded process $a>0$ and $\R^{d'}$-dimensional L\'evy measure $\ov F$, we have
			%		\begin{equation}\label{eq:A-bound} 
			%			\la(\om;t,z)F(dz)\leq a_t(\om) \ov F(dz)
			%		\end{equation}
			%	for all $t>0$, $\om\in\Om$ and Borel sets $B$. Moreover,
			\item The function $\mathcal{J}_n(\bz_{N+1})$ takes values in $\R^{(d+2)^{N+1}}$ and for all $n\in\N$, $j_1,\dots,j_{N+1}\in\{1,\dots,d+2\}$, $z_1,\dots,z_{N+1} \in\R^{d'}$ and $t<T_n$, we have that
			\begin{equation}\label{eq:theta-bound} 
				\lvert\theta(t,\bz_{N+1})_{j_1,\dots,j_{N+1}}\rvert^r \vee 	\lvert\theta(t,\bz_{N+1})_{j_1,\dots,j_{N+1}}\rvert^2 \leq \mathcal{J}_n(\bz_{N+1})_{j_1,\dots,j_{N+1}},
			\end{equation}
			\begin{equation}\label{eq:beta-L2} %\begin{dcases} \int_{(\R^{d'})^i} \theta(0,\bz_i)_{j_1,\dots,j_i}^2 \prod_{\ell: j_\ell \neq d+2} \delta_0(dz_{\ell})\prod_{\ell: j_\ell=d+2} J_n(z_\ell) F(dz_{\ell}) < \infty, \\ 
				%	\sup_{s\in [0,t]}  
				\int_{(\R^{d'})^{N+1}} 
				%\E[\theta(s,\bz_{N+1})_{j_1,\dots,j_{N+1}}^2] 
				\mathcal{J}_n(\bz_{N+1})_{j_1,\dots,j_{N+1}}
				\prod_{\ell: j_\ell \neq d+2} \delta_0(dz_{\ell})\prod_{\ell: j_\ell=d+2} J_n(z_\ell)  F(dz_{\ell}) < \infty. %\end{dcases}
			\end{equation}
			\item For any $0<t<t',s<t+1$ and $z\in\R^{d'}$, the random variable $\la_{t,t'}(s,z)$ is $\calf_t\vee \calg_{t'}$-measurable. Moreover, if $s<T_n$, then
			\begin{equation}\label{eq:la-smooth} 
				\E[(\la(s,z)-\la_{t,t'}(s,z))^2\wedge1]^{1/2}\leq (\lvert t'-t\rvert^{1/2}+ [( s-t')_+]^2) J_n(z).
			\end{equation}
		\end{enumerate} 
		%Still under $\Q$, w
		We also assume  \eqref{eq:x-1} and \eqref{eq:si-1} without loss of generality, so that $x$ and $\si$ are given by \eqref{eq:x} and \eqref{eq:si} (with $\Ga\equiv \Ga^\si\equiv0$), respectively.
	\end{Assumption}

	Note that under Assumption \ref{ass:main}, the integrals in \eqref{eq:xshort-2} are all well defined and the process $x$ and its coefficients in each layer are special It\^o semimartingales. % under $\Q$. This automatically implies that these coefficients are also It\^o semimartingales under $\P$, but not necessarily  special ones. In particular, under the true probability measure $\P$, the jumps of $x$ may not have a finite first moment and we may have $\Ga(t,z)\neq0$ or $\Ga^\si(t,z)\neq0$.
	Part 1 and 2 of Assumption~\ref{ass:main} are mild and impose some classical  integrability and regularity conditions on the last layer of the deep It\^o semimartingale $x$ as well as the coefficients of the intensity process $\la$. Part 3 of Assumption~\ref{ass:main} is also mild and is satisfied, for example, if the intensity process $\la(t,z)$ is also a deep It\^o semimartingale with at least three hidden layers. For example, assume that $\la(s,z)=\la_s z$, where $
		\la_s=\la_0+\int_0^s \si^\la_r dW_r$,	$\si^\la_r=\si^\la_0+\int_0^r \si(\si^\la)_u dW_u$,
		$\si(\si^\la)_u=\si(\si^\la)_0+\int_0^u \si(\si(\si^\la))_v dW_v$ and $
		 \si(\si(\si^\la))_v=\si(\si(\si^\la))_0+\int_0^v \si(\si(\si(\si^\la)))_w dW_w
	$
	for some locally bounded $\si(\si(\si(\si^\la)))$.
	Then one can choose $\la_{t,t'}(s,z)$ as
	\begin{align*}
		&\biggl(\la_t+\si_t^\la(W_s-W_{t'})+\si(\si^\la)_t\int_{t'}^s\int_{t'}^r dW_udW_r +\si(\si(\si^\la))_t\int_{t'}^s\int_{t'}^r\int_{t'}^u dW_v dW_udW_r\\
		&\quad+\si(\si(\si(\si^\la)))_t\int_{t'}^s\int_{t'}^r\int_{t'}^u \int_{t'}^v dW_wdW_v dW_udW_r\biggr)z,
	\end{align*} 
	which yields $\la(s,z)-\la_{t',t'}(s,z)=O((s-t')^2)$ and $\la_{t',t'}(s,z)-\la_{t,t'}(s,z)=O(\sqrt{t'-t})$ and hence \eqref{eq:la-smooth}.
	Clearly, this example easily extends to   cases where $\la$ has jumps and/or where $\la$ does not have a product form.
		We also consider a strengthening of Assumption~\ref{ass:main}, under which price jumps are  
	summable.
	%\addtocounter{Assumption}{-1}
	
	\begin{Assumption}\label{ass:main-1} In addition to Assumption~\ref{ass:main} (and with the notation from there), we have the following %under $\Q$
		for all $n\in\N$, $t<T_n$, $t'\in[t,t+1]$ and $z\in\R^{d'}$:
		\begin{equation}\label{eq:extra} 
			\lvert \ga(t,z)\rvert\leq J_n(z),\qquad
			\E[(\ga(t',z)-\ga(t,z))^2\wedge1]^{1/2}   \leq  \lvert t'-t\rvert^{1/2} J_n(z).
		\end{equation} 
	\end{Assumption}

	\section{The Main Expansion Theorem}\label{sec:expansion-2}
	
	Our main result is an expansion theorem for the $\calf_t$-conditional characteristic function %(under $\Q$) 
	of the normalized price change from $t$ to $t+T$, that is, 
	\begin{equation}\label{eq:call} 
		\call_{t,T}(u)= \E_t[e^{iu(x_{t+T}-x_t)/\sqrt{T}}] = \E_t[e^{iu_T(x_{t+T}-x_t)}],
	\end{equation} 
	where $u_T=u/\sqrt{T}$ and $\E_t=\E[\cdot \mid \calf_t]$ as before. In what follows, we   use $o=o_p$ and $O=O_p$ to indicate order in probability. A superscript ``uc'' as in $O^\uc$ or $o^\uc$ indicates uniformity in $u\in\calu$, where $\calu$ is an arbitrary compact subset of $(0,\infty)$. Moreover, with a slight abuse of notation, we let
	\begin{equation}\label{eq:tT} 
		t^n_i=t-i\Den,\quad T^n_i = T+i\Den,\quad T^{\prime n}_i = T'+i\Den,\quad i=1,2,\dots,
	\end{equation} 
	for some $\Den\to0$. Increments of a process $F_t$ are denoted by $\Delta^n_i F_t = F_{t^n_{i-1}}-F_{t^n_i}$.
	\begin{theorem}\label{thm:incr-L}
		For $s,t,T>0$ and $u\in\R\setminus\{0\}$, define  
		\begin{equation}\label{eq:La} \begin{split}
				\La_{t,T}(s,u)&=e^{iu\al_t T-\frac12 u^2\si_t^2 T+T\vp_t(s,u) -\frac12iu^3\si_t^2\si^\si_t T^2 + T^2 ( \psi_t(s,u)+\ov\psi_t(s,u)+\wt\psi_t(s,u) )},
			\end{split}
		\end{equation}
		where
		\begin{equation}\label{eq:Theta123}\begin{split}
				\vp_t(s,u)&=\int_{\R^{d'}} (e^{iu\ga(t,z)}-1-iu\ga(t,z))\nu_s(dz),\\
				%\ov\vp_t(s,u)&=\frac12 \int_{\R^{d'}} (e^{iu\ga(t,z)}-1-iu\ga(t,z))\biggl(iu\si_s\si^{\la,(1)}(s,z) \\
				%&\quad+ \int_{\R^{d'}\times\R}  (e^{iu\ga(s,z')}-1)\ga^\la(s,z,z',v')\bone_{\{0\leq v'\leq \la(s,z')\}} F(dz')dv'\biggr)F(dz),\\
				\psi_t(s,u) &= -\frac12 u^2\si_t\int_{\R^{d'}} (e^{iu\ga(t,z)}-1)(\si^\ga(t,z)+\ga^\si(t,z))\nu_s(dz),\\
				\ov\psi_t(s,u)&= \frac12iu\int_{\R^{d'}}\int_{\R^{d'}} (e^{iu\ga(t,z)}-1)(e^{iu\ga(t,z')}-1)\ga^\ga(t,z,z')\nu_s(dz)\nu_s(dz'),\\
				\wt\psi_t(s,u)&= \frac12\int_{\R^{d'}}(e^{iu\ga(t,z)}-1-iu\ga(t,z))\biggl(iu\si_t \si^{\la,(1)}(s,z)\\
				&\quad+\int_{\R^{d'}\times\R} (e^{iu\ga(t,z')}-1) \ga^\la(s,z,z',v')\bone_{\{0\leq v'\leq \la(s,z')\}} F(dz')dv'\biggr)F(dz).
		\end{split}\end{equation}
		Under Assumption~\ref{ass:main}, there is a finite number of It\^o semimartingales $v^{(k)}_t$ and $C^{(k)}_t(u)$, $k=1,\dots, K$,  such that $C^{(k)}_t(u)$ is uniformly bounded in $u$ on compacts and $\lvert \Delta^n_i v^{(k)}_t\rvert+\lvert \Delta^n_i C^{(k)}_t(u) \rvert= O^\uc (\sqrt{\Den})$, uniformly in $i$, and the following holds as $\Den\to0$, $T\to0$ and $\Den/T\to0$: writing 	  $\Delta^{n}_i \call_{t,T}(u)= \call_{t^n_{i-1},T^n_{i-1}}(u)-\call_{t^n_i,T^n_i}(u)$ and $\Delta^n_i \La_{t,T}(s,u)=\La_{t^n_{i-1},T^n_{i-1}}(s,u_{T^n_{i-1}}) - \La_{t^n_{i},T^n_{i}}(s,u_{T^n_{i}})$, we have 	for all $i$ with $i\Den=O(T)$ that
		\begin{equation}\label{eq:La-exp} \begin{split}
			\Delta^n_i \call_{t,T}(u)	&= \Delta^n_i \La_{t,T}(t^n_i,u)+ \sum_{k=1}^K \Delta^n_i v^{(k)}_t C^{(k)}_t(u) T^n_{i-1} \\&\quad+O^\uc(T^{N/2}) + o^\uc( \sqrt{\Den}T + \Den/\sqrt{T}).\end{split}
		\end{equation}
	\end{theorem}
	
	As we can see from \eqref{eq:La} and \eqref{eq:La-exp}, a high-frequency increment $\Delta^n_i\call_{t,T}(u)$ of the conditional characteristic function of the normalized price change until time to maturity is given, to first order, simply by the corresponding increment one would obtain if the process $x$ were a   Lévy process with characteristics given by those of $x$ frozen at $t^n_i$, the beginning of the high-frequency time interval. However, if one intends to use $\Delta^n_i\call_{t,T}(u)$ (or transforms thereof) in estimators typical of high-frequency statistics, then this first-order approximation is often not sufficient for proving central limit theorems. This is why higher-order terms have to be considered in both  \eqref{eq:La} and \eqref{eq:La-exp}. These terms are due to the time variation of drift, volatility, jump coefficient and jump intensity. In particular, we note that $\psi$, $\ov\psi$ and $\wt\psi$ in (\ref{thm:incr-L}) all depend on the time variation in the jump compensator of $x$ and the volatility jumps. The exact asymptotic order of these three components depends on the degree of jump activity (which is left unrestricted). 
	
	Several factors make the proof of \eqref{eq:La-exp}, which we detail in Section~\ref{sec:proofs},  nontrivial. First, when considering an increment from $t^n_i$ to $t^n_{i-1}$, the filtration to be conditioned on changes from $\calf_{t^n_i}$ to $\calf_{t^n_{i-1}}$. Second, due to the double asymptotics $\Delta_n\to0$ and $T\to0$ (typically such that $\Den=o(T)$), we cannot simply use expansion results in $T$ only for $\call_{t,T}(u)$. In particular, as can be seen in \eqref{eq:La-exp}, pure powers of $T$ only show up in the residual term $O^\uc(T^{N/2})$, in contrast to the expansion results of \cite{T21}, for example. This is because for all powers  up to $T^{N/2}$, we explicitly take the high-frequency differencing into account, in order to obtain an additional factor that is at least of order $\sqrt{\Den}$. 
	%Finally, due to the jump-related quantities in \eqref{eq:Theta123}, the expansion \eqref{eq:La-exp} is different from a Taylor-type expansion in $\Delta_n$ and $T$.
	
	To illustrate this more clearly, consider the following expansion result of $\call_{t,T}(u)$ at a fixed time $t$:
		\begin{lemma}\label{lem:exp}
		Under Assumption~\ref{ass:main}, we have 
		\begin{equation}\label{eq:L-exp-2} 
			\call_{t,T} (u) =\Theta_{t,T}(u_T)(1-\eta_{t,T}(u_{T}))+C_{t}(u)T+o^\uc(T)
		\end{equation}
		as $T\to0$, where
			\begin{equation}\label{eq:phi}
			\Theta_{t,T}(u)= e^{iu \al_{t }T  - \frac12u^2 \si_{t}^2T +T  \vp_{t}(u)}%,\quad 	\Theta'_{t,T}(u)= e^{iu \al_{t }T  - \frac12u^2 \si_{t}^2T +T  \vp_{t}(u)-\frac12 iu^3\si^2_t\si^\si_tT^2}
		\end{equation}
	and
	\begin{equation}\label{eq:eta}\begin{split}
			\eta_{t,T}(u)
			&=\frac{1}{2}iu^3 T^2\si_{t}^2\si^\si_{t}- T^2\Bigl(\psi_{t}(t,u)+\ov\psi_t(t,u)+\wt\psi_t(t,u)\Bigr)\\
			&=\frac{1}{2}\biggl(iu^3 T^2\si_{t}^2\si^\si_{t}+ u^2 \si_{t}T^2 \chi_{t}^{(1)}(u)\\
			&\qquad\qquad\qquad\qquad\qquad-  i uT^2\chi_{t}^{(2)}(u)-  iuT^2\si_t\chi^{(3)}_t(u)-  T^2\chi^{(4)}_t(u)\biggr).
	\end{split}\end{equation}
The processes $\chi_t^{(i)}(u)$, $i=1,\dots,4$, are introduced in \eqref{eq:chi} below 
		and $C_t(u)$ is an It\^o semimartingale in $t$ and a polynomial in $u$.
	\end{lemma}

Lemma~\ref{lem:exp}, which we prove in the appendix, is very similar in spirit to \cite{CT22} (specialized to $H=\frac12$) and to previous fixed-$t$ expansions mentioned in the Introduction. If we difference $\call_{t,T}(u)$ at high frequency, we clearly recover the  leading-order terms in Theorem~\ref{thm:incr-L}. The fact that Theorem~\ref{thm:incr-L} does not follow from the lemma is due to the nonexplicit $o^\uc(T)$-terms in \eqref{eq:L-exp-2}, which are too big for the expansions in Theorem~\ref{thm:incr-L} (and for the applications we have in mind). Other expansions in the literature are sometimes explicit up to a slightly higher order (e.g., $T^{3/2}$ in \cite{T21}), but this is still not enough. Therefore, the novelty of Theorem~\ref{thm:incr-L} really rests in showing bounds on the increments of the higher-order terms in $o^\uc(T)$ \emph{without} having access to their explicit form.
	
	\begin{remark}
		Our expansion results in this paper are for conditional characteristic functions of the \emph{normalized} price change $(x_{t+T}-x_t)/\sqrt{T}$. The normalization makes the diffusive component of the price be the dominant component in these expansions. As we show in the next section, such results are useful for studying the diffusive volatility and its dynamics.  The counterpart of Theorem~\ref{thm:incr-L} for the conditional characteristic function of the \emph{raw} price increment $x_{t+T}-x_t$ will be very different. In particular, the jumps in $x$ will play a more important role in such an expansion. 
	\end{remark}
	
	\section{Application to Option-based Volatility Estimators}\label{sec:app}
	
	Building upon the fact that $\call_{t,T}(u)=e^{-\frac12 u^2\si_t^2} + o(1)$ as $T\downarrow 0$ under mild assumptions on $x$ and $\si$, \cite{T19} constructs an option-based estimator of the spot variance $\si_t^2$   by defining
	\begin{equation}\label{eq:vol-est} 
		\si^2_{t,T}(u)=-\frac{2}{u^2} \log \lvert \call_{t,T}(u)\rvert
	\end{equation} 
	and using an option-based estimate of $\call_{t,T}(u)$. In order to remove biases of higher asymptotic order,  \cite{todorov2021bias} introduce a bias-corrected version of \eqref{eq:vol-est} by considering a second time-to-maturity
	\begin{equation}\label{eq:Tprime} 
		T'=\tau T,
	\end{equation}
	for some $\tau>1$ and setting
	\begin{equation}\label{eq:vol-est-20} 
		\si^2_{t,T,T'}(u)=\frac{T'\si^2_{t,T}(u)-T\si^2_{t,T'}(u)}{T'-T}.
	\end{equation}
	The main sources of errors of both $\si^2_{t,T}(u)$ and $\si^2_{t,T,T'}(u)$ are jump risks in price and volatility and the dynamics of their semimartingale characteristics. %By developing higher-order expansions of $\call_{t,T}(u)$, \cite{T19} and \cite{todorov2021bias} quantify these errors in terms of powers of $T$, the tenor of the considered options.
	
	The expansion result in Theorem~\ref{thm:incr-L} can be transformed into an expansion of increments of $\si^2_{t,T}(u)$ and $\si^2_{t,T,T'}(u)$, that is, of
	\begin{equation}\label{eq:si-incr} \begin{split}
			\Delta^n_i \si^2_{t,T}(u)	&=\si^2_{t^n_{i-1},T^n_{i-1}}(u)-\si^2_{t^n_i,T^n_i}(u),\\
			\Delta^n_i \si^2_{t,T,T'}(u)&=\si^2_{t^n_{i-1},T^n_{i-1}, T^{\prime n}_{i-1}}(u)-\si^2_{t^n_{i},T^n_{i}, T^{\prime n}_{i}}(u).
		\end{split}
	\end{equation} 
	The result for $\si^2_{t,T}(u)$ is stated in the following theorem:
	\begin{theorem}\label{thm:incr} 
		Recall the notation introduced in Theorem~\ref{thm:incr-L} and further define,	for $s,t,T>0$ and $u\in\R\setminus\{0\}$,   
		\begin{equation}\label{eq:Theta}  
			\Phi_{t}(s,u)=-\frac{2}{u^2}\Re\bigl( \vp_t(s,u)  \bigr),\quad
			\Psi_{t,T}(s,u)=-\frac{2T}{u^2}\Re\bigl( \psi_t(s,u)+\ov\psi_{t}(s,u)+\wt\psi_{t}(s,u)\bigr).
		\end{equation} 
		Then, under Assumption~\ref{ass:main},  there is a finite number of It\^o semimartingales $v^{(k)}_t$ and $C^{(k)}_t(u)$, $k=1,\dots, K$,  such that $C^{(k)}_t(u)$ is uniformly bounded in $u$ on compacts, $\lvert \Delta^n_i v^{(k)}_t\rvert+\lvert \Delta^n_i C^{(k)}_t(u) \rvert= O^\uc (\sqrt{\Den})$, uniformly in $i$, and the following holds as $\Den\to0$, $T\to0$ and $\Den/T\to0$:  %the processes \linebreak $\{(\Psi_{t,T}(u))_{t\geq0}, (\ov\Psi_{t,T}(u))_{t\geq0} : T\in[0,1], u\in U\}$ form a tight family of It\^o semimartingales. Moreover,
		\begin{equation}\label{eq:incr-si-0} \begin{split}
				\Delta^{n}_i \si^2_{t,T}(u)&=\Delta^n_i \si^2_t +\Delta^n_i \Phi_{t,T}(t^n_i,u) + \Delta^n_i \Psi_{t,T}(t^n_i,u)+ \sum_{k=1}^K \Delta^n_i v^{(k)}_t C^{(k)}_{t^n_i}(u) T^n_{i-1}\\
				&\quad +O^\uc(T^{N/2}) + o^\uc( \sqrt{\Den}T + \Den/\sqrt{T}),
			\end{split}
		\end{equation}
		for all $i$ such that $i\Den=O(T)$,	where  
		\begin{equation}\label{eq:not} \begin{split}
				%\Delta^n_i \si^2_{t,T}(u)	&=\si^2_{t^n_{i-1},T^n_{i-1}}(u)-\si^2_{t^n_i,T^n_i}(u)%=-\frac{2}{u^2}(\log \lvert \call_{t^n_{i-1},T^n_{i-1}}(u)\rvert- \log \lvert \call_{t^n_{i},T^n_{i}}(u)\rvert)	,\\
				\Delta^n_i \Phi_{t,T}(s,u) &=\Phi_{t^n_{i-1}}(s,u_{T^n_{i-1}})-\Phi_{t^n_i}(s,u_{T^n_i}),\\
				\Delta^n_i \Psi_{t,T}(s,u) &=\Psi_{t^n_{i-1},T^n_{i-1}}(s,u_{T^n_{i-1}})-\Psi_{t^n_i,T^n_i}(s,u_{T^n_i}).
			\end{split}
		\end{equation} 
		Under Assumption~\ref{ass:main-1}, we further have
		\begin{equation}\label{eq:incr-si-2} \begin{split}
				\Delta^{n}_i \si^2_{t,T}(u)&=\Delta^n_i \si^2_t + \sum_{k=1}^K \Delta^n_i v^{(k)}_t C^{(k)}_{t^n_i}(u) T^n_{i-1} +O^\uc(T^{N/2}+\sqrt{\Den T}) + o^\uc(  \Den/\sqrt{T}).
			\end{split}
		\end{equation}
	\end{theorem}
	
	Theorem~\ref{thm:incr} shows that   an increment of $\si^2_{t,T}(u)$ is given by the corresponding increment of spot variance plus several bias terms. Two of them, $\Delta^n_i \Phi_{t,T}(t^n_i,u)$ and $\Delta^n_i \Psi_{t,T}(t^n_i,u)$, depend on the intensity and symmetry of jumps in price, volatility and the jump intensity. Under the stronger Assumption~\ref{ass:main-1},  these two terms can be bounded by $O^\uc(\sqrt{\Den T})$. %While this assumption imposes finite-variation jumps in the price process $x$, it is already more general than similar conditions needed in previous work on estimating volatility of volatility. For instance, \cite{vetter2015estimation} assumes that both price and volatility are continuous, while \cite{li2022volatility} allow for finite variation jumps in the price (as we do) but assumes that volatility is continuous. 
	%This being said, it is possible to adjust the subsequent results so that they hold under Assumption~\ref{ass:main} only; see Remark~\ref{rem:infvar} below. 
	
	In addition to $\Delta^n_i \Phi_{t,T}(t^n_i,u)$ and $\Delta^n_i \Psi_{t,T}(t^n_i,u)$, there are also several bias terms of exact order $\sqrt{\Den} T$, which are due to the time variation of the characteristics of $x$. These biases, even though they can be absorbed into $O^\uc(\sqrt{\Den T})$, can be nontrivial in applications, which is why we singled them out. By using the estimator $\si^2_{t,T,T'}(u)$, we shall see in Corollary~\ref{cor:incr} that these bias terms can be effectively canceled out. Also, as Corollary~\ref{cor:finact} below shows, if both price and volatility have jumps of finite variation and price jumps are furthermore like those in a (possibly time-changed) Lévy process, then the bias terms  of exact order $\sqrt{\Den} T$ are the leading ones, so in this case  using $\si^2_{t,T,T'}(u)$ not only reduces bias in finite    samples but also   asymptotically.
	
	%	 provided one is willing to
	%	\begin{itemize}
	%		\item perform a bias correction (to remove $\Delta^n_i \Psi_{t,T}(u)$ from \eqref{eq:Theta}), and
	%		\item assume that the small jumps of $x$ are stable like (as in \cite{jacod2014efficient}) and perform another bias correction (to remove $\Delta^n_i \Phi_{t,T}(u)$ from \eqref{eq:Theta}).
	%	\end{itemize}

	In practice, one is often interested in estimating volatility or variance of a \emph{transform} of volatility, that is, of $V_t=F(\si^2_t)$, where $F$ is a $C^2$-function on $(0,\infty)$. Typical functions of interest include $F(x)=x$ (volatility of variance), $F(x)=\sqrt{x}$ (volatility of volatility), $F(x)=\log x$ (volatility of log-variance) and $F(x)= \log \sqrt{x}$ (volatility of log-volatility). For any fixed $0\leq\underline t < t<\overline t<\infty$, 
	since $\si$ is an It\^o semimartingale, $V$ is again an It\^o semimartingale on $[\underline t, \overline t]$ on the event $\{\inf_{s\in[\underline t,\overline t]}\si_s^2>0\}$. %We are interested in estimating $VV_t$, the spot variance of $V$ at time $t$. By It\^o's formula, we have
	%\begin{equation}\label{eq:VV} 
	%	VV_t= 4\si_t^2\lvert F'(\si_t^2)\rvert^2 ((\si^\si_t)^2+(\ov \si^\si_t)^2),
	%\end{equation}
	%which is the same under $\P$ and $\Q$.
	As natural estimators of $V_t$, we consider
	\begin{equation}\label{eq:V} 
		V_{t,T}(u) = F(\si^2_{t,T}(u)),\qquad 	V_{t,T,T'}(u)=\frac{T'V_{t,T}(u)-TV_{t,T'}(u)}{T'-T},
	\end{equation}
	where $\si^2_{t,T}(u)$ is defined in \eqref{eq:vol-est}. We have the following result for the increments of $V_{t,T}(u)$ and $V_{t,T,T'}(u)$:%As an immediate consequence of \eqref{eq:incr-si-01}, the mean-value theorem and It\^o's formula, we obtain
	\begin{corollary}\label{cor:incr} Suppose that $F$ is a $C^2$-function on $(0,\infty)$ and that Assumption~\ref{ass:main-1} holds. Writing $\Delta^{n}_i V_{t,T}(u)= V_{t^n_{i-1},T^n_{i-1}}(u)-V_{t^n_i,T^n_i}(u)$,  $\Delta^n_i V_{t,T,T'}(u)=V_{t^n_{i-1},T^n_{i-1}, T^{\prime n}_{i-1}}(u)-V_{t^n_{i},T^n_{i}, T^{\prime n}_{i}}(u)$ and $\Delta^n_i V_t = V_{t^n_{i-1}}-V_{t^n_i}$, we  have on the set $\{\inf_{s\in[\underline t,\overline t]}\si_s^2>0\}$ that
		\begin{equation}\label{eq:incr-V} \begin{split}
				\Delta^{n}_i V_{t,T}(u)&=\Delta^n_i V_t+ \sum_{k=1}^K \Delta^n_i v^{(k)}_t C^{(k)}_t(u) T^n_{i-1} +O^\uc(T^{N/2}+\sqrt{\Den T}) + o^\uc(  \Den/\sqrt{T}),\\ %F'(\si^2_{t^n_i})\Delta^n_i \si^2_t + o^\uc(\sqrt{\Den/k_n})=\Delta^n_i V_t + o^\uc(\sqrt{\Den/k_n}),\\
				\Delta^{n}_i V_{t,T,T'}(u)&=\Delta^n_i V_t+O^\uc(T^{N/2}+\sqrt{\Den T}) + o^\uc(  \Den/\sqrt{T}),
			\end{split}
		\end{equation}
		as $\Den\to0$, $T\to0$ and $\Den/T\to0$
		for some $v^{(k)}_t$ and $C^{(k)}_t(u)$ with the same properties as in Theorem~\ref{thm:incr}.
	\end{corollary}
	
	We note that $\Delta^{n}_i V_{t,T,T'}(u)$ is void of the terms $\sum_{k=1}^K \Delta^n_i v^{(k)}_t C^{(k)}_t(u) T^n_{i-1}$, which are due to the dynamics of the spot variance and which can be nontrivial in practice as illustrated in \cite{CT23_b}. Furthermore, in the case where price jumps  come from a (possibly time-changed) finite variation Lévy process, as is the case for many models used in applications, then the next corollary  shows that $\Delta^n_i V_{t,T,T'}(u)$ is a less biased approximation of $\Delta^n_i V_t$ than $\Delta^n_i V_{t,T}(u)$, not only in finite samples but also asymptotically.  
	\begin{corollary}\label{cor:finact}
		Besides the assumptions of Corollary~\ref{cor:incr}, suppose that $\ga(t,z)=z_1$ (where $z_1$ is the first coordinate of $z$) and $F(dz)=dz$ and assume that \eqref{eq:extra} remains true with $\ga^\si$ in place of $\ga$. Then
		on the set $\{\inf_{s\in[\underline t,\overline t]}\si_s^2>0\}$,
		\begin{equation}\label{eq:incr-V-2} \begin{split}
				\Delta^{n}_i V_{t,T}(u)&=\Delta^n_i V_t+ \sum_{k=1}^K \Delta^n_i v^{(k)}_t C^{(k)}_t(u) T^n_{i-1} -T^n_{i-1}\si_{t^n_i}\int_{\R^{d'}} \Delta^n_i \ga^\si(t,z)\la(t^n_i,z)dz \\
				&\quad+O^\uc(T^{N/2}) + o^\uc(\sqrt{\Den}T+  \Den/\sqrt{T}),\\ %F'(\si^2_{t^n_i})\Delta^n_i \si^2_t + o^\uc(\sqrt{\Den/k_n})=\Delta^n_i V_t + o^\uc(\sqrt{\Den/k_n}),\\
				\Delta^{n}_i V_{t,T,T'}(u)&=\Delta^n_i V_t+O^\uc(T^{N/2}) + o^\uc(\sqrt{\Den}T+  \Den/\sqrt{T}).
			\end{split}\raisetag{-2.5\baselineskip}
		\end{equation}
		In particular, if $\inf_{s\in[\un t,\ov t]} \{(\si^\si_s)^2+(\ov\si^\si_s)^2\}>0$, $T^{N/2}= o^\uc(\sqrt{\Den}T)$ and $\Den/\sqrt{T} =O^\uc(\sqrt{\Den}T)$,   the bias of $	\Delta^{n}_i V_{t,T,T'}(u)$ is of strictly smaller asymptotic order than the bias of $	\Delta^{n}_i V_{t,T}(u)$.
	\end{corollary}
	Using suitable bias-reduction techniques, the residual terms in \eqref{eq:incr-V} can be further reduced in size (in an asymptotic sense) under mild assumptions on the jump measure of $x$, see e.g., Theorem 5 in \cite{todorov2021bias}. With such corrections, one should be able to achieve the strong result of Corollary~\ref{cor:finact} even in the infinite variation jump case. To keep the analysis short, we do not consider such extensions here.

	\begin{remark}\label{rem:PQ}
		The results of this and the previous section hold as long as $x$ satisfies the conditions outlined in Assumptions~\ref{ass:main} or \ref{ass:main-1} with respect to a given  probability space $(\Om,\calf,\P)$. In financial applications, there are  typically two probability measures of interest: the physical probability measure $\P$ and the risk-neutral probability measure $\Q$, which is used for pricing of derivatives and is equivalent to $\P$. Even though we have used the notation $\P$ so far (to denote a generic probability measure and to stress the fact that the expansions obtained so far can be applied to any probability measure), in a financial context where $O_{t,T}(k)$ has the interpretation of an option price, \eqref{eq:opt} and consequently \eqref{eq:spanning} and \eqref{eq:vol-est} are to be understood under the risk-neutral measure $\Q$. Consequently, in order to apply the theorems and corollaries so far in this setting, we have to assume Assumptions~\ref{ass:main} or \ref{ass:main-1} under $\Q$. This distinction becomes important in the next section as we need assumptions under both $\P$ and $\Q$.
	\end{remark}

	\section{Testing whether volatility jumps have finite or infinite variation}\label{sec:test}
	
	In this section, we show how our previous results can be used to construct a test between
	\begin{equation}\label{eq:H} \begin{split}
		&H_0\colon \text{``Volatility jumps are of finite variation.''}\\
		 \text{and}\quad & H_1\colon \text{``Volatility jumps are of infinite variation.''}
		 \end{split}
	\end{equation}
From now on, we use $\P$ to denote the physical probability measure and $\Q$ to denote the equivalent  risk-neutral probability measure. We further write $\E$ and $\E_t$ for (conditional) expectation under $\P$ and put a superscript $\Q$ to denote the same under $\Q$.
	Before specifying the details of the test, let us remark that none of the estimators studied  in the previous section is feasible in practice. This is because the  estimators so far hinge on $\call_{t,T}(u)=\E^\Q_t[e^{iu(x_{t+T}-x_t)/\sqrt{T}}]$, which by \eqref{eq:spanning} is only a statistic if the option prices $O_{t,T}(k)$ were continuously observed for every log-strike $k$. In practice, of course, option prices are only  available on a  finite log-strike grid
	\begin{equation}\label{eq:logmoney}
		\underline{k}_{t,T} ~\equiv~ k_{1,t,T} \, < \, k_{2,t,T} \, < \, \cdots \, < \, k_{N_{t,T},t,T}~\equiv~ \overline{k}_{t,T},\qquad N_{t,T}\in\mathbb{N}_+,
	\end{equation}
	which may be random and vary in $t$ and $T$. We denote the gap between consecutive log-strikes  by $\delta_{j,t,T} = k_{j,t,T} - k_{j-1,t,T}$, for $j=2,\dots,N_{t,T}$. In addition, option prices may be observed with  errors, that is, we only observe
	\begin{equation}\label{eq:obs}
		\widehat{O}_{t,T}(k_{j,t,T}) = O_{t,T}(k_{j,t,T})+\epsilon_{j,t,T},\qquad j=1,\dots,N.
	\end{equation}
	We assume that the errors $\epsilon_{j,t,T}$ are defined on an auxiliary space $(\Omega^{(1)},\mathcal{F}^{(1)})$ equipped with a transition probability $\mathbb{P}^{(1)}(\omega,d\omega^{(1)})$ from $\Omega$, the probability space on which $x$ is defined, to $\Omega^{(1)}$. We further define
	\begin{equation}\label{eq:ext} 
		\ov\Omega \,=\, \Omega\times \Omega^{(1)},\quad\ov{\mathcal{F}} \,=\, \mathcal{F} \otimes \mathcal{F}^{(1)},\quad\ov{\mathbb{P}}(d\omega,d\omega^{(1)}) \,=~ \mathbb{P}(d\omega) \, \mathbb{P}^{(1)}(\omega,d\omega^{(1)}).
	\end{equation}   
	With $\widehat{O}_{t,T}(k_{j,t,T})$, we can now  form a Riemann sum approximation of the integral in (\ref{eq:spanning}) and obtain a feasible estimator of $\call_{t,T}(u)$ via
	\begin{equation}\label{eq:L_hat}
		\widehat{\mathcal{L}}_{t,T}(u) = 1 - \left(\frac{u^2}{T}+i\frac{u}{\sqrt{T}}\right)e^{-x_t}\sum_{j=2}^{N_{t,T}}e^{(iu/\sqrt{T}-1)(k_{j-1,t,T}-x_t)}\widehat{O}_{t,T}(k_{j-1,t,T})\delta_{j,t,T},
	\end{equation}
	for $u\in\R$. This in turn yields feasible versions of  \eqref{eq:V} via
	\begin{equation}\label{eq:V_hat} 
		\wh V_{t,T}(u)=F(\widehat{\si}^2_{t,T}(u)),	\qquad \widehat{\si}^2_{t,T}(u) = -\frac{2}{u^2}\log|\widehat{\mathcal{L}}_{t,T}(u)|,
	\end{equation}
	and
	\begin{equation}\label{eq:V_12}
		\wh V_{t,T,T'}(u)= \frac{T'\widehat{V}_{t,T}(u) - T\widehat{V}_{t,T'}(u)}{T'-T}.
	\end{equation}
	
	In order that $\wh\call_{t,T}(u)$  be a sufficiently good estimator of $\call_{t,T}(u)$, we work under the following  assumptions in the remainder of this section. They concern the existence of conditional moments of $x$ under $\Q$, the option observation scheme as well as the observation errors, and are similar to those  in \cite{CT22}.   %In what follows, if expectation is taken under $\mathbb{Q}$, no superscript is used in the notation; if expectation is taken under $\P$ or $\ov \P$, we put superscripts to signify this.
	
	\bass\label{ass:C} The observed option prices $\wh O_{t,T}(k_{j,t,T})$ from \eqref{eq:obs} are defined on $(\ov\Om,\ov\calf,\ov\P)$ and there is an $\F$-adapted c\`{a}dl\`{a}g process  $C_t$   such that the following holds:
	\benu
	\item  For all $0<t<u<\infty$, 
	\begin{equation}\label{a3:1}\begin{split}
			%			\begin{split}
			%				&\mathbb{E}_s[|\sigma_u|^4]+\mathbb{E}_s[e^{2|x_u|}] + \mathbb{E}_s \biggl[\biggl( \int_{\mathbb{R}} (e^{|\gamma(u,z)|}-1)\la(dx)  \biggr)^2 \biggr]\\&\qquad+ \mathbb{E}_s\biggl[\biggl( \int_{\mathbb{R}} (e^{|\delta(u,z)|}-1-|\delta(u,z)|)\la(dx)  \biggr)^2\biggr]<C_s.
			%			\end{split}
			&\E^\Q_t\biggl[ \al_u^4+\sigma_u^6+e^{4\lvert x_u\rvert}+\biggl(\int_\R (e^{3\lvert\Ga(u, z)\rvert}-1) \nu_u(dz)\biggr)^4\\
			&\qquad+\biggl(\int_\R (e^{3\lvert\ga(u, z)\rvert}-1-3\lvert\ga(u,z)\rvert) \nu_u(dz)\biggr)^4\biggr]<C_t.\end{split}
	\end{equation}
	\item The number of strikes $N_{t,T}$ and the log-strike grid $\{k_{j,t,T}\}_{j=1}^{N_{t,T}}$ are $\mathcal{F}_{{t}}$-measurable  and  
	\begin{equation}\label{a4:1}
		C^{-1}_{{t}}\delta\leq \delta_{j,t,T}\leq C_{{t}}\delta,\qquad j=2,\dots, N_{t,T},
	\end{equation}
	for a deterministic sequence $\delta=\delta(T)$. Moreover, there is a random function $\rho_{t,\tau}(k)$ such that $t\mapsto \rho_{t,\tau}(k)$ is $\F$-adapted,  locally bounded (uniformly in $k$) and satisfies
	\begin{equation}\label{eq:rho} 
		\E_t[(\rho_{u,\tau}(k)-\rho_{t,\tau}(k))^2\wedge 1]^{1/2} \leq C_t\sqrt{u-t},\quad 0<t<u<\infty,\ \tau>0,\ k\in\R;
	\end{equation}
	$\tau\mapsto \rho_{t,\tau}(k)$ is continuous in $\tau>0$ (uniformly in $k$ and locally uniformly in $t$) and $k\mapsto \rho_{t,\tau}(k)$ is differentiable with a derivative that is locally bounded in $t$ (uniformly in $k$ and locally uniformly in  $\tau$); and
	for any $\tau>0$,
\begin{equation}\label{eq:gridsize} 
	\sup_{j:\lvert k_{j,t,\tau T}-x_t\rvert < C_t^{-1}}  \lvert \delta_{j,t,\tau T}/\delta-\rho_{t,\tau}(k_{j-1,t,\tau T}-x_t) \rvert \leq C_t\delta.
\end{equation}
	\item For some $\iota>0$,
\begin{equation}\label{eq:kminmax} 
	\liminf_{T\to0} \frac{\inf_{n\in\N, i=1,\dots,k_n} ( \lvert\underline  k_{t^n_i,T^n_i} \rvert\wedge\lvert \underline k_{t^n_i,T^{\prime n}_i}\rvert\wedge \overline  k_{t^n_i,T^n_i}  \wedge  \overline k_{t^n_i,T^{\prime n}_i}  )}{(\delta/\sqrt{T})^\iota} = \infty.  
\end{equation}
	\item For $t,\tau>0$ and $j=1,\dots, N_{t,\tau T}$, we have
	\begin{equation}\label{eq:epsilon} 
		\epsilon_{j,t,\tau T}(k_{j,t,\tau T}) = \zeta_{t,\tau}(k_{j,t,\tau T}-x_t)O_{t,\tau T}(k_{j,t,\tau T})\overline{\epsilon}_{j,t,\tau T},
	\end{equation}
	where $t\mapsto \zeta_{t,\tau}(k)$ is $\F$-adapted, locally bounded (uniformly in $k$) and satisfies
	\begin{equation}\label{eq:zeta} 
		\E_t[(\zeta_{u,\tau}(k)-\zeta_{t,\tau}(k))^2\wedge 1]^{1/2} \leq C_t\sqrt{u-t},\quad 0<t<u<\infty,\ \tau>0,\ k\in\R;
	\end{equation}  $\tau\mapsto \zeta_{t,\tau}(k)$ is continuous in $\tau>0$ (uniformly in $k$ and locally uniformly in $t$) and $k\mapsto \zeta_{t,\tau}(k)$ is   differentiable with a derivative that is locally bounded in $t$ (uniformly in $k$ and locally uniformly in  $\tau$); and $\overline{\epsilon}_{j,t,T}$ is   $\ov\calf$-measurable, independent of $\calf$ under $\ov\P$ and i.i.d.\ as $j$, $t$ and $T$ vary. Moreover,
	\begin{equation}\label{eq:mom} 
		\mathbb{E}^{\ov \P}[\overline{\epsilon}_{j,t,T}\mid \mathcal{F}]= 0,\quad \mathbb{E}^{\ov \P}[(\overline{\epsilon}_{j,t,T})^2\mid \mathcal{F}] = 1,\quad \mathbb{E}^{\ov \P}[| \overline{\epsilon}_{j,t,T}|^{p}\mid\mathcal{F}] <\infty  \text{ for all } p>2.
	\end{equation} 
	\eenu
	\eass

	We now turn to  the details of our option-based test for infinite variation volatility jumps. 
	In order to describe the intuition behind our test, let us assume that we have access to $\{\Delta^n_i V_t : i=1,\dots, k_n, t\in\calt\}$, where $k_n$ increases to $\infty$ such that $k_n\Den\to0$ and  $\calt=\{t^{(i)}:i=1,\dots, N_n\}$ is a collection of time points such that $t^{(i)}+k_n\Den = t^{(i+1)}$ for all $i$ and $n$  and $k_n\lvert \calt\rvert \sim \tau/\Den$ for some $\tau>0$. Note that $\calt$ in principle depends on $n$ but we refrain from indicating this in the notation for the sake of  simplicity. Furthermore, let us assume
	that $V=F(\si^2_t)$ is a Lévy procress of the form $V_t=V_0+\al^V t +\si^V W_t + J^V_t$, where $J^V$ is a symmetric $\beta$-stable process with scale parameter $A$. The idea is now to look at the empirical characteristic function of normalized volatility increments given by
	\begin{equation}\label{eq:L1} 
		\L^n_1(U)= \frac1{k_n\lvert \calt\rvert}\sum_{t\in\calt} \sum_{i=1}^{k_n} e^{iU_n \Delta^n_i V},\quad U>0, \quad U_n=U/\sqrt{\Den}.
	\end{equation}
	%and a version of it computed with increments at one half of the original frequency:
	%	\begin{equation}\label{eq:L2} 
	%	\L^n_2(U)= \frac2n \sum_{j=1}^{\lfloor n/2\rfloor} e^{iU_n (\Delta^n_{2j-1}V+\Delta^n_{2j} V)}.
	%\end{equation}
	By the CLT and the Lévy--Khintchine formula, we have
	\[ \log \lvert \L^n_1(U)\rvert \sim -\frac12 U^2(\si^V)^2 -\lvert UA\rvert^\beta \Den^{1-\beta/2}, \]
	where the approximation holds at a rate of $1/\sqrt{\Den}$.
	Therefore, if $A=0$ or $\beta<1$ (i.e., $V$ has no jumps or only finite variation jumps),
	\begin{equation}\label{eq:aux1} 
	\Den^{-1/2}\bigl(\log \lvert \L^n_1(\sqrt{2} U)\rvert -2\log \lvert \L^n_1(U)\rvert \bigr) \limd N(0,\mathrm{Var}) 
	\end{equation}
	for some $\mathrm{Var}>0$, while if $\beta>1$ and $A\neq 0$, we have 
	\begin{equation}\label{eq:aux2} 
		\Den^{-1/2}\bigl(\log \lvert \L^n_1(\sqrt{2} U)\rvert -2\log \lvert \L^n_1(U)\rvert \bigr)\sim (2-2^{\beta/2})AU^\beta\Den^{1/2-\beta/2}\to\infty,
	\end{equation}
	which allows us to construct a test with asymptotic power converging to $1$. 
	
	In reality, of course, $V$ may not be a Lévy process (but we still know how to expand  \eqref{eq:L1} in this case thanks to the results of the previous sections) and we do not have access to $\Delta^n_j V$ but only to $\Delta^n_j \wh V_{t,T}(u)$ (or $\Delta^n_j \wh V_{t,T,T'}(u)$). As a result, we have to consider, for each $t$, a feasible version of \eqref{eq:L1} given by 
	\begin{equation}\label{eq:L1-feas} 
		\wh	\L^n_1(u,U)_{t,T}= \frac1{k_n }   \sum_{i=1}^{k_n} e^{iU_n \Delta^n_i\wh V_{t,T}(u)},\quad u,U>0. 
	\end{equation}
	%where $k_n\to\infty$ is an integer sequence and  $\calt=\{t^{(i)}:i=1,\dots, N\}$ is a collection of time points such that $t^{(i)}+k_n\Den \leq t^{(i+1)}$ for all $i$ and $n$ (e.g., $t^{(i)}$ denotes the end of trading day $i$ and $k_n$ is the number of high-frequency increments per day). Note that $t^{(i)}$ and $t^{(i+1)}$ can be asymptotically close to each other, as one day is considered short in our framework. Similarly, $N$ can be increasing (e.g., if $N$ is the number of trading days in a given year). Therefore, $\calt$ in principle depends on $n$ but we refrain from indicating this in the notation for the sake of  simplicity.
	Since $\wh	\L^n_1(u,U)_{t,T}$ contains  a contribution from the option observation errors, we cannot simply base our test on a feasible version of \eqref{eq:aux1} and \eqref{eq:aux2}. Instead, we consider a statistic similar to \eqref{eq:L1-feas} but computed at one half of the original frequency, that is,
	\begin{equation}\label{eq:L2-feas} 
		\wh	\L^n_2(u,U)_{t,T}= \frac2{k_n }   \sum_{i=1}^{\lfloor k_n/2\rfloor} e^{iU_n (\Delta^n_{2i-1} \wh V_{t,T}(u)+\Delta^n_{2i} \wh V_{t,T}(u))},\quad u,U>0, 
	\end{equation}
	and form the ratio $\wh	\L^n_2(u,U)_{t,T}/\wh	\L^n_1(u,U)_{t,T}$ first. As we will see below, this effectively removes the   contribution of the noise to \eqref{eq:L1-feas}. %If we define $\wh	\L^n_1(u,U)_{t,T}$ and $\wh	\L^n_2(u,U)_{t,T}$ analogously to \eqref{eq:L1-feas} and \eqref{eq:L2-feas}, the same applies to the ratio $\wh	\L^n_2(u,U)_{t,T,T'}/\wh	\L^n_1(u,U)_{t,T,T'}$ .
	
	To operationalize this idea, we consider for $j=1,2$ local estimators of the diffusive volatility of $V$ given by
	\begin{equation}\label{eq:Cnhat} 
		\wh \C^n_j(u,U)_{t,T}=\ov \C^n_j(u,U)_{t,T}-\wt \C^n_j(u,U)_{t,T},
	\end{equation}
where
\begin{equation} \ov \C^n_j(u,U)_{t,T}=-\frac{2}{U^2} \log \lvert \wh\L^n_j(u,U)_{t,T}\rvert,~ \wt \C^n_j(u,U)_{t,T}=\frac{2^{\bone_{\{j=2\}}}}{U^2k_n} \bigl(\lvert \wh \L^n_j(u,U)_{t,T}\rvert^2-1\bigr),\end{equation}
with $\ov \C^n_j(u,U)_{t,T}$ being the main estimator and $\wt \C^n_j(u,U)_{t,T}$ being a bias correction term. As $\ov \C^n_2(u,U)_{t,T}-\ov \C^n_1(u,U)_{t,T}=-2U^{-2}\log \lvert \wh\L^n_2(u,U)_{t,T}/\wh\L^n_1(u,U)_{t,T}\rvert$, taking   differences  such as $\wh \C^n_2(u,U)_{t,T}-\wh \C^n_1(u,U)_{t,T}$ mitigates the effect of noise due to the option observation errors.
	Our test statistic based on a single time-to-maturity $T$ is now given by
	\begin{equation}\label{eq:test} \begin{split}
			R^n_T(u,U)&=\sqrt{\frac{k_n\lvert \calt\rvert}{2\wh\avar^{n}_{T}(u,U)}}\frac{1}{\lvert\calt\rvert}\sum_{t\in\calt}\Bigl(\bigl[\wh \C^n_2(u,\sqrt{2}U)_{t,T}-\wh \C^n_1(u,\sqrt{2}U)_{t,T}\bigr]\\
			&\quad-  \bigl[\wh \C^n_2(u,U)_{t,T}-\wh \C^n_1(u,U)_{t,T}\bigr]\Bigr).
		%	R^n_{T,T'}(u,U)&=\sqrt{\frac{k_n\lvert \calt\rvert}{\wh\avar^{n}_{T,T'}(u,U)}}\\
		%	&\quad\times\frac{1}{\lvert\calt\rvert U^2}\sum_{t\in\calt}\biggl(\log \biggl\lvert \frac{ \wh	\L^n_2(u,\sqrt{2}U)_{t,T,T'}  }{   \wh	\L^n_1(u,\sqrt{2}U)_{t,T,T'}}\biggr\rvert - 2\log \biggl\lvert\frac{ \wh	\L^n_2(u,U)_{t,T,T'}  }{    \wh	\L^n_1(u,U)_{t,T,T'}}\biggr\rvert\biggr),
		\end{split}
	\end{equation}
Here,  the estimator $\wh\avar^{n}_{T}(u,U)$ of the asymptotic variance is defined by  
	\begin{equation}\label{eq:Avar} 
		\wh\avar^{n}_{T}(u,U)=\frac{1}{\lvert\calt\rvert} \sum_{t\in\calt}\caly^n_{t,T}(u,U)'\avar^{n}_{t,T}(u,U) \caly^n_{t,T}(u,U),\\
\end{equation}
where
\begin{equation}\label{eq:Y}\begin{split}
			\caly^n_{t,T}(u,U)&=\begin{pmatrix} -2\wt\caly^n_{t,T}(u,U)/U^2\\ \wt\caly^n_{t,T}(u,\sqrt{2}U)/U^2,\end{pmatrix}\\
			\wt\caly^n_{t,T}(u,U)&=\biggl(-\frac{\Re(\wh	\L^n_1(u,U)_{t,T})}{\lvert 	\wh	\L^n_1(u,U)_{t,T}\rvert^2},\frac{\Im(\wh	\L^n_1(u,U)_{t,T})}{\lvert 	\wh	\L^n_1(u,U)_{t,T}\rvert^2},\\
			&\qquad\qquad\qquad\qquad\qquad\frac{\Re(\wh	\L^n_2(u,U)_{t,T})}{\lvert 	\wh	\L^n_2(u,U)_{t,T}\rvert^2},-\frac{\Im(\wh	\L^n_2(u,U)_{t,T})}{\lvert 	\wh	\L^n_2(u,U)_{t,T}\rvert^2}\biggr)'
\end{split}\end{equation}
and
\begin{equation}\label{eq:Avar-2}\begin{split}
					\avar^{n}_{t,T}(u,U)&= \avar^{n,0}_{t,T}(u,U)+ \avar^{n,1}_{t,T}(u,U), \\
		\avar^{n,0}_{t,T}(u,U)&=\frac{2}{k_n }\sum_{i=1}^{\lfloor k_n/2\rfloor} \ov \calz^n_i(u,U)_{t,T}\ov \calz^n_i(u,U)'_{t,T},\\
	\avar^{n,1}_{t,T}(u,U)&=\frac{2}{k_n }\sum_{i=1}^{\lfloor k_n/2\rfloor}\bigl( \ov \calz^n_{i-1}(u,U)_{t,T}\ov \calz^n_i(u,U)'_{t,T}\\
	&\quad+\frac{2}{k_n } \sum_{i=1}^{\lfloor k_n/2\rfloor} \ov \calz^n_{i}(u,U)_{t,T}\ov \calz^n_{i-1}(u,U)'_{t,T}
\end{split}\end{equation}
with
	\begin{equation}\label{eq:Z}\begin{split}
			\wt \calz^n_i(u,U)_{t,T}&=\begin{pmatrix} \tfrac12 (\cos(U_n\Delta^n_{2i-1}\wh V_{t,T}(u))+\cos(U_n\Delta^n_{2i}\wh V_{t,T}(u)))\\ \tfrac12(\sin(U_n\Delta^n_{2i-1}\wh V_{t,T}(u))+\sin(U_n\Delta^n_{2i}\wh V_{t,T}(u)))\\ \cos(U_n(\Delta^n_{2i-1}\wh V_{t,T}(u)+\Delta^n_{2i}\wh V_{t,T}(u)))\\ \sin(U_n(\Delta^n_{2i-1}\wh V_{t,T}(u)+\Delta^n_{2i}\wh V_{t,T}(u)))\end{pmatrix},\\
			\calz^n_i(u,U)_{t,T} &= \begin{pmatrix} \wt \calz^n_i(u,U)_{t,T}\\ \wt \calz^n_i(u,\sqrt{2}U)_{t,T}\end{pmatrix},\\
			\ov \calz^n_i(u,U)_{t,T} &= \calz^n_i(u,U)_{t,T}- \frac{2}{k_n}\sum_{\ell=1}^{\lfloor k_n/2\rfloor} \calz^n_\ell(u,U)_{t,T}.
	\end{split}\end{equation}
For any of the quantities introduced in \eqref{eq:L1-feas}--\eqref{eq:Z} of the form $X_{t,T}$ or $X_T$, we define $X_{t,T,T'}$ or $X_{T,T'}$ as the same quantity but with $\Delta^n_{i} \wh V_{t,T}(u)$ replaced by $\Delta^n_{i} \wh V_{t,T,T'}(u)$ everywhere in the definition.

Next, we formally introduce the null and alternative hypotheses from \eqref{eq:H}, which are stated in terms of the coefficients of the canonical decomposition of $V_t=F(\si^2_t)$ under $\P$:
\begin{equation}\label{eq:V-dec}\begin{split}
	V_t	&=V_0 + \int_0^T \al^V_s ds +  \int_0^t   \si^V_s d  W^\P_s +\int_0^t\int_{\R} z\bone_{\{\lvert z\rvert\leq 1\}}(\mu^V-\nu^V)(ds,dz)\\
		&\quad + \int_0^t\int_\R z\bone_{\{\lvert z\rvert>1\}}\mu^V(ds,dz),
\end{split}\end{equation}
where $W^\P$ is an $\F$-Brownian motion under $\P$, $\mu^V$ is the jump measure of $V$ whose $\F$-compensator under $\P$ is assumed to have the form $\nu^V(dt,dz)=F^V_t(dz)dt$ for some transition kernel $F^V_t(\om;dz)$ from $\Om\times[0,\infty)$, equipped with the predictable $\si$-field, into $\R$.    Of course, if the density process $(\frac{d\Q}{d\P}|_{\calf_t})_{t\geq0}$ is sufficiently regular, It\^o's formula and Girsanov's theorem imply that conditions imposed on $\si$ under $\Q$ can be translated into  properties of $V$ under $\P$. But to keep things simple, below we impose  conditions on $x$ (including $\si$) under $\Q$ and on $V$ under $\P$ separately. Let us also stress  that since $\P$ and $\Q$ are equivalent, both the diffusive volatility $\si^V$ and the degree of jump activity of $V$ are measure-independent. 

\settheoremtag{H$_0$}
\begin{Assumption}\label{ass:H0} We have $\int_0^\tau \int_\R (1\wedge \lvert z\rvert)F^V_t(dz) dt<\infty$ a.s., where $\tau$ is the constant in $k_n\lvert\calt\rvert\sim \tau/\Den$.
\end{Assumption}

\settheoremtag{H$_1$}
\begin{Assumption}\label{ass:H1} We have $F^V_t(dz)=F'_t(dz)+F''_t(dz)$, where
\begin{equation}\label{eq:Ft}
F'_t(dz)	=\frac{  (\log \lvert z\rvert^{-1})^{\beta'_t}}{\lvert z\rvert^{1+\beta_t}}(c^+_t\bone_{\{0<z\leq\eps^+_t\}} + c^-_t\bone_{\{-\eps^-_t\leq z<0\}})dz
\end{equation}
and the following holds for all $t\in[0,\tau]$ and some locally bounded process $C_t\geq1$:
\begin{enumerate}
	\item The processes $\eps^+$, $\eps^-$, $c^+$ and $c^-$ are predictable and c\`agl\`ad and satisfy
\begin{equation}\label{eq:cond1} 
	C_t^{-1}\leq \eps^+_t, \eps_t^-\leq 1, \quad 0\leq c^+_t,c_t^-  \leq C_t,  \quad \Leb(A_\tau)>0,
\end{equation}
where $\Leb$ denotes the Lebesgue measure and $A_\tau= \{t\in[0,\tau]: c^+_t + c^-_t > 0\}$.
	\item  Both $\beta$ and $\beta'$ are  c\`agl\`ad predictable processes such that
	\begin{equation}\label{eq:cond2} 
C_t^{-1}\leq \beta_t\leq r,\quad \lvert \beta'_t\rvert \leq C_t,\quad \beta^\ast=\max_{t\in[0,\tau]} \beta_t\bone_{A_\tau}(t) >1\quad\text{a.s.},
	\end{equation}
where $r$ is the number from Assumption~\ref{ass:main}.
\item $F''_t(dz)=F''_t(\om,dz)$ is a signed transition kernel from $\Om\times[0,\infty)$, equipped with the predictable $\si$-field into $\R$, such that 
\begin{equation}\label{eq:cond3} 
	\int_\R (1\wedge \lvert z\rvert^{\rho\beta_t}) \lvert F''_t\rvert(dz) \leq C_t
\end{equation}
for some $\rho\in(0,1)$.
\end{enumerate}
\end{Assumption}

 Assumption~\ref{ass:H0} is equivalent to $V$ having only finite variation jumps on $[0,\tau]$. Assumption~\ref{ass:H1} is essentially Assumption 3 in \cite{AJ11}, adapted to the case where $\beta$, the instantaneous Blumenthal--Getoor index of $V$, exceeds $1$ on a set of positive Lebesgue measure. Apart from the assumption $\beta^\ast>1$, which leads to infinite variation volatility jumps (at least locally on a set of positive measure),  the other regularity conditions in Assumption~\ref{ass:H1} are very mild and satisfied by virtually all parametric jump models used in practice; see \cite{AJ11}.
	
	\begin{theorem}\label{thm:test}
		Suppose that Assumption~\ref{ass:main-1} is satisfied under $\Q$ and that we have the rate conditions
			\begin{equation}\label{eq:rates}% \begin{split}
				\frac{\delta}{\sqrt{T}}=o(\Delta_n^{1+\iota'}),\quad k_n^2\Den\to0,\quad k_n^{2+\iota''}\Den\to\infty,\quad k_n\lvert \calt\rvert \sim \tau/\Den,\quad \sqrt{\Den}\asymp T
				   %, \quad U\to0,\quad U\Den^{-\iota''}\to\infty % ,\quad T^N=o(\Den).
				%  \end{split}
			\end{equation}
		for some $\tau,\iota'>0$ and all $\iota''>0$.
		Further suppose that $(\si_t)_{t\geq0}$ is locally bounded away from zero and that $V=F(\si^2_t)$, for some $C^2$-function $F$, is an It\^o semimartingale of the form \eqref{eq:V-dec} under $\P$ with two hidden layers. Further assume that Assumption~\ref{ass:C} is in force. Then,
	under Assumption~\ref{ass:H0}, we have
	\begin{equation}\label{eq:null} 
		R^n_T(u,U) \limst N(0,1),\quad R^n_{T,T'}(u,U) \limst N(0,1),
	\end{equation}
while under Assumption~\ref{ass:H1}, we have
	\begin{equation}\label{eq:alt} 
	R^n_T(u,U) \limp  +\infty,\quad R^n_{T,T'}(u,U) \limp  +\infty.
\end{equation}
	\end{theorem}

Note that Assumption~\ref{ass:main-1} implies that the jumps of the log-price process $x$ are of finite variation. It is well known from \cite{JR14} that the minimax rate at which jumps and volatility can be separated from high-frequency return data  deteriorates as the jump activity of $x$ increases. As a similar result should hold for volatility estimation from option price data, an extension of Theorem~\ref{thm:test} to cover the case of infinite variation price jumps seems hopeless, at least in the nonparametric case. That said, even in the presence of infinite variation price jumps, we do believe that the test in Theorem~\ref{thm:test} retains its properties if used to decide whether $F(\si^2_t -2 u^{-2}T\vp^\Q_t(u_T))$ has finite or infinite variation jumps, where $\vp^\Q_t(u)$ is the local characteristic exponent of the small jumps of $x$ under $\Q$. We leave such an extension to future research.

\section{Numerical Experiments}\label{sec:num}
In this section we illustrate with simulated and real data the performance of the test for presence of infinite variation volatility jumps. 
\subsection{Monte Carlo Study}

We evaluate the finite sample performance of the test on simulated data from the following diffusive volatility model:
\begin{equation}\label{eq:v}
V_t = \exp\left(-3.94+\int_0^te^{-8(t-s)}dL_s\right),
\end{equation}
where $L_t$ is a L\'{e}vy process given by a standard Brownian motion under the null hypothesis and  by the sum of a Brownian motion with variance $0.4$ and a tempered stable jump process with L\'{e}vy density  
\begin{equation}
\nu^{\textrm{ts}}(x) = c\frac{e^{-\lambda|x|}}{|x|^{1+\alpha}},\quad\alpha = 1.5,\quad\lambda = 3, \quad c = 0.3\frac{\lambda^{2-\alpha}}{\Gamma(2-\alpha)},
\end{equation}
under the alternative hypothesis. In both considered cases for $L_t$, we have $\mathbb{E}[L_t] = 0$ and $\mathbb{E}[L_t^2] = t$. Under the alternative hypothesis, the Blumenthal--Getoor index of the jumps (captured by the parameter $\alpha$ above) is $1.5$.  The scale parameter in equation (\ref{eq:v}) is chosen such that the average value of $V_t$ in the model matches roughly  the one in our real application (our unit of time is one year). 

The above model for volatility is an exponential L\'{e}vy-driven Ornstein–Uhlenbeck process. Such parametric specification has been found to provide good fit empirically to volatility, see e.g., \cite{andersen2002empirical} and \cite{chernov2003alternative}. Unfortunately, closed-form solutions for option prices for such volatility models do not exist. For this reason, in our Monte Carlo experiment, we work with 
\begin{equation}
\widehat{V}_{t,T,T'}(u,U)= \wh V_{t,T,T'} = V_t + \sigma_{\epsilon}\times \epsilon_t^{v},
\end{equation} 
where $(\epsilon_t^{v})_{t\geq 1}$ is an i.i.d. sequence of standard normal random variables that are independent of $V_t$. According to our expansion results, this should hold asymptotically. We calibrate the value of $\sigma_{\epsilon}$ to match the noise in estimates of $\widehat{V}_{t,T,T'}$ in our application to real data. We do this on the basis of sample autocovariance estimates for lags 0 and 1 of the high-frequency increments of $\widehat{V}_{t,T,T'}$. This leads to $\sigma_{\epsilon} = \sqrt{0.18}\times \sqrt{\textrm{Var}(\Delta_i^n V_t))}$. 

The sampling frequency and the choice of our local block size matches exactly that in the real data application. Mainly, we sample every five minutes during regular hours of the trading day. This means $\Delta_n = (1/252)\times (1/77)$ (we use business day convention in which one year has 252 business days). We set $k_n = 77$ which results in a block of length one day and let $\mathcal{T}$ correspond to $20$ consecutive days. Finally, we set $U$ in the following data-adaptive way: 
\[\widehat{U}_t = \sqrt{\Delta_n}\times \frac{\sqrt{-2\log(0.75)}}{\sqrt{2 BV^n_t/k_n}},~~BV_t^n = \frac{\pi}{2}\sum_{i=2}^{k_n}|\Delta_{i-1}^n\widehat{V}_{t,T,T'}(u,U)||\Delta_i^n\widehat{V}_{t,T,T'}(u,U)|. \]

In Figure~\ref{fig:test_mc}, we plot the realization of the test statistic $R^n_{t,T,T'}(u,U)$ in the two simulation scenarios. Consistent with the asymptotic theory derived in the previous section, $R^n_{t,T,T'}(u,U)$ appears centered at zero under the null hypothesis and it appears centered above zero under the alternative hypothesis. The test has good finite sample properties. Indeed, the rejection rate of a one-sided test based on $R^n_{t,T,T'}(u,U)$ of size $5\%$ is $4.5\%$ under the null hypothesis and it is $55.5\%$ under the alternative hypothesis. In other words, the test can separate relatively well the null from the alternative hypothesis in finite samples. 

\begin{figure}[htbp]
\begin{center}
\includegraphics[width=140mm,height=70mm]{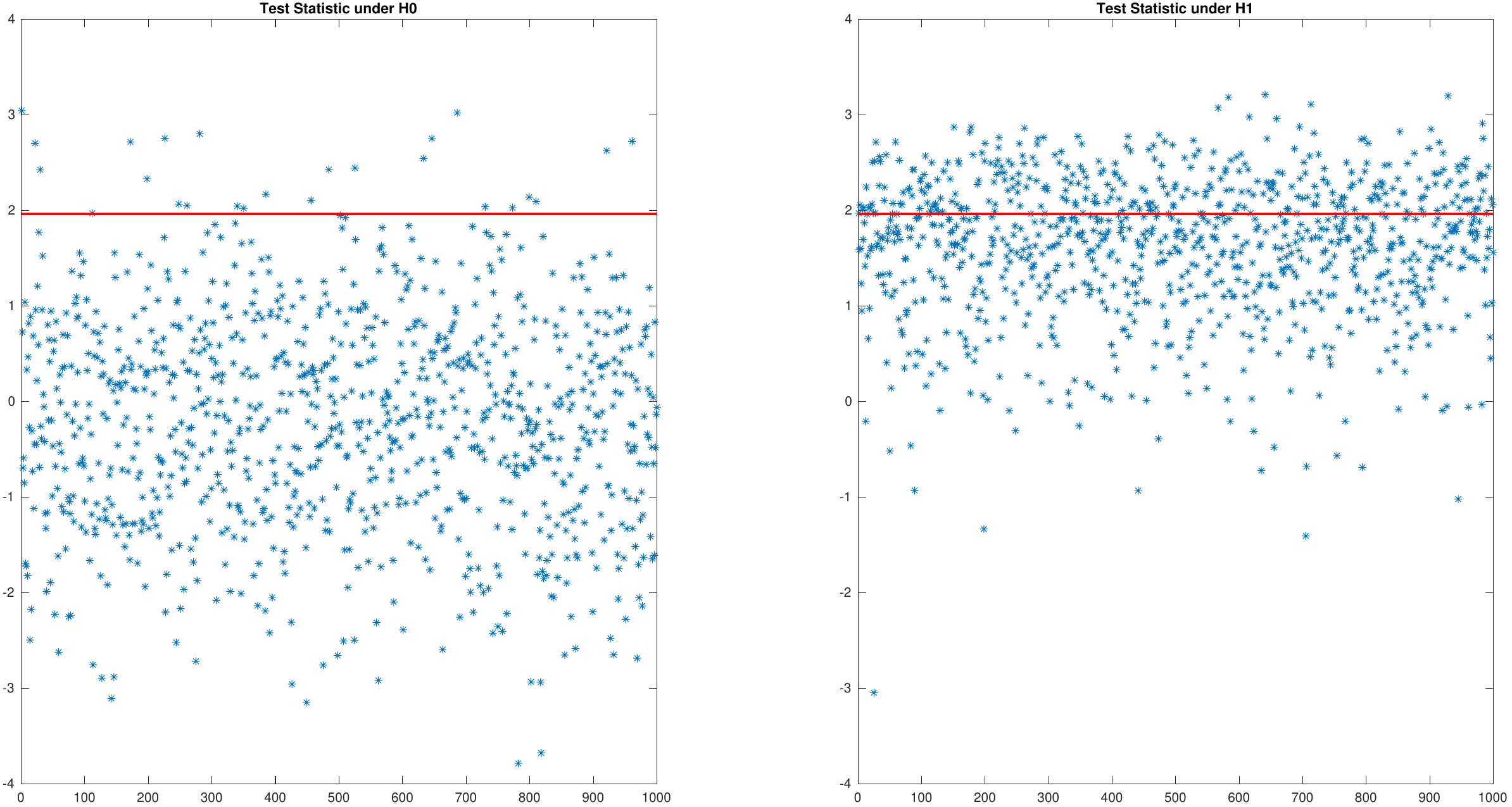}
\end{center}
\vspace{-4ex}
\caption{Values of $R^n_{t,T,T'}(u,U)$ in the Monte Carlo. Red line corresponds to $5\%$ critical value for a one-sided test of the null hypothesis.}
\label{fig:test_mc}
 \end{figure}

\subsection{Empirical Illustration}

In our empirical application, we study the activity of market volatility jumps using short-dated options written on the S\&P 500 index. Our sample covers the period 2016--2020 and we sample option prices during the trading days at five minute sampling frequency. The data is cleaned exactly as in \cite{CT23_b} and we refer to that paper for the details. Our choice of $k_n$, $\mathcal{T}$ and $U$ is exactly as in the Monte Carlo study, while the value of $u$ for estimating the spot volatility from the option data is set exactly as in \cite{CT23_b}.

In Figure~\ref{fig:test_emp}, we plot the time series of the test statistic $R^n_{t,T,T'}(u,U)$. As seen from the figure, the statistic is on average positive, which is consistent with the alternative hypothesis of infinite variation volatility jumps. The empirical rejection rate for a one-sided test of size $5\%$ of the null hypothesis is 23.3\%. This value is above the nominal size of the test by a nontrivial amount and thus provides evidence for the presence of infinite variation jumps in volatility.
	
\begin{figure}[htbp]
\begin{center}
\includegraphics[width=70mm,height=50mm]{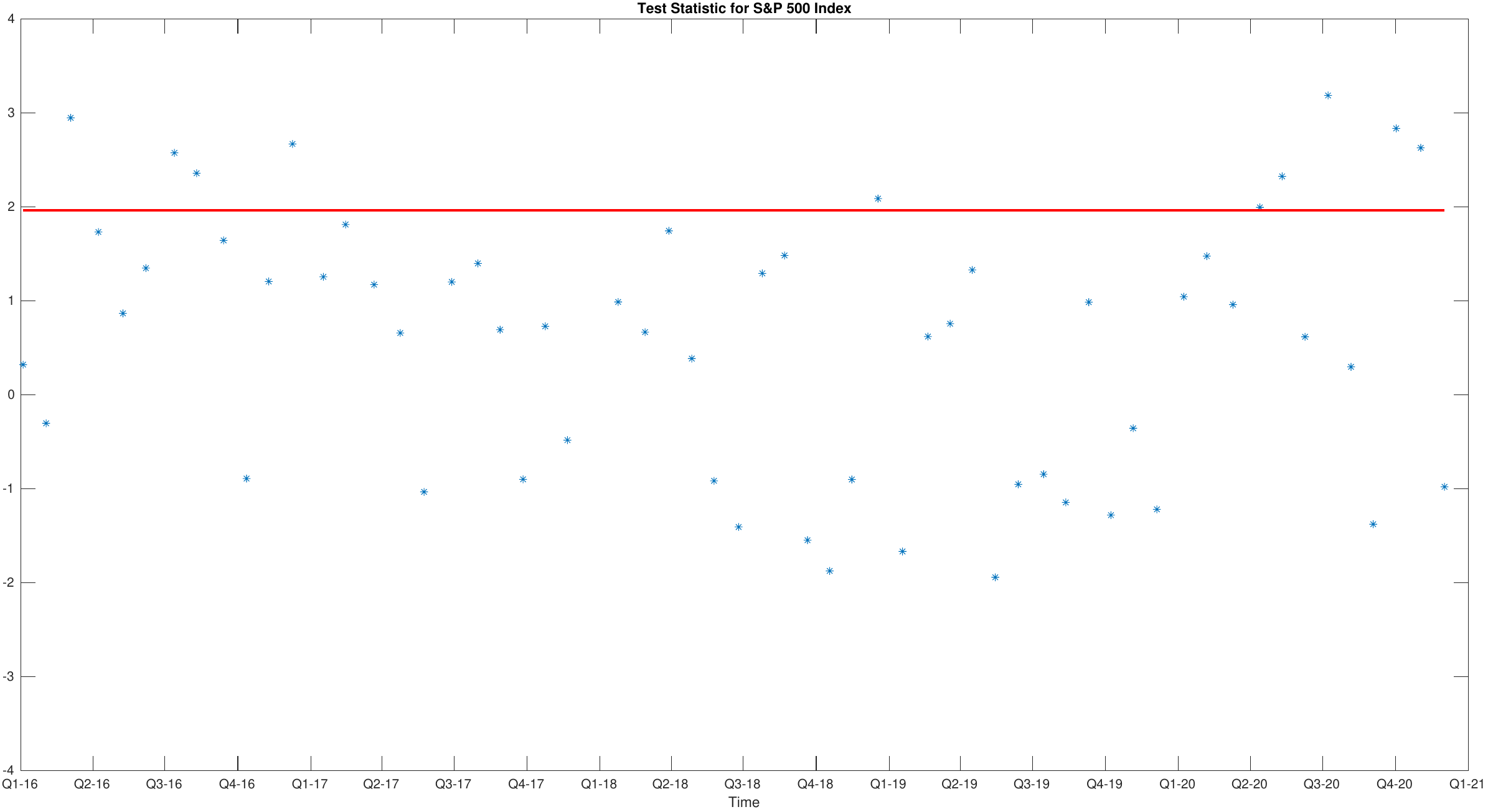}
\end{center}
\vspace{-4ex}
\caption{Values of $R^n_{t,T,T'}(u,U)$ for the S\&P 500 index. Red line corresponds to $5\%$ critical value for a one-sided test of the null hypothesis.}
\label{fig:test_emp}
 \end{figure}

\section{Proofs for Sections~\ref{sec:expansion-2} and \ref{sec:app}}\label{sec:proofs}

	%\subsection{Proofs for Section \ref{sec:expansion}}\label{sec:proof-1}
	%In what follows, if expectation is taken under $\mathbb{Q}$, no superscript is used in the notation; if expectation is taken under $\P$ or $\ov \P$, we put superscripts to signify this.
	In a first step, we show a technical version of Theorem~\ref{thm:incr-L}, in which $\Delta^n_i \La_{t,T}(t^n_i,u)$ is broken into its constituent pieces. %We will reassemble them and obtain Theorem~\ref{thm:incr-L} at the end of this section. %
	To this end, 
	we let  (recall that $\nu_t(dz)=\la(t,z)F(dz)$)
	%\[ \chi_t(u)=-\si^\si_t\si^2_tu^2 + iu\si_t\chi^{(1)}_t(u)+\chi_t^{(2)}(u), \]
	%where
	\begin{equation}\label{eq:chi}
		\begin{split}
			\chi^{(1)}_t(u)	&=\int_{\R^{d'}}  (e^{iu\ga(t,z)}-1)(\si^\ga(t,z)+\ga^\si(t,z))\nu_t(dz),\\
			\chi^{(2)}_{t}(u)	&=  \int_{\R^{d'}}  \int_{\R^{d'}}   (e^{iu\ga(t,z)}-1)(e^{iu\ga(t,z')}-1)  \ga^\ga(t,z,z')\nu_t(dz)\nu_t(dz'),\\
			\chi^{(3)}_t(u)&= \int_{\R^{d'}} (e^{iu\ga(t,z)}-1-iu\ga(t,z))\si^{\la,(1)}(t,z)F(dz),\\
			\chi^{(4)}_t(u)&=\int_{\R^{d'}\times\R^{d'}\times\R} (e^{iu\ga(t,z)}-1-iu\ga(t,z))(e^{iu\ga(t,z')}-1)\\
			&\qquad\qquad\qquad\qquad\qquad\times\ga^\la(t,z,z',v')\bone_{\{0\leq v'\leq \la(t,z')\}}F(dz)F(dz')dv'
		\end{split}
	\end{equation}
	and 
	\begin{align*}
		\xi^{(1)}_t(s,u)&=\int_{\R^{d'}} (e^{iu\ga(s,z)}-1)\ga(t,z)\nu_s(dz),\\
		%\xi^{(2)}_t(s,u)&=\int_\R e^{iu(\ga(s,z)+\Ga(s,z))}\Ga(t,z)\la(dz),\\
		\xi^{(2)}_t(s,u)&=\int_{\R^{d'}}  (e^{iu\ga(s,z)}-1)(e^{iu\ga(s,z')}-1)\ga^\ga(t,z,z')\nu_s(dz) \nu_s(dz'),\\
		\xi^{(3)}_t(s,u)	&=\int_{\R^{d'}\times\R^{d'}}  (\ga^\ga(s,z,z')+\ga^\ga(s,z',z))\\
		&\qquad\qquad\qquad\qquad \qquad\times(e^{iu\ga(s,z)}-1)e^{iu\ga(s,z')}\ga(t,z')\nu_s(dz)\nu_s(dz'),\\
		\xi^{(4)}_t(s,u)&=\int_{\R^{d'}}  (e^{iu\ga(s,z)}-1)(\si^\ga(t,z)+\ga^\si(t,z))\nu_s(dz),\\
		\xi^{(5)}_t(s,u)&=\int_{\R^{d'}}  e^{iu\ga(s,z)}  (\si^\ga(s,z)+\ga^\si(s,z))\ga(t,z)\nu_s(dz),\\
		\xi^{(6)}_t(s,u)&=\int_{\R^{d'}} (e^{iu\ga(s,z)}-1)\si^{\la,(1)}(s,z)\ga(t,z)F(dz),\\
		\xi^{(7)}_t(s,u)&=  \int_{\R^{d'}\times\R^{d'}\times\R}(e^{iu\ga(s,z)}-1)  (e^{iu\ga(s,z')}-1)\\
		&\qquad\qquad\qquad\qquad \qquad\times	\ga^\la(s,z,z',v')\bone_{\{0\leq v'\leq \la(s,z')\}}\ga(t,z)F(dz)F(dz')dv',\\
		\xi^{(8)}_t(s,u)&=  \int_{\R^{d'}\times\R^{d'}\times\R}(e^{iu\ga(s,z)}-1-iu\ga(s,z))e^{iu\ga(s,z')}\\
		&\qquad\qquad\qquad\qquad \qquad\times \ga(t,z') \ga^\la(s,z,z',v')\bone_{\{0\leq v'\leq \la(s,z')\}}F(dz)F(dz')dv'.
	\end{align*}
	Note that
	\begin{align*}
		\psi_t(t,u)&=-\frac12u^2\si_t\chi^{(1)}_t(u)	, \qquad\ov\psi_t(t,u)= \frac12iu\chi^{(2)}_{t}(u),\\
		 \wt\psi_t(t,u)&=\frac12iu\si_t\chi^{(3)}_t(u) +\frac12\chi^{(4)}_t(u).
	\end{align*} 
	We also use  the notation
	\begin{equation}\label{eq:incr} \begin{split}
			%	\Delta^n_i	F_{t,T}(u)&=F_{t^n_{i-1},T^n_{i-1}} (u_{ {T^n_{i-1}}} )-F_{t^n_{i},T^n_i} (u_{ {T^n_i}} ),\\
			%	\Delta^n_i F_{t,T}(s,u)&=F_{t^n_{i-1},T}(s,u)-F_{t^n_{i},T}(s,u),\\
			&\Delta^n_i	F_{t}=F_{t^n_{i-1}}-F_{t^n_{i}},~~	\Delta^n_i F(t)=F(t^n_{i-1})-F(t^n_i)	, ~~ \Delta^n_i F_t(u)=F_{t^n_{i-1}}(u)-F_{t^n_i}(u),\\
			&\Delta^n_i F_{t}(s,u)=F_{t^n_{i-1}}(s,u)-F_{t^n_{i}}(s,u),	~~ \Delta^n_i	F(t,z)=F(t^n_{i-1},z)-F({t^n_{i}},z) 
		\end{split}
	\end{equation}  
	for a general function $F_t$, $F(t)$, $F_t(u)$, $F_t(s,u)$ or $F(t,z)$ and define
	\begin{equation}\label{eq:vp-var} 
		\vp_t(u)=\vp_t(t,u).
	\end{equation}

	\begin{proposition}\label{prop:incr}
		Under Assumption~\ref{ass:main}, there is a finite number of It\^o semimartingales $v^{(k)}_t$ and $C^{(k)}_t(u)$, $k=1,\dots, K$,  such that $C^{(k)}_t(u)$ is uniformly bounded in $u$ on compacts and $\lvert \Delta^n_i v^{(k)}_t\rvert+\lvert \Delta^n_i C^{(k)}_t(u) \rvert= O^\uc (\sqrt{\Den})$, uniformly in $i$, and
		\begin{equation}\label{eq:incr2} \begin{split}
				\Delta^n_i \call_{t,T}(u)&= \Theta_{t^n_i,T^n_{i-1}}(u_{T^n_{i-1}})\bigl(\call^{n,i}_{t,T}(u) + \call^{\prime n,i}_{t,T}(u)\bigr) + \sum_{k=1}^K \Delta^n_i v^{(k)}_t C^{(k)}_t(u) T^n_{i-1}\\
				&\quad +O^\uc(T^{N/2}) + o^\uc( \sqrt{\Den}T + \Den/\sqrt{T})	\end{split}	\end{equation}
		for all $i$ such that $i\Den=O(T)$. In \eqref{eq:incr2}, $\Theta_{t,T}(u)$ is defined in \eqref{eq:phi}
		and $\call^{n,i}_{t,T}(u)=\sum_{j=1}^5 \call^{n,i,j}_{t,T}(u)$ where
		\begin{align*} 
			\call^{n,i,1}_{t,T}(u)&=  iu{\textstyle\sqrt{T^n_{i-1}}} \Bigl(\Delta^n_i \al_t+ \bigl(1+\tfrac12 (T^n_{i-1})^2 \chi^{(4)}_{t^n_i}(u_{T^n_{i-1}})\bigr) \Delta^n_i \xi^{(1)}_t(t^n_i, u_{T^n_{i-1}}) \\
			&\quad +\frac12 T^n_{i-1} \chi^{(3)}_{t^n_i}(u_{T^n_{i-1}})\Delta^n_i \si_t  \\
			&\quad  +\frac12 T^n_{i-1} \bigl(\Delta^n_i\xi^{(2)}_t(t^n_i,u_{T^n_{i-1}})+\Delta^n_i\xi^{(7)}_t(t^n_i,u_{T^n_{i-1}})+\Delta^n_i\xi^{(8)}_t(t^n_i,u_{T^n_{i-1}})\bigr) \Bigr),\\
			\call^{n,i,2}_{t,T}(u)&= - \frac12 u^2  \Bigl(\bigl(2\si_{t^n_i}(1+\tfrac12(T^n_{i-1})^2\chi^{(4)}_{t^n_i}(u_{T^n_{i-1}}))  + T^n_{i-1}\chi^{(1)}_{t^n_i}(u_{T^n_{i-1}})\bigr)\Delta^n_i \si_t  \\
			&\quad +  (T^n_{i-1})^2 \bigl(\chi^{(2)}_{t^n_i}(u_{T^n_{i-1}})+\si_{t^n_i}\chi^{(3)}_{t^n_i}(u_{T^n_{i-1}})\bigr)\Delta^n_i  \xi^{(1)}_t(t^n_i,u_{T^n_{i-1}})   \\
			& \quad   +T^n_{i-1}\Delta^n_i \xi^{(3)}_t(t^n_i,u_{T^n_{i-1}}) +T^n_{i-1}\si_{t^n_i}(\Delta^n_i\xi^{(4)}_t(t^n_i,u_{T^n_{i-1}})+\Delta^n_i\xi^{(6)}_t(t^n_i,u_{T^n_{i-1}}))\Bigr),\\
			\call^{n,i,3}_{t,T}(u)&= - \frac12iu^3{\textstyle\sqrt{T^n_{i-1}}}\si_{t^n_i} \Bigl(\bigl(2\si^\si_{t^n_i}+T^n_{i-1}\chi^{(2)}_{t^n_i}(u_{T^n_{i-1}})+T^n_{i-1}\chi^{(3)}_{t^n_i}(u_{T^n_{i-1}})\si_{t^n_i}\bigr)\Delta^n_i \si_t  \\
			& \quad  +\si_{t^n_i}\Delta^n_i \si^\si_t + T^n_{i-1}\chi^{(1)}_{t^n_i}(u_{T^n_{i-1}}) \Delta^n_i  \xi^{(1)}_t(t^n_i,u_{T^n_{i-1}})   + \Delta^n_i \xi^{(5)}_t(t^n_i, u_{T^n_{i-1}}) \Bigr), \\
			\call^{n,i,4}_{t,T}(u)&  = \frac12 u^4 T^n_{i-1} \si^2_{t^n_i}\Bigl( \chi^{(1)}_{t^n_i}(u_{T^n_{i-1}})\Delta^n_i \si_t + \si^\si_{t^n_i}\Delta^n_i  \xi^{(1)}_t(t^n_i, u_{T^n_{i-1}}) \Bigr),  \\
			\call^{n,i,5}_{t,T}(u)& =  \frac12 i u^5{\textstyle\sqrt{T^n_{i-1}}}\si_{t^n_i}^3\si^\si_{t^n_i}\Delta^n_i \si_t, \\
			\call^{\prime n,i}_{t,T}(u)&=  1-\frac{\Theta_{t^n_i, T^n_i}(u_{T^n_{i}})}{\Theta_{t^n_i,T^n_{i-1}}(u_{T^n_{i-1}})}  +\frac14iu^3\si^2_{t^n_i}\si^\si_{t^n_i}\frac{\Den}{\sqrt{T^n_{i-1}}}. 
		\end{align*}
		Furthermore, %$\log$ is the principal branch of the complex logarithm and 
		the $O^\uc$- and $o^\uc$-terms in \eqref{eq:incr2} are also uniform in  $i$.
	\end{proposition}

The proposition is proved in the appendix.

		\begin{proof}[Proof of Theorem~\ref{thm:incr-L}]
		We start by writing
		\begin{align*}
			\Delta^n_i \La_{t,T}(t^n_i,u_{T^n_{i-1}})&=\La_{t^n_{i-1},T^n_{i-1}}(t^n_i,u_{T^n_{i-1}}) - \La_{t^n_{i},T^n_{i-1}}(t^n_i,u_{T^n_{i-1}}) \\
			&\quad+\La_{t^n_{i},T^n_{i-1}}(t^n_i,u_{T^n_{i-1}}) -\La_{t^n_{i},T^n_{i}}(t^n_i,u_{T^n_{i}})	 \\
			&=D^{n,i,1}_{t,T}(u)+D^{n,i,2}_{t,T}(u).
		\end{align*}
		By a first-order approximation, we have as a direct consequence of \eqref{eq:La} that 
\begin{equation}\label{eq:help-5} \begin{split}
		D^{n,i,1}_{t,T}(u)	&=\La_{t^n_i,T^n_{i-1}}(t^n_i,u_{T^n_{i-1}})\bigg( iu{\textstyle\sqrt{T^n_{i-1}}}\Delta^n_i\al_t-u^2\si_{t^n_i}\Delta^n_i\si_t + T^n_{i-1}\Delta^n_i\vp_t(t^n_i,u_{T^n_{i-1}}) \\
	&\qquad-\frac12iu^3{\textstyle\sqrt{T^n_{i-1}}}(2\si_{t^n_i}\si^\si_{t^n_i}\Delta^n_i\si_t+\si^2_{t^n_i}\Delta^n_i \si^\si_t)\\
	&\qquad+ (T^n_{i-1})^2(\Delta^n_i\psi_t(t^n_i,u_{T^n_{i-1}})+\Delta^n_i\ov\psi_t(t^n_i,u_{T^n_{i-1}})+\Delta^n_i\wt\psi_t(t^n_i,u_{T^n_{i-1}}))\biggr) \\
	&\quad+O^\uc(\Den). 
\end{split}\raisetag{-3.5\baselineskip}\end{equation} 
		By another first-order  approximation of $\La_{t^n_i,T^n_{i-1}}(t^n_i,u_{T^n_{i-1}})$, we have that
		\begin{equation}\label{eq:help-4} 
			\La_{t^n_i,T^n_{i-1}}(t^n_i,u_{T^n_{i-1}})=\Theta_{t^n_i,T^n_{i-1}}(t^n_i,u_{T^n_{i-1}})(1-\eta_{t^n_i,T^n_{i-1}}(u_{T^n_{i-1}}))+O^\uc(T), 
		\end{equation}
		where $\eta_{t,T}(u)$ is defined in \eqref{eq:eta} and satisfies $\eta_{t^n_i,T^n_{i-1}}(u_{T^n_{i-1}})=O^\uc(\sqrt{T})$. Next, by similar arguments, we have the expansions
		\begin{align*}
			\Delta^n_i \psi_t(t^n_i,u_{T^n_{i-1}}) 
			&=-\frac12 u_{T^n_{i-1}}^2\Bigl(\chi^{(1)}_{t^n_i}(u_{T^n_{i-1}})\Delta^n_i \si_t +\si_{t^n_i}\Delta^n_i \xi^{(4)}_t(t^n_i,u_{T^n_{i-1}})\\
			&\quad+ iu_{T^n_{i-1}}\si_{t^n_i} \Delta^n_i\xi^{(5)}_t(t^n_i,u_{T^n_{i-1}})\Bigr)+O^\uc(\Den/T^2),\\
			\Delta^n_i \ov\psi_t(t^n_i,u_{T^n_{i-1}})	&=\frac12 iu_{T^n_{i-1}}\Bigl( \Delta^n_i \xi^{(2)}_t(t^n_i,u_{T^n_{i-1}}) + iu_{T^n_{i-1}}\Delta \xi^{(3)}_t(t^n_i,u_{T^n_{i-1}}) \Bigr)+O^\uc(\Den/T^2),\\
			\Delta^n_i \wt\psi_t(t^n_i,u_{T^n_{i-1}})	&= \frac12 iu_{T^n_{i-1}}\Bigl(\chi^{(3)}_{t^n_i}(u_{T^n_{i-1}})\Delta^n_i\si_t + \Delta^n_i\xi^{(8)}_t(t^n_i,u_{T^n_{i-1}}) \\
			&\quad+iu_{T^n_{i-1}}\si_{t^n_i}   \Delta^n_i \xi^{(6)}_t(t^n_i,u_{T^n_{i-1}})+ \Delta^n_i \xi^{(7)}_t(t^n_i,u_{T^n_{i-1}})\Bigr)+O^\uc(\Den/T^2).
		\end{align*}
		Combining this with \eqref{eq:help-5} and \eqref{eq:help-4}, we can verify after some  algebraic manipulations that 
		\begin{equation}\label{eq:Dni1} \begin{split}
			D^{n,i,1}_{t,T}(u)	&=\Theta_{t^n_i,T^n_{i-1}}(t^n_i,u_{T^n_{i-1}})\call^{ n,i}_{t,T}(u) +  \text{``$\Delta^n_i v_t C_{t^n_i}(u) T^n_{i-1}$''}\\
			&\quad + o^\uc(\sqrt{\Den}T+\Den/\sqrt{T}),\end{split}
		\end{equation}
		where $\call^{ n,i}_{t,T}(u)$ was defined in Proposition~\ref{prop:incr} and, as in the proof of Proposition~\ref{prop:incr}, the quotation marks indicate a sum of terms of the form $\Delta^n_i v_t C_{t^n_i}(u) T^n_{i-1}$.
		
		Next, as $\Den/T\to0$ by hypothesis,
		\begin{align*}
			&D^{n,i,2}_{t,T}(u)\\
				&\quad=-\La_{t^n_i,T^n_{i-1}}(t^n_i,u_{T^n_{i-1}})\biggl(\frac12 iu\al_{t^n_i} \frac{\Den}{\sqrt{T^n_{i-1}}}-\frac14iu^3\si^2_{t^n_i}\si^\si_{t^n_i}\frac{\Den}{\sqrt{T^n_{i-1}}} +\Bigl(T^n_i\vp_{t^n_i}(u_{T^n_i})\\
			&\qquad-T^n_{i-1}\vp_{t^n_i}(u_{T^n_{i-1}})\Bigr) + \Bigl( (T^n_i)^2(\psi_{t^n_i}(t^n_i,u_{T^n_i})+\ov\psi_{t^n_i}(t^n_i,u_{T^n_i})+\ov\psi_{t^n_i}(t^n_i,u_{T^n_i})) \\
			&\qquad-(T^n_{i-1})^2(\psi_{t^n_i}(t^n_i,u_{T^n_{i-1}})+\ov\psi_{t^n_i}(t^n_i,u_{T^n_{i-1}})+\ov\psi_{t^n_i}(t^n_i,u_{T^n_{i-1}}))\Bigr)\biggr)+ o^\uc(\Den/\sqrt{T})\\
			&\quad=-\La_{t^n_i,T^n_{i-1}}(t^n_i,u_{T^n_{i-1}})\biggl(\frac12 iu\al_{t^n_i} \frac{\Den}{\sqrt{T^n_{i-1}}}-\frac14iu^3\si^2_{t^n_i}\si^\si_{t^n_i}\frac{\Den}{\sqrt{T^n_{i-1}}} \\
			&\qquad+\Bigl(T^n_i\vp_{t^n_i}(u_{T^n_i})-T^n_{i-1}\vp_{t^n_i}(u_{T^n_{i-1}})\Bigr) + o^\uc(\Den/\sqrt{T}).
		\end{align*}
		By \eqref{eq:help-4} and the fact that $T^n_i\vp_{t^n_i}(u_{T^n_i})-T^n_{i-1}\vp_{t^n_i}(u_{T^n_{i-1}})=o^\uc(\Den/T)$, we may exchange $\La_{t^n_i,T^n_{i-1}}(t^n_i,u_{T^n_{i-1}})$ for $\Theta_{t^n_i,T^n_{i-1}}(t^n_i,u_{T^n_{i-1}})$ in the last line (at a cost of $o^\uc(\Den/\sqrt{T})$). This shows that $D^{n,i,2}_{t,T}(u)=\Theta_{t^n_i,T^n_{i-1}}(t^n_i,u_{T^n_{i-1}})\call^{\prime n,i}_{t,T}(u)$, which in conjunction with \eqref{eq:help-5} and Proposition~\ref{prop:incr} completes the proof of Theorem~\ref{thm:incr-L}.
	\end{proof}
	
	\begin{proof}[Proof of Theorem~\ref{thm:incr}]
		%	Since $\si$, $\ga$, $\ga^\si$, $\si^\ga$ and $\ga^\ga$ are all It\^o semimartingales, so are $\psi_t(u)$, $\psi^{(21)}_t(u)$ and $\psi^{(22)}_t(u)$ and ultimately $\Psi_{t,T}(u)$ and $\ov\Psi_{t,T}(u)$. The tightness of $\{(\Psi_{t,T}(u))_{t\geq0},(\ov\Psi_{t,T}(u))_{t\geq0}\}$ follows  from the pointwise bounds
		%	\begin{align*}
		%		\lvert T\psi_t(u)\rvert	&\leq \frac12u^2\si_t^2 + u^2\int_\R \ga(t,z)^2 \la(dz), \\
		%		\lvert T^2\psi^{(21)}_t(u)\rvert& \leq \frac12 T^{1/2} u^3\lvert \si_t\rvert \int_R \lvert \ga(t,z)(\si^\ga(t,z)+\ga^\si(t,z)\rvert \la(dz),\\
		%		\lvert T^2\psi^{(22)}_t(u)\rvert& \leq \frac12 T^{1/2} u^3  \int_\R\int_\R \lvert \ga(t,z)\ga(t,z')\ga^\ga(t,z,z')\rvert\la(dz)\la(dz')
		%	\end{align*}
		%	and \eqref{eq:int}.
		%	
		Recall the notation introduced at the beginning of Section~\ref{sec:proofs}.
		We first show that
		\begin{equation}\label{eq:incr-si-1}\begin{split}
				\Delta^{n}_i \si^2_{t,T}(u)&= -\frac{2}{u^2} \Re\biggl\{iu  {\textstyle\sqrt{T^n_{i-1}} }\Bigl(\Delta^n_i \xi^{(1)}_t(t^n_i, u_{T^n_{i-1}})+\frac12T^n_{i-1}\chi^{(3)}_{t^n_i}(u_{T^n_{i-1}})\Delta^n_i\si_t\\
				&\quad\quad +\frac12 T^n_{i-1}(\Delta^n_i \xi^{(2)}_t(t^n_i,u_{T^n_{i-1}})+\Delta^n_i \xi^{(7)}_t(t^n_i,u_{T^n_{i-1}})+\Delta^n_i \xi^{(8)}_t(t^n_i,u_{T^n_{i-1}})) \Bigr) \\
				& \quad- \frac12u^2\Bigl((2\si_{t^n_i}  + T^n_{i-1}\chi^{(1)}_{t^n_i}(u_{T^n_{i-1}}))\Delta^n_i \si_t +T^n_{i-1}\Delta^n_i \xi^{(3)}_t(t^n_i,u_{T^n_{i-1}})  \\
				& \quad\quad   +T^n_{i-1}\si_{t^n_i}(\Delta^n_i\xi^{(4)}_t(t^n_i,u_{T^n_{i-1}})+\Delta^n_i\xi^{(6)}_t(t^n_i,u_{T^n_{i-1}}))\Bigr)\\
				&\quad-\frac12  iu^3\textstyle \sqrt{T^n_{i-1}}\si_{t^n_i}   \Delta^n_i \xi^{(5)}_t(t^n_i, u_{T^n_{i-1}}) 
				%-\psi_{t^n_i}(u_{T^n_{i-1}})\Den\biggr\}+ o^\uc(\sqrt{\Den/k_n}).
				-\Bigl(T^n_i\vp_{t^n_i}(u_{T^n_i})-T^n_{i-1}\vp_{t^n_i}(u_{T^n_{i-1}})\Bigr)\biggr\}\\
				&\quad + \text{``$\Delta^n_i v_t C_{t^n_i}(u) T^n_{i-1}$''}+O^\uc(T^{N/2})+o^\uc(\sqrt{\Den}T+\Den/\sqrt{T}),
		\end{split}\raisetag{-5\baselineskip}\end{equation}
		where, as before, the quotation marks signify a sum of terms of the form $\Delta^n_i v_t C_{t^n_i}(u) T^n_{i-1}$.
		Once \eqref{eq:incr-si-1} is established, we can
		reverse engineer \eqref{eq:incr-si-1} into \eqref{eq:incr-si-0} by using the expansions derived after \eqref{eq:eta} and noting that (recall the notation \eqref{eq:vp-var})
		\begin{align*}
			%	\Delta^n_i \si^2_t &= 2\si_{t^n_i}\Delta^n_i\si_t + o(\Den),\\
			&	\Delta^n_i \Phi_{t,T}(t^n_i,u)\\
			&	\quad	=-\frac{2}{u^2}\Re\Bigl(T^n_{i-1}\Delta^n_i \vp_t(t^n_i,u_{T^n_{i-1}})  -\bigl(T^n_i\vp_{t^n_i}(u_{T^n_i})-T^n_{i-1}\vp_{t^n_i}(u_{T^n_{i-1}})\bigr)\Bigr)\\
			&	\quad	=-\frac{2}{u^2}\Re \Bigl(iu{\textstyle{\sqrt{T^n_{i-1}}}}\Delta^n_i \xi^{(1)}_t(t^n_i,u_{T^n_{i-1}})  -(T^n_i\vp_{t^n_i}(u_{T^n_i})-T^n_{i-1}\vp_{t^n_i}(u_{T^n_{i-1}})) \Bigr)+O^\uc(\Den),\\
			&\Delta^n_i \Psi_{t,T}(t^n_i,u)	\\
			&\quad=-\frac{2T^n_{i-1}}{u_{T^n_{i-1}}^2}\Re\Bigl( \Delta^n_i\psi_{t}(t^n_i,u_{T^n_{i-1}})+\Delta^n_i\ov\psi_{t}(t^n_i,u_{T^n_{i-1}})+\Delta^n_i\wt\psi_{t}(t^n_i,u_{T^n_{i-1}}) \Bigr)  +o^\uc(\Den/\sqrt{T}).
			%	&=-\frac{2(T^n_{i-1})^2}{u^2}\Re\Bigl( \Delta^n_i \psi_t(t^n_i,u_{T^n_{i-1}})+\Delta^n_i\ov\psi_t(t^n_i,u_{T^n_{i-1}})+\Delta^n_i\wt\psi_t(t^n_i,u_{T^n_{i-1}}) \Bigr) \\
			%	&\quad+o^\uc(\Den/\sqrt{T})
		\end{align*}
		%	we use \eqref{eq:rel} and Taylor's theorem to approximate
		%		\begin{equation}\label{eq:incr-Th-0} 
		%			\begin{split}
		%				\Delta^n_i \Psi_{t,T}(t,u_T) 	&= -\frac2{u^2}\Re\Bigl\{(T+(t-t^n_i))\Delta^n_i\psi_t (u_T)-\psi_{t^n_i}(u_T)\Den\\
		%				&\quad + (T+(t-t^n_i))^2\bigl(\Delta^n_i\ov\psi_t(u_T) + \Delta^n_i\psi^{(3)}_t(u_T)\bigr)\\
		%				&\quad+2(T+(t-t^n_i))\bigl( \ov\psi_{t^n_i}(u_T) + \psi^{(3)}_{t^n_i}(u_T)\bigr)\Den\Bigr\}+ o^\uc(\sqrt{\Den/k_n})\\
		%				&=-\frac2{u^2}\Re\Bigl\{T^n_{i}\Delta^n_i\psi_t (u_T)-\psi_{t^n_i}(u_T)\Den\\
		%				&\quad+ (T^n_{i})^2\bigl(\Delta^n_i\ov\psi_t(u_T) + \Delta^n_i\psi^{(3)}_t(u_T)\bigr)\Bigr\}+ o^\uc(\sqrt{\Den/k_n})\\
		%				&=-\frac2{u^2}\Re\Bigl\{T^n_{i-1}\Delta^n_i\psi_t (u_T)-\psi_{t^n_i}(u_T)\Den\\
		%				&\quad+ (T^n_{i-1})^2\bigl(\Delta^n_i\ov\psi_t(u_T) + \Delta^n_i\psi^{(3)}_t(u_T)\bigr)\Bigr\}+ o^\uc(\sqrt{\Den/k_n}).
		%			\end{split}
		%		\end{equation}
		%		In a second step,  

		Proving \eqref{eq:incr-si-1} is therefore our main task.
		By Lemma~\ref{lem:exp},
		\begin{equation}\label{eq:L-exp}\begin{split}
				\call_{t^n_{i},T^n_{i}} (u ) &=\Theta_{t^n_i,T^n_i}(u_{T^n_i})(1-\eta_{t^n_i,T^n_i}(u_{T^n_i}))+C_{t^n_i}(u)T^n_i+o^\uc(T)\\
				&=e^{- \frac12u^2\si_{t^n_i}^2}+o^\uc(1),
			\end{split}
		\end{equation}
		where $C_t(u)$ is an It\^o semimartingale in $t$ and a polynomial in $u$ and $\eta_{t,T}(u)$ is defined in \eqref{eq:eta}.
		The last equality in \eqref{eq:L-exp} implies that $\lvert \call_{t^n_{i},T^n_{i}} (u )\rvert<1$ uniformly in $u$, $n$ and $i$ for all sufficiently small values of $T$. 
		
		%Next, let $\Log$ be the principal branch of the complex logarithm, which for every $r<1$ satisfies 
		%\[ \biggl\lvert\Log (1+z) - \sum_{i=1}^{k} \frac{(-1)^{k+1}}{k}z^k\biggr\rvert = O(\lvert z\rvert^{k+1}),\quad z\to0, \]
		%with a constant that depends on $r$ (but of course not on $z$). Therefore,
		%\begin{align*}
		%	\Delta^n_i \si^2_{t,T}(u)&=-\frac{2}{u^2} \Bigl(\log \Bigl\lvert \call_{t^n_{i-1},T^n_{i-1}}\bigl(u_T\sqrt{ {T^n_{i-1}}}\bigr)\Bigr\rvert-\log\Bigl \lvert \call_{t^n_{i},T^n_{i}}\bigl(u_T\sqrt{ {T^n_{i}}}\bigr)\Bigr\rvert\Bigr)\\
		%	&=\frac{2}{u^2} \Re\Bigl(   \Log \call_{t^n_{i},T^n_{i}}\bigl(u_T\sqrt{ {T^n_{i}}}\bigr)-\Log   \call_{t^n_{i-1},T^n_{i-1}}\bigl(u_T\sqrt{ T^n_{i-1}}\bigr)\Bigr) \\
		%	&=\frac{2}{u^2} \Re\Log \biggl(1-\frac{ \Delta^n_i \call_{t,T}(u)}{   \call_{t^n_{i-1},T^n_{i-1}}\bigl(u_T\sqrt{ T^n_{i-1}}\bigr)} \biggr)\\ 
		%	&=-\frac{2}{u^2}\Re \frac{\Delta^n_i \call_{t,T}(u)}{   \call_{t^n_{i-1},T^n_{i-1}}\bigl(u_T\sqrt{ T^n_{i-1}}\bigr)}-\frac{2}{u^2}\Re \sum_{j=2}^{N'} \frac{1}{j}\biggl(\frac{ \Delta^n_i \call_{t,T}(u)}{   \call_{t^n_{i-1},T^n_{i-1}}\bigl(u_T\sqrt{ T^n_{i-1}}\bigr)}\biggr)^j \\
		%	&\quad+ o^\uc(\sqrt{\Den/k_n}).
		%\end{align*}
		%In the last step, we have used the fact that $\Delta^n_i \call_{t,T}(u) = O^\uc(\sqrt{\Den} \vee \frac{\Den}{T})$ by Proposition~\ref{prop:incr} and hence, $(\Delta^n_i \call_{t,T}(u))^{N'+1}= o^\uc(\sqrt{\Den/k_n})$ by \eqref{eq:rel}.
		Next, let $\Log$ be the principal branch of the complex logarithm, which for  $\lvert z\rvert<1$ has the series representation $\Log (1+z) =\sum_{j=1}^{\infty} \frac{(-1)^{j+1}}{j}z^j$.
		Therefore,
		\begin{equation}\label{eq:Delta-si} 
			\begin{split}
				\Delta^n_i \si^2_{t,T}(u)&=-\frac{2}{u^2} \Bigl(\log \bigl\lvert \call_{t^n_{i-1},T^n_{i-1}} (u )\bigr\rvert-\log\bigl \lvert \call_{t^n_{i},T^n_{i}} (u  )\bigr\rvert\Bigr)\\
				&=\frac{2}{u^2} \Re\Bigl(   \Log \call_{t^n_{i},T^n_{i}} (u )-\Log   \call_{t^n_{i-1},T^n_{i-1}} (u  )\Bigr) \\
				&=\frac{2}{u^2} \Re\Log \biggl(1-\frac{ \Delta^n_i \call_{t,T}(u)}{   \call_{t^n_{i-1},T^n_{i-1}}(u)} \biggr) =-\frac{2}{u^2}\Re \sum_{j=1}^{\infty} \frac{1}{j}\biggl(\frac{ \Delta^n_i \call_{t,T}(u)}{   \call_{t^n_{i-1},T^n_{i-1}} (u)}\biggr)^j .
			\end{split}	\raisetag{-2.5\baselineskip}
		\end{equation}
		%		In the last step, we have used the fact that 
		%		\begin{equation}\label{eq:Delta-L} 
		%			\Delta^n_i \call_{t,T}(u) = O^\uc(\sqrt{\Den} \vee \frac{\Den}{T})
		%		\end{equation} 
		%	by Proposition~\ref{prop:incr} and hence, $(\Delta^n_i \call_{t,T}(u))^{N'+1}= o^\uc(\sqrt{\Den/k_n})$ by \eqref{eq:rel}.
		
		By \eqref{eq:chi}, we have that $\eta_{t,T}(u_T)=O^\uc(T^{1/2})$ and thus,
		\begin{align*}
			\frac{ \Delta^n_i \call_{t,T}(u)}{   \call_{t^n_{i-1},T^n_{i-1}} (u )} &= \frac{ \Delta^n_i \call_{t,T}(u)}{\Theta_{t^n_{i-1},T^n_{i-1}}(u_{T^n_{i-1}})}+\frac{\Delta^n_i \call_{t,T}(u)\eta_{t^n_{i-1},T^n_{i-1}}(u_{T^n_{i-1}})}{ \Theta_{t^n_{i-1},T^n_{i-1}}(u_{T^n_{i-1}}) }\\
			&\quad+ \frac{\Delta^n_i \call_{t,T}(u)(\eta_{t^n_{i-1},T^n_{i-1}}(u_{T^n_{i-1}}))^2}{ \Theta_{t^n_{i-1},T^n_{i-1}}(u_{T^n_{i-1}}) }+o^\uc(\sqrt{\Den}T)\\
			&=\frac{ \Delta^n_i \call_{t,T}(u)}{\Theta_{t^n_{i},T^n_{i-1}}(u_{T^n_{i-1}})}+\frac{\Delta^n_i \call_{t,T}(u)\eta_{t^n_{i-1},T^n_{i-1}}(u_{T^n_{i-1}})}{ \Theta_{t^n_{i},T^n_{i-1}}(u_{T^n_{i-1}}) }\\
			&\quad-\frac14 u^6T^n_{i-1}\si^4_{t^n_i}(\si^\si_{t^n_i})^2e^{\frac12 u^2\si^2_{t^n_i}}\Delta^n_i \call_{t,T}(u)+o^\uc(\sqrt{\Den}T).
		\end{align*}
		By Proposition~\ref{prop:incr}, it follows that 
		\begin{align*}
			\frac{ \Delta^n_i \call_{t,T}(u)}{   \call_{t^n_{i-1},T^n_{i-1}} (u )} 
			&=(\call^{n,i}_{t,T}(u)+ \call^{\prime n,i}_{t,T}(u)) (1+\eta_{t^n_{i-1},T^n_{i-1}}(u_{T^n_{i-1}}))\\
			&\quad + \text{``$\Delta^n_i v_t C_{t^n_i}(u) T^n_{i-1}$''} +O^\uc(T^{N/2})+o^\uc(\sqrt{\Den}T+\Den/\sqrt{T}).\end{align*}
		Furthermore, 
		$
		\call^{n,i}_{t,T}(u) = O^\uc(\sqrt{\Den})$ and $\call^{\prime n,i}_{t,T}(u)=o^\uc(\Den/T)$, %and $\Delta^n_i \call_{t,T}(u)=O^\uc(\sqrt{\Den}+\Den/T)$, 
		so
		the last   display is $\call^{\prime n,i}_{t,T}(u)+O^\uc(\sqrt{\Den}+T^{N/2}+ \Den/\sqrt{T})$. Hence, by \eqref{eq:Delta-si},
		\begin{equation}\label{eq:incr-si}\begin{split}
				\Delta^n_i \si^2_{t,T}(u)&=-\frac{2}{u^2}\Re\biggl((\call^{n,i}_{t,T}(u)+ \call^{\prime n,i}_{t,T}(u)) (1+\eta_{t^n_{i-1},T^n_{i-1}}(u_{T^n_{i-1}})) + \sum_{j=2}^{\infty} \frac{(\call^{\prime n,i}_{t,T}(u))^j}{j}\biggr)\\
				&\quad+ \text{``$\Delta^n_i v_t C_{t^n_i}(u) T^n_{i-1}$''} +O^\uc(T^{N/2})+o^\uc(\sqrt{\Den}T+\Den/\sqrt{T}). \end{split}\raisetag{-2.5\baselineskip}
		\end{equation}
		
		%Next we take double increments and notice that
		%\begin{equation}\label{eq:incr-prime}\begin{split}
		%&\call^{\prime n,i+1}_{t,T}(u)-\call^{\prime n,i}_{t,T}(u)\\
		%&\qquad=\Theta_{t^n_{i+1},\Den}(u_T)-\Theta_{t^n_{i},\Den}(u_T) -iu_T^3\Den\bigl(T^n_i\si^2_{t^n_{i+1}}\si^\si_{t^n_{i+1}}-T^n_{i-1}\si^2_{t^n_{i}}\si^\si_{t^n_{i}}\bigr)  \\
		%&\quad\qquad+\tfrac14 i u_T^5\Den\bigl((T^n_i)^2\si^4_{t^n_{i+1}}\si^\si_{t^n_{i+1}}-T^n_{i-1}\si^4_{t^n_{i}}\si^\si_{t^n_{i}}\bigr)\\
		%&\qquad=O^\uc(\tfrac{\Den^{3/2}}{T}\vee \tfrac{\Den^2}{T^{3/2}})= o^\uc(\sqrt{\Den/k_n}),
		%\end{split}\end{equation}
		%which shows that double increments of the sum over $j$ in \eqref{eq:incr-si} are negligible.
		Next, 	%and
		%\begin{align*}
		% \Theta_{t^n_i, T^n_i}(u_T)-\Theta_{t^n_{i-1}, T^n_{i-1}}(u_T)&= O^\uc(\tfrac{\Den}{T}\vee \sqrt{\Den}), \\
		% \eta_{t^n_i, T^n_i}(u_T)-\eta_{t^n_{i-1}, T^n_{i-1}}(u_T)&= O^\uc(\tfrac{\Den}{\sqrt{T}}\vee \sqrt{T\Den}). 
		%\end{align*}
		by definition (see Proposition~\ref{prop:incr}), we have that
		\begin{equation}\label{eq:help}
			\begin{split}
				\call^{n,i}_{t,T}(u)\eta_{t^n_{i-1},T^n_{i-1}}(u_{T^n_{i-1}}) 	&= \call^{n,i}_{t,T}(u)\eta_{t^n_{i},T^n_{i-1}}(u_{T^n_{i-1}}) +O^\uc(\Den)	\\
				&=\bigl(iu{\textstyle\sqrt{T^n_{i-1} }}\Delta^n_i \xi^{(1)}_t(t^n_i,u_{T^n_{i-1}})-u^2 \si_{t^n_i}\Delta^n_i \si_t\bigr)\eta_{t^n_{i},T^n_{i-1}}(u_{T^n_{i-1}}) \\
				&\quad+ \text{``$\Delta^n_i v_t C_{t^n_i}(u) T^n_{i-1}$''}+o^\uc(\sqrt{\Den}T + \Den/\sqrt{T}) 
			\end{split}\raisetag{-2.15\baselineskip}\end{equation}
		and $	\call^{\prime n,i}_{t,T}(u)\eta_{t^n_{i-1},T^n_{i-1}}(u_{T^n_{i-1}}) =o^\uc(\Den/\sqrt{T})$.
		Expanding the product in the second line of \eqref{eq:help}, removing  pure imaginary terms and dropping all expressions that are negligible, we obtain \eqref{eq:incr-si-1} from  the definition of $\call^{n,i}_{t,T}(u)$, \eqref{eq:incr-si} and the fact that 
		\begin{align*} \call^{\prime n,i}_{t,T}(u) +\sum_{j=2}^\infty\frac{(\call^{\prime n,i}_{t,T}(u))^j}{j} &= -\Log\bigl(1-\call^{\prime n,i}_{t,T}(u)\bigr) \\
			&=- iu\al_{t^n_i}(\sqrt{T^n_i}-\textstyle\sqrt{T^n_{i-1}}) - \Bigl( T^n_i\vp_{t^n_i}(u_{T^n_i})-T^n_{i-1}\vp_{t^n_i}(u_{T^n_{i-1}})\Bigr)\\
			&\quad- \frac14iu^3\si^2_{t^n_i}\si^\si_{t^n_i}\frac{\Den}{\sqrt{T^n_{i-1}}}+o^\uc(\Den/\sqrt{T}).\end{align*}

		It remains to verify \eqref{eq:incr-si-2} under Assumption~\ref{ass:main-1}. Recall the expansions we derived right after \eqref{eq:eta} and \eqref{eq:incr-si-1}. Under Assumption~\ref{ass:main}, it is easy to show that
		\begin{align*}
			\bigl\lvert  \chi^{(1)}_{t^n_i}(u_{T^n_{i-1}})\bigr\rvert 	 & \leq u_{T^n_{i-1}} \biggl(\int_{\R^{d'}}   \lvert\ga(t^n_i,z)\rvert^2\ov F(dz)\int_{\R^{d'}}   \bigl \lvert\si^\ga(t^n_i,z)+\ga^\si(t^n_i,z)\bigr\rvert^2  \ov F(dz)\biggr)^{1/2} \\
			&=O^\uc(1/\sqrt{T}),\\
			\bigl\lvert  \chi^{(3)}_{t^n_i}(u_{T^n_{i-1}})\bigr\rvert 	 & \leq \frac12u^2_{T^n_{i-1}} \int_{\R^{d'}}     \lvert\ga(t^n_i,z)\rvert^2 \ov F(dz) =O^\uc(1/T)
		\end{align*}
		as well as
		\begin{align*}
			\E_{t^n_i}\bigl[\bigl\lvert\Delta^n_i \xi^{(1)}_t(t^n_i,u_{T^n_{i-1}})\bigr\rvert\bigr] & \leq      \int_{\R^{d'}}   \E_{t^n_i}\bigl[ \lvert \Delta^n_i \ga(t,z) \rvert\bigr] \ov F(dz) =O^\uc(  \sqrt{\Den}),\\
			\E_{t^n_i}\bigl[\bigl\lvert\Delta^n_i \xi^{(2)}_t(t^n_i,u_{T^n_{i-1}})\bigr\rvert\bigr]& \leq u^2_{T^n_{i-1}}  \int_{\R^{d'}}   \lvert \ga(t^n_i,z)\rvert^2\ov F(dz) \\ 
			&\qquad\qquad \times\biggl(\int_{\R^{d'}}  \int_{\R^{d'}}    \E_{t^n_i}\bigl[\lvert\Delta^n_i \ga^\ga(t,z,z') \rvert^2\bigr]\ov F(dz)\ov F(dz') \biggr)^{1/2}\\
			&=O^\uc(  \sqrt{\Den}/T),\\
			%\leq u_{T^n_{i-1}} \int_\R \lvert \ga(t^n_i,z')\rvert \E_{t^n_i}\biggl[ \int_\R\lvert \Delta^n_i \ga^\ga(t,z,z') \rvert\la(dz) \biggr] \la(dz')\\
			% & \leq u_{T^n_{i-1}} \Den^{1/2}\int_\R \lvert \ga(t^n_i,z')\rvert C_{t^n_i}(z')\la(dz')\\
			%& \leq u_{T^n_{i-1}} \Den^{1/2}\biggl(\int_\R \lvert \ga(t^n_i,z')\rvert^2 \la(dz')\int_\R  ( C_{t^n_i}(z') )^2\la(dz')\biggr)^{1/2}\\
			%&\leq C_{t^n_i}u_{T^n_{i-1}} \Den^{1/2},\\
			\E_{t^n_i}\bigl[\bigl\lvert\Delta^n_i \xi^{(3)}_t(t^n_i,u_{T^n_{i-1}})\bigr\rvert\bigr] & \leq u_{T^n_{i-1}}\biggl(\int_{\R^{d'}}  \int_{\R^{d'}}     \lvert\ga^\ga(t^n_i,z,z')+\ga^\ga(t^n_i,z',z) \rvert^2\ov F(dz)\ov F(dz') \\
			&\quad \times \int_{\R^{d'}}    \lvert\ga(t^n_i,z) \rvert^2\ov F(dz)\int_{\R^{d'}}   \E_{t^n_i}\bigl[ \lvert\Delta^n_i\ga(t,z') \rvert^2\bigr]\ov F(dz')\biggr)^{1/2}\\
			&=O^\uc(  \sqrt{\Den/T}),\\
			\E_{t^n_i}\bigl[\bigl\lvert\Delta^n_i \xi^{(4)}_t(t^n_i,u_{T^n_{i-1}})\bigr\rvert\bigr]  & \leq u_{T^n_{i-1}} \biggl(\int_{\R^{d'}}   \lvert \ga(t^n_i,z)\rvert^2\ov F(dz) \\
			&\qquad\qquad \times \int_{\R^{d'}}   \E_{t^n_i}\bigl[\lvert \Delta^n_i\si^\ga(t,z)+\Delta^n_i \ga^\si(t,z)\rvert^2\bigr]\ov F(dz)\biggr)^{1/2}\\
			&=O^\uc(  \sqrt{\Den/T}),\\
			\E_{t^n_i}\bigl[\bigl\lvert\Delta^n_i \xi^{(5)}_t(t^n_i,u_{T^n_{i-1}})\bigr\rvert\bigr]  & \leq \biggl(\int_{\R^{d'}}   \bigl \lvert\si^\ga(t^n_i,z)+\ga^\si(t^n_i,z)\bigr\rvert^2  \ov F(dz)   \\
			&\qquad\qquad\qquad\qquad\qquad  \times \int_{\R^{d'}}   \E_{t^n_i}\bigl[\lvert \Delta^n_i \ga(t,z)\rvert^2\bigr]\ov F(dz)\biggr)^{1/2}\\
			&=O^\uc(  \sqrt{\Den}),\\
			\E_{t^n_i}\bigl[\bigl\lvert\Delta^n_i \xi^{(6)}_t(t^n_i,u_{T^n_{i-1}})\bigr\rvert\bigr]  & \leq  u_{T^n_{i-1}}\biggl(\int_{\R^{d'}}   \bigl \lvert \ga (t^n_i,z)\bigr\rvert^2  \ov F(dz)  \int_{\R^{d'}}   \E_{t^n_i}\bigl[\lvert \Delta^n_i \ga(t,z)\rvert^2\bigr]\ov F(dz)\biggr)^{1/2}\\
			&=O^\uc(  \sqrt{\Den/T})
		\end{align*}
		and
		\begin{align*}
			&\E_{t^n_i}\bigl[\bigl\lvert\Delta^n_i \xi^{(7)}_t(t^n_i,u_{T^n_{i-1}})\bigr\rvert\bigr]  +\E_{t^n_i}\bigl[\bigl\lvert\Delta^n_i \xi^{(8)}_t(t^n_i,u_{T^n_{i-1}})\bigr\rvert\bigr]\\
			& \leq  u_{T^n_{i-1}}^2\int_{\R^{d'}}   \bigl \lvert \ga (t^n_i,z)\bigr\rvert^2  \ov F(dz)\biggl(\int_{\R^{d'}\times\R}   j(z',v') \ov F(dz')dv' \int_{\R^{d'}}   \E_{t^n_i}\bigl[\lvert \Delta^n_i \ga(t,z)\rvert^2\bigr]\ov F(dz)\biggr)^{1/2}\\
			&=O^\uc(  \sqrt{\Den}/T).
		\end{align*}
		Moreover, under Assumption~\ref{ass:main-1},
		\begin{align*}
			\lvert T^n_i\vp_{t^n_i}(u_{T^n_i})-T^n_{i-1}\vp_{t^n_i}(u_{T^n_{i-1}})\rvert	&\leq T^n_{i-1}\bigl\lvert\Re\vp_{t^n_i}(u_{T^n_i})-\Re\vp_{t^n_i}(u_{T^n_{i-1}})\bigr\rvert+\Den\lvert\Re\vp_{t^n_i}(u_{T^n_i})\rvert \\
			&\leq \biggl(T^n_{i-1}\lvert u_{T^n_i}-u_{T^n_{i-1}}\rvert   + \Den u_{T^n_i}\biggr)\int_\R \lvert\ga(t^n_i,z)\rvert \ov F(dz)\\
			&\leq C_{t^n_i} \Den/\sqrt{T},
		\end{align*}
		which shows that $\sqrt{T}/\Den\lvert T^n_i\vp_{t^n_i}(u_{T^n_i})-T^n_{i-1}\vp_{t^n_i}(u_{T^n_{i-1}})\rvert = o^\uc(1)$ by  dominated convergence.
		Thus, from the expansions after \eqref{eq:eta} and \eqref{eq:incr-si-1}, we obtain $\Delta^n_i \Psi_{t,T}(t^n_i,u)+\Delta^n_i\Phi_{t,T}(t^n_i,u)=O^\uc(\sqrt{\Den T})+o^\uc(\Den/\sqrt{T})$, uniformly in $i$.
	\end{proof}

	\begin{proof}[Proof of Corollary \ref{cor:incr}] 
		By the mean-value theorem,
		\[ \Delta^n_i V_{t,T}(u) = F'(\si^2_{t^n_i,T^n_i}(u))\Delta^n_i \si^2_{t,T}(u) + O^\uc(\Den+T^N) \]
		By Theorem~1 in \cite{todorov2021bias}, we have under Assumption~\ref{ass:main-1} that $\si^2_{t^n_i,T^n_i}(u)=\si^2_{t^n_i} + o^\uc(T^{1/2})$, which in conjunction with Theorem~\ref{thm:incr} shows that 
		\begin{align*}
			\Delta^n_i V_{t,T}(u) &= F'(\si^2_{t^n_i})\Delta^n_i \si^2_{t,T}(u) + O^\uc(\Den+T^N+\sqrt{\Den T})\\
			&=F'(\si^2_{t^n_i})	\Delta^n_i \si^2_t + \sum_{k=1}^K \Delta^n_i v^{(k)}_t C^{(k)}_{t^n_i}(u) T^n_{i-1}\\
			&\quad+O^\uc(T^{N/2}+\sqrt{\Den T})  + o^\uc(  \Den/\sqrt{T}).
		\end{align*}
		The first equality in \eqref{eq:incr-V} follows now from the fact that $F'(\si^2_{t^n_i})	\Delta^n_i \si^2_t = \Delta^n_i V_t + O(\Den)$. Regarding the second equality, we use  \eqref{eq:V} to obtain
		\begin{align*}
			\Delta^n_i V_{t,T,T'}(u)	&=\frac{T^{\prime n}_{i-1}\Delta^n_i V_{t,T}(u)-T^n_{i-1}\Delta^n_i V_{t,T'}(u)+\Den(V_{t^n_i,T^{\prime n}_i}(u)-V_{t^n_i, T^n_i}(u))}{T'-T} \\
			&=\frac{(T^{\prime n}_{i-1}-T^n_{i-1})\Delta^n_i V_t + \Den(V_{t^n_i,T^{\prime n}_i}(u)-V_{t^n_i, T^n_i}(u))}{T'-T}\\
			&\quad + O^\uc(T^{N/2}+\sqrt{\Den T})+ o^\uc(\Den/\sqrt{T})\\
			&=\Delta^n_i V_t+ O^\uc(T^{N/2}+\sqrt{\Den T})+ o^\uc(\Den/\sqrt{T})
		\end{align*}
		(for the first equality, note that $T^{\prime n}_{i-1}-T^n_{i-1}=T'-T$; for second equality, notice that the terms $\sum_{k=1}^K \Delta^n_i v^{(k)}_t C^{(k)}_{t^n_i}(u) T^n_{i-1} $ from \eqref{eq:incr-si-2} perfectly cancel out; for the last equality, use that $V_{t,T}(u)=V_t+o^\uc(T^{1/2})$).
	\end{proof}
	
	\begin{proof}[Proof of Corollary~\ref{cor:finact}]
		It suffices to  consider the case $F(x)=x$, as the general case follows from an additional application of the mean-value theorem.	Since $\ga(t,z)=z_1$, we have $\si^\ga(t,z)=\ga^\ga(t,z,z')=0$ and consequently, $\chi^{(4)}_t(u)=0$ and $\xi^{(2)}_t(s,u)=\xi^{(3)}_t(s,u)=0$. Also, again because $\ga(t,z)=z_1$, we have that $\Delta^n_i \xi^{(j)}_t(s,u)=0$ for all $j=1,5,6,7,8$. Since both $x$ and $\si$ have summable jumps by assumption, we have 
		\begin{equation}\label{eq:1}\begin{split}
			&-\frac{2(T^n_{i-1})^2}{u^2}\Re\biggl(-\frac12 u_{T^n_{i-1}}^2\chi^{(1)}_{t^n_i}(u_{T^n_{i-1}})\Delta^n_i\si_t+\frac12 iu_{T^n_{i-1}}\chi^{(3)}_{t^n_i}(u_{T^n_{i-1}})\Delta^n_i \si_t\biggr) \\
			&\quad=T^n_{i-1}\Delta^n_i\si_t \Re\biggl(\int_{\R^{d'}} (e^{iu_{T^n_{i-1}}z_1}-1)\ga^\si(t^n_i,z)\la(t^n_i,z)dz - \int_{\R^{d'}} z_1\si^{\la,(1)}(t^n_i,z)dz\biggr)\\
			&\qquad+o^\uc(\sqrt{\Den}T) \\
			&\quad=C_{t^n_i} \Delta^n_i\si_t T^n_{i-1} + o^\uc(\sqrt{\Den}T),\end{split}\!\! 
		\end{equation}
		where the last step follows from the fact that $\int_{\R^{d'}} e^{iu_{T^n_{i-1}}z_1}\ga^\si(t^n_i,z)\la(t^n_i,z)dz =o^\uc(1)$ by the Riemann--Lebesgue lemma and the integrability of $\ga^\si(t^n_i,\cdot)\la(t^n_i,\cdot)$.  By the same reason,
		\begin{equation}\label{eq:2}\begin{split}
				&\frac{(T^n_{i-1})^2}{u^2}\Re\biggl(  u_{T^n_{i-1}}^2 \si_{t^n_i}\Delta^n_i\xi^{(4)}_t(t^n_i,u_{T^n_{i-1}})\biggr)\\
				&\quad=T^n_{i-1}\Re\biggl(\si_{t^n_i}\int_{\R^{d'}} (e^{iu_{T^n_{i-1}}z_1}-1)\Delta^n_i\ga^\si(t,z)\la(t^n_i,z)dz  \biggr)\\
				&\quad =-T^n_{i-1}\si_{t^n_i}\int_{\R^{d'}}  \Delta^n_i\ga^\si(t,z)\la(t^n_i,z)dz  + o^\uc(\sqrt{\Den}T). 
			\end{split}
		\end{equation}
		Next,
		\begin{align*}
			&-\frac{2}{u^2}\Re\Bigl(T^n_i\vp_{t^n_i}(u_{T^n_i})-T^n_{i-1}\vp_{t^n_i}(u_{T^n_{i-1}})\Bigr)\\	&\qquad=-\frac{2T^n_{i-1}}{u^2} \Re\biggl(\int_{\R^{d'}} (e^{iu_{T^n_i}z_1}-e^{iu_{T^n_{i-1}}z_1}-i(u_{T^n_i}-u_{T^n_{i-1}})z_1)\la(t^n_i,z)dz \\
			&\qquad\quad-\frac{2\Den}{u^2}\int_{\R^{d'}} (e^{iu_{T^n_i}z_1}-1-iu_{T^n_i}z_1)\la(t^n_i,z)dz\biggr)\\
			&\qquad=-\frac{2T^n_{i-1}}{u^2}  \int_{\R^{d'}} (\cos(u_{T^n_i}z_1)-\cos(u_{T^n_{i-1}}z_1))\la(t^n_i,z)dz+O^\uc(\Den),
		\end{align*}
		so by the mean-value theorem, there is $\tau^n_i\in[T^n_{i-1},T^n_i]$ such that the last integral is
		\[\frac{2T^n_{i-1}}{u^2} (u_{T^n_i}-u_{T^n_{i-1}}) \int_{\R^{d'}} \sin(u_{\tau^n_i}z_1)z_1\la(t^n_i,z)dz. \]
		The Riemann--Lebesgue lemma and the estimate $u_{T^n_i}-u_{T^n_{i-1}}=O^\uc(\Den/T^{3/2})$ imply that
		\[ -\frac{2}{u^2}\Re\Bigl(T^n_i\vp_{t^n_i}(u_{T^n_i})-T^n_{i-1}\vp_{t^n_i}(u_{T^n_{i-1}})\Bigr) = O^\uc(\Den) + o^\uc(\Den/\sqrt{T}). \]
		The assertions of the corollary now follow from Theorem~\ref{thm:incr} and the expansions of $\Delta^n_i\Phi_{t,T}(t^n_i,u)$ and $\Delta^n_i\Psi_{t,T}(t^n_i,u)$ derived after \eqref{eq:incr-si-1}, together with \eqref{eq:1} and \eqref{eq:2} and the fact that differencing with two maturities cancels out terms of the form $\Delta^n_i v^{(k)}_t C^{(k)}_t(u) T^n_{i-1}$ as well as the leading terms in \eqref{eq:1} and \eqref{eq:2}.
	\end{proof}

\section{Proof of Theorem~\ref{thm:test}}\label{sec:proof51}

 Throughout this section,  we write $\E=\E^{\overline \P}$. For the proof, we need two auxiliary lemmas, which we prove in the appendix.
 
 \begin{lemma}\label{lemma:bounds}
 	Suppose that Assumptions \ref{ass:main}--\ref{ass:C} hold. For any $0<\un t<\ov t<\infty$, there is a process  $C$ with c\`{a}dl\`{a}g paths   such that on the set $\{\inf_{s\in[\underline t,\overline t]} \si_s^2>0\}$, we have 
 	\begin{equation}\label{bounds_1}
 		O_{s,T}(k)\leq C_s\biggl(Te^{-\lvert k-x_s\rvert} \bone_{\{\lvert k-x_s\rvert>1\}}+\biggl(\sqrt{T}\wedge\frac{T}{|k-x_s|}\biggr)\bone_{\{|k-x_s|<1\}}\biggr)
 	\end{equation}
 	and
 	\begin{equation}\label{bounds_2}
 		|O_{s,T}(k_1)-O_{s,T}(k_2)|\leq C_s\biggl[\frac{T}{(k_2-x_s)^4}\wedge\frac{T}{(k_2-x_s)^2}\wedge 1\biggr]|e^{k_1}-e^{k_2}|.
 	\end{equation}
 	Moreover, if  $\sqrt{\Den}\asymp T$ and $(\si_t)_{t\geq0}$ is bounded away from zero, we further have
 	\begin{equation}\label{bounds_3} 
 		\E^\Q_r[\lvert e^{-x_s}O_{s,T}(k+x_s)-e^{-x_r}O_{r,T+(s-r)}(k+x_r)\rvert ] \leq C_r{\sqrt{\Den T}}  \log \Den^{-1}
 	\end{equation}
 	for $r,s\in[\underline{t},\overline{t}]$, $0\leq  s-r\leq \Den$, $k\in\R$, $k_1<k_2<x_s$ or $k_1>k_2>x_s$, and  small enough $T$.
 \end{lemma}

\begin{lemma}\label{lem:char}
Let $m_p=\E[(\ov\eps_{j,t,T})^p]^{1/p}$ (which does not depend on $j$, $t$ or $T$) and define
	\begin{equation}\label{eq:ov-eps} 
	\eps_{t,T}(u)	=-\frac{2}{u^2}F'\bigl(\si^2_{t,T}(u)\bigr)\frac{\Re\bigl( \call_{t,T}(u)\bigr)\Re \bigl(Z_{t,T}(u)\bigr)+\Im\bigl( \call_{t,T}(u)\bigr)\Im \bigl(Z_{t,T}(u)\bigr)}{\lvert \call_{t,T}(u)\rvert^2}
\end{equation}
and
	\begin{equation}\label{eq:beta-k}\begin{split}
			\beta_{2,t,T}(u)_k&=\Biggl\{\frac{2F'(\si^2_{t,T}(u))}{u^2\lvert \call_{t,T}(u)\rvert^2}\biggl[ \Re\bigl(\call_{t,T}(u)\bigr)\Re\biggl(\biggl(\frac{u^2}{T} +i \frac{u}{\sqrt{T}}\biggr)e^{(iu/\sqrt{T}-1)(k-x_t)}\biggr) \\
			&\quad \qquad+\Im\bigl(\call_{t,T}(u)\bigr)\Im\biggl(\biggl(\frac{u^2}{T} +i \frac{u}{\sqrt{T}}\biggr)e^{(iu/\sqrt{T}-1)(k-x_t)}\biggr)\biggr] \Biggr\}^2\\ &\quad\times e^{-2x_t}\zeta_{t,1}(k-x_t)^2O_{t,T}(k)^2\rho_{t,1}(k-x_t) \delta. 
	\end{split}\end{equation}
	Under the assumptions of Theorem~\ref{thm:test}, we have
	\begin{equation}\label{eq:char-approx} 
		\E[e^{iU_n\eps_{t,T}(u)}\mid \calf]=\exp\biggl(  -\frac12m_2^2 U_n^2 \int_\R\beta_{2,t,T}(u)_k dk \biggr)+o^\uc(\sqrt{\Den}),
	\end{equation}
	where the remainder term is also locally uniform in $t$ and uniformly on compacts in $U$.
\end{lemma}

\begin{proof}[Proof of Theorem~\ref{thm:test}] We only prove the theorem for $R^n_T(u,U)$. The proof for $R^n_{T,T'}(u,U)$ is completely analogous. By a classical localization argument, we can and will assume that all locally bounded processes (in particular, all semimartingale processes) that appear in the proof are uniformly bounded by a finite deterministic constant $C$. %In the following, we write $a_n= o^\uc_\eps(b_n)$ (with $a_n$ possibly random) if $a_n=O^\uc(b_n\Den^\eps)$ for some $\eps>0$. Also, 
	Let $\calg_t=\si(\ov \epsilon_{j,s,\tau T}: s< t,\, j=1,\dots, N_{s,\tau T},\, \tau,T>0)$ and $\ov\calf_t = \calf_t \otimes \calg_t$ and define
	\begin{equation}\label{eq:ov-L1} 
		\ov	\L^n_1(u,U)_{t,T}= \frac1{k_n } \sum_{i=1}^{k_n} \E[e^{iU_n \Delta^n_i \wh V_{t,T}(u)}\mid \ov \calf_{t^n_i}].
	\end{equation}
	By the derivation of Equation~(B.19) in \cite{CT23_b}, we can use Corollary~\ref{cor:incr} to show
	\begin{equation}\label{eq:Delta-Vhat-2} \begin{split}
			\Delta^n_i \wh V_{t,T}(u)
			&=\Delta^n_i V_t + \Delta^n_i \epsilon_{t,T}(u)\\
			&\quad+O^\uc\biggl(\frac{\delta}{\sqrt{T} }\log T +T^{N/2}+\sqrt{\Den T}\biggr)+o^\uc(\Den/\sqrt{T}),
		\end{split}
	\end{equation}
	where $\Delta^n_i \epsilon_{t,T}(u)=\eps_{t^n_{i-1},T^n_{i-1}}(u)-\eps_{t^n_i,T^n_i}(u)$ (with $\eps_{t,T}(u)$ defined in \eqref{eq:ov-eps}) and
	\begin{equation}\label{eq:Z-2}\begin{split}
			Z_{t,T}(u) = -\biggl(\frac{u^2}{T}+i\frac{u}{\sqrt{T}}\biggr)e^{-x_t}\sum_{j=2}^{N_{t,T}}e^{(iu/\sqrt{T}-1)(k_{j-1,t,T}-x_t)}\epsilon_{j-1,t,T}\delta_{j,t,T}.
	\end{split}\end{equation}
	By \eqref{eq:rates}, we have $O^\uc (\frac{\delta}{\sqrt{T} }\log T +T^{N/2}+\sqrt{\Den T} )+o^\uc(\Den/\sqrt{T})=o^\uc(\sqrt{\Den})=o^\uc(1/\sqrt{k_n\lvert\calt\rvert})$. 
	Recalling \eqref{eq:epsilon} and the assumption that $\ov\epsilon_{j,t^n_i,T}$ and $\ov\epsilon_{j,t^n_{i-1},T}$ are independent of $\calf$ and $\calg_{t^n_{i}}$, we  have
	\begin{align*}
		\ov	\L^n_1(u,U)_{t,T}&= \frac1{k_n}  \sum_{i=1}^{k_n} \E[e^{iU_n \Delta^n_i V_{t}} \E[ e^{iU_n\Delta^n_i \epsilon_{t,T}(u)}\mid \calf \otimes \calg_{t^n_{i}} ]\mid \ov \calf_{t^n_{i}}]  + o^\uc(\sqrt{\Den})\\
		&	= \frac1{k_n }  \sum_{i=1}^{k_n} \E[e^{iU_n \Delta^n_i V_{t}} \E[ e^{iU_n\Delta^n_i \epsilon_{t,T}(u)}\mid \calf ]\mid   \calf_{t^n_{i}}] +o^\uc(\sqrt{\Den}).
	\end{align*}
	By the independence properties of $\ov\epsilon_{j,t,T}$, we clearly have 
	\[ \E[ e^{iU_n\Delta^n_i \epsilon_{t,T}(u)}\mid \calf ] =\E[ e^{iU_n\eps_{t^n_{i-1},T^n_{i-1}}(u)}\mid \calf ]\E[ e^{-iU_n\eps_{t^n_{i},T^n_{i}}(u)}\mid \calf ].  \]
	Using Lemma~\ref{lem:char} (by a change of variables, we can replace $\beta_{2,t,T}(u)_k$ in \eqref{eq:char-approx} by $\beta'_{2,t,T}(u)_k=\beta_{2,t,T}(u)_{x_t+k}$) and the notation \eqref{eq:beta-k} from below, it follows that 
	\begin{align*}
		\E[ e^{iU_n\Delta^n_i \epsilon_{t,T}(u)}\mid \calf ] &=\exp\biggl(-\frac{m_2^2U_n^2}2 \int_\R(\beta'_{2,t^n_{i-1},T^n_{i-1}}(u)_k+\beta'_{2,t^n_{i},T^n_{i}}(u)_k) dk \biggr) +o^\uc(\sqrt{\Den})
	\end{align*}
	and consequently,
	\begin{align*}
		\ov \L^n_1(u,U)_{t,T}	&=\frac1{k_n} \sum_{i=1}^{k_n} \E_{t^n_i}\biggl[\exp\biggl(iU_n \Delta^n_i V_t\\
		&\quad-\frac12m_2^2U_n^2 \int_\R(\beta'_{2,t^n_{i-1},T^n_{i-1}}(u)_k+\beta'_{2,t^n_{i},T^n_{i}}(u)_k) dk\biggr)\biggr] +o^\uc(\sqrt{\Den}).
	\end{align*}
	
	We claim that
	\begin{equation}\label{eq:toshow} 
		\ov \L^n_1(u,U)_{t,T}	=\frac1{k_n}  \sum_{i=1}^{k_n}\exp\biggl( -m_2^2U_n^2 \int_\R \beta'_{2,t^n_{i},T^n_{i}}(u)_k dk\biggr) \E_{t^n_{i}} [e^{iU_n \Delta^n_i V_t}] +o^\uc(\sqrt{\Den}).\!\!\!
	\end{equation}
	To prove this, observe that $F'$ is $C^1$ and  that  $\Delta^n_i \call_{t,T}(u)=O^\uc(\sqrt{\Den})$ and $\Delta^n_i \si^2_{t,T}(u) = O^\uc(\sqrt{\Den})$ by Theorems~\ref{thm:incr-L} and \ref{thm:incr}. Moreover, $ (\frac{u^2}{T^n_{i-1}} +i \frac{u}{\sqrt{T^n_{i-1}}} )e^{(iu/\sqrt{T^n_{i-1}}-1)(k-x_{t^n_{i}})}-(\frac{u^2}{T^n_{i}} +i \frac{u}{\sqrt{T^n_{i}}} )e^{(iu/\sqrt{T^n_{i}}-1)(k-x_{t^n_{i}})}=O^\uc(\frac{\Den}{T^2}+ \frac{\Den}{T^{5/2}}\lvert k-x_{t^n_i}\rvert)$. Therefore, recalling the rate conditions in \eqref{eq:rates}, the smoothness properties \eqref{eq:rho} and \eqref{eq:zeta} as well as the option price bounds \eqref{bounds_1} and \eqref{bounds_3}, we can apply Taylor's theorem (which is tedious but straightforward) to show that
	\begin{equation}\label{eq:diffbeta}
		U_n^2\int_\R(\beta'_{2,t^n_{i-1},T^n_{i-1}}(u)_k-\beta'_{2,t^n_{i},T^n_{i}}(u)_k) dk	=%-U_n^2 \beta^{n,i}_{t,T}(u))+
		o^\uc(\sqrt{\Den}),  
	\end{equation}
	which implies \eqref{eq:toshow}. 
	
	Next, using Lemma~\ref{lem:exp} (with $T$ replaced by $\Den$), one can show  that  $\E_{t^n_{i}}[e^{iU_n \Delta^n_i V_t}  ]=\exp(iU\al^V_{t^n_i}\sqrt{\Den}-\frac12U^2(\si^V_{t^n_i})^2 + \Den\vp_{t^n_i}^V(U_n)-\frac12iU^3(\si^V_{t^n_i})^2\si^{\si^V}_{t^n_i}\sqrt{\Den}) + o^\uc(\sqrt{\Den})$, where  $\vp^V_t(U)=\int_\R (e^{iUz}-1-iUz) F^V_t(dz)$  and $\si^{\si^V}$ is the volatility of $\si^V$ with respect to $W^\P$. %Under the null hypothesis, we have $\Den\vp_{t^n_i}^V(U_n)-\Den\vp_t^V(U_n) = O^\uc(\Den^{1-r/2})=o^\uc(\sqrt{\Den})$. 
	Therefore,
	\begin{equation}\label{eq:ovLn1} 
		\ov \L^n_1(u,U)_{t,T} = \ov \L^{n,1}_1(u,U)_{t,T}+\ov \L^{n,2}_1(u,U)_{t,T}+\ov \L^{n,3}_1(u,U)_{t,T}+o^\uc(\sqrt{\Den}),
	\end{equation}
	where (with the notation $t_n = t-k_n\Den$)
	\begin{align*}
		\ov \L^{n,1}_1(u,U)_{t,T}&= \frac1{k_n}  \sum_{i=1}^{k_n}\exp\biggl( -m_2^2U_n^2 \int_\R \beta'_{2,t^n_{i},T^n_{i}}(u)_k dk+iU\al^V_{t_n}\sqrt{\Den} -\frac12U^2(\si^V_{t_n})^2\\
		&\quad  +\Den\vp_{t_n}^V(U_n) -\frac12iU^3(\si^V_{t_n})^2\si^{\si^V}_{t_n}\sqrt{\Den}\biggr),\\
		\ov \L^{n,2}_1(u,U)_{t,T}&=- \frac{U^2}{k_n}  \sum_{i=1}^{k_n}\exp\biggl( -m_2^2U_n^2 \int_\R \beta'_{2,t^n_{i},T^n_{i}}(u)_k dk+iU\al^V_{t_n}\sqrt{\Den}-\frac12U^2(\si^V_{t_n})^2\\
		&\quad    +\Den\vp_{t_n}^V(U_n)-\frac12iU^3(\si^V_{t_n})^2\si^{\si^V}_{t_n}\sqrt{\Den}\biggr)\si^V_{t_n}(\si^V_{t^n_i}-\si^V_{t_n}),\\
		\ov \L^{n,3}_1(u,U)_{t,T}&=  \frac{1}{k_n}  \sum_{i=1}^{k_n}\exp\biggl( -m_2^2U_n^2 \int_\R \beta'_{2,t^n_{i},T^n_{i}}(u)_k dk+iU\al^V_{t_n}\sqrt{\Den}-\frac12U^2(\si^V_{t_n})^2\\
		&\quad    +\Den\vp_{t_n}^V(U_n)-\frac12iU^3(\si^V_{t_n})^2\si^{\si^V}_{t_n}\sqrt{\Den}\biggr)\Den(\vp^V_{t^n_i}(U_n)-\vp^V_{t_n}(U_n)).
	\end{align*}
	
	Our next goal is to derive an expansion of $\wh \C^n_1(u,U)_{t,T}$. To this end, note that 
	\begin{equation}\label{eq:Ldiff} 
		\wh \L^n_1(u,U)_{t,T}-\ov \L^n_1(u,U)_{t,T}=k_n^{-1}\sum_{i=1}^{k_n} \Bigl(e^{iU_n \Delta^n_i \wh V_{t,T}(u)} -\E[e^{iU_n \Delta^n_i \wh V_{t,T}(u)}\mid \ov \calf_{t^n_i}]\Bigr),
	\end{equation}
	where the $i$th term is $\ov \calf_{t^n_{i-1}+\Den}$-measurable and has a zero $\ov\calf_{t^n_i}$-conditional expectation. Therefore, if we split the sum in the last display into two, one summing over odd indices and one summing over even indices, these two become martingale sums, showing that $\wh \L^n_1(u,U)_{t,T}-\ov \L^n_1(u,U)_{t,T}=O^\uc(1/\sqrt{k_n})$.
	Since $k_n^{-3/2}=o(\sqrt{\Den})$ by \eqref{eq:rates},  Taylor's theorem implies that
	\begin{align*}
		\frac{1}{\lvert\calt\rvert} \sum_{t\in\calt}\ov \C^n_1(u,U)_{t,T}	&=-\frac{2U^{-2}}{\lvert\calt\rvert} \sum_{t\in\calt}\Re \biggl(\Log \ov \L^n_1(u,U)_{t,T}  +  \frac{\wh \L^n_1(u,U)_{t,T}-\ov \L^n_1(u,U)_{t,T}}{ \ov \L^n_1(u,U)_{t,T}}\\
		&\quad-\frac12\biggl(\frac{\wh \L^n_1(u,U)_{t,T}-\ov \L^n_1(u,U)_{t,T}}{ \ov \L^n_1(u,U)_{t,T}}\biggr)^2 \biggr) + o^\uc(\sqrt{\Den}). 
	\end{align*}
	We can further expand $\ov \L^n_1(u,U)_{t,T}$ using \eqref{eq:ovLn1} and obtain
	\[\frac{1}{\lvert\calt\rvert} \sum_{t\in\calt}\ov \C^n_1(u,U)_{t,T}= E^n_1(u,U)_T+F^{n}_{1}(u,U)_T +\sum_{j=1}^3  G^{n,j}_{1}(u,U)_T +o^\uc(\sqrt{\Den}),  \]
	where 
	\begin{equation}\label{eq:EFG}\begin{split}
			E^n_1(u,U)_T 	&=-\frac{2U^{-2}}{\lvert\calt\rvert} \sum_{t\in\calt} \Re \biggl(   \frac{\wh \L^n_1(u,U)_{t,T}-\ov \L^n_1(u,U)_{t,T}}{ \ov \L^n_1(u,U)_{t,T}}\biggr),\\
			F^{n}_1(u,U)_T &=-\frac{2U^{-2}}{\lvert\calt\rvert} \sum_{t\in\calt} \Re (   \Log \ov\L^{n,1}_1(u,U)_{t,T}),\\
			G^{n,1}_{1}(u,U)_T	&= \frac{U^{-2}}{\lvert\calt\rvert} \sum_{t\in\calt} \Re \biggl(   \biggl(\frac{\wh \L^n_1(u,U)_{t,T}-\ov \L^n_1(u,U)_{t,T}}{ \ov \L^n_1(u,U)_{t,T}}\biggr)^2\biggr),\\
			%	G^{n,2}_{1}(u,U)_T	&=- \frac{U^{-2}}{2\lvert\calt\rvert} \sum_{t\in\calt} \Re \biggl(   \biggl(\frac{\ov \L^{n,2}_1(u,U)_{t,T}}{ \ov \L^{n,1}_1(u,U)_{t,T}}\biggr)^2\biggr),\\
			G^{n,2}_{1}(u,U)_T	&= -\frac{2U^{-2}}{\lvert\calt\rvert} \sum_{t\in\calt} \Re \biggl(\frac{\ov \L^{n,2}_1(u,U)_{t,T}}{ \ov \L^{n,1}_1(u,U)_{t,T}}\biggr),\\
			G^{n,3}_{1}(u,U)_T	&= -\frac{2U^{-2}}{\lvert\calt\rvert} \sum_{t\in\calt} \Re \biggl(\frac{\ov \L^{n,3}_1(u,U)_{t,T}}{ \ov \L^{n,1}_1(u,U)_{t,T}}\biggr).
	\end{split}\end{equation}
	
	%It is easy to see that $G^{n,2}_1(u,U)_T=O^\uc(k_n\Den/U^2)=O^\uc(\sqrt{\Den}/U^2)$ and therefore only $o^\uc(1)$ away from being negligible. This means we can simplify $G^{n,2}_1(u,U)_T$ by only keeping  leading order terms, which yields
	%\begin{align*}
	%G^{n,2}_1(u,U)_T	&= -\frac{U^2}{2\lvert\calt\rvert}\sum_{t\in\calt} \biggl( \frac1{k_n}\sum_{i=1}^{k_n} \si^V_{t_n}(\si^V_{t^n_i}-\si^V_{t_n}) \biggr)^2  + o^\uc(\sqrt{\Den})\\
	%&= -\frac{U^2}{2\lvert\calt\rvert}\sum_{t\in\calt}  \biggl( \frac1{k_n\Den}\int_0^{k_n\Den} \si^V_{t_n}(\si^V_{t-s}-\si^V_{t_n}) ds \biggr)^2  + o^\uc(\sqrt{\Den}).
	%\end{align*}
	%Since $U\to 0$, this shows
	%\begin{equation} 
	%	\label{eq:Gn2}
	%	G^{n,2}_1(u,U)_T=o^\uc(\sqrt{\Den}).
	%\end{equation}
	
	Note that $G^{n,1}_{1}(u,U)_T=O^\uc(k_n^{-1})$, so we can replace $\ov \L^n_1(u,U)_{t,T}$ in the denominator by $\exp(-\frac12U^2(\si^V_{t_n})^2)$, as this only leads to an overall error of order $O^\uc(k_n^{-1}((\delta/\sqrt{T})/\Den+\sqrt{\Den/k_n})+\Den^{1-r/2}/k_n)=O^\uc(\Den^{(\iota'\wedge(1-r/2))}/k_n+\sqrt{\Den/k_n})$, which is $o^\uc(\sqrt{\Den})$ by \eqref{eq:rates}. Moreover, we can take conditional expectation with respect to $\calf_{t_n}$ as the difference is a martingale sum in $t$ and therefore of order $O^\uc(\lvert \calt\rvert^{-1/2}k_n^{-1})=o^\uc(\sqrt{\Den})$. This shows that
	\begin{align*}
		G^{n,1}_1(u,U)_T	&= \frac{U^{-2}}{\lvert\calt\rvert} \sum_{t\in\calt}e^{U^2(\si^V_{t_n})^2}  \E_{t_n}\Bigl[\Re \bigl(   (\wh \L^n_1(u,U)_{t,T}-\ov \L^n_1(u,U)_{t,T})^2\bigr) \Bigr] +o^\uc(\sqrt{\Den})\\
		&=\frac{U^{-2}}{\lvert\calt\rvert k_n^2} \sum_{t\in\calt}e^{U^2(\si^V_{t_n})^2} \\
		&\quad\times \biggl(\sum_{i=1}^{k_n} \Re \Bigl( \E_{t_n}\Bigl[\bigl(e^{iU_n \Delta^n_i \wh V_{t,T}(u)}-\E[e^{iU_n \Delta^n_i \wh V_{t,T}(u)}\mid \ov \calf_{t^n_i}]\bigr)^2 \Bigr] \Bigr)\\
		&\qquad+2\sum_{i=2}^{k_n} \Re \Bigl( \E_{t_n}\Bigl[\bigl(e^{iU_n \Delta^n_i \wh V_{t,T}(u)}-\E[e^{iU_n \Delta^n_i \wh V_{t,T}(u)}\mid \ov \calf_{t^n_i}]\bigr)\\
		&\quad\qquad\times\bigl(e^{iU_n \Delta^n_{i-1} \wh V_{t,T}(u)}-\E[e^{iU_n \Delta^n_{i-1} \wh V_{t,T}(u)}\mid \ov \calf_{t^n_{i-1}}]\bigr) \Bigr] \Bigr)\biggr)+o^\uc(\sqrt{\Den}).
	\end{align*}
	As $\Delta^n_i \eps_{t,T}(u) = O^\uc(\delta^{1/2}/T^{1/4})=O^\uc(\Den^{1/2+\iota'})$ by \eqref{eq:rates}, all terms involving  the noise variables can be omitted, after which the second sum above is identically zero. Regarding the first sum, we only need to keep the diffusive part of $V$, so 
	\begin{equation}\label{eq:Gn1}\begin{split}
			G^{n,1}_1(u,U)_T	&=\frac{U^{-2}}{\lvert\calt\rvert k_n} \sum_{t\in\calt}(e^{-U^2(\si^V_{t_n})^2}-1)+o^\uc(\sqrt{\Den}) \\
			&= \frac{1}{\lvert\calt\rvert}\sum_{t\in\calt}\wt \C^n_1(u,U)_{t,T}+o^\uc(\sqrt{\Den}),
		\end{split}
	\end{equation}
	where the last step follows from the fact that $\wh \L^n_1(u,U)_{t,T} = \ov \L^n_1(u,U)_{t,T} + O^\uc(1/\sqrt{k_n}) = e^{-U^2(\si^V_{t_n})^2/2} + O^\uc(\Delta_n^{\iota'})$ (so $\wt \C^n_1(u,U)_{t,T}=1/(U^2k_n)(e^{-U^2(\si^V_{t_n})^2}-1) + O^\uc(k_n^{-1}\Den^{\iota'})$).
	
	Next, we consider $G^{n,2}_1(u,U)_T$, which is already $O^\uc(\sqrt{k_n\Den})$. Since changing the term $U_n^2\int_\R \beta'_{2,t^n_i, T^n_i}(u)_k dk$ to $U_n^2\int_\R \beta'_{2,t-k_n\Den,T+k_n\Den}(u)_k dk$ in the definitions of $\ov \L^{n,1}_1(u,U)_{t,T}$ and $\ov \L^{n,2}_1(u,U)_{t,T}$ only incurs an error of order $o^\uc(\sqrt{k_n\Den})$ by \eqref{eq:diffbeta} and $k_n\Den = o(\sqrt{\Den})$ by \eqref{eq:rates}, it follows that we can make these changes in our analysis of $G^{n,2}_1(u,U)_T$. Afterwards, the term corresponding to $t$ in the sum defining $G^{n,2}_1(u,U)_T$ will be $\calf_t$-measurable with a zero $\calf_{t_n}$-conditional expectation. Therefore, the sum over $t$ is a martingale and we obtain 
	\begin{equation}\label{eq:Gn2} 
		G^{n,2}_1(u,U)_T = O^\uc(\sqrt{k_n\Den/\lvert \calt\rvert}) = o^\uc(\sqrt{\Den})	
	\end{equation}
	by \eqref{eq:rates}. Analogous arguments show that
	\begin{equation}\label{eq:Gn3} 
		G^{n,3}_1(u,U)_T = O^\uc(\sqrt{k_n\Den/\lvert \calt\rvert}) = o^\uc(\sqrt{\Den}).	
	\end{equation}
	%This shows that we can further make  $o^\uc(1)$-modifications, leading to
	%\begin{equation}\label{eq:Gn3}\begin{split}
	%	G^{n,3}_1(u,U)_T&=-\frac{1}{k_n\lvert \calt\rvert}\sum_{t\in\calt} \si^V_{t_n} \sum_{i=1}^{k_n} (\si^V_{t^n_i}-\si^V_{t_n})+o^\uc(\sqrt{\Den})\\ &= -\frac{1}{\lvert\calt\rvert k_n\Den} \sum_{t\in\calt} \si^V_{t_n} \int_0^{k_n\Den} (\si^V_{t-s}-\si^V_{t_n})+o^\uc(\sqrt{\Den}). 
	%	\end{split}
	%\end{equation}
	
	Concerning $F^n_1(u,U)_T$, note that $\Den \vp_{t_n}^V(U_n)=o^\uc(\sqrt{\Den})$ under Assumption~\ref{ass:H0}. Moreover, $U_n^2\int_\R \beta'_{2,t^n_i,T^n_i}(u)_k dk-U_n^2\int_\R \beta'_{2,t-s\Den,Ts\Den}(u)_k dk = o^\uc(\sqrt{\Den})$ by \eqref{eq:diffbeta}, so under the null hypothesis,
	\begin{equation}\label{eq:Fn}\begin{split}
			F^n_1(u,U)_T&	=-\frac{2U^{-2}}{\lvert \calt\rvert}\sum_{t\in\calt} \biggl(\log\biggl(\frac1{k_n\Den} \int_0^{k_n\Den}\exp\biggl(-m_2^2U_n^2\int_\R \beta'_{2,t-s,T+s}(u)_k dk\biggr) ds\biggr)\\
			&\quad-\frac{U^2}{2} (\si^V_{t_n})^2  \biggr)+o^\uc(\sqrt{\Den}).\end{split}\raisetag{-2.5\baselineskip}
	\end{equation}
	
	The previous arguments can be used to derive an analogous decomposition of 
	\[\frac{1}{\lvert\calt\rvert} \sum_{t\in\calt}\ov \C^n_2(u,U)_{t,T}= E^n_2(u,U)_T+F^{n}_{2}(u,U)_T +\sum_{j=1}^3  G^{n,j}_{2}(u,U)_T +o^\uc(\sqrt{\Den}),  \]
	where
	\begin{equation*}
		E^n_2(u,U)_T= -\frac{2U^{-2}}{\lvert\calt\rvert} \sum_{t\in\calt} \Re \biggl(   \frac{\wh \L^n_2(u,U)_{t,T}-\ov \L^n_2(u,U)_{t,T}}{ \ov \L^n_2(u,U)_{t,T}}\biggr)
	\end{equation*}
	with 
	\begin{align*}
		\ov \L^n_2(u,U)_{t,T} &= \ov \L^{n,1}_2(u,U)_{t,T}+\ov \L^{n,2}_2(u,U)_{t,T}+\ov \L^{n,3}_2(u,U)_{t,T}+o^\uc(\sqrt{\Den}),\\
		\ov \L^{n,1}_2(u,U)_{t,T}&= \frac2{k_n}  \sum_{i=1}^{\lfloor k_n/2\rfloor}\exp\biggl( -m_2^2U_n^2 \int_\R \beta'_{2,t^n_{2i},T^n_{2i}}(u)_k dk+2iU\al^V_{t_n}\sqrt{\Den}-U^2(\si^V_{t_n})^2 \\
		&\quad +2\Den\vp_{t_n}^V(U_n) -2iU^3(\si^V_{t_n})^2\si^{\si^V}_{t_n}\sqrt{\Den}\biggr),\\
		\ov \L^{n,2}_2(u,U)_{t,T}&=- \frac{4U^2}{k_n}     \sum_{i=1}^{\lfloor k_n/2\rfloor} \exp\biggl( -m_2^2U_n^2 \int_\R \beta'_{2,t^n_{2i},T^n_{2i}}(u)_k dk+2iU\al^V_{t_n}\sqrt{\Den}-U^2(\si^V_{t_n})^2\\
		&\quad  +2\Den\vp_{t_n}^V(U_n)  -2iU^3(\si^V_{t_n})^2\si^{\si^V}_{t_n}\sqrt{\Den}\biggr)\si^V_{t_n}(\si^V_{t^n_{2i}}-\si^V_{t_n}),\\
		\ov \L^{n,3}_2(u,U)_{t,T}&= \frac4{k_n}  \sum_{i=1}^{\lfloor k_n/2\rfloor}\exp\biggl( -m_2^2U_n^2 \int_\R \beta'_{2,t^n_{2i},T^n_{2i}}(u)_k dk+2iU\al^V_{t_n}\sqrt{\Den}-U^2(\si^V_{t_n})^2 \\
		&\quad +2\Den\vp_{t_n}^V(U_n) -2iU^3(\si^V_{t_n})^2\si^{\si^V}_{t_n}\sqrt{\Den}\biggr)\Den(\vp^V_{t^n_i}(U_n)-\vp^V_{t_n}(U_n))
	\end{align*}
	and
	\begin{equation}\label{eq:Fn2}\begin{split}
			F^n_2(u,U)_T&=-\frac{2U^{-2}}{\lvert \calt\rvert}\sum_{t\in\calt} \biggl(\log\biggl(\frac1{k_n\Den} \int_0^{k_n\Den}\exp\biggl(-m_2^2U_n^2\int_\R \beta'_{2,t-s,T+s}(u)_k dk\biggr) ds\biggr)\\
			&\quad- {U^2} (\si^V_{t_n})^2  \biggr)+o^\uc(\sqrt{\Den}),\\
		\end{split}\raisetag{-2.5\baselineskip}\end{equation}
	and 
	\begin{align*}
		G^{n,1}_2(u,U)_T&=\frac{1}{\lvert \calt\rvert}\sum_{t\in\calt}\wt \C^n_2(u,U)_{t,T}+o^\uc(\sqrt{\Den}),\\	G^{n,2}_2(u,U)_T&=G^{n,3}_2(u,U)_T=o^\uc(\sqrt{\Den}).
	\end{align*}
	Upon realizing that
	\begin{equation}\label{eq:F} 
		F^n_{2}(u,\sqrt{2}U)_T -F^n_{1}(u,\sqrt{2}U)_T-	F^n_{2}(u,U)_T +F^n_{1}(u,U)_T  =o^\uc(\sqrt{\Den})
	\end{equation}
	(all terms that are not $o^\uc(\sqrt{\Den})$ cancel each other perfectly), we arrive at 
	\begin{equation}\label{eq:R}\begin{split}
			R^n_T(u,U)	&=\frac{\wh R^n_T(u,U)}{\sqrt{\wh{\avar}^n_T(u,U)}}+o^\uc(1/\sqrt{k_n\lvert\calt\rvert}),\quad\text{where} \\
			\wh R^n_T(u,U)&=\sqrt{k_n\lvert\calt\rvert/2}\bigl (E^n_2(u,\sqrt{2} U)_T - E^n_1(u,\sqrt{2} U)_T-E^n_2(u,U)_T+E^n_1(u, U)_T \bigr).
	\end{split}\end{equation}
	
	Next, we want to derive a simpler expression for $E^n_1(u,U)_T$ and $E^n_2(u,U)_T$. To this end, we define
	\begin{align*}
		\wt \L^{n,1}_1(u,U)_{t,T}&=  \exp\biggl( -m_2^2U_n^2 \int_\R \beta'_{2,t-k_n\Den,T+k_n\Den}(u)_k dk+iU\al^V_{t_n}\sqrt{\Den}-\frac12U^2(\si^V_{t_n})^2\\
		&\quad+\Den\vp_{t_n}^V(U_n) -\frac12iU^3(\si^V_{t_n})^2\si^{\si^V}_{t_n}\sqrt{\Den}\biggr),\\
		\wt \L^{n,1}_2(u,U)_{t,T}&=  \exp\biggl( -m_2^2U_n^2 \int_\R \beta'_{2,t-k_n\Den,T+k_n\Den}(u)_k dk+2iU\al^V_{t_n}\sqrt{\Den}-U^2(\si^V_{t_n})^2\\
		&\quad +2\Den\vp_{t_n}^V(U_n)-2iU^3(\si^V_{t_n})^2\si^{\si^V}_{t_n}\sqrt{\Den}\biggr).
	\end{align*}
	By \eqref{eq:diffbeta}, we have $\ov \L^{n,1}_1(u,U)_{t,T}-\wt \L^{n,1}_1(u,U)_{t,T}=o^\uc(\sqrt{k_n\Den})$ and $\ov \L^{n,1}_2(u,U)_{t,T}-\wt \L^{n,1}_2(u,U)_{t,T}=o^\uc(\sqrt{k_n\Den})$. Together with the fact that the term corresponding to $t$ in the sum defining $E^n_1(u,U)_T$ (resp.,  $E^n_2(u,U)_T$) is $\ov\calf_{t+\Den}$-measurable with a zero $\ov\calf_{t-k_n\Den}$-conditional expectation, it follows that replacing the denominator in this sum by $\wt \L^{n,1}_1(u,U)_{t,T}+\ov \L^{n,2}_1(u,U)_{t,T}$ (resp., $\wt \L^{n,1}_2(u,U)_{t,T}+\ov \L^{n,2}_2(u,U)_{t,T}$) results in an error of order $o^\uc(\lvert\calt\rvert^{-1/2}\sqrt{k_n\Den})$, which is $o^\uc(\sqrt{\Den})$ by \eqref{eq:rates} and therefore negligible. In other words,
	\begin{align*}
		E^n_1(u,U)_T 	&=-\frac{2U^{-2}}{\lvert\calt\rvert} \sum_{t\in\calt} \Re \biggl(   \frac{\wh \L^n_1(u,U)_{t,T}-\ov \L^n_1(u,U)_{t,T}}{ \wt \L^{n,1}_1(u,U)_{t,T}+\ov \L^{n,2}_1(u,U)_{t,T}+\ov \L^{n,3}_1(u,U)_{t,T}}\biggr) \\
		&\quad+ o^\uc(\sqrt{\Den})\\
		&=-\frac{2U^{-2}}{\lvert\calt\rvert} \sum_{t\in\calt} \Re \biggl(   \frac{\wh \L^n_1(u,U)_{t,T}-\ov \L^n_1(u,U)_{t,T}}{ \wt \L^{n,1}_1(u,U)_{t,T}}\biggr)\\
		&\quad-\frac{2U^{-2}}{\lvert\calt\rvert} \sum_{t\in\calt} \Re \biggl(   \frac{\wh \L^n_1(u,U)_{t,T}-\ov \L^n_1(u,U)_{t,T}}{ \wt \L^{n,1}_1(u,U)_{t,T}}\ov \L^{n,2}_1(u,U)_{t,T} \biggr)+ o^\uc(\sqrt{\Den}).
	\end{align*}
	For the last step, we used the fact that $\wh \L^n_1(u,U)_{t,T}-\ov \L^n_1(u,U)_{t,T}= O^\uc(k_n^{-1/2})$ and $\ov \L^{n,3}_1(u,U)_{t,T}=O^\uc(\sqrt{k_n\Den}\Den^{1-r/2})$.
	
	Consider the   term  in the last line of the previous display. %In \eqref{eq:Ldiff}, we can clearly replace both exponential terms by $1-e^{iU_n\Delta^n_i\wh V_{t,T}(u)}$, which shows that 
	As $\wh \L^n_1(u,U)_{t,T}-\ov \L^n_1(u,U)_{t,T}= O^\uc(k_n^{-1/2})$ and $\ov \L^{n,2}_1(u,U)_{t,T}=O^{\uc}(\sqrt{k_n\Den})$, this term is $O^\uc(\sqrt{\Den})$, which means we can make any additional $o^\uc(1)$-modification we want. In particular, we can take conditional expectation with respect to $\E_{t_n}$ and omit all but the dominating parts of  $\wh \L^n_1(u,U)_{t,T}-\ov \L^n_1(u,U)_{t,T}$ and $\ov \L^{n,2}_1(u,U)_{t,T}$, which yields the expression
	\begin{align*}
		\frac{2}{\lvert\calt\rvert k_n^2} \sum_{t\in\calt} \si^V_{t_n}\sum_{i,j=1}^{k_n} \E_{t_n}\Bigl[(\si^V_{t^n_i}-\si^V_{t_n})  \bigl( \cos(U_n\Delta^n_j V_t)- \E_{t^n_j}[\cos(U_n\Delta^n_j V_t)]\bigr)\Bigr]. 
	\end{align*}
	As we only need to keep the dominating part of $\si^V_{t^n_i}-\si^V_{t_n}$, there is no loss of generality to assume that $\si^V$ is a continuous  martingale. Then both $\si^V_{t^n_i}-\si^V_{t_n}$ and $\cos(U_n\Delta^n_j V_t)- \E_{t^n_j}[\cos(U_n\Delta^n_j V_t)]$ are martingale increments and the $\E_{t_n}$-conditional expectation is only nonzero provided that $j\leq i$. Decomposing $\si^V_{t^n_i}-\si^V_{t_n}=(\si^V_{t^n_i}-\si^V_{t^n_{j-1}})+(\si^V_{t^n_{j-1}}-\si^V_{t^n_j})+(\si^V_{t^n_j}-\si^V_{t_n})$, we further see that only the middle term contributes (to remove the contribution of $\si^V_{t^n_i}-\si^V_{t^n_{j-1}}$, condition on $\calf_{t^n_{j-1}}$ first; to remove the contribution of $\si^V_{t^n_j}-\si^V_{t_n}$, condition on $\calf_{t^n_{j}}$ first). Since this middle term is $O(\sqrt{\Den})$, what remains from the expression of the previous display is 
	\begin{equation*}
		\frac{2}{\lvert\calt\rvert k_n^2} \sum_{t\in\calt} \si^V_{t_n}\sum_{j=1}^{k_n} (k_n-j+1)\E_{t_n}\Bigl[\Delta^n_j \si^V_t  \bigl( \cos(U_n\Delta^n_j V_t)- \E_{t^n_j}[\cos(U_n\Delta^n_j V_t)]\bigr)\Bigr]. 
	\end{equation*}
	To keep notation simple, we further assume that $\Delta^n_j \si^V_t=\int_{t^n_i}^{t^n_{i-1}} \si^{\si^V}_s dW^\P_s$. The case where $\si^V$   is driven by several $\P$-Brownian motions can be treated analogously. Then, because the expression in the previous display is already $O^\uc(\sqrt{\Den})$, we can freeze $\si^{\si^V}_s$ (and also $\si^V_s$) at $s=t_n$, which leaves us with
	\begin{align*}
		&\frac{2}{\lvert\calt\rvert k_n^2} \sum_{t\in\calt} \si^V_{t_n}\sum_{j=1}^{k_n}(k_n-j+1) \\
		&\qquad\times\E_{t_n}\Bigl[\si^{\si^V}_{t_n} \Delta^n_j W^\P_t  \bigl( \cos(U_n\si^V_{t_n}\Delta^n_j W^\P_t)- \E_{t^n_j}[\cos(U_n\si^V_{t_n}\Delta^n_j W^\P_t)]\bigr)\Bigr]\\
		&\quad=\frac{2}{\lvert\calt\rvert k_n^2} \sum_{t\in\calt} \si^V_{t_n}\sum_{j=1}^{k_n}(k_n-j+1) \E_{t_n}\Bigl[\si^{\si^V}_{t_n} \Delta^n_j W^\P_t   \cos(U_n\si^V_{t_n}\Delta^n_j W^\P_t)\Bigr].
	\end{align*}
	The last $\E_{t_n}$-term is the conditional expectation of an odd function a centered Gaussian distribution and therefore identically zero. All in all, we have shown that
	\begin{align*}
		E^n_1(u,U)_T 	&= -\frac{2U^{-2}}{\lvert\calt\rvert} \sum_{t\in\calt} \Re \biggl(   \frac{\wh \L^n_1(u,U)_{t,T}-\ov \L^n_1(u,U)_{t,T}}{ \wt \L^{n,1}_1(u,U)_{t,T}}\biggr)+o^\uc(\sqrt{\Den})\\
		&=-\frac{2U^{-2}}{k_n\lvert\calt\rvert} \sum_{t\in\calt} \sum_{i=1}^{k_n}  \Re\biggl(\frac{e^{iU_n \Delta^n_i \wh V_{t,T}(u)} -\E[e^{iU_n \Delta^n_i \wh V_{t,T}(u)}\mid \ov \calf_{t^n_i}]}{\wt \L^{n,1}_1(u,U)_{t,T}}\biggr)+o^\uc(\sqrt{\Den}).
		%		E^n_2(u,U)_T 	&= \frac{U^{-2}}{\lvert\calt\rvert} \sum_{t\in\calt} \Re \biggl(   \frac{\wh \L^n_2(u,U)_{t,T}-\ov \L^n_2(u,U)_{t,T}}{ \wt \L^{n,1}_2(u,U)_{t,T}}\biggr),
	\end{align*}
	In the last line,  the numerator for two different values of $i$, say, $i_1$ and $i_2$ becomes conditionally independent as soon as $\lvert i_1-i_2\rvert \geq2$. Hence, the last term is of order $O^\uc(1/\sqrt{k_n\lvert \calt\rvert})$, and we can make any $o^\uc(1)$ modification that we want. For example, we can replace the denominator by $e^{-U^2(\si^V_{t_n})^2/2}$ and remove everything but the diffusive component of $\Delta^n_i \wh V_{t,T}(u)$, which we denote by $\Delta^n_i V^c_t = \int_{t^n_i}^{t^n_{i-1}} \si^{V}_s dW^\P_s$. As a result,
	\begin{equation}\label{eq:E}\begin{split}
			E^n_1(u,U)_T &= -\frac{2U^{-2}}{k_n\lvert\calt\rvert} \sum_{t\in\calt} \sum_{i=1}^{k_n} e^{\frac12 U^2(\si^V_{t_n})^2}  \bigl  (\cos(U_n \Delta^n_i V^c_t) -\E_{t^n_i}[\cos(U_n \Delta^n_i   V^c_t)] \bigr)\\
			&\quad+o^\uc(1/\sqrt{k_n\lvert \calt\rvert}).
			% 	&=\frac{1}{k_n\lvert\calt\rvert} \sum_{t\in\calt} \sum_{i=1}^{k_n}  \biggl\{\biggl(\frac{\Delta^n_i V^c_t}{\sqrt{\Den}}\biggr)^2 -\E_{t^n_i}\biggl[\biggl(\frac{\Delta^n_i V^c_t}{\sqrt{\Den}}\biggr)^2 \biggr] \biggr\}+o^\uc(1/\sqrt{k_n\lvert \calt\rvert}),
		\end{split}
	\end{equation}
	%where the last step makes use of the expansion $\cos(x)-1=-\frac12x^2 + O(x^4)$.
	
	Analogously, one can show that
	\begin{equation}\label{eq:E2}\begin{split}
			E^n_2(u,U)_T &= -\frac{4U^{-2}}{k_n\lvert\calt\rvert} \sum_{t\in\calt} \sum_{i=1}^{\lfloor k_n/2\rfloor}e^{  U^2(\si^V_{t_n})^2}  \bigl  (\cos(U_n (\Delta^n_{2i-1} V^c_t+\Delta^n_{2i} V^c_t))\\
			&\quad -\E_{t^n_{2i}}[\cos(U_n (\Delta^n_{2i-1} V^c_t+\Delta^n_{2i} V^c_t))] \bigr) +o^\uc(1/\sqrt{k_n\lvert \calt\rvert}).
			% 	&=\frac{1}{k_n\lvert\calt\rvert} \sum_{t\in\calt} \sum_{i=1}^{k_n}  \biggl\{\biggl(\frac{\Delta^n_i V^c_t}{\sqrt{\Den}}\biggr)^2 -\E_{t^n_i}\biggl[\biggl(\frac{\Delta^n_i V^c_t}{\sqrt{\Den}}\biggr)^2 \biggr] \biggr\}+o^\uc(1/\sqrt{k_n\lvert \calt\rvert}),
		\end{split}
	\end{equation}
	At the same time, it is easy to see that we can modify \eqref{eq:E} and write 
	\begin{equation}\label{eq:E1}\begin{split}
			E^n_1(u,U)_T &= -\frac{2U^{-2}}{k_n\lvert\calt\rvert} \sum_{t\in\calt} \sum_{i=1}^{\lfloor k_n/2\rfloor} e^{\frac12 U^2(\si^V_{t_n})^2}  \bigl  (\cos(U_n \Delta^n_{2i-1} V^c_t)+\cos(U_n \Delta^n_{2i} V^c_t) \\
			&\quad-\E_{t^n_{2i}}[\cos(U_n \Delta^n_{2i-1} V^c_t)+\cos(U_n \Delta^n_{2i} V^c_t)] \bigr) +o^\uc(1/\sqrt{k_n\lvert \calt\rvert}).
			% 	&=\frac{1}{k_n\lvert\calt\rvert} \sum_{t\in\calt} \sum_{i=1}^{k_n}  \biggl\{\biggl(\frac{\Delta^n_i V^c_t}{\sqrt{\Den}}\biggr)^2 -\E_{t^n_i}\biggl[\biggl(\frac{\Delta^n_i V^c_t}{\sqrt{\Den}}\biggr)^2 \biggr] \biggr\}+o^\uc(1/\sqrt{k_n\lvert \calt\rvert}),
		\end{split}
	\end{equation}
	A straightforward (but tedious) application of Theorem 2.2.15 of \cite{JP12} therefore shows that
	\begin{equation}\label{eq:CLT-E} 
		\sqrt{k_n\lvert \calt\rvert/2}\begin{pmatrix} E^n_1(u,U)_T  \\ E^n_2(u,U)_T\\ E^n_1(u,\sqrt{2}U)_T \\ E^n_2(u,\sqrt{2}U)_T \end{pmatrix} \limst \calv^{1/2}\calz,
	\end{equation}
	where $\calz$ is a  four-dimensional vector of independent standard normal random variables, defined on an extension of $(\ov\Om,\ov\calf,\ov\P)$ and independent from it, and $\calv=(\calv_{ij})_{i,j=1}^4$ is given by $\calv_{ij}= \tau^{-1}U^{-4} \calv^{(1)}_{ij}\int_0^\tau \sinh^2(\calv^{(2)}_{ij} U^2(\si^V_s)^2) ds$,
	where
	\begin{equation}\label{eq:Vs} 
		\calv^{(1)}=\begin{pmatrix} 4 & 8&2 &4\\8&8&4&4\\ 2 & 4&1&2 \\ 4 & 4&2&2\end{pmatrix}, \qquad \calv^{(2)}=\begin{pmatrix} \frac12 & \frac12&\frac{1}{\sqrt{2}} &\frac{1}{\sqrt{2}}\\[3pt] \frac12&1&\frac{1}{\sqrt{2}}&\sqrt{2}\\[3pt]  \frac{1}{\sqrt{2}} & \frac{1}{\sqrt{2}}&1&1 \\[3pt] \frac{1}{\sqrt{2}}&\sqrt{2}& 1&2\end{pmatrix}.
	\end{equation}
	By \eqref{eq:R}, it follows that 
	\begin{equation}\label{eq:R-conv} 
		\wh R^n_T(u,U) \limst N(0,a'\calv a),\quad\text{where } a = (1,-1,-1,1)'.
	\end{equation}
	
	To conclude the proof under the null hypothesis, it remains to show that $\wh\avar^n_T(u,U)\limp a'\calv a$. Since we are only interested in a limit in probability, only the leading order terms in $\wh\avar^n_T(u,U)$ matter. As
	\begin{align*}
		\wt \calz^n_i(u,U)_{t,T}&= \begin{pmatrix} \frac12(\cos(U_n\Delta^n_{2i-1} V^c_t) + \cos(U_n\Delta^n_{2i} V^c_t))\\ \frac12(\sin(U_n\Delta^n_{2i-1} V^c_t) + \sin(U_n\Delta^n_{2i} V^c_t))\\ \cos(U_n(\Delta^n_{2i-1} V^c_t+\Delta^n_{2i} V^c_t)) \\ \sin(U_n(\Delta^n_{2i-1} V^c_t+\Delta^n_{2i} V^c_t))\end{pmatrix}+o^\uc(1),%\\
		%\wt\caly^n_{t,T}(u,U)&=(-e^{\frac12 U^2(\si^V_{t_n})^2},0,e^{U^2(\si^V_{t_n})^2},0)'+o^\uc(1),
	\end{align*}
	we obtain
	\begin{align*}
		(\avar^{n,0}_{t,T}(u,U))_{ij}&=\calc^{(1)}_{ij}e^{-\calc^{(2)}_{ij} U^2(\si^V_{t_n})^2}\sinh^2(\calc^{(3)}_{ij}U^2(\si^V_{t_n})^2) + o^\uc(1) ,\quad i,j=1,\dots,8,\\ \avar^{n,1}_{t,T}(u,U)&=o^\uc(1),
	\end{align*}
	where $\calc^{(1)}_{ij}=\calc^{(2)}_{ij}=\calc^{(3)}_{ij}=0$ whenever $i$ or $j$ is even and $\ov\calc^{(1)}=(\calc^{(1)}_{2i-1,2j-1})_{i,j=1}^4$, $\ov\calc^{(2)}=(\calc^{(2)}_{2i-1,2j-1})_{i,j=1}^4$ and $\ov\calc^{(3)}=(\calc^{(3)}_{2i-1,2j-1})_{i,j=1}^4$ are given by
	\[ \ov\calc^{(1)} = \begin{pmatrix} 1 & 2&1 &2\\2&2&2&2\\ 1 & 2&1&2 \\ 2 & 2&2&2\end{pmatrix},\quad \ov\calc^{(2)}=\begin{pmatrix} 1 & \frac32&\frac32 &\frac52\\[3pt] \frac32&2&2&3\\[3pt] \frac32 & 2&2&3 \\[3pt] \frac52 & 3&3&4\end{pmatrix},\quad \ov \calc^{(3)}=\calv^{(2)},  \]
	respectively. Moreover, we have
	\[ \wt\caly^n_{t,T}(u,U)=(-e^{\frac12 U^2(\si^V_{t_n})^2},0,e^{U^2(\si^V_{t_n})^2},0)'+o^\uc(1),\]
	which shows   
	\begin{equation}\label{eq:avar} 
		\wh\avar^n_{T}(u,U) \limp   a'\calv a,
	\end{equation}
	completing the proof of the theorem under the null hypothesis.
	
	Turning to the alternative hypothesis, we note that the majority of our previous analysis of $R^n_T(u,U)$ remains valid under \ref{ass:H1}. The only exception is \eqref{eq:F}, which was derived assuming \ref{ass:H0} and no longer holds under \ref{ass:H1}. Indeed, under the alternative hypothesis, we need to replace \eqref{eq:Fn} and \eqref{eq:Fn2}  by 
	\begin{equation}\label{eq:Fn-alt}\begin{split}
			F^n_1(u,U)_T&	=-\frac{2U^{-2}}{\lvert \calt\rvert}\sum_{t\in\calt} \biggl(\log\biggl(\frac1{k_n\Den} \int_0^{k_n\Den}\exp\biggl(-m_2^2U_n^2\int_\R \beta'_{2,t-s,T+s}(u)_k dk\biggr) ds\biggr)\\
			&\quad-\frac{U^2}{2} (\si^V_{t_n})^2+ \Den\vp^V_{t_n}(U_n)  \biggr)+o^\uc(\sqrt{\Den})\end{split}\raisetag{-2.5\baselineskip}
	\end{equation}
	and
	\begin{equation}\label{eq:Fn2-alt}\begin{split}
			F^n_2(u,U)_T&=-\frac{2U^{-2}}{\lvert \calt\rvert}\sum_{t\in\calt} \biggl(\log\biggl(\frac1{k_n\Den} \int_0^{k_n\Den}\exp\biggl(-m_2^2U_n^2\int_\R \beta'_{2,t-s,T+s}(u)_k dk\biggr) ds\biggr)\\
			&\quad- {U^2} (\si^V_{t_n})^2 + 2\Den\vp^V_{t_n}(U_n)  \biggr)+o^\uc(\sqrt{\Den}),\\
		\end{split}\raisetag{-2.5\baselineskip}\end{equation}
	respectively. This implies that instead of \eqref{eq:F}, we now have
	\begin{align*} 
		&F^n_{2}(u,\sqrt{2}U)_T -F^n_{1}(u,\sqrt{2}U)_T-F^n_{2}(u,U)_T +F^n_{1}(u,U)_T \\
		&\quad=\frac{U^{-2}}{\lvert \calt\rvert}\sum_{t\in\calt} \Re (2\Den\vp^V_{t_n}(U_n) -\Den\vp^V_{t_n}(2U_n))+o^\uc(\sqrt{\Den}).
	\end{align*}
	For any $t$, %$t\in A_\tau$
	\begin{equation*}
		\Re(\vp^V_t(U_n))=	\int_{[-1,1]} (\cos(U_nz)-1)F'_t(dz) + O^\uc(\Den^{-\rho\beta_t/2}),
	\end{equation*}
	which shows that the  contribution of $F''_t(dz)$ to \eqref{eq:F} is $O^\uc(\Den^{1-\rho\beta^\ast/2})$. Therefore, its total contribution to $R^n_T(u,U)$ under \ref{ass:H1} is $O^\uc(\sqrt{k_n\lvert\calt}\rvert\Den^{1-\rho\beta^\ast /2})=O^\uc(\Den^{(1-\rho\beta^\ast)/2})$. Regarding the contribution of $F'_t(dz)$, we note that for any $t\in A_\tau$ such that $\beta_t>0$, 
	\begin{align*}
		\Re(\vp^V_t(U_n))&= \int_{-\eps^-_t}^{\eps^+_t} \frac{\cos(U_n z)-1}{\lvert z\rvert^{1+\beta_t}} (\log \lvert z\rvert^{-1})^{\beta'_t}(c_t^+\bone_{\{z>0\}}+c_t^-\bone_{\{z<0\}}) dz+ O^\uc(\Den^{-\rho\beta_t/2})\\
		&=-U_n^{\beta_t}\int_{-U_n\eps^-_t}^{U_n\eps^+_t} \frac{1-\cos z}{\lvert z\rvert^{1+\beta_t}} (\log (U_n/\lvert z\rvert))^{\beta'_t}(c_t^+\bone_{\{z>0\}}+c_t^-\bone_{\{z<0\}}) dz\\
		&\quad+ O^\uc(\Den^{-\rho\beta_t/2})\\
		&=-U_n^{\beta_t}\int_\R \frac{1-\cos z}{\lvert z\rvert^{1+\beta_t}} \lvert \log (U_n/\lvert z\rvert)\rvert^{\beta'_t}(c_t^+\bone_{\{z>0\}}+c_t^-\bone_{\{z<0\}}) dz\\
		&\quad+ O^\uc(\Den^{-\rho\beta_t/2}).
	\end{align*}
	Therefore, under \ref{ass:H1},
	\begin{align*}
		&\sqrt{k_n\lvert\calt\rvert} \bigl(F^n_{2}(u,\sqrt{2}U)_T -F^n_{1}(u,\sqrt{2}U)_T-F^n_{2}(u,U)_T +F^n_{1}(u,U)_T\bigr) \\
		&\quad=	\frac{\sqrt{k_n\lvert\calt\rvert}}{\lvert \calt\rvert}\sum_{t\in\calt} U^{\beta_{t_n}-2}\Den^{1-\beta_{t_n}/2} \int_\R \frac{1-\cos z}{\lvert z\rvert^{1+\beta_{t_n}}}(c_{t_n}^+\bone_{\{z>0\}}+c_{t_n}^-\bone_{\{z<0\}})\\
		&\qquad\times\bigl[2^{\beta_{t_n}} \lvert\log ( 2U_n/\lvert z\rvert)\rvert^{\beta'_{t_n}}-2\lvert\log (U_n/\lvert z\rvert)\rvert^{\beta'_{t_n}}\bigr] dz+O^\uc(\Den^{(1-\rho\beta^\ast)/2}).
	\end{align*}
	There is no loss of generality to assume that $\rho\beta^\ast >1$. In this case, if we choose $\eps^\ast<\rho\beta^\ast-1$ and restrict the previous sum over $t$ to times $t\in\calt$ for which $\beta_{t_n}>1+\eps^\ast$, what we omit are terms that are $O^\uc(\Den^{-\eps^\ast/2}(\log \Den^{-1})^{\beta'_{t_n}})=o^\uc(\Den^{(1-\rho\beta^\ast)/2})$ and hence negligible. As a result,
	\begin{equation}\label{eq:aux}\begin{split}
			&\sqrt{k_n\lvert\calt\rvert} \bigl(F^n_{2}(u,\sqrt{2}U)_T -F^n_{1}(u,\sqrt{2}U)_T-F^n_{2}(u,U)_T +F^n_{1}(u,U)_T\bigr) \\
			&\quad=  \frac{\sqrt{k_n\lvert\calt\rvert}}{\lvert \calt\rvert}\sum_{t\in\calt} \bone_{\{\beta_{t_n}>1+\eps^\ast\}}U^{\beta_{t_n}-2}\Den^{1-\beta_{t_n}/2} \int_\R \frac{1-\cos z}{\lvert z\rvert^{1+\beta_{t_n}}}(c_{t_n}^+\bone_{\{z>0\}}+c_{t_n}^-\bone_{\{z<0\}})\\
			&\qquad\times\bigl[2^{\beta_{t_n}} \lvert\log ( 2U_n/\lvert z\rvert)\rvert^{\beta'_{t_n}}-2\lvert\log (U_n/\lvert z\rvert)\rvert^{\beta'_{t_n}}\bigr] dz  +O^\uc(\Den^{(1-\rho\beta^\ast)/2}).
	\end{split}\!\!\!\raisetag{-2.5\baselineskip}\end{equation}
	
	Next,  we claim that the $dz$-integral is eventually nonnegative for large $n$, uniformly for $t_n$ such that $\beta_{t_n}>1+\eps^\ast$. This is clear if $\beta'_{t_n}=0$ or if $c_{t_n}^+=c_{t_n}^-=0$. If all other cases,   write
	\begin{align*}
		&\int_\R \frac{1-\cos z}{\lvert z\rvert^{1+\beta_{t_n}}}(c_{t_n}^+\bone_{\{z>0\}}+c_{t_n}^-\bone_{\{z<0\}}) \bigl[2^{\beta_{t_n}} \lvert\log( 2U_n/\lvert z\rvert)\rvert^{\beta'_{t_n}}-2\lvert\log (U_n/\lvert z\rvert)\rvert^{\beta'_{t_n}}\bigr] dz  \\
		&\quad=\int_\R \frac{1-\cos z}{\lvert z\rvert^{1+\beta_{t_n}}}(c_{t_n}^+\bone_{\{z>0\}}+c_{t_n}^-\bone_{\{z<0\}}) (2^{\beta_{t_n}} -2) \lvert\log (U_n/\lvert z\rvert)\rvert^{\beta'_{t_n}} dz\\
		&\qquad+2^{\beta_{t_n}}\int_\R \frac{1-\cos z}{\lvert z\rvert^{1+\beta_{t_n}}}(c_{t_n}^+\bone_{\{z>0\}}+c_{t_n}^-\bone_{\{z<0\}}) \bigl[\lvert\log ( 2U_n/\lvert z\rvert)\rvert^{\beta'_{t_n}}-\lvert\log (U_n/\lvert z\rvert) \rvert^{\beta'_{t_n}}\bigr]  dz
	\end{align*}
	and note that the first term on the right-hand side is  
	\begin{align*}
		&\geq \int_{-\sqrt{U_n}}^{\sqrt{U_n}} \frac{1-\cos z}{\lvert z\rvert^{1+\beta_{t_n}}}(c_{t_n}^+\bone_{\{z>0\}}+c_{t_n}^-\bone_{\{z<0\}}) (2^{\beta_{t_n}} -2) (\log U_n - \log \lvert z\rvert)^{\beta'_{t_n}} dz \\
		&\geq 2^{-(\beta'_{t_n})_+}(2^{\beta_{t_n}} -2)(\log U_n)^{\beta'_{t_n}}\int_{-\sqrt{U_n}}^{\sqrt{U_n}} \frac{1-\cos z}{\lvert z\rvert^{1+\beta_{t_n}}}(c_{t_n}^+\bone_{\{z>0\}}+c_{t_n}^-\bone_{\{z<0\}})    dz\\
		&=2^{-(\beta'_{t_n})_+}(2^{\beta_{t_n}} -2)(\log U_n)^{\beta'_{t_n}}\int_\R \frac{1-\cos z}{\lvert z\rvert^{1+\beta_{t_n}}}(c_{t_n}^+\bone_{\{z>0\}}+c_{t_n}^-\bone_{\{z<0\}})    dz + O^\uc(\Den^{1/4}),
	\end{align*}
	while the second one is, in absolute value,
	\begin{align*}
		&\leq 2^{\beta_{t_n}}\int_{-\sqrt{U_n}}^{\sqrt{U_n}} \frac{1-\cos z}{\lvert z\rvert^{1+\beta_{t_n}}}(c_{t_n}^+\bone_{\{z>0\}}+c_{t_n}^-\bone_{\{z<0\}}) \bigl\lvert(\log ( 2U_n/\lvert z\rvert))^{\beta'_{t_n}}-(\log (U_n/\lvert z\rvert) )^{\beta'_{t_n}}\bigr\rvert  dz\\
		&\quad + O^\uc(\Den^{1/4})\\
		&\leq 2^{\beta_{t_n}}\lvert\beta'_{t_n}\rvert\int_{-\sqrt{U_n}}^{\sqrt{U_n}}\frac{1-\cos z}{\lvert z\rvert^{1+\beta_{t_n}}}(c_{t_n}^+\bone_{\{z>0\}}+c_{t_n}^-\bone_{\{z<0\}}) (\log(2^{\bone_{\{\beta'_{t_n}\geq1\}}}U_n/\lvert z\rvert))^{\beta'_{t_n}-1}\\
		&\quad\times (\log (2U_n/\lvert z\rvert)-\log(U_n/\lvert z\rvert))^{ \beta'_{t_n}} dz+ O^\uc(\Den^{1/4})
	\end{align*} 
	by the mean-value theorem. As this is $O^\uc((\log U_n)^{\beta'_{t_n}-1})$, the first integral dominates and our claim is proved. In the following, we assume that $n$ is large enough such that the second integral is at most one half of the first one in absolute value.
	
	In order to conclude, note that $\beta$, $c^+$ and $c^-$ are c\`agl\`ad and $\beta^\ast>1$ by assumption. Therefore, for $n$ large there exists (a possibly random) interval $I=[\tau_1,\tau_2]\subseteq[0,\tau]$ and some small $\delta>0$ such that $\Leb(I)>0$, $c_I = \inf_{t\in [\tau_1-\delta,\tau_2]} (c^+_t + c^-_t) >0$ and $\beta_I = \inf_{t\in [\tau_1-\delta,\tau_2]} \beta_t > \rho\beta^\ast$. Note that $\beta_I>1+\eps^\ast$ by the definition of $\eps^\ast$. Hence,  dividing both sides of \eqref{eq:aux} by $\Den^{(1-\rho\beta^\ast)/2}$, we have for sufficiently large $n$ that (recall that $k_n\lvert \calt\rvert \sim \tau/\Den$)
	\begin{align*}
		&\frac{\sqrt{k_n\lvert\calt\rvert}}{\Den^{(1-\rho\beta^\ast)/2}} \bigl(F^n_{2}(u,\sqrt{2}U)_T -F^n_{1}(u,\sqrt{2}U)_T-F^n_{2}(u,U)_T +F^n_{1}(u,U)_T\bigr) \\
		&~ \geq\frac{\sqrt{\tau}(2^{\beta_I} -2)c_I}{2^{1+C} }(U^{-1}\vee 1) \Den^{(\rho\beta^\ast-\beta_{I})/2}(\log U_n)^{-C} \int_0^\infty \frac{1-\cos z}{  z ^2 \vee z^3}  dz\sum_{t\in\calt \cap I} \frac{1}{\lvert\calt\rvert}  +O^\uc(1)\\
		&~\sim\frac{(2^{\beta_I} -2)c_I}{2^{1+C}\sqrt{\tau}  }(U^{-1}\vee 1) \Den^{(\rho\beta^\ast-\beta_{I})/2}(\log U_n)^{-C} \int_0^\infty \frac{1-\cos z}{  z ^2 \vee z^3}  dz, 
	\end{align*}
	which diverges in probability to $+\infty$. This completes the proof of the theorem under \ref{ass:H1}.
\end{proof}

	\begin{acks}[Acknowledgments]
		The authors would like to thank an anonymous referee, an Associate
		Editor and the Editor for their comments that improved the
		quality of this paper.
	\end{acks}

	\bibliographystyle{abbrvnat}
	\bibliography{ovv}

\begin{thebibliography}{33}
\providecommand{\natexlab}[1]{#1}
\providecommand{\url}[1]{\texttt{#1}}
\expandafter\ifx\csname urlstyle\endcsname\relax
  \providecommand{\doi}[1]{doi: #1}\else
  \providecommand{\doi}{doi: \begingroup \urlstyle{rm}\Url}\fi

\bibitem[A\"{\i}t-Sahalia and Jacod(2011)]{AJ11}
Y.~A\"{\i}t-Sahalia and J.~Jacod.
\newblock Testing whether jumps have finite or infinite activity.
\newblock \emph{Ann. Statist.}, 39\penalty0 (3):\penalty0 1689--1719, 2011.

\bibitem[A{\"i}t-Sahalia and Jacod(2014)]{AJ14}
Y.~A{\"i}t-Sahalia and J.~Jacod.
\newblock \emph{High-Frequency Financial Econometrics}.
\newblock Princeton University Press, Princeton, 2014.

\bibitem[A\"{\i}t-Sahalia et~al.(2017)A\"{\i}t-Sahalia, Fan, Laeven, Wang, and
  Yang]{AFLWY17}
Y.~A\"{\i}t-Sahalia, J.~Fan, R.~J.~A. Laeven, C.~D. Wang, and X.~Yang.
\newblock Estimation of the continuous and discontinuous leverage effects.
\newblock \emph{J. Amer. Statist. Assoc.}, 112\penalty0 (520):\penalty0
  1744--1758, 2017.

\bibitem[Andersen et~al.(2002)Andersen, Benzoni, and
  Lund]{andersen2002empirical}
T.~G. Andersen, L.~Benzoni, and J.~Lund.
\newblock An empirical investigation of continuous-time equity return models.
\newblock \emph{J. Finance}, 57\penalty0 (3):\penalty0 1239--1284, 2002.

\bibitem[Bentata and Cont(2012)]{bentata2012short}
A.~Bentata and R.~Cont.
\newblock Short-time asymptotics for marginal distributions of semimartingales.
\newblock \emph{arXiv preprint arXiv:1202.1302}, 2012.

\bibitem[Carr and Madan(2001)]{CM01}
P.~Carr and D.~Madan.
\newblock Optimal positioning in derivative securities.
\newblock \emph{Quant. Finance}, 1\penalty0 (1):\penalty0 19--37, 2001.

\bibitem[Chernov et~al.(2003)Chernov, Gallant, Ghysels, and
  Tauchen]{chernov2003alternative}
M.~Chernov, A.~R. Gallant, E.~Ghysels, and G.~Tauchen.
\newblock Alternative models for stock price dynamics.
\newblock \emph{J. Econometrics}, 116\penalty0 (1-2):\penalty0 225--257, 2003.

\bibitem[Chong and Todorov(2024)]{CT23_b}
C.~H. Chong and V.~Todorov.
\newblock Volatility of volatility and leverage effect from options.
\newblock \emph{J. Econometrics}, 240\penalty0 (1):\penalty0 105669, 2024.

\bibitem[Chong and Todorov(2025)]{CT22}
C.~H. Chong and V.~Todorov.
\newblock Short-time expansion of characteristic functions in a rough
  volatility setting with application.
\newblock \emph{Bernoulli}, 2025.
\newblock Forthcoming.

\bibitem[Clinet and Potiron(2021)]{clinet2021estimation}
S.~Clinet and Y.~Potiron.
\newblock Estimation for high-frequency data under parametric market
  microstructure noise.
\newblock \emph{Ann. Inst. Statist. Math.}, 73\penalty0 (4):\penalty0 649--669,
  2021.

\bibitem[Curato(2019)]{C19}
I.~V. Curato.
\newblock Estimation of the stochastic leverage effect using the {Fourier}
  transform method.
\newblock \emph{Stochastic Process. Appl.}, 129\penalty0 (9):\penalty0
  3207--3238, 2019.

\bibitem[Duffie et~al.(2000)Duffie, Pan, and Singleton]{DPS00}
D.~Duffie, J.~Pan, and K.~Singleton.
\newblock Transform analysis and asset pricing for affine jump-diffusions.
\newblock \emph{Econometrica}, 68\penalty0 (6):\penalty0 1343--1376, 2000.

\bibitem[Duffie et~al.(2003)Duffie, Filipovi{\'c}, and
  Schachermayer]{duffie2003affine}
D.~Duffie, D.~Filipovi{\'c}, and W.~Schachermayer.
\newblock Affine processes and applications in finance.
\newblock \emph{Ann. Appl. Probab.}, 13\penalty0 (3):\penalty0 984--1053, 2003.

\bibitem[Figueroa-L{\'o}pez and Houdr{\'e}(2009)]{figueroa2009small}
J.~E. Figueroa-L{\'o}pez and C.~Houdr{\'e}.
\newblock Small-time expansions for the transition distributions of {L}{\'e}vy
  processes.
\newblock \emph{Stochastic Process. Appl.}, 119\penalty0 (11):\penalty0
  3862--3889, 2009.

\bibitem[Figueroa-L{\'o}pez et~al.(2012)Figueroa-L{\'o}pez, Gong, and
  Houdr{\'e}]{figueroa2012small}
J.~E. Figueroa-L{\'o}pez, R.~Gong, and C.~Houdr{\'e}.
\newblock Small-time expansions of the distributions, densities, and option
  prices of stochastic volatility models with {L}{\'e}vy jumps.
\newblock \emph{Stochastic Process. Appl.}, 122\penalty0 (4):\penalty0
  1808--1839, 2012.

\bibitem[Gatheral et~al.(2018)Gatheral, Jaisson, and
  Rosenbaum]{gatheral2018volatility}
J.~Gatheral, T.~Jaisson, and M.~Rosenbaum.
\newblock Volatility is rough.
\newblock \emph{Quant. Finance}, 18\penalty0 (6):\penalty0 933--949, 2018.

\bibitem[Jacod and Protter(2012)]{JP12}
J.~Jacod and P.~Protter.
\newblock \emph{Discretization of Processes}, volume~67 of \emph{Stochastic
  Modelling and Applied Probability}.
\newblock Springer, Heidelberg, 2012.

\bibitem[Jacod and Reiss(2014)]{JR14}
J.~Jacod and M.~Reiss.
\newblock A remark on the rates of convergence for integrated volatility
  estimation in the presence of jumps.
\newblock \emph{Ann. Statist.}, 42\penalty0 (3):\penalty0 1131--1144, 2014.

\bibitem[Jacod and Shiryaev(2003)]{JS03}
J.~Jacod and A.~N. Shiryaev.
\newblock \emph{Limit Theorems for Stochastic Processes}, volume 288 of
  \emph{Grundlehren der mathematischen Wissenschaften [Fundamental Principles
  of Mathematical Sciences]}.
\newblock Springer-Verlag, Berlin, second edition, 2003.

\bibitem[Jacod and Todorov(2014)]{jacod2014efficient}
J.~Jacod and V.~Todorov.
\newblock Efficient estimation of integrated volatility in presence of infinite
  variation jumps.
\newblock \emph{Ann. Statist.}, 42\penalty0 (3):\penalty0 1029--1069, 2014.

\bibitem[Jacod and Todorov(2018)]{jacod2018limit}
J.~Jacod and V.~Todorov.
\newblock Limit theorems for integrated local empirical characteristic
  exponents from noisy high-frequency data with application to volatility and
  jump activity estimation.
\newblock \emph{Ann. Appl. Probab.}, 28\penalty0 (1):\penalty0 511--576, 2018.

\bibitem[Kalnina and Xiu(2017)]{KX17}
I.~Kalnina and D.~Xiu.
\newblock Nonparametric estimation of the leverage effect: a trade-off between
  robustness and efficiency.
\newblock \emph{J. Amer. Statist. Assoc.}, 112\penalty0 (517):\penalty0
  384--396, 2017.

\bibitem[Li et~al.(2022)Li, Liu, and Zhang]{li2022volatility}
Y.~Li, G.~Liu, and Z.~Zhang.
\newblock Volatility of volatility: Estimation and tests based on noisy high
  frequency data with jumps.
\newblock \emph{J. Econometrics}, 229\penalty0 (2):\penalty0 422--451, 2022.

\bibitem[Muhle-Karbe and Nutz(2011)]{MN11}
J.~Muhle-Karbe and M.~Nutz.
\newblock Small-time asymptotics of option prices and first absolute moments.
\newblock \emph{J. Appl. Probab.}, 48\penalty0 (4):\penalty0 1003--1020, 2011.

\bibitem[Qin and Todorov(2019)]{QT19}
L.~Qin and V.~Todorov.
\newblock Nonparametric implied {L}{\'e}vy densities.
\newblock \emph{Ann. Statist.}, 47\penalty0 (2):\penalty0 1025--1060, 2019.

\bibitem[R\"uschendorf and Woerner(2002)]{ruschendorf2002expansion}
L.~R\"uschendorf and J.~H.~C. Woerner.
\newblock Expansion of transition distributions of {L}\'evy processes in small
  time.
\newblock \emph{Bernoulli}, 8\penalty0 (1):\penalty0 81--96, 2002.

\bibitem[Sanfelici et~al.(2015)Sanfelici, Curato, and
  Mancino]{sanfelici2015high}
S.~Sanfelici, I.~V. Curato, and M.~E. Mancino.
\newblock High-frequency volatility of volatility estimation free from spot
  volatility estimates.
\newblock \emph{Quant. Finance}, 15\penalty0 (8):\penalty0 1331--1345, 2015.

\bibitem[Todorov(2019)]{T19}
V.~Todorov.
\newblock Nonparametric spot volatility from options.
\newblock \emph{Ann. Appl. Probab.}, 29\penalty0 (6):\penalty0 3590--3636,
  2019.

\bibitem[Todorov(2021)]{T21}
V.~Todorov.
\newblock Higher-order small time asymptotic expansion of {I}t{\^o}
  semimartingale characteristic function with application to estimation of
  leverage from options.
\newblock \emph{Stochastic Process. Appl.}, 142:\penalty0 671--705, 2021.

\bibitem[Todorov and Zhang(2021)]{todorov2021bias}
V.~Todorov and Y.~Zhang.
\newblock Bias reduction in spot volatility estimation from options.
\newblock \emph{J. Econometrics}, 2021.

\bibitem[Toscano et~al.(2022)Toscano, Livieri, Mancino, and
  Marmi]{toscano2022volatility}
G.~Toscano, G.~Livieri, M.~E. Mancino, and S.~Marmi.
\newblock Volatility of volatility estimation: {C}entral limit theorems for the
  {F}ourier transform estimator and empirical study of the daily time series
  stylized facts.
\newblock \emph{J. Financ. Economet.}, 2022.

\bibitem[Vetter(2015)]{vetter2015estimation}
M.~Vetter.
\newblock Estimation of integrated volatility of volatility with applications
  to goodness-of-fit testing.
\newblock \emph{Bernoulli}, 21\penalty0 (4):\penalty0 2393--2418, 2015.

\bibitem[Wang and Mykland(2014)]{WM14}
C.~D. Wang and P.~A. Mykland.
\newblock The estimation of leverage effect with high-frequency data.
\newblock \emph{J. Amer. Statist. Assoc.}, 109\penalty0 (505):\penalty0
  197--215, 2014.

\end{thebibliography}
	
	\newpage
	\begin{appendix}
		
		\section*{Proof of Technical Results}

		\begin{proof}[Proof of Proposition 7.1]  Recall that $x$ is a special deep It\^o semimartingale under Assumption~1, which means that we have (2.12) and (2.13). 
			Sometimes we want to single out the   martingale components, so we define
			\begin{equation}\label{eq:not2}
				\begin{split}
					&\theta^i_\cc(t)=(\theta(t,\bz_i)_{j_1,\dots,j_i})_{j_1,\dots,j_i\in\{1,\dots,d\}},	\quad	\theta_\mm(t,\bz_i)=(\theta(t,\bz_i)_{j_1,\dots,j_i})_{j_1,\dots,j_i\in\{1,\dots,d,d+2\}},	 \\
					&y_\cc(ds)=d\mathbb{W}_s,\quad	y_\mm(ds,dz)=(dW^{(1)}_s\delta_0(dz),\dots,dW^{(d)}_s\delta_0(dz), \wh\mu(ds,dz))^\top,\\
					&\by_\cc(d\bs_i)=y_\cc(ds_i)\cdots y_\cc(ds_1),\quad \by_\mm(d\bs_i,d\bz_i)=y_\mm(ds_i,dz_i)\cdots y_\mm(ds_1,dz_1)
				\end{split}  
			\end{equation}
			and denote the square and angle bracket of $y_\cc$ and $y_\mm$  by
			\begin{equation}\label{eq:QVy} \begin{split}
					[y_\cc](ds)&=\langle y_\cc\rangle(ds)=\mathrm{diag}\bigl(ds,\dots, ds\bigr),\\
					[y_\mm](ds,dz)&=\mathrm{diag}\bigl(ds\delta_0(dz),\dots, ds\delta_0(dz), \mu(ds,dz)\bigr),\\
					\langle y_\mm\rangle(ds,dz)&=\mathrm{diag}\bigl(ds\delta_0(dz),\dots, ds\delta_0(dz),\la(s,z) dsF(dz)\bigr),
				\end{split}
			\end{equation}
			which are $d\times d$-dimensional  and $(d+1)\times(d+1)$-dimensional measures, respectively. Under Assumption~1, we can employ a classical localization argument (cf.\ Lemma 3.4.5 in \cite{JP12}) and assume without loss of generality that $T_n=\infty$ for all $n\in\N$. % and $\la(t,z)\leq AJ_1(z)$ for some constant $A>0$ and all $t>0$ and $z\in\R^{d'}$. 
			Furthermore, we can assume that (2.15) holds without $\wedge 1$. To simplify notation, we write $J(z)=J_1(z)$, $j(z)=j_1(z,v)$ and $\mathcal{J}(\bz_{N+1})=\mathcal{J}_1(\bz_{N+1})$ as well as $\ov F(dz)=J(z)F(dz)$ in the following.
			
			With this preparation and the notation
			\begin{equation*}
				\S^{(i)}_{t,T} = \{ (\bs_i,\bz_i) \in [0,\infty)^i \times (\R^{d'})^i: t<s_N<\dots<s_1<t+T \},
			\end{equation*}
			we can decompose an increment of $x$ into its $N+1$ layers as follows:
			\begin{equation*}
				x_{t+T}-x_t = \sum_{k=1}^{N} \int_{\S^{(k)}_{t,T}} \theta(t,\bz_k) \by(d\bs_k,d\bz_k) + \int_{\S^{(N+1)}_{t,T}} \theta(\bs_{N+1},\bz_{N+1}) \by(d\bs_{N+1},d\bz_{N+1}).
			\end{equation*}
			So writing $\ov x_{t,T}= \sum_{k=1}^N \int_{\S^{(k)}_{t,T}} \theta(t,\bz_k) \by(d\bs_k,d\bz_k)$ %,\qquad %\ov\calr_{t,T}(u)= \E_t[\cos(u_T\ov x_{t,T})], 
			and
			$	\ov\call_{t,T}(u)= \E_t[e^{iu_T\ov x_{t,T}}]$, 
			we  have that 
			\begin{equation}\label{eq:remain} 
				\begin{split}
					\lvert\call_{t,T}(u)-\ov \call_{t,T}(u)\rvert	&=\lvert\E_t[e^{iu_T(x_{t+T}-x_t)}-e^{iu_T \ov x_{t,T}}]\rvert\\
					& \leq u_T\E_t\biggl[\biggl\lvert\int_{\S^{(N+1)}_{t,T}} \theta(s_{N+1},\bz_{N+1}) \by(d\bs_{N+1},d\bz_{N+1})\biggr\rvert\biggr]\\
					& \leq u_T\E_t\biggl[\biggl\lvert\int_{\S^{(N+1)}_{t,T}} \theta(s_{N+1},\bz_{N+1}) \by(d\bs_{N+1},d\bz_{N+1})\biggr\rvert^2\biggr]^{1/2}\\
					&= O^\uc(T^{N/2}).
				\end{split}
			\end{equation}
			Informally speaking, the last step follows from  the fact that the size of an increment of an It\^o semimartingale is locally dominated by the Gaussian part. In other words, for each $\ell=1,\dots,N+1$, integration with respect to $y(ds_\ell,dz_\ell)$ gives a factor of $T^{1/2}$, leading to the last bound above.
			To obtain it rigorously, we first use the triangle inequality for the $L^2$-norm to break the last integral in \eqref{eq:remain} into a sum of scalar  integrals and then combine It\^o's isometry  with (2.16) and the bounds $\E[\lvert \int_t^{t+T} H_s ds\rvert^2]^{1/2} \leq \int_t^{t+T}\E[\lvert H_s\rvert^2]^{1/2} ds \leq\bigl( \int_t^{t+T}  \E[\lvert H_s\rvert^2]ds \bigr)^{1/2}$ (for $T\leq1$) and $\la(s,z)\leq AJ(z)$  to get 
			\begin{align*}
				&\E_t\biggl[\biggl\lvert\int_{\S^{(N+1)}_{t,T}} \theta(s_{N+1},\bz_{N+1})_{j_1,\dots,j_{N+1}} y_{j_{N+1}}(ds_{N+1},dz_{N+1})\cdots y_{j_1}(ds_1,dz_1)\biggr\rvert^2\biggr]\\
				%&\quad\leq \int_t^{t+T} \E_t\biggl[\biggl\lvert\int_{\S^{(N)}_{t,s_1-t}} \theta(s_{N+1},\bz_{N+1})_{j_1,\dots,j_{N+1}} y_{j_{N+1}}(ds_{N+1},dz_{N+1})\cdots y_{j_2}(ds_2,dz_2) \biggr\rvert^2\biggr] ds_1\delta_0(dz_1)\\
				&\quad \leq\int_{\S^{(N+1)}_{t,T}} \E_t[ \theta(s_{N+1},\bz_{N+1})_{j_1,\dots,j_{N+1}}^2] ds_{N+1}\cdots ds_1  \prod_{\ell: j_\ell \neq d+2} \delta_0(dz_{j_\ell})\prod_{\ell: j_\ell=d+2} A \ov F(dz_{j_\ell}) \\
				&\quad \leq\frac{T^{N+1}}{(N+1)!}\int_{(\R^{d'})^{N+1}} \mathcal{J}(\bz_{N+1})_{j_1,\dots,j_{N+1}}    \prod_{\ell: j_\ell \neq d+2} \delta_0(dz_{j_\ell})\prod_{\ell: j_\ell=d+2} A \ov F(dz_{j_\ell})\biggr),		
			\end{align*}
			which is $O(T^{N+1})$ by (2.17) and thus implies the last step in \eqref{eq:remain}. %by \eqref{eq:rel}. 
			
			Next, consider the main part, which is
			\[\Delta^n_i \ov\call_{t,T}(u)= A^{n,i}_{t,T}(u)-B^{n,i}_{t,T}(u) + C^{n,i}_{t,T}(u),\]
			where
			\begin{equation}\label{eq:ABC}\begin{split}
					A^{n,i}_{t,T}(u)	&=\E_{t^n_{i-1}}\biggl[\exp\biggl(iu_{T^n_{i-1}} \sum_{k=1}^N \int_{\S^{(k)}_{t^n_{i-1},T^n_{i-1}}} \theta(t^n_{i-1},\bz_k)\by(d\bs_k,d\bz_k)\biggr)\\
					&\quad\qquad\qquad \qquad-\exp\biggl(iu_{T^n_{i-1}}  \sum_{k=1}^N \int_{\S^{(k)}_{t^n_{i-1},T^n_{i-1}}} \theta(t^n_i,\bz_k)\by(d\bs_k,d\bz_k)\biggr) \biggr],\\
					B^{n,i}_{t,T}(u)	&=\E_{t^n_i}\biggl[\exp\biggl(iu_{T^n_{i-1}}  \sum_{k=1}^N \int_{\S^{(k)}_{t^n_{i},T^n_{i}}} \theta(t^n_i,\bz_k)\by(d\bs_k,d\bz_k)\biggr)\\
					& \qquad\qquad\qquad -\exp\biggl(iu_{T^n_{i-1}}  \sum_{k=1}^N \int_{\S^{(k)}_{t^n_{i-1},T^n_{i-1}}} \theta(t^n_i,\bz_k)\by(d\bs_k,d\bz_k)\biggr) \biggr],\\
					C^{n,i}_{t,T}(u)	&=\E_{t^n_i}\biggl[\exp\biggl(iu_{T^n_{i-1}}  \sum_{k=1}^N \int_{\S^{(k)}_{t^n_{i},T^n_{i}}} \theta(t^n_i,\bz_k)\by(d\bs_k,d\bz_k)\biggr)\\
					& \qquad\qquad\qquad -\exp\biggl(iu_{T^n_{i}}  \sum_{k=1}^N \int_{\S^{(k)}_{t^n_{i},T^n_{i}}} \theta(t^n_i,\bz_k)\by(d\bs_k,d\bz_k)\biggr)\biggr].
			\end{split}\end{equation}
			Let us start with the first term. Using the elementary inequality $$\lvert e^{ix}-1-ix\rvert = \sqrt{(1-\cos x)^2+ (x-\sin x)^2} \leq \frac12\lvert x\rvert^2,$$ we have
			\begin{equation}\label{eq:A} 
				\begin{split}
					&\Biggl\lvert	A^{n,i}_{t,T}(u)	-  \E_{t^n_{i-1}}\biggl[ iu_{T^n_{i-1}} \exp\biggl(iu_{T^n_{i-1}} \sum_{k=1}^N \int_{\S^{(k)}_{t^n_{i-1},T^n_{i-1}}} \theta(t^n_i,\bz_k)\by(d\bs_k,d\bz_k)\biggr)\\
					& \qquad\qquad\qquad\qquad\qquad \times\sum_{k=1}^N \int_{\S^{(k)}_{t^n_{i-1},T^n_{i-1}}} (\theta(t^n_{i-1},\bz_k)-\theta(t^n_{i},\bz_k))\by(d\bs_k,d\bz_k)\biggr]\Biggr\rvert\\
					&\quad\leq \frac12 \E_{t^n_{i-1}}\biggl[u^2_{T^n_{i-1}}\biggl(\sum_{k=1}^N \int_{\S^{(k)}_{t^n_{i-1},T^n_{i-1}}} (\theta(t^n_{i-1},\bz_k)-\theta(t^n_{i},\bz_k))\by(d\bs_k,d\bz_k)\biggr)^2\biggr].
				\end{split}
			\end{equation}
			Because $\theta(t,\bz_k)$ is an It\^o semimartingale, we have similarly to \eqref{eq:remain} that the last integral in the previous display is $O^\uc(\sqrt{\Den}T^{k/2})= O^\uc(\sqrt{\Den T})$. Thus, the last line above is $O^\uc(\Den)$. Furthermore, in the summation over $k$ in the second line of \eqref{eq:A}, terms corresponding to $k\geq4$ only contribute a term of order $O^\uc(\sqrt{\Den}T^{3/2})=o^\uc(\sqrt{\Den}T)$ to $A^{n,i}_{t,T}(u)$. Similarly, if we fix $k=1$ (resp., $k=2$ or $k=3$) in the second line, we only have to keep, in the summation in the first line of \eqref{eq:A}, the terms corresponding to $k=1,2,3$ (resp., $k=1,2$ or $k=1$), if we allow for an $o^\uc(\sqrt{\Den}T)$-error. %And lastly, we note that if we integrate with respect to $y(ds,dz)$ over an interval of length $O(T)$, then the diffusive parts of $y$ give rise to a term of order $O(\sqrt{T})$, while the drift and jump part of $y$ give rise to terms of order $O(T)$ and $o(\sqrt{T})$, respectively; see Equations (2.1.33) and (2.1.34) and Corollary 2.1.9 in \cite{JP12}. 
			%	again by the same reason, in any double integral that remains, we only need to keep martingale components (a drift part would yield an extra $T^{1/2}$  making the resulting contribution $O^\uc(\Den^{1/2}T)$). 
			Thus,
			\begin{equation}\label{eq:A-dec} 
				A^{n,i}_{t,T}(u)=A^{n,i,1}_{t,T}(u)+A^{n,i,2}_{t,T}(u)+A^{n,i,3}_{t,T}(u)+O^\uc(\Den) + o^\uc(\sqrt{\Den}T),
			\end{equation}
			where (with the notation in (7.2))
			\begin{align*}
				A^{n,i,1}_{t,T}(u)	&=  \E_{t^n_{i-1}}\biggl[ iu_{T^n_{i-1}} \iint_{t^n_{i-1}}^{t+T} \Delta^n_i\theta(t,z)y(ds,dz)\exp\biggl(iu_{T^n_{i-1}} \iint_{t^n_{i-1}}^{t+T} \theta(t^n_i,z)y(ds,dz) \\
				& \quad +iu_{T^n_{i-1}}\iint_{t^n_{i-1}}^{t+T}\iint_{t^n_{i-1}}^{s-} \theta(t^n_i,z,z')y(dr,dz')y(ds,dz)\\
				&\quad +iu_{T^n_{i-1}} \iint_{t^n_{i-1}}^{t+T}\iint_{t^n_{i-1}}^{s-}\iint_{t^n_{i-1}}^{r-} \theta_\mm(t^n_i,z,z',z'')\\
				&\quad\qquad\qquad\qquad\quad\qquad\qquad\qquad\qquad\times y_\mm(dv,dz'')y_\mm(dr,dz')y_\mm(ds,dz) \biggr)\biggr],\\
				A^{n,i,2}_{t,T}(u)	&= \E_{t^n_{i-1}}\biggl[ iu_{T^n_{i-1}} \iint_{t^n_{i-1}}^{t+T}\iint_{t^n_{i-1}}^{s-} \Delta^n_i\theta(t,z,z') y(dr,dz')y(ds,dz) \\
				& \quad\times\exp\biggl(iu_{T^n_{i-1}} \biggl[\iint_{t^n_{i-1}}^{t+T} \theta(t^n_i,z)y(ds,dz) \\
				&\quad\qquad\qquad\qquad+   \iint_{t^n_{i-1}}^{t+T}\iint_{t^n_{i-1}}^{s-} \theta_\mm(t^n_i,z,z') y_\mm(dr,dz')y_\mm(ds,dz) \biggr]\biggr)\biggr],\\
				A^{n,i,3}_{t,T}(u)	&= \E_{t^n_{i-1}}\biggl[ iu_{T^n_{i-1}}  \iint_{t^n_{i-1}}^{t+T}\iint_{t^n_{i-1}}^{s-}\iint_{t^n_{i-1}}^{r-} \Delta^n_i\theta_\mm(t,z,z',z'')\\
				&\quad\qquad\qquad\qquad\quad\qquad\qquad\qquad\qquad\times y_\mm(dv,dz'')y_\mm(dr,dz')y_\mm(ds,dz) \\
				&\quad\times   \exp\biggl(iu_{T^n_{i-1}} \iint_{t^n_{i-1}}^{t+T} \theta_\mm(t^n_i,z)y_\mm(ds,dz)  \biggr)\biggr].
			\end{align*}
			Note that we  replaced some $\theta$ and $y$ by the martingale parts $\theta_\mm$ and $y_\mm$. This was possible because integration with respect to the drift part of $y$ over an interval of length $T^n_{i-1}$ yields a factor of $O(T)$, which, as the reader can quickly check, only results in negligible terms of order $O^\uc(\sqrt{\Den}T^{3/2})=o^\uc(\sqrt{\Den}T)$. 
			
			Consider the last term, $A^{n,i,3}_{t,T}(u)$. Since we condition on $\calf_{t^n_{i-1}}$, both $\Delta^n_i \theta_\mm(t,z,z',z'')$ and $\theta_\mm(t^n_i,z)$ are known. To proceed further, we introduce
			\begin{equation}\label{eq:muni} 
				\mu^{n,i}(dt,dz)=\int_{\R} \bone_{\{0\leq v\leq \la((t\wedge t^n_i)- ,z)\}}  \pf(dt,dz,dv),
			\end{equation}  
			which is equal to $\mu$ up to time $t^n_i$ and a (conditional) Poisson random measure with  intensity measure $\la(t^n_i,z)dt F(dz)$ after time $t^n_i$. Define $\wh\mu^{n,i}$, $y^{n,i}$ and $y^{n,i}_\mm$ accordingly. The idea is now to split each $y_\mm$ appearing in $A^{n,i,3}_{t,T}(u)$ into $y^{n,i}_\mm$ and $y_\mm-y^{n,i}_\mm$ (the only nonzero component of the latter is $\wh\mu-\wh\mu^{n,i}$). By (2.14),
			\begin{equation}\label{eq:14}
				\begin{split}
					&\E_{t^n_{i-1}}\biggl[\biggl\lvert \iint_{t^n_{i-1}}^{t+T} \ga(t^n_i,z) (\wh\mu-\wh\mu^{n,i})(ds,dz)\biggr\rvert^2\biggr] \\
					&\quad\leq \iiint_{t^n_{i-1}}^{t+T} \lvert \ga(t^n_i,z)\rvert^2\E_{t^n_{i-1}}\Bigl[\bone_{\{\la(s,z)\wedge \la(t^n_i,z) <v\leq \la(s,z)\vee \la(t^n_i,z)  \}}\Bigr] ds F(dz) dv \\
					&\quad= \iint_{t^n_{i-1}}^{t+T} \lvert \ga(t^n_i,z)\rvert^r \E_{t^n_{i-1}} [\lvert \la(s,z)-\la(t^n_i,z)\rvert ] ds F(dz)   \\
					&\quad\leq  C \sqrt{T}\iint_{t^n_{i-1}}^{t+T} \lvert \ga(t^n_i,z)\rvert^2 ds J(z)F(dz)  
				\end{split}
			\end{equation}
			for some numerical constant $C$. Since the last bound is $O(T^{3/2})$, replacing $y_\mm$ by $y^{n,i}_\mm$ in the complex exponential in the definition of $A^{n,i,3}_{t,T}(u)$ results in a negligible error of size $O^\uc(\sqrt{\Den}T^{5/4})$. A similar argument shows that this is also true if we replace $y_\mm$ by $y^{n,i}_\mm$ in the triple integral of $A^{n,i,3}_{t,T}(u)$. We conclude that
			%	By assumption, 
			%\begin{equation}\label{eq:la-incr} 
			%	\lvert \la(t',z)-\la(t^n_i,z)\rvert \leq (A_{t'}-A_{t^n_i})J(z),
			%\end{equation} 
			%where $A$ is the Lévy process
			%$$A_t=t+\sum_{i=1}^d W^{(i)}_t + \iiint_0^t \sqrt{j(z',v')} (\pf-\qf)(ds,dz',dv') + \iiint_0^t j(z',v')  \pf(ds,dz',dv').$$ 
			%Since $A_{t'}-A_{t^n_i} = O(\sqrt{T})$, uniformly for all $t'\in[t^n_i,t+T]$, the $L^r$-norm of an integral with respect to $y_\mm-y^{n,i}_\mm$ on an interval of length $T^n_{i-1}$ will be of order $O(T^{3/(2r)})=o(T^{3/4})$. Thus,
			\begin{align}
				A^{n,i,3}_{t,T}(u)	&=\E_{t^n_{i-1}}\biggl[ iu_{T^n_{i-1}}  \iint_{t^n_{i-1}}^{t+T}\iint_{t^n_{i-1}}^{s-}\iint_{t^n_{i-1}}^{r-} \Delta^n_i\theta_\mm(t,z,z',z'')\nonumber\\ &\quad\qquad\qquad\qquad\quad\qquad\qquad\qquad\qquad\times y^{n,i}_\mm(dv,dz'')y^{n,i}_\mm(dr,dz')y^{n,i}_\mm(ds,dz) \label{eq:help-2}\\
				&\quad\times   \exp\biggl(iu_{T^n_{i-1}} \iint_{t^n_{i-1}}^{t+T} \theta_\mm(t^n_i,z)y^{n,i}_\mm(ds,dz)  \biggr)\biggr] +o^\uc(\sqrt{\Den}T).
				\nonumber\end{align}
			Conditionally on $\calf_{t^n_{i-1}}$, the process $s\mapsto  \ov L^{n,i}_t(u,s)= \iint_{t^n_{i-1}}^{s} \theta_\mm(t^n_i,z)y^{n,i}_\mm(dr,dz)$ is a L\'evy process for $s\geq t^n_{i-1}$. Thus, 
			\[\exp\biggl(iu_{T^n_{i-1}} \iint_{t^n_{i-1}}^{t+T} \theta_\mm(t^n_i,z)y^{n,i}_\mm(ds,dz)  \biggr) = e^{-\frac12u^2\si^2_{t^n_i}+T^n_{i-1}\vp_{t^n_i}(u_{T^n_{i-1}})}\ov Z^{n,i}_t(u_{T^n_{i-1}},t+T), \]
			where $\ov Z^{n,i}_t(u,s)$ is the stochastic exponential of $\ov L^{n,i}_t(u,s)$, that is, the solution to the SDE $$\ov Z^{n,i}_t(u,ds)=\ov Z^{n,i}_t(u,s-)\ov L^{n,i}_t(u,ds),\qquad \ov Z^{n,i}_t(u,t^n_{i-1})=1.$$ Consequently, we have the chaos representation
			\begin{equation}\label{eq:chaos-1}\begin{split}
					&	\exp\biggl(iu_{T^n_{i-1}} \iint_{t^n_{i-1}}^{t+T} \theta_\mm(t^n_i,z)y^{n,i}_\mm(ds,dz)  \biggr)  \\
					&\quad = e^{-\frac12u^2\si^2_{t^n_i}+T^n_{i-1}\vp_{t^n_i}(u_{T^n_{i-1}})}\biggl(1+\sum_{k=1}^\infty \int_{\S^{(k)}_{t^n_{i-1},T^n_{i-1}}}   \wh\theta_{t^n_i}(u_{T^n_{i-1}},\bz_k) \by^{n,i}_\mm(d\bs_k,d\bz_k)\biggr), 
			\end{split}\end{equation}
			where 
			\begin{equation}\label{eq:wh-theta} 
				\wh\theta_s(u,\bz_k)_{j_1,\dots,j_k} = \prod_{\ell=1}^k \wh\theta_s(u,z_\ell)_{j_\ell},\quad \wh\theta_s(u,z)=(iu\si_s,0,\ldots,0,e^{iu\ga(s,z)}-1). 
			\end{equation}
			Since multiple Wiener or Poisson integrals of different orders are uncorrelated, only the term for $k=3$ in \eqref{eq:chaos-1} matters for $A^{n,i,3}_{t,T}(u)$. This term   in turn gives rise to $(d+1)^3$ covariance terms since we have three integrals with respect to $y_\mm^{n,i}$ and there is a choice between taking one of the $d$ continuous martingale components or the discontinuous martingale component. %Each time we take the latter, we also get a corresponding factor $e^{iu_{T^n_{i-1}}\ga(t^n_i,z)}-1$ (or with $z'$ or $z''$) from $\wh\theta_{t^n_i}(u_{T^n_{i-1}},\bz_k)$. Bounding its absolute value by $\lvert u_{T^n_{i-1}} \ga(t^n_i,z)\rvert^r
			We claim that only one case contributes asymptotically, namely when we take $W$ for $y_\mm^{n,i}$ every single time, in both \eqref{eq:help-2} and \eqref{eq:chaos-1}. By the scaling properties of Brownian motion and the fact that $\exp(-\frac12u^2\si^2_{t^n_i}+T^n_{i-1}\vp_{t^n_i}(u_{T^n_{i-1}})) = \exp(-\frac12 u^2\si^2_{t^n_i})+o^\uc(1)$,  we would then have (recall \eqref{eq:not2})
			\begin{equation}\label{eq:A3} 
				A^{n,i,3}_{t,T}(u) = \Delta^n_i \theta^3_\cc(t)_{111} C_{t^n_i}(u)T^n_{i-1} + o^\uc(\sqrt{\Den}T)
			\end{equation}
			for some $C_t(u)$ that is a polynomial in $u$ and $\si_t$. Note that  $\theta^3_\cc(t)_{111}$ is $\si^{\si^\si}_t$, if one wants to extend the notation in (2.10) consistently. To see that no other combination  matters asymptotically, first notice that \eqref{eq:chaos-1} does not contain any of the Brownian motions $W^{(2)},\dots,W^{(d)}$. Thus, as soon as one of the triple integrals in \eqref{eq:help-2} is taken with respect to $W^{(i)}$ for some $i\geq2$, we immediately get a zero covariance. In all remaining cases, there is at least one integral with respect to $\wh\mu^{n,i}$. Since they can be treated analogously, we only discuss the covariance (using the notation $\si^{\si^\ga}(t,z)=\theta(t,z,0,0)_{d+2,1,1}$)
			\begin{align*}
				&\E_{t^n_{i-1}}\biggl[ iu_{T^n_{i-1}}  \iint_{t^n_{i-1}}^{t+T}\int_{t^n_{i-1}}^{s-}\int_{t^n_{i-1}}^{r-} \Delta^n_i\si^{\si^\ga}(t,z) dW_vdW_r\wh\mu^{n,i}_\mm(ds,dz) \\
				&\quad\times (iu_{T^n_{i-1}}\si_{t^n_i})^2  \iint_{t^n_{i-1}}^{t+T}\int_{t^n_{i-1}}^{s-}\int_{t^n_{i-1}}^{r-} (e^{iu_{T^n_{i-1}}\ga(t^n_i,z)}-1) dW_vdW_r\wh\mu^{n,i}_\mm(ds,dz)\biggr], 
			\end{align*}
			which can be evaluated to
			\[
			-iu^3_{T^n_{i-1}}\si^2_{t^n_i}  \frac{(T^n_{i-1})^3}{3!}\int_{\R^{d'}} \Delta^n_i\si^{\si^\ga}(t,z)(e^{iu_{T^n_{i-1}}\ga(t^n_i,z)}-1)\nu_{t^n_i}(dz).
			\]
			Since $\lvert e^{ix}-1\rvert\leq 2\lvert x\rvert^{r/2}$ for any $r\in[1,2]$, we can apply the Cauchy--Schwarz inequality to get the bound
			\begin{align*} &\frac{\sqrt{2}u^3\si^2_{t^n_i} (T^n_{i-1})^{3/2}}{6}\biggl(\int_{\R^{d'}} (\Delta^n_i\si^{\si^\ga}(t,z))^2\nu_{t^n_i}(dz)\biggr)^{1/2}\biggl(\int_{\R^{d'}}   \lvert u_{T^n_{i-1}}\ga(t^n_i,z)\rvert^r\nu_{t^n_i}(dz)\biggr)^{1/2} \\
				&\qquad =O^\uc(\sqrt{\Den}T^{3/2-r/4})=o^\uc(\sqrt{\Den}T).\end{align*}
			This shows that we have \eqref{eq:A3}.
			In what follows, we use $v_t$ and $C_t(u)$ to denote generic processes (which may change from line to line) that are It\^o semimartingales in $t$, uniformly bounded in $u$ on compact subsets of $(0,\infty)$, and satisfy $\lvert \Delta^n_i v_t\rvert+\lvert\Delta^n_i C_t(u)\rvert=O(\sqrt{\Den})$, uniformly in $i$ and $u$.
			
			With similar reasoning, we can identify components of $A^{n,i,1}_{t,T}(u)$ and $A^{n,i,2}_{t,T}(u)$ that are $o^\uc(\sqrt{\Den}T)$ or of an analogous form to \eqref{eq:A3}. For example, it can be shown that $y_\mm$ and $\theta_\mm$ in $A^{n,i,1}_{t,T}(u)$ and $A^{n,i,2}_{t,T}(u)$ can be replaced by $W$ and coefficients with respect to $W$. Also, if we use the expansion $e^{i(a+b)}=e^{ia} + ie^{ia}b + O(b^2)$ for the last two lines in the definition of $A^{n,i,2}_{t,T}(u)$, then $b$ is of order $O^\uc(\sqrt{T})$. So the $O(b^2)$-part leads to a contribution to $A^{n,i,2}_{t,T}(u)$ that is $O^\uc(\sqrt{\Den}T^{3/2})=o^\uc(\sqrt{\Den}T)$, while the  $ie^{ia}b$-part leads to a contribution given by
			\begin{align*}
				&\E_{t^n_{i-1}}\biggl[ iu_{T^n_{i-1}} \iint_{t^n_{i-1}}^{t+T}\iint_{t^n_{i-1}}^{s-} \Delta^n_i\theta(t,z,z') y(dr,dz')y(ds,dz) \\
				& \qquad\times iu_{T^n_{i-1}}\exp\biggl(iu_{T^n_{i-1}}  \iint_{t^n_{i-1}}^{t+T} \theta(t^n_i,z)y(ds,dz)  \biggr)\int_{t^n_{i-1}}^{t+T}\int_{t^n_{i-1}}^s \si^\si_{t^n_i} W(dr)W(ds)\biggr] \\
				&\quad=-u_{T^n_{i-1}}^2\E_{t^n_{i-1}}\biggl[  \int_{t^n_{i-1}}^{t+T}\int_{t^n_{i-1}}^{s} \Delta^n_i\theta^2_\cc(t) y_\cc(dr)y_\cc(ds)\exp\biggl(iu_{T^n_{i-1}}  \int_{t^n_{i-1}}^{t+T} \si_{t^n_i} dW_s  \biggr) \\
				& \qquad\qquad\qquad\qquad\qquad\qquad\qquad\times \int_{t^n_{i-1}}^{t+T}\int_{t^n_{i-1}}^s \si^\si_{t^n_i} W(dr)W(ds)\biggr]+o^\uc(\sqrt{\Den}T)\\
				&\quad=\Delta^n_i v_t C_{t^n_i}(u) T^n_{i-1} +o^\uc(\sqrt{\Den}T).
			\end{align*}
			Arguing like this, we obtain 
			\begin{equation}\label{eq:A-simple} \begin{split}
					A^{n,i,1}_{t,T}(u)&=\ov A^{n,i,11}_{t,T}(u) +\ov A^{n,i,12}_{t,T}(u)+ \text{``$\Delta^n_i v_t C_{t^n_i}(u) T^n_{i-1}$''}+o^\uc(\sqrt{\Den}T),\\
					A^{n,i,2}_{t,T}(u)&= \ov A^{n,i,2}_{t,T}(u)+ \text{``$\Delta^n_i v_t C_{t^n_i}(u) T^n_{i-1}$''}+o^\uc(\sqrt{\Den}T),
				\end{split}
			\end{equation}
			where the quotation marks represent a finite \emph{sum} of terms of the form $\Delta^n_i v_t C_{t^n_i}(u) T^n_{i-1}$
			and
			\begin{align*}
				\ov	A^{n,i,11}_{t,T}(u)	&= \E_{t^n_{i-1}}\biggl[ iu_{T^n_{i-1}} \iint_{t^n_{i-1}}^{t+T} \Delta^n_i\theta(t,z)y(ds,dz) e^{iu_{T^n_{i-1}} \iint_{t^n_{i-1}}^{t+T} \theta(t^n_i,z)y(ds,dz)}\biggr], \\
				\ov A^{n,i,12}_{t,T}(u)	&= \E_{t^n_{i-1}}\biggl[ iu_{T^n_{i-1}} \iint_{t^n_{i-1}}^{t+T} \Delta^n_i\theta_\mm(t,z)y_\mm(ds,dz)e^{iu_{T^n_{i-1}} \iint_{t^n_{i-1}}^{t+T} \theta_\mm(t^n_i,z)y_\mm(ds,dz)} \\
				& \quad\times iu_{T^n_{i-1}}\iint_{t^n_{i-1}}^{t+T}\iint_{t^n_{i-1}}^{s-} \theta_\mm(t^n_i,z,z')y_\mm(dr,dz')y_\mm(ds,dz)\biggr],\\
				\ov	A^{n,i,2}_{t,T}(u)	&= \E_{t^n_{i-1}}\biggl[ iu_{T^n_{i-1}} \iint_{t^n_{i-1}}^{t+T}\iint_{t^n_{i-1}}^{s-} \Delta^n_i\theta_\mm(t,z,z') y_\mm(dr,dz')y_\mm(ds,dz) \\
				& \quad\times\exp\biggl(iu_{T^n_{i-1}} \iint_{t^n_{i-1}}^{t+T} \theta_\mm(t^n_i,z)y_\mm(ds,dz) \biggr)\biggr].
			\end{align*}

			Using integration by parts multiple times, we can express
			\begin{equation}\label{eq:5terms} 
				\begin{split}
					& \iint_{t^n_{i-1}}^{t+T} \Delta^n_i\theta_\mm(t,z)y_\mm(ds,dz)\iint_{t^n_{i-1}}^{t+T}\iint_{t^n_{i-1}}^{s-} \theta_\mm(t^n_i,z,z')y_\mm(dr,dz')y_\mm(ds,dz)	 \\
					&\quad= \int_{\S^{(3)}_{t^n_{i-1},T^n_{i-1}}}  \theta_\mm(t^n_i,z',z'') y_\mm(du,dz'')y_\mm(dr,dz')\Delta^n_i\theta_\mm(t,z)y_\mm(ds,dz)\\
					&\qquad+\int_{\S^{(3)}_{t^n_{i-1},T^n_{i-1}}}  \theta_\mm(t^n_i,z,z'') y_\mm(du,dz'')\Delta^n_i\theta_\mm(t,z')y_\mm(dr,dz')y_\mm(ds,dz)\\
					&\qquad+\int_{\S^{(3)}_{t^n_{i-1},T^n_{i-1}}}  \Delta^n_i\theta_\mm(t,z'') y_\mm(du,dz'')\theta_\mm(t^n_i,z,z')y_\mm(dr,dz')y_\mm(ds,dz)\\
					&\qquad+\iint_{t^n_{i-1}}^{t+T} \iint_{t^n_{i-1}}^{s}  \theta_\mm(t^n_i,z,z')  y_\mm(dr,dz')[y_\mm](ds,dz)\Delta^n_i\theta_\mm(t,z)^\top\\
					&\qquad+\iint_{t^n_{i-1}}^{t+T} \iint_{t^n_{i-1}}^{s}  \theta_\mm(t^n_i,z,z')  [y_\mm](dr,dz')\Delta^n_i\theta_\mm(t,z')^\top y_\mm(ds,dz).
				\end{split}
			\end{equation}
			%Multiplying by the cosine term and taking conditional expectation, the contributions of the last two terms in previous display  can be determined as before and are given by
			%\[ \frac{u^2}2\Re[e^{T\vp_{t^n_i}(u_T)}T\Delta^n_i\chi^{(3)}_{t}(t^n_i,u_T)]+ O(\Den).\]
			In order to compute the conditional expectation  in $\ov A^{n,i,12}_{t,T}(u)$, we argue similarly to the paragraph following \eqref{eq:help-2}. We define
			$Z^{n,i}_t(u,s)=\exp(iu \iint_{t^n_{i-1}}^{s} \theta_\mm(t^n_i,z)y_\mm(dr,dz))$
			and use   It\^o's formula   (see Theorem I.4.57 in \cite{JS03}) to obtain
			\begin{align*}
				&Z^{n,i}_{t}(u,t+T)	\\
				&\quad=1+  iu\iint_{t^n_{i-1}}^{t+T}Z^{n,i}_t(u,s-)\theta_\mm(t^n_i,z)y_\mm(ds,dz)-\frac12u^2 \si_{t^n_i}^2\int_{t^n_{i-1}}^{t+T}Z^{n,i}_t(u,s) ds\\
				&\qquad+\iint_{t^n_{i-1}}^{t+T} Z^{n,i}_t(u,s-)(e^{iu\ga(t^n_i,z)}-1-iu\ga(t^n_i,z))\mu(ds,dz).
			\end{align*}
			We realize that $Z^{n,i}_t(u,s)$ is the stochastic exponential of 
			\begin{align*}
				L^{n,i}_{t}(u,s)&= \int_{t^n_{i-1}}^s \biggl(\vp_{t^n_i}(r,u)-\frac12u^2\si^2_{t^n_i}\biggr)dr +iu\si_{t^n_i}(W_s-W_{t^n_{i-1}}) \\
				&\quad+\iint_{t^n_{i-1}}^{s} (e^{iu\ga(t^n_i,z)}-1)\wh\mu(dr,dz) 
			\end{align*}
			for $s\geq t^n_{i-1}$.
			Hence, $Z^{n,i}_{t}(u,s)$ is the unique solution to the SDE
			\[
			Z^{n,i}_{t}(u,ds)=Z^{n,i}_{t}(u,s-)L^{n,i}_{t}(u,ds),\qquad Z^{n,i}_{t}(u,t^n_{i-1})=1.
			\]
			Equivalently,  $\wh Z^{n,i}_{t}(u,s)=\exp\bigl(-\int_{t^n_{i-1}}^s  (\vp_{t^n_i}(r,u)-\frac12u^2\si^2_{t^n_i} )dr\bigr)Z^{n,i}_{t}(u,s)$ satisfies
			\begin{equation}\label{eq:SDE} 
				\wh Z^{n,i}_{t}(u,ds)=\wh Z^{n,i}_{t}(u,s-)\wh L^{n,i}_{t}(u,ds),\qquad \wh Z^{n,i}_{t}(u,t^n_{i-1})=1,
			\end{equation} 
			with 
			\begin{equation}\label{eq:L} 
				\wh	L^{n,i}_{t}(u,s)=iu\si_{t^n_i}(W_s-W_{t^n_{i-1}}) +\iint_{t^n_{i-1}}^{s} (e^{iu\ga(t^n_i,z)}-1)\wh\mu(dr,dz). 
			\end{equation}
			The SDE \eqref{eq:SDE} immediately yields the series expansion 
			\begin{equation}\label{eq:chaos-3}\begin{split}
					&Z^{n,i}_{t}(u_{T^n_{i-1}},t+T) = e^{-\frac12u^2\si^2_{t^n_i}+\int_{t^n_{i-1}}^{t+T}  \vp_{t^n_i}(r,u_{T^n_{i-1}})dr }\wh Z^{n,i}_{t}(u_{T^n_{i-1}},t+T)\\
					&\quad=e^{-\frac12u^2\si^2_{t^n_i}+\int_{t^n_{i-1}}^{t+T}  \vp_{t^n_i}(r,u_{T^n_{i-1}})dr }\biggl(1+\sum_{k=1}^\infty \int_{\S^{(k)}_{t^n_{i-1},T^n_{i-1}}}   \wh\theta_{t^n_i}(u_{T^n_{i-1}},\bz_k) \by_\mm(d\bs_k,d\bz_k)\biggr),
				\end{split}\!\!\!\end{equation}
			where $\wh\theta_s(u,\bz_k)$ is the same as in \eqref{eq:wh-theta}.
			By (2.15) and the fact that 
			\begin{equation}\label{eq:o1} \begin{split}
					&T\int_{\R^{d'}} \lvert e^{iu_T\ga(t,z)}-1-iu_T\ga(t,z)\rvert J(z)F(dz)\\
					&\quad \leq 2u^{r/2}T^{1-r/2} \int_{\R^{d'}} \lvert \ga(t,z)\rvert^r J(z)F(dz)=o^\uc(1), 		
				\end{split}
			\end{equation}
			it follows that
			\begin{equation}\label{eq:chaos} \begin{split}
					Z^{n,i}_{t}(u_{T^n_{i-1}},t+T) 
					&=e^{-\frac12u^2\si^2_{t^n_i}+T^n_{i-1}\vp_{t^n_i}(u_{T^n_{i-1}}) + o^\uc(\sqrt{T}) }\\
					&\quad\times\biggl(1+\sum_{k=1}^\infty \int_{\S^{(k)}_{t^n_{i-1},T^n_{i-1}}}   \wh\theta_{t^n_i}(u_{T^n_{i-1}},\bz_k) \by_\mm(d\bs_k,d\bz_k)\biggr).
				\end{split}
			\end{equation} 
			
			As the left-hand side of \eqref{eq:5terms} is $O(\sqrt{\Den}T^{3/2})$, the $o^\uc(\sqrt{T})$-term in \eqref{eq:chaos} can be ignored in the computation of $\ov A^{n,i,12}_{t,T}(u)$ if one allows for an $o^\uc(\sqrt{\Den}T)$-error. With this in mind,
			we are now in the position to find the contributions to $\ov A^{n,i,12}_{t,T}(u)$ coming from  the first three terms on the right-hand side of \eqref{eq:5terms}. As they are $\calf_{t^n_{i-1}}$-conditionally centered, the constant term in \eqref{eq:chaos} has no contribution. For the remaining terms, recall $\mu^{n,i}$ from \eqref{eq:muni} and the idea to split $y_\mm$ appearing in \eqref{eq:chaos} or in any of the first three terms on the right-hand side of \eqref{eq:5terms} into $y^{n,i}_\mm$ and $y_\mm-y^{n,i}_\mm$.
			
			Let us first consider  the case where \emph{all} such $y_\mm$'s are replaced by $y^{n,i}_\mm$. In this case,  the first three terms in \eqref{eq:5terms}  are triple integrals with respect to (conditional) Wiener and Poisson measures. Since  multiple Wiener or Poisson integrals of different order are uncorrelated to each other, only the term corresponding to $k=3$ in \eqref{eq:chaos} yields a nonzero covariance. Thus, each of the first three terms in \eqref{eq:5terms} (with $y^{n,i}_\mm$ instead of $y_\mm$) has an identical contribution to $\ov A^{n,i,12}_{t,T}(u)$ that is equal to
			\begin{equation}\label{eq:A12-1}\begin{split}
					\qquad	&-u_{T^n_{i-1}}^2 e^{-\frac12u^2\si_{t^n_i}^2+T^n_{i-1}\vp_{t^n_i}(u_{T^n_{i-1}})}  \int_{\mathbb{S}^{(3)}_{t^n_{i-1},T^n_{i-1}}} \Delta^n_i\theta_\mm(t,z)\langle y^{n,i}_\mm\rangle(ds_1,dz_1)(\wh\theta_{t^n_i}(u_{T^n_{i-1}},z_1))^\top \\
					&\qquad\times \theta_\mm(t^n_i,z_2,z_3) \langle y^{n,i}_\mm\rangle(ds_3,dz_3)(\wh\theta_{t^n_i}(u_{T^n_{i-1}},z_3))^\top\langle y^{n,i}_\mm\rangle(ds_2,dz_2)(\wh\theta_{t^n_i}(u_{T^n_{i-1}},z_2))^\top\\
					&\quad=-\frac{u^2}6(T^n_{i-1})^2  \Theta_{t^n_i,T^n_{i-1}}(u_{T^n_{i-1}})\Bigl(iu_{T^n_{i-1}}\si_{t^n_i}\Delta^n_i\si_t +\Delta^n_i \xi^{(1)}_t(t^n_i,u_{T^n_{i-1}}) \Bigr)  \\
					&\qquad\times\Bigl( -\si^\si_{t^n_i}u^2_{T^n_{i-1}}\si^2_{t^n_i} +iu_{T^n_{i-1}}\si_{t^n_i} \chi^{(1)}_{t^n_i}(u_{T^n_{i-1}})+\chi^{(2)}_{t^n_i}(u_{T^n_{i-1}}) \Bigr) +\text{``$\Delta^n_i v_t C_{t^n_i}(u) T^n_{i-1}$''}. %\\
					%&\quad=-\frac{u^2}6 e^{-\frac12u^2\si_{t^n_i}^2 +T\vp_{t^n_i}(u_T)} \Bigl(iu\sqrt{T}\si_{t^n_i}\Delta^n_i\si_t +T\Delta^n_i \xi^{(1)}_t(t^n_i,u_T) \Bigr)\\
					%&\qquad\times\Bigl(-\si^\si_{t^n_i}\si^2_{t^n_i}u^2 + iu\sqrt{T}\si_{t^n_i}\chi^{(1)}_{t^n_i}(u_T)+T\chi_{t^n_i}^{(2)}(u_T)\Bigr) +O^\uc(\Den).}	 
				\end{split}\raisetag{-3.5\baselineskip}\end{equation}
			Next, consider the case where \emph{at least two}  $y_\mm$'s are replaced by $y_\mm-y_\mm^{n,i}$. As seen in \eqref{eq:14}, because $\lvert \la(t',z)-\la(t^n_{i},z)\rvert = O(\sqrt{T})$ uniformly in $t'\in[t^n_i,t+T]$, each of these two substitutions leads to an extra factor of $O(T^{1/4})$, rendering the contributions of the first three terms of \eqref{eq:5terms} to $\ov A^{n,i,12}_{t,T}(u)$ negligible in size. 
			
			It thus remains to consider the case where \emph{exactly one} $y_\mm$ is replaced by $y_\mm-y_\mm^{n,i}$. Suppose this happens for one of the $y_\mm(ds_i,dz_i)$'s in 
			\begin{align*}
				&	\int_{\S^{(k)}_{t^n_{i-1},T^n_{i-1}}}   \wh\theta_{t^n_i}(u_{T^n_{i-1}},\bz_k)\by_\mm(d\bs_k,d\bz_k)\\ &\quad=\int_{\S^{(k)}_{t^n_{i-1},T^n_{i-1}}}   \wh\theta_{t^n_i}(u_{T^n_{i-1}},\bz_k)y_\mm(ds_k,dz_k)\cdots y_\mm(ds_1,dz_1)
			\end{align*}
			from \eqref{eq:chaos} (otherwise the argument is similar). Let $i_0$ be the position where this substitution occurs. 
			Since the terms from \eqref{eq:5terms} are triple Wiener/Poisson integrals and $y_\mm$ is centered, we must have $i_0\in\{1,2,3\}$. Let us further consider the case $i_0=1$ and the resulting covariance with the first term on the right-hand of \eqref{eq:5terms} (the other $2+2 \times 3 = 8$ cases can be treated analogously). Then
			\begin{align*}
				&-u_{T^n_{i-1}}^2 \E_{t^n_{i-1}}\Biggl[  \int_{\S^{(3)}_{t^n_{i-1},T^n_{i-1}}}  \theta_\mm(t^n_i,z',z'') y^{n,i}_\mm(du,dz'')y^{n,i}_\mm(dr,dz')\Delta^n_i\theta_\mm(t,z)y^{n,i}_\mm(ds,dz)\\
				&\qquad\times	\int_{\S^{(k)}_{t^n_{i-1},T^n_{i-1}}}   \wh\theta_{t^n_i}(u_{T^n_{i-1}},\bz_k)y^{n,i}_\mm(ds_k,dz_k)\cdots y^{n,i}_\mm(ds_2,dz_2)(y_\mm-y^{n,i}_\mm)(ds_1,dz_1) \Biggr]\\
				&\quad =u_{T^n_{i-1}}^2 \iiint_{t^n_{i-1}}^{t+T} \E_{t^n_{i-1}}\Biggl[\Biggl( \int_{\S^{(2)}_{t^n_{i-1},s_1-t^n_{i-1}-}}  \theta_\mm(t^n_i,z_2,z_3) y^{n,i}_\mm(ds_3,dz_3)y^{n,i}_\mm(ds_2,dz_2)\\
				&\qquad	\times	\int_{\S^{(k-1)}_{t^n_{i-1},s_1-t^n_{i-1}-}}  \prod_{i=2}^{k} \wh\theta_{t^n_i}(u_{T^n_{i-1}},z_i)y^{n,i}_\mm(ds_i,dz_i)\Biggr)\bone_{\{\la(s_1,z_1)< v\leq \la(t^n_i,z_1)\}} \Biggr]\\
				&\qquad\times\Delta^n_i\ga(t,z_1)(e^{iu_{T^n_{i-1}}\ga(t^n_i,z_1)}-1)ds_1F(dz_1) dv.
			\end{align*}
			If $h^{n,i}(u_{T^n_{i-1}},s_1)$ denotes the big parenthesis, the above equals
			\begin{equation*}
				u_{T^n_{i-1}}^2 \iint_{t^n_{i-1}}^{t+T} \E_{t^n_{i-1}} [ h^{n,i}(u_{T^n_{i-1}},s)(\la(s,z)-\la(t^n_i,z))_- ] \Delta^n_i\ga(t,z)(e^{iu_{T^n_{i-1}}\ga(t^n_i,z)}-1) ds F(dz),
			\end{equation*}
			where $x_-=-(x\wedge 0)$. 
			%We  replace $\la(s,z)$ by
			%\begin{equation}\label{eq:lani} \begin{split}
			%	\la^{n,i}(s,z)&=\la(t^n_i,z)+\int_0^s\al^\la(t^n_i,z)dr+\sum_{i=1}^d 	\int_0^s\si^{\la,(i)}(t^n_i,z) dW^{(i)}_r \\
			%	&\quad+ \iiint_0^s  \ga^\la(t^n_i,z,z',v') (\pf-\qf)(dr,dz',dv')+ \iiint_0^s \Ga^\la(t^n_i,z,z',v') \pf(dr,dz',dv'),\end{split}
			%\end{equation}
			%which is $\la(s,z)$ but with characteristics frozen at $t=t^n_i$ (and thus a Lévy process). Since $(\cdot)_-$ is Lipschitz and $\lvert \la(s,z)-\la(t^n_i,z)\rvert = O(T)$ for all $s\in[t^n_i,t+T]$, the error resulting from doing so is $O^\uc(T)$
			Since $\la(s,z)-\la(t^n_i,z) = O(\sqrt{T})$, $h^{n,i}(u_{T^n_{i-1}},s) = O^\uc(T)$ and $$\iint_{t^n_{i-1}}^{t+T} \Delta^n_i\ga(t,z)(e^{iu_{T^n_{i-1}}\ga(t^n_i,z)}-1) ds F(dz) = O^\uc(\sqrt{\Den}T^{1-r/4}),$$
			it follows that the penultimate display is $O^\uc(\sqrt{\Den}T^{3/2-r/4})=o^\uc(\sqrt{\Den}T)$. In conclusion, each of the first three terms on the right-hand side of \eqref{eq:5terms} contributes \eqref{eq:A12-1} to $\ov A^{n,i,12}_{t,T}(u)$ plus some $o^\uc(\sqrt{\Den}T)$-error.
			
			We argue similarly for the last two terms in \eqref{eq:5terms}. After passing to  $\mu^{n,i}$, we  decompose $[y^{n,i}_\mm]=\langle y^{n,i}_\mm\rangle+ \mathrm{diag}(0,\dots,0,\wh\mu^{n,i})$ and then observe that only  the summands with $k=1$ or $k=2$ in \eqref{eq:chaos} need to be kept. As a result, the reader can check that the total contribution of  the  last two terms in \eqref{eq:5terms} to $\ov A^{n,i,12}_{t,T}(u)$ is
			\begin{equation}\label{eq:A12-2}	\begin{split}
					&-\frac{u^2}{2}  T^n_{i-1}  \Theta_{t^n_i,T^n_{i-1}}(u_{T^n_{i-1}})\Bigl(2iu_{T^n_{i-1}}\si^\si_{t^n_i}\si_{t^n_i}\Delta^n_i \si_t+iu_{T^n_{i-1}}\si_{t^n_i}\Delta^n_i \xi^{(5)}_t(t^n_i,u_{T^n_{i-1}}) \\
					&\quad+ \chi^{(1)}_{t^n_i}(u_{T^n_{i-1}})\Delta^n_i \si_t+\Delta^n_i\xi^{(3)}_t(t^n_i,u_{T^n_{i-1}})\Bigr) +\text{``$\Delta^n_i v_t C_{t^n_i}(u) T^n_{i-1}$''} + o^\uc(\sqrt{\Den}T).
			\end{split}\end{equation}
			%In conclusion, we have shown that
			%\begin{equation}\label{eq:A12} \begin{split}
			%	A^{n,i,12}_{t,T}(u)&=-\frac{u^2}{2} e^{-\frac12u^2\si_{t^n_i}^2 +T\vp_{t^n_i}(u_T)}\Bigl( \bigl(iu\sqrt{T}\si_{t^n_i}\Delta^n_i\si_t +T\Delta^n_i\wh\xi^{(1)}_t(t^n_i,u_T)  \bigr)T\chi_{t^n_i}(u_T)\\
			%	&\quad+2iu\sqrt{T}\si^\si_{t^n_i}\si_{t^n_i}\Delta^n_i \si_t+iu\sqrt{T}\si_{t^n_i}\Delta^n_i \xi^{(5)}_t(t^n_i,u_T) \\
			%	&\quad+ T\chi^{(1)}_{t^n_i}(u_T)\Delta^n_i \si_t+T\Delta^n_i\xi^{(3)}_t(t^n_i,u_T)\Bigr) + O^\uc(\Den).\end{split}
			%\end{equation}
			By similar methods (replacing $\mu$ by $\mu^{n,i}$ and  using the chaos representation \eqref{eq:chaos-1}), we obtain
			\begin{align}
				\ov	A^{n,i,2}_{t,T}(u)	&= \frac12iu(T^n_{i-1})^{3/2} \Theta_{t^n_i,T^n_{i-1}}(u_{T^n_{i-1}})\Bigl(-u^2_{T^n_{i-1}}\si^2_{t^n_i}\Delta^n_i \si^\si_t  +iu_{T^n_{i-1}} \si_{t^n_i}\Delta^n_i\xi^{(4)}_t(t^n_i,u_{T^n_{i-1}})\nonumber\\
				&\quad+\Delta^n_i\xi^{(2)}_t(t^n_i,u_{T^n_{i-1}})\Bigr)+\text{``$\Delta^n_i v_t C_{t^n_i}(u) T^n_{i-1}$''} + o^\uc(\sqrt{\Den}T).\label{eq:A2} 
			\end{align}
			
			We turn to $\ov A^{n,i,11}_{t,T}(u)$. By definition,   $	\ov	A^{n,i,11}_{t,T}(u)=\ov A^{n,i,112}_{t,T}(u)+\ov A^{n,i,111}_{t,T}(u)$, where
			\begin{align*}
				\ov	A^{n,i,111}_{t,T}(u)	&=iu {\textstyle\sqrt{T^n_{i-1}}}\Delta^n_i \al_t \E_{t^n_{i-1}} \Bigl[ e^{iu_{T^n_{i-1}} \iint_{t^n_{i-1}}^{t+T} \theta(t^n_i,z)y(ds,dz)}\Bigr], \\
				%	&\quad+iu_T  \E_{t^n_{i-1}}\biggl[ \iint_{t^n_{i-1}}^{t+T} \Delta^n_i\Ga(t,z)\mu(ds,dz)\exp\biggl(iu_T \iint_{t^n_{i-1}}^{t+T} \theta(t^n_i,z)y(ds,dz)\biggr)\biggr] \\
				\ov A^{n,i,112}_{t,T}(u)	&= iu_{T^n_{i-1}} \E_{t^n_{i-1}}\biggl[ \iint_{t^n_{i-1}}^{t+T} \Delta^n_i\theta_\mm(t,z)y_\mm(ds,dz) e^{iu_{T^n_{i-1}} \iint_{t^n_{i-1}}^{t+T} \theta(t^n_i,z)y(ds,dz)}\biggr].
			\end{align*}
			By  (3.7) and \eqref{eq:chaos},
			\begin{equation}\label{eq:A111} \begin{split}
					\ov A^{n,i,111}_{t,T}(u)&=iu{\textstyle\sqrt{ T^n_{i-1}}} \Theta_{t^n_i,T^n_{i-1}}(u_{T^n_{i-1}})\Delta^n_i \al_t\\
					&\quad\times \E_{t^n_{i-1}}\biggl[ e^{o^\uc(\sqrt{T})}\biggl(1+\sum_{k=1}^\infty \int_{\S^{(k)}_{t^n_{i-1},T^n_{i-1}}}   \wh\theta_{t^n_i}(u_{T^n_{i-1}},\bz_k) \by_\mm(d\bs_k,d\bz_k)\biggr) \biggr]\\
					&=iu{\textstyle\sqrt{ T^n_{i-1}}} \Theta_{t^n_i,T^n_{i-1}}(u_{T^n_{i-1}})\Delta^n_i \al_t+o^\uc(\sqrt{\Den}T).
				\end{split}
			\end{equation}
			%Next, 
			%\begin{align*}
			%&A^{n,i,112}_{t,T}(u)\\
			%%&\quad=	iu_T  \E_{t^n_{i-1}}\biggl[ \iint_{t^n_{i-1}}^{t+T} \Delta^n_i\Ga(t,z)\mu(ds,dz)\exp\biggl(iu_T \iint_{t^n_{i-1}}^{t+T} \theta_\mm(t^n_i,z)y_\mm(ds,dz)\biggr)\biggr]\\
			%%&\qquad+O^\uc(\Den^{1/2}T)\\
			%	&\quad=iu_TT^n_{i-1}\int_\R \Delta^n_i\Ga(t,z)\la(dz) e^{iu_T\al_{t^n_i}T^n_{i-1}-\frac12u^2_T\si^2_{t^n_i}T^n_{i-1}+T^n_{i-1}\vp_{t^n_i}(u_T)+T^n_{i-1}\phi_{t^n_i}(u_T)} \\
			%	&\qquad+ i	u_T  \E_{t^n_{i-1}}\biggl[ \iint_{t^n_{i-1}}^{t+T} \Delta^n_i\Ga(t,z)\wh\mu(ds,dz)\exp\biggl(iu_T \iint_{t^n_{i-1}}^{t+T} \theta(t^n_i,z)y(ds,dz)\biggr)\biggr]\\
			%	&\quad\quad+O^\uc(\Den^{1/2}T\vee \Den).
			%\end{align*}
			%In the last $y(ds,dz)$-integral, there is no harm to compensate the big jumps, as $1-\exp(iu_T T^n_{i-1} \int_\R \Ga(t^n_i,z)\la(dz)) = O^\uc(\sqrt{T})$ renders the resulting difference for the second term on the right-hand side $O^\uc(\Den^{1/2}T)=o^\uc(\sqrt{\Den/k_n})$.
			%By \eqref{eq:chaos}, it then follows that
			%\begin{equation*}
			%	A^{n,i,112}_{t,T}(u) = 	iu_TT^n_{i-1}\eta_{t^n_i,T^n_{i-1}}(u_T)\Delta^n_i\xi^{(2)}_t(t^n_i,u_T) +o^\uc(\sqrt{\Den/k_n}).
			%\end{equation*}
			Next,  by %first compensating the big jumps and then extracting the drift part from the complex exponential and using 
			\eqref{eq:chaos-3}, we  get
			\begin{align*}
				&	\ov A^{n,i,112}_{t,T}(u)	=iu_{T^n_{i-1}}\Theta_{t^n_i,T^n_{i-1}}(u_{T^n_{i-1}})\E_{t^n_{i-1}}\biggl[\iint_{t^n_{i-1}}^{t+T} \Delta^n_i\theta_\mm(t,z)y_\mm(ds,dz) \\
				&\qquad\times\biggl(1+\sum_{k=1}^\infty \int_{\S^{(k)}_{t^n_{i-1},T^n_{i-1}}}   \wh\theta_{t^n_i}(u_{T^n_{i-1}},\bz_k) \by_\mm(d\bs_k,d\bz_k)\biggr)\\ 
				&\qquad\times \exp\biggl(\iint_{t^n_{i-1}}^{t+T} (e^{iu_{T^n_{i-1}}\ga(t^n_i,z)}-1-iu_{T^n_{i-1}}\ga(t^n_i,z))(\la(r,z)-\la(t^n_i,z))drF(dz)\biggr)\biggr].
			\end{align*}
			Now, by (2.4), (2.14) and (2.15), we have for $r\in [t^n_i,t+T]$ that
			\begin{align}
				\la(r,z)	&=\la(t^n_i,z)+\sum_{j=1}^d\int_{t^n_i}^r \si^{\la,(j)}(s,z)dW^{(j)}_s + \iiint_{t^n_i}^r \ga^\la(s,z,z',v')(\pf-\qf)(ds,dz',dv')\nonumber\\
				&\quad+O(TJ(z)) \label{eq:la-expand}  \\
				&=\la(t^n_i,z)+\la^{n,i}(r,z)+O(TJ(z)),\nonumber
			\end{align}
			where
			\[ \la^{n,i}(r,z)=\sum_{j=1}^d\si^{\la,(j)}(t^n_i,z)(W^{(j)}_r-W^{(j)}_{t^n_i} )+ \iiint_{t^n_i}^r \ga^\la(t^n_i,z,z',v')(\pf-\qf)(ds,dz',dv').\]
			Combined with \eqref{eq:o1}, it follows that 
			\begin{align*}
				\ov A^{n,i,112}_{t,T}(u)	&=iu_{T^n_{i-1}}\Theta_{t^n_i,T^n_{i-1}}(u_{T^n_{i-1}})\E_{t^n_{i-1}}\biggl[\iint_{t^n_{i-1}}^{t+T} \Delta^n_i\theta_\mm(t,z)y_\mm(ds,dz) \\
				&\qquad\times\biggl(1+\sum_{k=1}^\infty \int_{\S^{(k)}_{t^n_{i-1},T^n_{i-1}}}   \wh\theta_{t^n_i}(u_{T^n_{i-1}},\bz_k) \by_\mm(d\bs_k,d\bz_k)\biggr)\\ 
				&\qquad\times \exp\biggl(\iint_{t^n_{i-1}}^{t+T} (e^{iu_{T^n_{i-1}}\ga(t^n_i,z)}-1-iu_{T^n_{i-1}}\ga(t^n_i,z))\la^{n,i}(r,z)drF(dz)\biggr)\biggr]\\
				&\quad+o^\uc(\sqrt{\Den}T)\\
				&=	\ov A^{n,i,1121}_{t,T}(u)+	\ov A^{n,i,1122}_{t,T}(u)+o^\uc(\sqrt{\Den}T),
			\end{align*}
			where
			\begin{align*}
				\ov A^{n,i,1121}_{t,T}(u)&= iu_{T^n_{i-1}}\Theta_{t^n_i,T^n_{i-1}}(u_{T^n_{i-1}})\E_{t^n_{i-1}}\biggl[\iint_{t^n_{i-1}}^{t+T} \Delta^n_i\theta_\mm(t,z)y_\mm(ds,dz) \\
				&\quad\times\biggl(1+\sum_{k=1}^\infty \int_{\S^{(k)}_{t^n_{i-1},T^n_{i-1}}}   \wh\theta_{t^n_i}(u_{T^n_{i-1}},\bz_k) \by_\mm(d\bs_k,d\bz_k)\biggr)\biggr],\\
				\ov A^{n,i,1122}_{t,T}(u)&=iu_{T^n_{i-1}}\Theta_{t^n_i,T^n_{i-1}}(u_{T^n_{i-1}})\E_{t^n_{i-1}}\biggl[\iint_{t^n_{i-1}}^{t+T} \Delta^n_i\theta_\mm(t,z)y_\mm(ds,dz) \\
				&\quad\times\biggl(1+\sum_{k=1}^\infty \int_{\S^{(k)}_{t^n_{i-1},T^n_{i-1}}}   \wh\theta_{t^n_i}(u_{T^n_{i-1}},\bz_k) \by_\mm(d\bs_k,d\bz_k)\biggr)\\ 
				&\quad\times \iint_{t^n_{i-1}}^{t+T} (e^{iu_{T^n_{i-1}}\ga(t^n_i,z)}-1-iu_{T^n_{i-1}}\ga(t^n_i,z))\la^{n,i}(r,z)drF(dz)\biggr].
			\end{align*}
			
			By It\^o's isometry,
			\begin{align*} 
				\ov A^{n,i,1121}_{t,T}(u)& =-u^2\Theta_{t^n_i,T^n_{i-1}}(u_{T^n_{i-1}})\si_{t^n_i}\Delta^n_i\si_t\\ &\quad+iu_{T^n_{i-1}}\Theta_{t^n_i,T^n_{i-1}}(u_{T^n_{i-1}}) \sum_{k=1}^\infty \iint_{t^n_{i-1}}^{t+T} (e^{iu_{T^n_{i-1}}\ga(t^n_i,z)}-1)\Delta^n_i \ga(t,z)\\
				& \quad\times \E_{t^n_{i-1}}\biggl[\int_{\S^{(k-1)}_{t^n_{i-1},s-t^n_{i-1}}}    \wh \theta_{t^n_i}(u_{T^n_{i-1}},\bz_{k-1})\by_\mm(d\bs_{k-1},d\bz_{k-1}) \la(s,z)\biggr]dsF(dz)  \\
				& = \ov A^{\prime n,i,1121}_{t,T}(u)+\ov A^{\prime\prime n,i,1121}_{t,T}(u),
			\end{align*}
			where
			\begin{equation}\label{eq:A1121-1}
				\ov A^{\prime n,i,1121}_{t,T}(u)	= iu{\textstyle\sqrt{ T^n_{i-1}}} \Theta_{t^n_i,T^n_{i-1}}(u_{T^n_{i-1}})\Delta^n_i\xi^{(1)}_t(t^n_i,u_{T^n_{i-1}}) -u^2  \si_{t^n_i} \Theta_{t^n_i,T^n_{i-1}}(u_{T^n_{i-1}})\Delta^n_i \si_t
			\end{equation}
			and
			\begin{align*}
				&	\ov A^{\prime\prime n,i,1121}_{t,T}(u)=iu_{T^n_{i-1}}\Theta_{t^n_i,T^n_{i-1}}(u_{T^n_{i-1}}) \sum_{k=1}^\infty \iint_{t^n_{i-1}}^{t+T} (e^{iu_{T^n_{i-1}}\ga(t^n_i,z)}-1)\Delta^n_i \ga(t,z)\\
				&\quad\times \E_{t^n_{i-1}}\biggl[\int_{\S^{(k-1)}_{t^n_{i-1},s-t^n_{i-1}}}    \wh \theta_{t^n_i}(u_{T^n_{i-1}},\bz_{k-1})\by_\mm(d\bs_{k-1},d\bz_{k-1}) (\la(s,z)-\la(t^n_i,z))\biggr]dsF(dz).
			\end{align*}
			Using \eqref{eq:la-expand}, we can simplify the latter and obtain
			\begin{align*}
				\ov A^{\prime\prime n,i,1121}_{t,T}(u)&=iu_{T^n_{i-1}}\Theta_{t^n_i,T^n_{i-1}}(u_{T^n_{i-1}}) \sum_{k=1}^\infty \iint_{t^n_{i-1}}^{t+T} (e^{iu_{T^n_{i-1}}\ga(t^n_i,z)}-1)\Delta^n_i \ga(t,z)\\
				&\quad\times \E_{t^n_{i-1}}\biggl[\int_{\S^{(k-1)}_{t^n_{i-1},s-t^n_{i-1}}}    \wh \theta_{t^n_i}(u_{T^n_{i-1}},\bz_{k-1})\by_\mm(d\bs_{k-1},d\bz_{k-1}) \la^{n,i}(s,z)\biggr]dsF(dz)\\
				&\quad+o^\uc(\sqrt{\Den}T).
			\end{align*}
			Conditionally on $\calf_{t^n_{i-1}}$, the process $s\mapsto\la^{n,i}(s,z)$ is equal to $\la^{n,i}(t^n_{i-1},z)$ plus a centered Lévy process. Therefore, only the terms corresponding to $k=1$ and $k=2$ survive the last conditional expectation. Because $\la^{n,i}(t^n_{i-1},z)=O(\sqrt{\Den})$, it is easy to show that the term with $k=1$ leads to a negligible contribution of size $o^\uc(\Den)$. In the term corresponding to $k=2$, we can replace $y_\mm$ by $y_\mm^{n,i}$ if we allow for an $o^\uc(\sqrt{\Den}T)$-error. Hence,  
			\begin{equation}\label{eq:A1121-2} 
				\begin{split}
					\ov A^{\prime\prime n,i,1121}_{t,T}(u)&=-\frac12 T^n_{i-1}u^2\si_{t^n_i}\Theta_{t^n_i,T^n_{i-1}}(u_{T^n_{i-1}}) \Delta^n_i \xi^{(6)}_t(t^n_i,u_{T^n_{i-1}})\\
					&\quad+\frac12iu(T^n_{i-1})^{3/2}\Theta_{t^n_i,T^n_{i-1}}(u_{T^n_{i-1}}) \Delta^n_i \xi^{(7)}_t(t^n_i,u_{T^n_{i-1}})+o^\uc(\sqrt{\Den}T).
				\end{split}
			\end{equation}
			
			Regarding $\ov A^{n,i,1122}_{t,T}(u)$, we can argue similarly to how we did in the paragraph after \eqref{eq:chaos} to %note that we already have $\ov A^{n,i,1122}_{t,T}(u)=o^\uc(\sqrt{\Den T})$, so we are free to make any change of magnitude $O^\uc(\sqrt{T})$. In particular, we can 
			replace all $y_\mm$ by $y^{n,i}_\mm$. Since $\iint_{t^n_{i-1}}^{t+T} \Delta^n_i\theta_\mm(t,z)y^{n,i}_\mm(ds,dz)$ is a Wiener-type integral and $\iint_{t^n_{i-1}}^{t+T} (e^{iu_{T^n_{i-1}}\ga(t^n_i,z)}-1-iu_{T^n_{i-1}}\ga(t^n_i,z))\la^{n,i}(r,z)drF(dz)$ is a constant plus a Wiener-type integral (conditionally on $\calf_{t^n_{i-1}}$), only the constant term $1$ and the terms for $k=1$ and $k=2$ in the second line of the definition of $\ov A^{n,i,1122}_{t,T}(u)$ have nonzero contributions. As above, because $\la^{n,i}(t^n_{i-1},z)=O(\sqrt{\Den})$, the contribution corresponding to $k=1$ is negligible. After a tedious but entirely straightforward computation for the other two terms, we obtain 
			\begin{equation}\label{eq:A1122}\begin{split}
					\ov A^{n,i,1122}_{t,T}(u)	&=\frac12iu(T^n_{i-1})^{3/2}\Theta_{t^n_i,T^n_{i-1}}(u_{T^n_{i-1}})\Bigl(\chi^{(3)}_{t^n_i}(u_{T^n_{i-1}})\Delta^n_i \si_t + \Delta^n_i \xi^{(8)}_t(t^n_i,u_{T^n_{i-1}})\Bigr)\\
					&\quad-\frac12iu^3(T^n_{i-1})^{3/2} \Theta_{t^n_i,T^n_{i-1}}(u_{T^n_{i-1}})\chi^{(3)}_{t^n_i}(u_{T^n_{i-1}})\si^2_{t^n_i}\Delta^n_i\si_t\\
					&\quad-\frac12u^2(T^n_{i-1})^2\Theta_{t^n_i,T^n_{i-1}}(u_{T^n_{i-1}})\chi_{t^n_i}^{(4)}(u_{T^n_{i-1}})\si_{t^n_i}\Delta^n_i\si_t\\
					&\quad-\frac12u^2(T^n_{i-1})^2\Theta_{t^n_i,T^n_{i-1}}(u_{T^n_{i-1}})\chi^{(3)}_{t^n_i}(u_{T^n_{i-1}})\si_{t^n_i}\Delta^n_i\xi^{(1)}_t(t^n_i,u_{T^n_{i-1}})\\
					&\quad+\frac12iu(T^n_{i-1})^{5/2}\Theta_{t^n_i,T^n_{i-1}}(u_{T^n_{i-1}})\chi^{(4)}_{t^n_i}(u_{T^n_{i-1}}) \Delta^n_i\xi^{(1)}_t(t^n_i,u_{T^n_{i-1}})+o^\uc(\sqrt{\Den}T).
				\end{split}\raisetag{-4\baselineskip}\end{equation}

			We proceed to $B^{n,i}_{t,T}(u)$ from \eqref{eq:ABC}. Since $e^{i(a+b)}-1 = (e^{ia}-1) + i(e^{ia}-1)b + ib + O(b^2)$, we have
			\begin{equation}\label{eq:Bni}\begin{split}
					B^{n,i}_{t,T}(u)	&%\Re\biggl[iu_T\sum_{k=1}^N\int_{\S^{(k)}_{t^n_i,T^n_i}\setminus \S^{(k)}_{t^n_{i-1},T^n_{i-1}}} \theta(t^n_i,\bz_k)\by (d\bs_k,d\bz_k)\\
					%& \qquad\qquad\qquad\qquad\times\exp\biggl(iu_T \sum_{k=1}^N \int_{\S^{(k)}_{t^n_{i-1},T^n_{i-1}}} \theta(t^n_i,\bz_k)\by(d\bs_k,d\bz_k)\biggr) \\
					%	&\quad+
					=\E_{t^n_i}\biggl[\biggl(\exp\biggl(iu_{T^n_{i-1}}\sum_{k=1}^N\int_{\S^{(k)}_{t^n_i,T^n_i}\setminus \S^{(k)}_{t^n_{i-1},T^n_{i-1}}} \theta(t^n_i,\bz_k)\by (d\bs_k,d\bz_k)\biggr)-1\biggr)\\
					& \quad   \times\exp\biggl(iu_{T^n_{i-1}} \sum_{k=1}^N \int_{\S^{(k)}_{t^n_{i-1},T^n_{i-1}}} \theta(t^n_i,\bz_k)\by(d\bs_k,d\bz_k)\biggr)  \biggr]\\
					&=B^{n,i,1}_{t,T}(u)+B^{n,i,2}_{t,T}(u)+B^{n,i,3}_{t,T}(u)+O^\uc(\Den),
			\end{split}\end{equation}
			where
			\begin{align*}
				B^{n,i,1}_{t,T}(u)&=\E_{t^n_i}\biggl[ \biggl(\exp\biggl(iu_{T^n_{i-1}}\iint_{t^n_i}^{t^n_{i-1}}\theta(t^n_i,z)y (ds,dz)\biggr)-1\biggr)\\
				& \quad   \times\exp\biggl(iu_{T^n_{i-1}} \sum_{k=1}^N \int_{\S^{(k)}_{t^n_{i-1},T^n_{i-1}}} \theta(t^n_i,\bz_k)\by(d\bs_k,d\bz_k)\biggr)  \biggr],\\
				B^{n,i,2}_{t,T}(u)&=\E_{t^n_i}\biggl[  iu_{T^n_{i-1}}\sum_{k=2}^N \int_{\S^{(k)}_{t^n_i,T^n_i}\setminus \S^{(k)}_{t^n_{i-1},T^n_{i-1}}} \theta(t^n_i,\bz_k)\by (d\bs_k,d\bz_k)\\
				&\quad\times\biggl( \exp\biggl(iu_{T^n_{i-1}}\iint_{t^n_i}^{t^n_{i-1}}\theta(t^n_i,z)y (ds,dz)\biggr)-1\biggr)\\
				&\quad\times \exp\biggl(iu_{T^n_{i-1}} \sum_{k=1}^N \int_{\S^{(k)}_{t^n_{i-1},T^n_{i-1}}} \theta(t^n_i,\bz_k)\by(d\bs_k,d\bz_k)\biggr)\biggr],\\
				B^{n,i,3}_{t,T}(u)&=\E_{t^n_i}\biggl[  iu_{T^n_{i-1}}\sum_{k=2}^N \int_{\S^{(k)}_{t^n_i,T^n_i}\setminus \S^{(k)}_{t^n_{i-1},T^n_{i-1}}} \theta(t^n_i,\bz_k)\by (d\bs_k,d\bz_k)\\
				&\quad\times \exp\biggl(iu_{T^n_{i-1}} \sum_{k=1}^N \int_{\S^{(k)}_{t^n_{i-1},T^n_{i-1}}} \theta(t^n_i,\bz_k)\by(d\bs_k,d\bz_k)\biggr)\biggr].
			\end{align*}
			For the analysis of $B^{n,i,3}_{t,T}(u)$, we  recall the process $\la_{t,t'}(s,z)$ from Assumption~1 and define % For the second part,
			%%recall from \eqref{eq:la-expand} that $\la(s,z)=\la(t^n_i,z)+\la^{n,i}(s,z)+O(TJ(z))$ and note that $\la^{n,i}(s,z)=\ov\la^{n,i}(s,z)+O(\sqrt{\Den})$, where
			%\begin{equation}\label{eq:ovlani} 
			%	\ov\la^{n,i}(s,z)=\sum_{j=1}^d\si^{\la,(j)}(t^n_i,z)(W^{(j)}_s-W^{(j)}_{t^n_{i-1}} )+ \iiint_{t^n_{i-1}}^s \ga^\la(t^n_i,z,z',v')(\pf-\qf)(dr,dz',dv')
			%\end{equation} 
			%for $s\in [t^n_{i-1},t+T]$ and $\ov \la^{n,i}(s,z)=0$ for $s\in[t^n_i,t^n_{i-1}]$, 
			%which satisfies $\la^{n,i}(s,z)=\ov\la^{n,i}(s,z)+O(\sqrt{\Den})$ and thus
			%\begin{equation}\label{eq:la-expand-2} 
			%	\la(s,z)=\la(t^n_i,z)+\ov \la^{n,i}(s,z)+O((T+\sqrt{\Den})J(z))
			%\end{equation}
			%by \eqref{eq:la-expand} for all $s\in[t^n_i,t+T]$. Correspondingly, we define 
			\begin{equation}\label{eq:muni-2} 
				\ov \mu^{n,i}(ds,dz)=\int_\R \bone_{\{0\leq v\leq \la_{t^n_i,t^n_{i-1}}(s,z)\}}\pf(ds,dz,dv)
			\end{equation}
			and similarly $\wh{\ov\mu}^{n,i}$, $\ov y^{n,i}$ and $\ov\by^{n,i}$. By (2.18), 
			\begin{equation}\label{eq:la-expand-2} 
				\la(s,z)=\la_{t^n_i,t^n_{i-1}}(s,z)+O((T^2+\sqrt{\Den})J(z))
			\end{equation}
			for all $s\in[t^n_i,t+T]$, so
			an integral with respect to $y-\ov y^{n,i}$ over an interval of length $O(T)$ is $o(T^{3/2}+\sqrt{T}\Den^{1/4})$ (by a similar calculation to \eqref{eq:14}). Also, recall that  $\calg_{t^n_{i-1}}$ denotes the $\si$-field generated by the increments of $\mathbb{W}$ and $\pf$ after $t^n_{i-1}$ and that $\ov y^{n,i}(ds, dz)$ restricted to $[t^n_{i-1},\infty)\times\R^{d'}$ is $\calf_{t^n_i}\vee \calg_{t^n_{i-1}}$-measurable.
			
			With that in mind, we now have the decomposition $B^{n,i,3}_{t,T}(u)=B^{n,i,31}_{t,T}(u)+B^{n,i,32}_{t,T}(u)+B^{n,i,33}_{t,T}(u)+o^\uc(\sqrt{\Den}T+\Den)$,  where
			\begin{align*}
				B^{n,i,31}_{t,T}(u)	&= \E_{t^n_i}\biggl[  iu_{T^n_{i-1}}\sum_{k=2}^N \int_{\S^{(k)}_{t^n_i,T^n_i}\setminus \S^{(k)}_{t^n_{i-1},T^n_{i-1}}} \theta(t^n_i,\bz_k)\ov\by^{n,i} (d\bs_k,d\bz_k)\\
				&\quad\times \exp\biggl(iu_{T^n_{i-1}} \sum_{k=1}^N \int_{\S^{(k)}_{t^n_{i-1},T^n_{i-1}}} \theta(t^n_i,\bz_k)\ov\by^{n,i}(d\bs_k,d\bz_k)\biggr)\biggr],\\
				B^{n,i,32}_{t,T}(u)	&= \E_{t^n_i}\biggl[  iu_{T^n_{i-1}}\sum_{k=2}^N \int_{\S^{(k)}_{t^n_i,T^n_i}\setminus \S^{(k)}_{t^n_{i-1},T^n_{i-1}}} \theta(t^n_i,\bz_k)(\by-\ov\by^{n,i}) (d\bs_k,d\bz_k)\\
				&\quad\times \exp\biggl(iu_{T^n_{i-1}} \sum_{k=1}^N \int_{\S^{(k)}_{t^n_{i-1},T^n_{i-1}}} \theta(t^n_i,\bz_k)\ov\by^{n,i}(d\bs_k,d\bz_k)\biggr)\biggr],\\
				B^{n,i,33}_{t,T}(u)	&= \E_{t^n_i}\biggl[  iu_{T^n_{i-1}}\sum_{k=2}^N \int_{\S^{(k)}_{t^n_i,T^n_i}\setminus \S^{(k)}_{t^n_{i-1},T^n_{i-1}}} \theta(t^n_i,\bz_k)\by (d\bs_k,d\bz_k)\\
				&\quad\times iu_{T^n_{i-1}} \sum_{k=1}^N \int_{\S^{(k)}_{t^n_{i-1},T^n_{i-1}}} \theta(t^n_i,\bz_k)(\by-\ov\by^{n,i})(d\bs_k,d\bz_k)\\
				&\quad\times \exp\biggl(iu_{T^n_{i-1}} \sum_{k=1}^N \int_{\S^{(k)}_{t^n_{i-1},T^n_{i-1}}} \theta(t^n_i,\bz_k)\ov\by^{n,i}(d\bs_k,d\bz_k)\biggr)\biggr].
			\end{align*}
			In $B^{n,i,31}_{t,T}(u)$,
			the integral $\int_{\S^{(k)}_{t^n_i,T^n_i}\setminus \S^{(k)}_{t^n_{i-1},T^n_{i-1}}}$ is a sum of integrals of the form $$\iint_{t^n_{i-1}}^{t+T}\dotsi\iint_{t^n_{i-1}}^{t+T}\iint_{t^n_i}^{t^n_{i-1}}\dotsi \iint_{t^n_i}^{t^n_{i-1}}$$ (where the change point can occur anywhere). So if we further condition on $\calg_{t^n_{i-1}}$ in the computation of $B^{n,i,31}_{t,T}(u)$, the exponential term becomes known and we can first compute $\E[\iint_{t^n_i}^{t^n_{i-1}}\dotsi \iint_{t^n_i}^{t^n_{i-1}} (\cdots) \mid \calf_{t^n_i}\vee \calg_{t^n_{i-1}}] = \E_{t^n_i}[\iint_{t^n_i}^{t^n_{i-1}}\dotsi \iint_{t^n_i}^{t^n_{i-1}} (\cdots)]$, which removes the martingale components of $y^{n,i}$ and leaves us with a drift for these integrals. Two cases can occur: we either have at least two integrals of the form $\iint_{t^n_i}^{t^n_{i-1}}$, in which case the expectation under consideration is $O(\Den^2/\sqrt{T})=o(\Den)$, or we only have one integral of the form $\iint_{t^n_i}^{t^n_{i-1}}$ but then we must also have an integral of the form $\iint_{t^n_{i-1}}^{t+T}$ because $k\geq2$, in which case the expectation is $O^\uc(\Den)$. This shows that $
			B^{n,i,31}_{t,T}(u)=O^\uc(\Den)$.

			Next, we consider $B^{n,i,32}_{t,T}(u)$ and a plain size estimate already shows that $B^{n,i,32}_{t,T}(u)=o^\uc(\sqrt{\Den}T+\Den^{3/4})$, so we can make any $O^\uc(\sqrt{T})$ change  if we allow for an $o^\uc(\sqrt{\Den} T+\Den)$-error. In particular, in the two sums over $k$, it suffices to keep the term corresponding to $k=2$ and $k=1$, respectively, and only the martingale parts. Also, we can replace $\int_{\S^{(2)}_{t^n_i,T^n_i}\setminus \S^{(2)}_{t^n_{i-1},T^n_{i-1}}}$ by $\iint_{t^n_{i-1}}^{t+T}\iint_{t^n_i}^{t^n_{i-1}}$ and $y_\mm(dr,dz')y_\mm(ds,dz)-\ov y^{n,i}_\mm(dr,dz')\ov y^{n,i}_\mm(ds,dz)$ by $\ov y^{n,i}_\mm(dr,dz')(y_\mm-\ov y^{n,i}_\mm)(ds,dz)+(y_\mm-\ov y^{n,i}_\mm)(dr,dz')\ov y^{n,i}_\mm(ds,dz)$. %Moreover, since $\ov y^{n,i}_\mm(ds,dz)= y^{n,i}_\mm(ds,dz)$ if restricted to $[t^n_i,t^n_{i-1}]\times\R^{d'}$, 
			It follows that 
			\begin{equation}\label{eq:help-3} 
				\begin{split}
					B^{n,i,32}_{t,T}(u)	&=\E_{t^n_i}\biggl[  iu_{T^n_{i-1}} \iint_{t^n_{i-1}}^{t+T}\iint_{t^n_i}^{t^n_{i-1}}  \theta_\mm(t^n_i,z,z')\\
					&\quad \times(\ov  y^{n,i}_\mm(dr,dz')(y_\mm-\ov y^{n,i}_\mm)(ds,dz)+(y_\mm-  \ov y^{n,i}_\mm)(dr,dz')\ov y^{n,i}_\mm(ds,dz))\\
					&\quad\times \exp\biggl(iu_{T^n_{i-1}}  \iint_{t^n_{i-1}}^{t+T} \theta_\mm(t^n_i,z)\ov y_\mm^{n,i}(ds,dz)\biggr)\biggr]+o^\uc(\sqrt{\Den}T+\Den).
				\end{split}\!\!
			\end{equation}
			Similarly to \eqref{eq:chaos}, we have the chaos representation
			\begin{equation}\label{eq:chaos-2} \begin{split}
					&\exp\biggl(iu_{T^n_{i-1}}  \iint_{t^n_{i-1}}^{t+T} \theta_\mm(t^n_i,z)\ov y_\mm^{n,i}(ds,dz)\biggr)\\	& =e^{-\frac12u^2\si^2_{t^n_i}+T^n_{i-1}\vp_{t^n_i}(u_{T^n_{i-1}}) + o^\uc(T^2+\sqrt{\Den}) }\biggl(1+\sum_{k=1}^\infty \int_{\S^{(k)}_{t^n_{i-1},T^n_{i-1}}}   \wh\theta_{t^n_i}(u_{T^n_{i-1}},\bz_k) \ov\by^{n,i}_\mm(d\bs_k,d\bz_k)\biggr).
			\end{split}\end{equation} 
			We can remove the $o^\uc(T^2+\sqrt{\Den})$-term. Afterwards, the conditional expectation in \eqref{eq:help-3} wipes out the constant %and the term corresponding to $k=1$ 
			in \eqref{eq:chaos-2}. If $k\geq1$, %the part $-y^{n,i}_\mm(dr,dz')y^{n,i}_\mm(ds,dz)$ in the previous display can be ignored (as it gives rise to a double Wiener or Poisson integral, which is orthogonal to multiple integrals of length $\geq3$). Hence, 
			the $k$th term in \eqref{eq:chaos-2} leads to $\exp(-\frac12u^2\si^2_{t^n_i}+T^n_{i-1}\vp_{t^n_i}(u_{T^n_{i-1}}))$ times
			\begin{align*}
				&-u^2_{T^n_{i-1}}\si_{t^n_i} \int_{t^n_{i-1}}^{t+T}   \E_{t^n_i}\biggl[ \iint_{t^n_i}^{t^n_{i-1}} \ga^\si(t^n_i,z')(\wh\mu-\wh{\ov\mu}^{n,i})(dr,dz')\\
				& \quad\quad\quad\quad\quad\quad\quad \quad\quad \quad\quad\quad\quad\times\int_{\S^{(k-1)}_{t^n_{i-1},s-t^n_{i-1}}} \wh \theta_{t^n_i}(u_{T^n_{i-1}},\bz_{k-1})\ov\by^{n,i}_\mm(d\bs_{k-1},d\bz_{k-1}) \biggr] ds\\
				& +iu_{T^n_{i-1}} \iint_{t^n_{i-1}}^{t+T}(e^{iu_{T^n_{i-1}}\ga(t^n_i,z)}-1)\E_{t^n_i}\biggl[  \iint_{t^n_i}^{t^n_{i-1}} \ga^\ga(t^n_i,z,z') (\wh\mu-\wh{\ov\mu}^{n,i})(dr,dz')\\
				& \quad\quad\quad\quad \quad\times\int_{\S^{(k-1)}_{t^n_{i-1},s-t^n_{i-1}}} \wh \theta_{t^n_i}(u_{T^n_{i-1}},\bz_{k-1})\ov\by^{n,i}_\mm(d\bs_{k-1},d\bz_{k-1})\la_{t^n_i,t^n_{i-1}}(s,z)  \biggr]F(dz)ds\\
				& -iu_{T^n_{i-1}} \iint_{t^n_{i-1}}^{t+T}(e^{iu_{T^n_{i-1}}\ga(t^n_i,z)}-1)\E_{t^n_i}\biggl[ \biggl( \si^\ga(t^n_i,z)\Delta^n_i W_t\\
				&\quad\quad\quad\quad\quad\quad\quad\quad \quad\quad \quad\quad\quad\quad\quad\quad\quad\quad\quad\quad\quad\quad+ \iint_{t^n_i}^{t^n_{i-1}} \ga^\ga(t^n_i,z,z') \wh{\ov\mu}^{n,i}(dr,dz')\biggr)\\
				& \quad     \times\int_{\S^{(k-1)}_{t^n_{i-1},s-t^n_{i-1}}} \wh \theta_{t^n_i}(u_{T^n_{i-1}},\bz_{k-1})\ov\by^{n,i}_\mm(d\bs_{k-1},d\bz_{k-1})(\la(s,z)-\la_{t^n_i,t^n_{i-1}}(s,z) )_-  \biggr]F(dz)ds.
			\end{align*}%(\la(s,z)\wedge \la(t^n_i,z))
			The first part is clearly equal zero, as we can condition on $\calf_{t^n_{i-1}}$ first. 
			The second part is also equal to zero
			%recall from \eqref{eq:la-expand} that $\la(s,z)=\la(t^n_i,z)+\la^{n,i}(s,z)+O(TJ(z))$ and note that $\la^{n,i}(s,z)=\ov\la^{n,i}(s,z)+O(\sqrt{\Den})$, where
			%
			%So in the second part of the penultimate display, replacing $\la(s,z)$ by $\la(t^n_i,z)+\ov\la^{n,i}(s,z)$ only incurs a total error of size $o^\uc(\sqrt{\Den}T+\Den)$, which is negligible. Once this change has been made, 
			when we first take conditional expectation with respect to $\calf_{t^n_i} \vee \calg_{t^n_{i-1}}$. And the last part is $o^\uc(\sqrt{\Den}T^2+\Den)$ by \eqref{eq:la-expand-2}. %It remains to study the term that corresponds to $k=2$ in \eqref{eq:chaos-1} and its contribution to $B^{n,i,32}_{t,T}(u)$. The same argument applies except that now it is a different reason why the $-y^{n,i}_\mm(dr,dz')y^{n,i}_\mm(ds,dz)$ part can be ignored: for $k=2$, this is because we can first condition on $\calf_{t^n_{i-1}}$, which makes the double integral  a single integral, which in turn is orthogonal to the double integral coming from \eqref{eq:chaos-1}. 
			We have thus proved that $B^{n,i,32}_{t,T}(u)=o^\uc(\sqrt{\Den}T+\Den)$.
			
			The term $B^{n,i,33}_{t,T}(u)$ is already $o^\uc(\sqrt{\Den}T+\Den^{3/4})$, so it suffices to keep the lowest order term in all three sums over $k$. Moreover, we can replace $y$ in the first line by $\ov y^{n,i}_\mm$ and only keep the martingale terms. % and, similarly to above, $y-y^{n,i}$ by a random measure with intensity measure $\ov\la^{n,i}(s,z)ds F(dz)$. 
			After making these changes, we obtain
			\begin{align*}
				B^{n,i,33}_{t,T}(u)	&= -u_{T^n_{i-1}}^2\E_{t^n_i}\biggl[  \iint_{t^n_{i-1}}^{t+T}\iint_{t^n_i}^{t^n_{i-1}} \theta_\mm(t^n_i,z,z')\ov y^{n,i}_\mm(dr,dz')\ov y_\mm^{n,i}(ds,dz)\\
				&\quad\quad \times      \iint_{t^n_{i-1}}^{t+T} \theta_\mm(t^n_i,z)(y_\mm-\ov y^{n,i}_\mm)(ds,dz)\\
				&\qquad\times \exp\biggl(iu_{T^n_{i-1}}    \iint_{t^n_{i-1}}^{t+T}\theta_\mm(t^n_i,z) \ov y_\mm^{n,i}(ds,dz)\biggr)\biggr] +o^\uc(\sqrt{\Den} T+\Den).
			\end{align*}
			We use \eqref{eq:chaos-2} and integration by parts to get
			\begin{align*}
				&B^{n,i,33}_{t,T}(u)\\
				&~	= -u_{T^n_{i-1}}^2e^{-\frac12u^2\si^2_{t^n_i}+T^n_{i-1}\vp_{t^n_i}(u_{T^n_{i-1}})}\E_{t^n_i}\biggl[ \biggl(1+\sum_{k=1}^\infty \int_{\S^{(k)}_{t^n_{i-1},T^n_{i-1}}}   \wh\theta_{t^n_i}(u_{T^n_{i-1}},\bz_k) \ov\by^{n,i}_\mm(d\bs_k,d\bz_k)\biggr)\\ &~\quad\times\biggl(\iint_{t^n_{i-1}}^{t+T}\iint_{t^n_{i-1}}^{s-}\iint_{t^n_i}^{t^n_{i-1}} \theta_\mm(t^n_i,z,z'')\theta_\mm(t^n_i,z')\\
				&~\quad\qquad\qquad\qquad\qquad\qquad\qquad\qquad\times\ov y^{n,i}_\mm(dw,dz'') (y_\mm-\ov y^{n,i}_\mm)(dr,dz')\ov y_\mm^{n,i}(ds,dz)\\
				&~\qquad+\iint_{t^n_{i-1}}^{t+T}\iint_{t^n_{i-1}}^{s-}\iint_{t^n_i}^{t^n_{i-1}} \theta_\mm(t^n_i,z',z'')\theta_\mm(t^n_i,z)\\
				&~\quad\qquad\qquad\qquad\qquad\qquad\qquad\qquad\times\ov y^{n,i}_\mm(dw,dz'') \ov y_\mm^{n,i}(dr,dz')(y_\mm-\ov y^{n,i}_\mm)(ds,dz)\\
				&~\qquad-\iint_{t^n_{i-1}}^{t+T}\iint_{t^n_i}^{t^n_{i-1}} \theta_\mm(t^n_i,z,z')\ga(t^n_i,z)\\
				&~\quad\qquad\qquad\qquad\qquad\qquad\qquad \times\ov y^{n,i}_\mm(dr,dz')(\la(s,z)-\la_{t^n_i,t^n_{i-1}}(s,z))_-F(dz)ds\biggl)\biggr]\\
				&\quad\quad+o^\uc(\sqrt{\Den} T+\Den^{3/4}).
			\end{align*}
			The quadratic variation term in the fifth line is negligibly small. For the other two terms, we keep $y_\mm-\ov y_\mm^{n,i}$ unchanged but replace the other $\ov y_\mm^{n,i}$'s (including those appearing in the first line of the above display) by $y_\mm^{n,i}$, which was defined after \eqref{eq:muni}. Since $\la(t^n_i,z)-\la_{t^n_{i},t^n_{i-1}}(s,z) = O(\sqrt{T})$ for $s\in[t^n_{i-1},t+T]$, the error incurred through this modification is asymptotically negligible. After this change, when we calculate the $\calf_{t^n_i}$-conditional covariance between the term $\int_{\S^{(k)}_{t^n_{i-1},T^n_{i-1}}}   \wh\theta_{t^n_i}(u_{T^n_{i-1}},\bz_k)  \by^{n,i}_\mm(d\bs_k,d\bz_k)$ and one of the two triple integrals, at some point, we encounter the predictable covariation between $y_\mm-\ov y_\mm^{n,i}$ and $y_\mm^{n,i}$, which is equal to 
			\begin{equation}\label{eq:qv} 
				\Bigl([(\la(t^n_i,z)\wedge \la(s,z))-\la_{t^n_{i-1},t^n_i}(s,z)]_+ - [(\la(t^n_i,z)\wedge \la_{t^n_{i-1},t^n_i}(s,z))-\la(s,z)]_+\Bigr) F(dz)ds.
			\end{equation}
			Since the parenthesis is $O(T^2+\sqrt{\Den})$, it follows that the contribution of the two triple integrals to $B^{n,i,33}_{t,T}(u)$ is $o^\uc(\sqrt{\Den}T^2+\Den)$. We have shown that 
			\begin{equation}\label{eq:B3} 
				B^{n,i,3}_{t,T}(u)=o^\uc(\sqrt{\Den}T)+O^\uc(\Den).
			\end{equation}

			We move to $B^{n,i,2}_{t,T}(u)$ and notice that the second line in its definition is $O^\uc(\sqrt{\Den/T})$. %Thus, $B^{n,i,2}_{t,T}(u)=O^\uc(\Den/\sqrt{T})$. 
			Thus, if we allow for an $o^\uc(\Den/\sqrt{T})$-error, we only need to keep the dominating part in each of the three factors defining $B^{n,i,2}_{t,T}(u)$. Because $\Den/T\to0$, we obtain
			\begin{align*}
				B^{n,i,2}_{t,T}(u)	&= \E_{t^n_i}\biggl[ iu_{T^n_{i-1}}   \int_{t^n_{i-1}}^{t+T} \int_{t^n_i}^{t^n_{i-1}} \theta^2_\cc(t^n_i) y_\cc(dr)y_\cc(ds) \\
				&\quad\times iu_{T^n_{i-1}}\si_{t^n_i}(W_{t^n_{i-1}}-W_{t^n_i}) e^{iu_{T^n_{i-1}}\si_{t^n_i}(W_{t+T}-W_{t^n_{i-1}})}\biggr] + o^\uc(\Den/\sqrt{T})\\
				&=-u_{T^n_{i-1}}^2 \si_{t^n_i}\si^\si_{t^n_i} \E_{t^n_i}\biggl[ (W_{t+T}-W_{t^n_{i-1}}) (W_{t^n_{i-1}}-W_{t^n_i})^2 e^{iu_{T^n_{i-1}}\si_{t^n_i}(W_{t+T}-W_{t^n_{i-1}})}\biggr]\\
				&\quad + o^\uc(\Den/\sqrt{T}),
			\end{align*}
			where the second step follows by conditioning on $W$. Since $W$ has independent increments, 
			\begin{equation}\label{eq:B2}\begin{split}
					B^{n,i,2}_{t,T}(u)	&= -u^2_{T^n_{i-1}} \si_{t^n_i}\si^\si_{t^n_i}  \Den\E_{t^n_i}\biggl[ (W_{t+T}-W_{t^n_{i-1}})  e^{iu_{T^n_{i-1}}\si_{t^n_i}(W_{t+T}-W_{t^n_{i-1}})}\biggr] \\
					&\quad+ o^\uc(\Den/\sqrt{T}),\\
					&=-iu^3 \si^2_{t^n_i}\si^\si_{t^n_i} \frac{\Den}{\sqrt{T^n_{i-1}}} e^{-\frac12 u^2\si^2_{t^n_i}} + o^\uc(\Den/\sqrt{T})\\
					&= -iu^3 \si^2_{t^n_i}\si^\si_{t^n_i} \frac{\Den}{\sqrt{T^n_{i-1}}} \Theta_{t^n_i,T^n_{i-1}}(u_{T^n_{i-1}}) + o^\uc(\Den/\sqrt{T}).
			\end{split}\end{equation}
			
			Finally, we analyze $B^{n,i,1}_{t,T}(u)$. To this end, we first note that similarly to \eqref{eq:chaos}, we have
			\begin{align*}
				&\exp\biggl(iu_{T^n_{i-1}}\iint_{t^n_i}^{t^n_{i-1}}\theta(t^n_i,z)y (ds,dz)\biggr)-1\\
				&\quad	=\Theta_{t^n_i,\Den}(u_{T^n_{i-1}})\exp\biggl(\int_{t^n_i}^{t^n_{i-1}}( \vp_{t^n_i}(r,u_{T^n_{i-1}})-\vp_{t^n_i}(u_{T^n_{i-1}}))dr\biggr)\\
				&\quad\quad\times\biggl(1+\sum_{k=1}^\infty \int_{\S^{(k)}_{t^n_{i},t^n_{i-1}-t^n_i}}   \wh\theta_{t^n_i}(u_{T^n_{i-1}},\bz_k) \by_\mm(d\bs_k,d\bz_k)\biggr)-1\\
				&\quad=\Bigl(\Theta_{t^n_i,\Den}(u_{T^n_{i-1}})-1\Bigr)+\Theta_{t^n_i,\Den}(u_{T^n_{i-1}})\sum_{k=1}^\infty \int_{\S^{(k)}_{t^n_{i},t^n_{i-1}-t^n_i}}   \wh\theta_{t^n_i}(u_{T^n_{i-1}},\bz_k) \by_\mm(d\bs_k,d\bz_k)\\
				&\quad\quad+o^\uc(\Den/\sqrt {T}),
			\end{align*}
			where the last step holds because $\vp_{t^n_i}(r,u_{T^n_{i-1}})-\vp_{t^n_i}(u_{T^n_{i-1}}) = o^\uc(\lvert r-t^n_i\rvert^{1/2}/ {T})$. Thus, $B^{n,i,1}_{t,T}(u)= B^{n,i,11}_{t,T}(u)+B^{n,i,12}_{t,T}(u)+o^\uc(\Den/\sqrt{T})$, where
			\begin{align*}
				B^{n,i,11}_{t,T}(u)&= \Bigl(\Theta_{t^n_i,\Den}(u_{T^n_{i-1}})-1\Bigr)\E_{t^n_i}\biggl[\exp\biggl(iu_{T^n_{i-1}} \sum_{k=1}^N \int_{\S^{(k)}_{t^n_{i-1},T^n_{i-1}}} \theta(t^n_i,\bz_k)\by(d\bs_k,d\bz_k)\biggr)  \biggr],\\
				B^{n,i,12}_{t,T}(u)	&=\Theta_{t^n_i,\Den}(u_{T^n_{i-1}})\sum_{k=1}^\infty \E_{t^n_i}\biggl[\int_{\S^{(k)}_{t^n_{i},t^n_{i-1}-t^n_i}}   \wh\theta_{t^n_i}(u_{T^n_{i-1}},\bz_k) \by_\mm(d\bs_k,d\bz_k)\\
				&\quad\quad\quad\quad\quad\quad\quad\quad\quad\quad \times\exp\biggl(iu_{T^n_{i-1}} \sum_{k=1}^N \int_{\S^{(k)}_{t^n_{i-1},T^n_{i-1}}} \theta(t^n_i,\bz_k)\by(d\bs_k,d\bz_k)\biggr)  \biggr].
			\end{align*}
			For $B^{n,11}_{t,T}(u)$, because $ \Theta_{t^n_i,\Den}(u_{T^n_{i-1}}) -1=O^\uc(\Den/T)$, we only need the keep the terms corresponding to $k=1$ and $k=2$  (and for $k=2$ only integrals with respect to $W$), if we allow for an $o^\uc(\Den/\sqrt{T})$-error. Consequently,
			\begin{align*}
				B^{n,i,11}_{t,T}(u)&= \Bigl( \Theta_{t^n_i,\Den}(u_{T^n_{i-1}}) -1\Bigr)
				\E_{t^n_i}\Bigl[e^{iu_{T^n_{i-1}}  \iint_{t^n_{i-1}}^{t+T} \theta(t^n_i,z)y(ds,dz)} \Bigr]  +\Bigl( \Theta_{t^n_i,\Den}(u_{T^n_{i-1}}) -1\Bigr)\\
				& \times\E_{t^n_i}\biggl[e^{iu_{T^n_{i-1}} \si_{t^n_i}(W_{t+T}-W_{t^n_{i-1}}) }  iu_{T^n_{i-1}}  \si^{\si}_{t^n_i}\int_{t^n_{i-1}}^{t+T}\int_{t^n_{i-1}}^s dW_r dW_s \biggr]  +o^\uc(\Den/\sqrt{T}).
			\end{align*}
			Let us denote the two summands on the right-hand side by $B^{n,i,111}_{t,T}(u)$ and $	B^{n,i,112}_{t,T}(u)$, respectively.
			By \eqref{eq:chaos}, we have 
			\begin{equation}\label{eq:B11} 
				B^{n,i,111}_{t,T}(u)	
				=(\Theta_{t^n_i, \Den}(u_{T^n_{i-1}})-1)\Theta_{t^n_i, T^n_{i-1}}(u_{T^n_{i-1}})+o^\uc(\Den/\sqrt{T}).
			\end{equation}
			The analysis of $B^{n,i,112}_{t,T}(u)$ is very similar to how we computed the terms in \eqref{eq:5terms}. Using \eqref{eq:chaos-1} and the assumption that $\Den/T\to0$, we find that
			%	\begin{align*}
			%		B^{n,i,112}_{t,T}(u)	&=-\frac12iu^3{\textstyle\sqrt{T^n_{i-1}}}\si_{t^n_i}^\si\si_{t^n_i}^2\Bigl(e^{- \frac12u^2_{T^n_{i-1}}\si_{t^n_i}^2T^n_i+T^n_i\vp_{t^n_i}(u_{T^n_{i-1}})}  -e^{- \frac12u^2_{T^n_{i-1}}\si_{t^n_i}^2T^n_{i-1}+T^n_{i-1}\vp_{t^n_i}(u_{T^n_{i-1}})}\Bigr)\\
			%		&\quad+O^\uc(\Den).
			%	\end{align*}
			%	Upon noting that $T\vp_{t^n_i}(u_T)=o^\uc(1)$ as $T\to0$ by the dominated convergence theorem,
			\begin{equation}\label{eq:B12}\begin{split}
					B^{n,i,112}_{t,T}(u)&=\frac14iu^5\si_{t^n_i}^\si\si_{t^n_i}^4 \frac{ \Den}{\sqrt{T^n_{i-1}}} e^{-\frac12 u^2 \si^2_{t^n_i} }+o^\uc(\Den/\sqrt{T})\\
					&		=\frac14iu^5\si_{t^n_i}^\si\si_{t^n_i}^4 \frac{ \Den}{\sqrt{T^n_{i-1}}} 	\Theta_{t^n_i,T^n_{i-1}}(u_{T^n_{i-1}})+o^\uc(\Den/\sqrt{T}). 	
			\end{split}\end{equation}

			We proceed to $B^{n,i,12}_{t,T}(u)$ and recall the fact that integrals with respect to $y-\ov y^{n,i}$  come with an additional factor of $o(T+\Den^{1/4})$ (see \eqref{eq:muni-2}). Using the formula
			\begin{align*}
				e^{i(a_1+a_2+b_1+b_2)}&=e^{i(a_1+b_1)}+e^{i(a_1+b_1)}i(a_2+b_2)+O((a_2+b_2)^2) 	\\
				&=e^{i(a_1+b_1)}+e^{ia_1}ia_2-e^{ia_1}a_2b_1+e^{ia_1}ib_2+O((a_2+b_2)^2+a_2b_1^2+b_1b_2) 
			\end{align*}
			and omitting higher-order terms where possible, we have $B^{n,i,12}_{t,T}(u)= \sum_{j=1}^5 B^{n,i,12j}_{t,T}(u)+o^\uc(\Den/\sqrt{T}+\sqrt{\Den}T)$, where
			\begin{align*}
				B^{n,i,121}_{t,T}(u)&=\Theta_{t^n_i,\Den}(u_{T^n_{i-1}})\sum_{k=1}^\infty \E_{t^n_i}\biggl[\int_{\S^{(k)}_{t^n_{i},t^n_{i-1}-t^n_i}}   \wh\theta_{t^n_i}(u_{T^n_{i-1}},\bz_k) \by_\mm(d\bs_k,d\bz_k)\\
				&\quad\quad\quad\quad\quad\quad\quad\quad \times\exp\biggl(iu_{T^n_{i-1}} \sum_{k=1}^N \int_{\S^{(k)}_{t^n_{i-1},T^n_{i-1}}} \theta(t^n_i,\bz_k)\ov\by^{n,i}(d\bs_k,d\bz_k)\biggr)  \biggr],\\
				B^{n,i,122}_{t,T}(u)&=\Theta_{t^n_i,\Den}(u_{T^n_{i-1}})\sum_{k=1}^\infty \E_{t^n_i}\biggl[\int_{\S^{(k)}_{t^n_{i},t^n_{i-1}-t^n_i}}   \wh\theta_{t^n_i}(u_{T^n_{i-1}},\bz_k) \by_\mm(d\bs_k,d\bz_k)\\
				&\quad\times e^{iu_{T^n_{i-1}}   \iint_{t^n_{i-1}}^{t+T} \theta(t^n_i,z)\ov y^{n,i}(ds,dz)}  iu_{T^n_{i-1}}\iint_{t^n_{i-1}}^{t+T} \ga(t^n_i,z)(\wh\mu-\wh{\ov\mu}^{n,i})(ds,dz) \biggr],\\
				B^{n,i,123}_{t,T}(u)&=  \E_{t^n_i}\biggl[\iint_{t^n_{i}}^{t^n_{i-1}}  \wh\theta_{t^n_i}(u_{T^n_{i-1}},z) y_\mm(ds,dz)e^{iu_{T^n_{i-1}}   \iint_{t^n_{i-1}}^{t+T} \theta_\mm(t^n_i,z)\ov y^{n,i}_\mm(ds,dz)} \\
				&\quad\times iu_{T^n_{i-1}}\iint_{t^n_{i-1}}^{t+T} \ga(t^n_i,z)(\wh\mu-\wh{\ov\mu}^{n,i})(ds,dz)\\
				&\quad\times iu_{T^n_{i-1}} \iint_{t^n_{i-1}}^{t+T} \iint_{t^n_{i-1}}^{s} \theta(t^n_i,z,z')\ov y^{n,i}(dr,dz')\ov y^{n,i}(ds,dz)\biggr],\\
				B^{n,i,124}_{t,T}(u)&=  \E_{t^n_i}\biggl[\iint_{t^n_{i}}^{t^n_{i-1}}   \wh\theta_{t^n_i}(u_{T^n_{i-1}},z) y_\mm(ds,dz) e^{iu_{T^n_{i-1}}   \iint_{t^n_{i-1}}^{t+T} \theta_\mm(t^n_i,z)\ov y^{n,i}_\mm(ds,dz)}\\
				&\quad\times iu_{T^n_{i-1}}\iint_{t^n_{i-1}}^{t+T} \iint_{t^n_{i-1}}^{s-} \theta_\mm(t^n_i,z,z')\ov y^{n,i}_\mm(dr,dz')(y_\mm-\ov y_\mm^{n,i})(ds,dz) \biggr],\\
				B^{n,i,125}_{t,T}(u)&=  \E_{t^n_i}\biggl[\iint_{t^n_{i}}^{t^n_{i-1}}   \wh\theta_{t^n_i}(u_{T^n_{i-1}},z) y_\mm(ds,dz) e^{iu_{T^n_{i-1}}   \iint_{t^n_{i-1}}^{t+T} \theta_\mm(t^n_i,z)\ov y^{n,i}_\mm(ds,dz)}\\
				&\quad\times iu_{T^n_{i-1}}\iint_{t^n_{i-1}}^{t+T} \iint_{t^n_{i-1}}^{s-} \theta_\mm(t^n_i,z,z')  (y_\mm-\ov y_\mm^{n,i})(dr,dz') \ov y^{n,i}_\mm(ds,dz)\biggr].
			\end{align*}
			Since the exponential term in $B^{n,i,121}_{t,T}(u)$ is $\calf_{t^n_i}\vee \calg_{t^n_{i-1}}$-measurable, we  see that $B^{n,i,121}_{t,T}(u)=0$ by further conditioning on $\calg_{t^n_{i-1}}$. For the other four terms,
			we first rewrite  the exponential term as in \eqref{eq:chaos-2}, ignoring the $o^\uc$-term in \eqref{eq:chaos-2} (but with another $\exp(iu\sqrt{T^n_{i-1}}\al_{t^n_i})$-factor in the case of $B^{n,i,122}_{t,T}(u)$). Then we make use of the following consequence of the integration by parts formula for semimartingales:
			% suppose that $L$ and $N$ are martingales and that $X=M+A$ where $M$ is another martingale and $A$ is a continuous finite variation process. If all processes have bounded third moments, then for all $s<\tau$,
			%\begin{align*}
			%\E_s[L_\tau X_\tau N_\tau]	&=L_sX_s N_s + \E_s\biggl[ \int_s^\tau \int_s^r L_u dN_u dA_u + \int_s^\tau L_r d[M,N]_r + \int_s^\tau\int_s^r N_u dL_u dA_u \\
			%&\quad+ \int_s^\tau N_r d[L,M]_r + \int_s^\tau [L,N]_r dA_r + \int_s^\tau X_rd[L,N]_r +\int_s^\tau d[[L,N],M]_r\biggr]. 
			%\end{align*}
			If $L$, $M$ and $N$ are square-integrable martingales, then for all $s<\tau$,
			\begin{align*}
				\E_s[L_\tau M_\tau N_\tau]	&=L_sX_s N_s + \E_s\biggl[  \int_s^\tau L_r d[M,N]_r   + \int_s^\tau N_r d[L,M]_r  + \int_s^\tau M_rd[L,N]_r \\
				&\quad+\int_s^\tau d[[L,N],M]_r\biggr]. 
			\end{align*}
			For $B^{n,i,122}_{t,T}(u)$ and $B^{n,i,124}_{t,T}(u)$, we apply this rule to $s=t^n_i$, $\tau=t+T$, $L_\tau$ given by $L_\tau=\int_{\S^{(k)}_{t^n_{i},t^n_{i-1}\wedge \tau-t^n_i}}  \wh\theta_{t^n_i}(u_{T^n_{i-1}},\bz_k) \by_\mm(d\bs_k,d\bz_k)$ (for general $k$ in the first case and $k=1$ in the second), %$X_t=\exp (iu_{T^n_{i-1}}   \iint_{t^n_{i-1}}^\tau \theta(t^n_i,z)\ov y^{n,i}(ds,dz) )$ or $X_t=\exp (iu_{T^n_{i-1}}   \iint_{t^n_{i-1}}^\tau \theta_\mm(t^n_i,z)\ov y^{n,i}_\mm(ds,dz) )$ 
			$M_\tau$ given by one of the terms in parenthesis in \eqref{eq:chaos-2},
			and $N_\tau$ given by what follows the exponential term in the definition of $B^{n,i,122}_{t,T}(u)$ and $B^{n,i,124}_{t,T}(u)$. With this choice, we have $L_s=L_{t^n_i}=0$ and $[L,M]=[L,N]=0$ (because $L$ only lives on $[t^n_i,t^n_{i-1}]$ while $M$ and $N$ only live on $[t^n_{i-1},t+T]$), %. By the same reason, we have $\int_s^r N_u dL_u=0$, 
			so that the above display reduces to
			\begin{equation*}
				\E_s[L_\tau X_\tau N_\tau]	=  \E_s\biggl[  \int_s^\tau L_r d[M,N]_r \biggr]=\E_s\biggl[  \int_s^\tau L_r d\langle M,N\rangle_r \biggr]
			\end{equation*}
			in our setting. Moreover, in our example, we either have $d\langle M,N\rangle_s = -\int_{\R^{d'}} (\cdots)(\la(s,z)-\la_{t^n_i,t^n_{i-1}}(s,z))_-ds F(dz)$ %or $d\langle M,N\rangle_s = \int_{\R^{d'}} (\cdots)(\la(s,z)-\la_{t^n_i,t^n_{i-1}}(s,z))_+ds F(dz)$ 
			or, if we pick the constant term in \eqref{eq:chaos-2}, $d\langle M,N\rangle_s =0$. A power counting argument now shows that the above conditional expectation is $o^\uc(\Den/\sqrt{T})$ and thus,
			$B^{n,i,122}_{t,T}(u)+B^{n,i,124}_{t,T}(u)=o^\uc(\Den/\sqrt{T})$.
			
			For $B^{n,i,125}_{t,T}(u)$, we perform an additional step and replace $\ov y^{n,i}_\mm(ds,dz)$, which appears twice in its definition, by $y^{n,i}_\mm(ds,dz)$ as defined after \eqref{eq:muni} (note that we do not change $y_\mm(ds,dz)$ or  $(y_\mm-y_\mm^{n,i})(dr,dz')$). This change leads to an error of size $o^\uc(\sqrt{\Den}(T+\Den^{1/4})(T+\Den^{1/4}+T^{1/4}))=o^\uc(\sqrt{\Den}T + \Den + \Den^{3/4}T^{1/4})=o^\uc(\sqrt{\Den}T+\Den/\sqrt{T})$, which is negligible. With this modification, we can now repeat the argument above and obtain
			\begin{align*}
				&B^{n,i,125}_{t,T}(u)\\
				&\quad= e^{-\frac12u^2\si^2_{t^n_i}+T^n_{i-1}\vp_{t^n_i}(u_{T^n_{i-1}})} iu_{T^n_{i-1}}\sum_{k=1}^\infty \E_{t^n_i}\biggl[ \iint_{t^n_{i-1}}^{t+T}\iint_{t^n_{i}}^{t^n_{i-1}}   \wh\theta_{t^n_i}(u_{T^n_{i-1}},z) y_\mm(ds,dz)  \\
				&\qquad\times \int_{\S^{(k-1)}_{t^n_{i-1},s-t^n_{i-1}}}   \wh\theta_{t^n_i}(u_{T^n_{i-1}},\bz_{k-1})  \by^{n,i}_\mm(d\bs_{k-1},d\bz_{k-1}) \\
				&\qquad\times  \biggl(iu_{T^n_{i-1}}\si_{t^n_i}\iint_{t^n_{i-1}}^{s} \ga^\si(t^n_i,z')  (\wh\mu-\wh{\ov\mu}^{n,i})(dr,dz')\delta_0(dz)ds\\
				&\quad\qquad +(e^{iu_{T^n_{i-1}}\ga(t^n_i,z)}-1)\iint_{t^n_{i-1}}^{s} \ga^\ga(t^n_i,z,z')  (\wh\mu-\wh{\ov\mu}^{n,i})(dr,dz')\la(t^n_i,z)F(dz)ds\biggr)\biggr]\\
				&\qquad+o^\uc(\sqrt{\Den}T+\Den/\sqrt{T}).
			\end{align*}
			In order to evaluate the conditional expectation,  we first condition on $\calf_{t^n_{i-1}}$. If $k=1$, the resulting conditional expectation vanishes.
			If $k\geq2$, we have a product of the two martingale terms $ \int_{\S^{(k-1)}_{t^n_{i-1},s-t^n_{i-1}}}   \wh\theta_{t^n_i}(u_{T^n_{i-1}},\bz_{k-1})  \by^{n,i}_\mm(d\bs_{k-1},d\bz_{k-1})$ and $\iint_{t^n_{i-1}}^{s} (\cdots)  (\wh\mu-\wh{\ov\mu}^{n,i})(dr,dz')$. By the integration-by-parts formula, only the quadratic variation term survives the $\calf_{t^n_{i-1}}$-conditional expectation. The resulting predictable covariation is an integral with respect to \eqref{eq:qv}, from which the reader can verify that $B^{n,i,125}_{t,T}(u)=o^\uc(\sqrt{\Den}T+\Den/\sqrt{T})$.
			
			For $B^{n,i,123}_{t,T}(u)$ the argument is similar: we only need to use integration by parts to write the product of the two terms after the exponential as a sum of two martingale terms (to which we can apply the analysis of $B^{n,i,125}_{t,T}(u)$ or $B^{n,i,125}_{t,T}(u)$) plus a drift term (which includes a predictable covariation term of the form $- (\cdots)(\la(s,z)-\la_{t^n_i,t^n_{i-1}}(s,z))_-ds F(dz)$, so the previous argument applies again). The upshot is that $B^{n,i,123}_{t,T}(u)=o^\uc(\sqrt{\Den}T+\Den/\sqrt{T})$ and hence,
			\begin{equation}\label{eq:B12-2} 
				B^{n,i,12}_{t,T}(u)=o^\uc(\Den/\sqrt{T}+\sqrt{\Den}T).
			\end{equation}

			Finally, we discuss $C^{n,i}_{t,T}(u)$, which can be written as
			\begin{align*}
				C^{n,i}_{t,T}(u)	&=\E_{t^n_i}\biggl[\exp\biggl(iu_{T^n_{i}}  \sum_{k=1}^N \int_{\S^{(k)}_{t^n_{i},T^n_{i}}} \theta(t^n_i,\bz_k)\by(d\bs_k,d\bz_k)\biggr)\\
				&\quad\times\biggl(\exp\biggl(i(u_{T^n_{i-1}} -u_{T^n_i}) \sum_{k=1}^N \int_{\S^{(k)}_{t^n_{i},T^n_{i}}} \theta(t^n_i,\bz_k)\by(d\bs_k,d\bz_k)\biggr)-1\biggr)\biggr].
			\end{align*}
			Since $u_{T^n_{i-1}}-u_{T^n_i}=O^\uc(\Den/T^{3/2})$ and $\int_{\S^{(k)}_{t^n_{i},T^n_{i}}} \theta(t^n_i,\bz_k)\by(d\bs_k,d\bz_k)=O(T^{k/2})$, we can discard terms corresponding to $k\geq3$ in both sums in the previous display in exchange of an $O^\uc(\Den)$-error. Moreover, 
			\begin{align*}
				&\exp\biggl(i(u_{T^n_{i-1}} -u_{T^n_i}) \sum_{k=1}^2 \int_{\S^{(k)}_{t^n_{i},T^n_{i}}} \theta(t^n_i,\bz_k)\by(d\bs_k,d\bz_k)\biggr)-1 \\
				&\quad=\Bigl(e^{i(u_{T^n_{i-1}} -u_{T^n_i})  \iint_{t^n_{i}}^{t+T} \theta(t^n_i,z)y(ds,dz) }-1\Bigr)+e^{i(u_{T^n_{i-1}} -u_{T^n_i})\iint_{t^n_{i}}^{t+T} \theta(t^n_i,z)y(ds,dz) } \\
				&\quad\quad \times i(u_{T^n_{i-1}} -u_{T^n_i})\int_{\S^{(2)}_{t^n_{i},T^n_{i}}} \theta(t^n_i,\bz_2)\by(d\bs_2,d\bz_2) + O^\uc( \Den^2/T)\\
				&\quad=\Bigl(e^{i(u_{T^n_{i-1}} -u_{T^n_i})  \iint_{t^n_{i}}^{t+T} \theta(t^n_i,z)y(ds,dz) }-1\Bigr)+e^{i(u_{T^n_{i-1}} -u_{T^n_i})\int_{t^n_i}^{t+T}\si_{t^n_i}dW_s} \\
				&\quad\quad \times i(u_{T^n_{i-1}} -u_{T^n_i})  \int_{t^n_{i}}^{t+T} \int_{t^n_i}^{s} \theta^2_\cc(t^n_i)y_\cc(dr)y_\cc(ds) +o^\uc(\Den/\sqrt{T}).
			\end{align*}
			Since only choosing $dW_rdW_s$ for $y_\cc(dr)y_\cc(ds)$ yields a nonzero conditional expectation, we obtain from \eqref{eq:chaos} that
			\begin{align*}
				&C^{n,i}_{t,T}(u)\\
				&\quad	=\E_{t^n_i}\Bigl[e^{i u_{T^n_{i}}\int_{t^n_{i}}^{t+T} \theta(t^n_i,z)y(ds,dz)}\Bigl(e^{i(u_{T^n_{i-1}} -u_{T^n_i})  \int_{t^n_{i}}^{t+T} \theta(t^n_i,z)y(ds,dz) }-1\Bigr)\Bigr] \\ %e^{i u_{T^n_{i-1}}\int_{t^n_{i}}^{t+T} \theta(t^n_i,z)y(ds,dz)}-e^{i u_{T^n_i}\int_{t^n_{i}}^{t+T} \theta(t^n_i,z)y(ds,dz)} \Bigr] \\
				&\qquad+\E_{t^n_i}\biggl[e^{i u_{T^n_{i}}\int_{t^n_i}^{t+T}\si_{t^n_i}dW_s}\Bigl(e^{i(u_{T^n_{i-1}} -u_{T^n_i}) \int_{t^n_i}^{t+T}\si_{t^n_i}dW_s}-1\Bigr)  iu_{T^n_i}\si^\si_{t^n_i}\int_{t^n_{i}}^{t+T} \int_{t^n_i}^{s} dW_{r}dW_{s}\biggr]\\
				&\qquad + \E_{t^n_i}\biggl[e^{i u_{T^n_{i-1}}\int_{t^n_i}^{t+T}\si_{t^n_i}dW_s}i(u_{T^n_{i-1}} -u_{T^n_i})\si^\si_{t^n_i}\int_{t^n_{i}}^{t+T} \int_{t^n_i}^{s} dW_{r}dW_{s}\biggr]+o^\uc(\Den/\sqrt{T})\\
				&\quad=\E_{t^n_i}\Bigl[e^{i u_{T^n_{i}}\int_{t^n_{i}}^{t+T} \theta(t^n_i,z)y(ds,dz)}\Bigl(e^{i(u_{T^n_{i-1}} -u_{T^n_i})  \int_{t^n_{i}}^{t+T} \theta(t^n_i,z)y(ds,dz) }-1\Bigr)\Bigr]\\
				&\qquad +iu_{T^n_{i-1}}\si^\si_{t^n_i} \E_{t^n_i}\biggl[e^{i u_{T^n_{i-1}}\int_{t^n_i}^{t+T}\si_{t^n_i}dW_s}\int_{t^n_{i}}^{t+T} \int_{t^n_i}^{s} dW_rdW_{s}\biggr]\\
				&\qquad-iu_{T^n_{i}}\si^\si_{t^n_i} \E_{t^n_i}\biggl[e^{i u_{T^n_{i}}\int_{t^n_i}^{t+T}\si_{t^n_i}dW_s}\int_{t^n_{i}}^{t+T} \int_{t^n_i}^{s} dW_{r}dW_{s}\biggr]+o^\uc(\Den/\sqrt{T}).
			\end{align*}
			Note that $$\exp\biggl(i u_{T^n_{i}}\int_{t^n_{i}}^{t+T} \theta(t^n_i,z)y(ds,dz)\biggr)=\exp\biggl(i u_{T^n_{i}}\int_{t^n_{i}}^{t+T} \theta(t^n_i,z)y^{n,i}(ds,dz)\biggr)+o^\uc(\sqrt{T})$$ by \eqref{eq:chaos-1} and \eqref{eq:chaos}. By a similar argument,  we have that \begin{align*}
				&\exp\biggl(i(u_{T^n_{i-1}} -u_{T^n_i})  \int_{t^n_{i}}^{t+T} \theta(t^n_i,z)y(ds,dz)\biggr )-1\\
				&\quad=\exp\biggl(i(u_{T^n_{i-1}} -u_{T^n_i})  \int_{t^n_{i}}^{t+T} \theta(t^n_i,z)y^{n,i}(ds,dz) \biggr)-1 + o^\uc(\Den/\sqrt{T})\\ &\quad=O^\uc(\Den/T).
			\end{align*} Thus, we can replace $y$ by $y^{n,i}$ on the right-hand side of the last display and, using the assumption that $\Den/T\to0$, obtain
			\begin{equation}\label{eq:C}\begin{split}
					C^{n,i}_{t,T}(u)	& =\Theta_{t^n_i,T^n_i}(u_{T^n_{i-1}})-\Theta_{t^n_i,T^n_i}(u_{T^n_{i}})-\frac12 i u^3\si^2_{t^n_i} \si^\si_{t^n_i}\\
					&\quad\times\biggl( \frac{(T^n_i)^2}{(T^n_{i-1})^{3/2}} e^{-\frac12 u^2\si^2_{t^n_i} T^n_i/T^n_{i-1}}-\sqrt{T^n_i} e^{-\frac12 u^2\si^2_{t^n_i} } \biggr)+o^\uc(\Den/\sqrt{T})\\
					&=\Theta_{t^n_i,T^n_i}(u_{T^n_{i-1}})-\Theta_{t^n_i,T^n_i}(u_{T^n_{i}})-\frac34 i u^3\si^2_{t^n_i} \si^\si_{t^n_i}e^{-\frac12 u^2\si^2_{t^n_i}} \frac{\Den}{\sqrt{T^n_{i}}} \\
					&\quad+ \frac14iu^5\si^4_{t^n_i}\si^\si_{t^n_i}e^{-\frac12 u^2\si^2_{t^n_i}} \frac{\Den}{\sqrt{T^n_{i}}}+o^\uc(\Den/\sqrt{T})\\
					&=\Theta_{t^n_i,T^n_i}(u_{T^n_{i-1}})-\Theta_{t^n_i,T^n_i}(u_{T^n_{i}})-\frac34 i u^3\si^2_{t^n_i} \si^\si_{t^n_i}	\Theta_{t^n_i,T^n_{i-1}}(u_{T^n_{i-1}})\frac{\Den}{\sqrt{T^n_{i-1}}} \\
					&\quad+ \frac14iu^5\si^4_{t^n_i}\si^\si_{t^n_i}	\Theta_{t^n_i,T^n_{i-1}}(u_{T^n_{i-1}}) \frac{\Den}{\sqrt{T^n_{i-1}}}+o^\uc(\Den/\sqrt{T}).
			\end{split}\end{equation}
			Gathering the approximations   in \eqref{eq:A3}, \eqref{eq:A-simple}, \eqref{eq:A12-1}--\eqref{eq:A111}, \eqref{eq:A1121-1}--\eqref{eq:A1122}, \eqref{eq:B3}--\eqref{eq:C}, we derive (7.4).
		\end{proof}
		
		\begin{proof}[Proof of Lemma~3.2]
			The statement can be shown along the lines of the proof of Theorem 3.1 in \cite{CT22} (specialized to the case $H=\frac12$). However, some modifications are needed due to the stochastic intensity $\la(t,z)$ in our setting. In a first step, let us define
			\begin{align*}
				x'_{t+T}	&=x_t+\al_tT +\si_t(W_{t+T}-W_t)+\int_t^{t+T}\Bigl( \si^\si_t (W_s-W_t)+\ov\si^\si_t (\ov W_s-\ov W_t)\Bigr)dW_s    \\
				&\quad+ \int_t^{t+T}\iint_t^s\ga^\si(t,z)\wh\mu(dr,dz)dW_s+\iint_t^{t+T} \ga(s,z)\wh\mu(ds,dz),\\
				x''_{t+T}	&=x'_{t+T}+\iint_t^{t+T} (\ga(t,z)-\ga(s,z))\wh\mu(ds,dz).
			\end{align*}
			Then $x'_{t+T}$ and $x''_{t+T}$ are equal to $x'_{t+T}$ and $x''_{t+T}$ as defined after the proof of Lemma 7.1 in \cite{CT22} (specialized to $H=\frac12$), except that terms   of exact order $T^{3/2}$ have been omitted. As in the proof of the cited lemma, one can show that
			\[ \call_{t,T}(u)=\E_t[e^{iu_T(x'_{t+T}-x_t)}]+C_t(u)T+o^\uc(T). \]
			(The omitted terms    of   order $T^{3/2}$ enter the $C_t(u)T$-bin.) As a next step, we pass from $x'_{t+T}$ to $x''_{t+T}$ by expanding
			\begin{equation*}
				\E_t[e^{iu_T(x'_{t+T}-x_t)}]	=\E_t[e^{iu_T(x''_{t+T}-x_t)}]+\E_t[iu_Te^{iu_T(x''_{t+T}-x_t)}(x'_{t+T}-x''_{t+T})]+o^\uc(T).
			\end{equation*}
			Using the fact that 
			\begin{align*} 
				e^{iu_T(x''_{t+T}-x_t)}&=e^{iu_T\iint_t^{t+T} \theta(t,z)y(ds,dz)}+O^\uc(\sqrt{T})\\
				&=\Theta_{t,T}(u_T)\biggl(1+\sum_{k=1}^\infty \int_{\S^{(k)}_{t,T}}\wh\theta_t(u_T,\bz_k)\by_\mm(d\bs_k,d\bz_k)\biggr)+O^\uc(\sqrt{T}),
			\end{align*} 
			we derive
			\begin{align*}
				&iu_T\E_t[e^{iu_T(x''_{t+T}-x_t)}(x'_{t+T}-x''_{t+T})]	\\
				&\quad= iu_T\Theta_{t,T}(u_T)\sum_{k=1}^\infty\E_t\biggl[  \int_{\S^{(k)}_{t,T}}\wh\theta_t(u_T,\bz_k)\by_\mm(d\bs_k,d\bz_k) \\
				&\quad\quad\qquad\qquad\qquad\qquad\qquad\qquad\qquad\times\iint_t^{t+T}(\ga(s,z)-\ga(t,z))\wh\mu(ds,dz)\biggr]+o^\uc(T)\\
				&\quad=iu_T\Theta_{t,T}(u_T)\sum_{k=1}^\infty\iint_t^{t+T}(e^{iu_T\ga(t,z)}-1)\E_t\biggl[  \int_{\S^{(k-1)}_{t,s-t}}\wh\theta_t(u_T,\bz_{k-1})\by_\mm(d\bs_{k-1},d\bz_{k-1})\\
				&\quad\quad\qquad\qquad\qquad\qquad\qquad\qquad\qquad\times (\ga(s,z)-\ga(t,z))\la(s,z)\biggr]F(dz)ds+o^\uc(T)\\
				&\quad=iu_T\Theta_{t,T}(u_T)\sum_{k=1}^\infty\iint_t^{t+T}(e^{iu_T\ga(t,z)}-1)\E_t\biggl[  \int_{\S^{(k-1)}_{t,s-t}}\wh\theta_t(u_T,\bz_{k-1})\by_\mm(d\bs_{k-1},d\bz_{k-1})\\
				&\quad\quad\qquad\qquad\qquad\qquad\qquad\qquad\qquad\times (\ga(s,z)-\ga(t,z))\biggr]\la(t,z)F(dz)ds+o^\uc(T).
			\end{align*}
			By writing $\ga(s,z)-\ga(t,z)$ as a semimartingale increment with characteristics frozen at time $t$ plus an $O(\sqrt{T})$-term, we realize that only $k=2$ yields a contribution that is not $o^\uc(T)$. A short calculation results in
			\begin{equation}\label{eq:exp-1} 
				\begin{split}
					&	iu_T\E_t[e^{iu_T(x''_{t+T}-x_t)}(x'_{t+T}-x''_{t+T})]\\
					&\quad=-\frac12u^2T\Theta_{t,T}(u_T)\si_t\int_{\R^{d'}}(e^{iu_T\ga(t,z)}-1)\si^\ga(t,z)\nu_t(dz)\\
					&\qquad+\frac12iuT^{3/2}\Theta_{t,T}(u_T)\chi^{(2)}_t(u_T).
				\end{split}
			\end{equation}
			
			Next, we evaluate $\E_t[e^{iu_T(x''_{t+T}-x_t)}]$ by expanding
			\begin{align*}
				&	\E_t[e^{iu_T(x''_{t+T}-x_t)}]	=\E_t[e^{iu_T\iint_t^{t+T}\theta(t,z)y(ds,dz)}]+ iu_T\E_t\biggl[e^{iu_T\iint_t^{t+T}\theta(t,z)y(ds,dz)}\\
				&\qquad\times\int_t^{t+T}\biggl( \si^\si_t (W_s-W_t)+\ov\si^\si_t (\ov W_s-\ov W_t)+\iint_t^s\ga^\si(t,z)\wh\mu(dr,dz)\biggr)dW_s\biggr]. 
			\end{align*}
			As in the proof of Lemma~7.5 in \cite{CT22}, one can show that the second term equals
			\begin{equation}\label{eq:exp-2} 
				-\frac{1}{2}iu^3{\textstyle\sqrt{T}}\Theta_{t,T}(u_T)\si_{t}^2\si^\si_{t}-\frac12u^2T\Theta_{t,T}(u_T)\si_t\int_{\R^{d'}}(e^{iu_T\ga(t,z)}-1)\ga^\si(t,z)\nu_t(dz).
			\end{equation}
			The last term to be analyzed is $\E_t[e^{iu_T\iint_t^{t+T}\theta(t,z)y(ds,dz)}]$. To this end, we use the fact that
			\begin{align*}
				&e^{iu_T\iint_t^{t+T} \theta(t,z)y(ds,dz)}\\
				&\quad=\exp\biggl(\iint_t^{t+T} (e^{iu_T\ga(t,z)}-1-iu_T\ga(t,z))(\la(s,z)-\la(t,z))F(dz)ds \biggr)\\
				&\qquad\times \Theta_{t,T}(u_T)\biggl(1+\sum_{k=1}^\infty \int_{\S^{(k)}_{t,T}}\wh\theta_t(u_T,\bz_k)\by_\mm(d\bs_k,d\bz_k)\biggr)
			\end{align*}
			to get
			\begin{align*}
				&\E_t[e^{iu_T\iint_t^{t+T} \theta(t,z)y(ds,dz)}]\\
				&\quad	=\Theta_{t,T}(u_T)+\Theta_{t,T}(u_T)\E_t\biggl[ \biggl(1+\sum_{k=1}^\infty \int_{\S^{(k)}_{t,T}}\wh\theta_t(u_T,\bz_k)\by_\mm(d\bs_k,d\bz_k)\biggr)\\
				&\qquad\times\iint_t^{t+T} (e^{iu_T\ga(t,z)}-1-iu_T\ga(t,z))(\la(s,z)-\la(t,z))F(dz)ds\biggr] +o^\uc(T).
			\end{align*}
			It can be shown that only the term $k=1$ has  a nonnegligible contribution to the last conditional expectation. In computing the covariance between the $y_\mm(ds,dz)$-integral and $\la(s,z)-\la(t,z)$, we can further ``freeze'' characteristics at $t$, which shows that the last conditional expectation above equals
			\begin{equation}\label{eq:exp-3} \frac12iuT^{3/2}\si_t\chi^{(3)}_t(u_T)+\frac12T^2\chi^{(4)}_t(u_T)+o^\uc(T). \end{equation}
			We deduce (3.6) from \eqref{eq:exp-1}, \eqref{eq:exp-2} and \eqref{eq:exp-3}.
		\end{proof}
		
		\begin{proof}[Proof of Lemma 8.1]
			For   (8.1) and (8.2), we refer the reader to Lemma 1 in \cite{T19}. For (8.3), we consider $k> 0$ first. Using the notation following \eqref{eq:not2} and the elementary inequality $\lvert (Y-K)_+ - (X-K)_+\rvert \leq \lvert Y-X\rvert$, where $X,Y,K\in\R$ and $x_+=x\vee 0$, we have that 
			\begin{equation*}
				\lvert O_{s,T}(k+x_s)-\E^\Q_s[(e^{x_s+\ov x_{s,T}}-e^{k+x_s})_+] \rvert	 \leq \E^\Q_s[\lvert e^{x_{s+T}} - e^{x_s+\ov x_{s,T}}\rvert].
			\end{equation*}
			Arguing as in the proof of Lemma~3 of \cite{QT19} and using the last bound in \eqref{eq:remain}, one can show that the right-hand side is bounded by $C_s\E[(x_{s+T}-x_s-\ov x_{s,T})^2]^{1/2}\leq C_s T^{(N+1)/2}$, which in turn is bounded by $C_s\sqrt{\Den T}$. A similar argument shows the same bound for $k<0$. Therefore, it suffices to prove (8.3) with $O_{s,T}(k+x_s)$ replaced by $\ov O_{s,T}(k+x_s) = \E^\Q_s[(e^{x_s+\ov x_{s,T}}-e^{k+x_s})_+]$ and $O_{r,T+(s-r)}(k+x_r)$ replaced by $\ov O_{r,T+(s-r)}(k+x_r)$.

			By the Carr--Madan formula (1.2) (applied to $x_s+\ov x_{s,T}$) and Fourier inversion, we have
			\begin{align*}
				e^{-k-x_s}\ov O_{s,T}(k+x_s)&=\frac{1}{2\pi} \int_\R e^{-iuk}\frac{1-\E^\Q_s[e^{iu\ov x_{s,T}}]}{u^2+iu} du \\
				&= \frac{1}{2\pi} \int_\R e^{-iuk/\sqrt{T}}\frac{1-\E^\Q_s[e^{iu\ov x_{s,T}/\sqrt{T}}]}{u^2/\sqrt{T}+iu} du.
			\end{align*}
			In what follows, we will reuse estimates from the proof of Proposition~7.1, which is why we   consider the case $s=t^n_{i-1}$, $T=T^n_{i-1}$ and $r=t^n_i$ only. The proof, of course, applies to general $s$, $T$ and $r$ as well. By the formula in the previous display and \eqref{eq:ABC},
			\begin{align*}
				&e^{-k-x_{t^n_{i-1}}}\ov O_{t^n_{i-1},T^n_{i-1}}(k+x_{t^n_{i-1}})-e^{-k-x_{t^n_i}}\ov O_{{t^n_i},T^n_i}(k+x_{t^n_i})	\\
				&\quad= -\frac1{2\pi}\int_\R e^{-iu_{T^n_{i-1}}k}\frac{\E^\Q_{t^n_{i-1}}[e^{iu_{T^n_{i-1}}\ov x_{t^n_{i-1},T^n_{i-1}}}]-\E^\Q_{t^n_i}[e^{iu_{T^n_{i-1}}\ov x_{t^n_{i},T^n_i}}]}{u^2/\sqrt{T^n_{i-1}}+iu} du\\
				&\quad=-\frac1{2\pi}\int_\R e^{-iu_{T^n_{i-1}}k}\frac{ A^{n,i}_{t,T}(u)}{u^2/\sqrt{T^n_{i-1}}+iu} du + \frac1{2\pi}\int_\R e^{-iu_{T^n_{i-1}}k}\frac{B^{n,i}_{t,T}(u)}{u^2/\sqrt{T^n_{i-1}}+iu} du.
			\end{align*}
			Denote these two terms by $\wt A^{n,i}_{t,T}$ and $\wt B^{n,i}_{t,T}$. Since $\E^\Q_{t^n_i}[\lvert A^{n,i}_{t,T}(u)\rvert] \leq C_{t^n_i}( \lvert u\rvert \sqrt{\Den}\wedge 1)$, we have
			\begin{align*}
				\E^\Q_{t^n_i}[\lvert  \wt A^{n,i}_{t,T} \rvert ] 	&\leq C_{t^n_i} \int_\R \frac{\lvert u\rvert \sqrt{\Den}\wedge 1}{u^2/\sqrt{T^n_{i-1}}\vee \lvert u\rvert} du\\
				& \leq 2C_{t^n_i}\biggl(\int_0^{\sqrt{T^n_{i-1}}} \sqrt{\Den} du + \int_{\sqrt{T^n_{i-1}}}^{1/\sqrt{\Den}} \frac{\sqrt{\Den T^n_{i-1}}}{u} du + \int_{1/\sqrt{\Den}}^\infty \frac{\sqrt{T^n_{i-1}}}{u^2} du\biggr)\\
				&\leq C_{t^n_i}\sqrt{\Den T}\log \Den^{-1}.
			\end{align*}
			Next, as in   \eqref{eq:Bni}, one can show that
			$
			\E^\Q_{t^n_{i}}[\lvert B^{n,i}_{t,T}(u) - B^{n,i,1}_{t,T}(u) \rvert ] 	\leq C_{t^n_i}(\lvert u\rvert  \sqrt{\Den}\wedge 1)
			$,
			so arguing as above, we conclude that only the contribution of $B^{n,i,1}_{t,T}(u)$ to $\wt B^{n,i}_{t,T}$ needs to be considered further. Because $\E^\Q_{t^n_{i}}[\lvert B^{n,i,1}_{t,T}(u)-B^{n,i,11}_{t,T}(u)-B^{n,i,12}_{t,T}(u)\rvert]\leq C_{t^n_i}(u^2\Den^{3/2}/T\wedge 1) \leq C_{t^n_i}(u^2\Den\wedge 1)\leq C_{t^n_i}(\lvert u\rvert \sqrt{\Den}\wedge 1)$, it actually suffices to analyze the contributions of $B^{n,i,11}_{t,T}(u)$ and $B^{n,i,12}_{t,T}(u)$ to $\wt B^{n,i}_{t,T}$, which we denote by $\wt B^{n,i,11}_{t,T}$ and $\wt B^{n,i,12}_{t,T}$ in the following.
			
			Regarding $\wt B^{n,i,11}_{t,T}$, note that we can remove terms that correspond to $k\geq3$ in the definition of $B^{n,i,11}_{t,T}(u)$. Indeed, this leads to an error of order  $O((u^2\Den/T \wedge 1)(\lvert u\rvert T\wedge 1))$ in $B^{n,i,11}_{t,T}(u)$ and to an error of order 
			\begin{align*}
				&\int_\R \frac{(u^2\Den/T\wedge 1)(uT \wedge1)}{u^2/\sqrt{T}\vee \lvert u\rvert} du\\
				& \quad\leq 2\biggl(\int_0^{\sqrt{T/\Den}} u\Den\sqrt{T} du + \int_{\sqrt{T/\Den}}^{1/T} \frac{T^{3/2}}{u} du + \int_{1/T}^\infty \frac{\sqrt{T}}{u^2} du\biggr)\\
				&\quad\leq CT^{3/2}\log (\sqrt{\Den}/T^{3/2})= O(\sqrt{\Den T}\log \Den^{-1})
			\end{align*}
			in $\wt B^{n,i,11}_{t,T}$  (recall that $\sqrt{\Den}\asymp T$) . As a result, we can substitute
			$$ (\Theta_{t^n_i,\Den}(u_{T^n_{i-1}})-1 )\E^\Q_{t^n_i}\biggl[\exp\biggl(iu_{T^n_{i-1}} \sum_{k=1}^2 \int_{\S^{(k)}_{t^n_{i-1},T^n_{i-1}}} \theta(t^n_i,\bz_k)\by (d\bs_k,d\bz_k)\biggr)  \biggr]$$
			for $B^{n,i,11}_{t,T}(u)$. In fact, this can be further simplified: with similar estimates, one can show that we only need to keep the drift term if $k=1$ and, by Assumption~2, the jump component $y(ds_1,dz_1)$ if $k=1$. Thus, we can replace $\wt B^{n,i,11}_{t,T}$ by $\wh B^{n,i,11}_{t,T}$, which is defined in the same way but uses
			\begin{align*}
				\wh B^{n,i,11}_{t,T}(u)	&= (\Theta_{t^n_i,\Den}(u_{T^n_{i-1}})-1 )\E^\Q_{t^n_i}\biggl[\exp\biggl(iu_{T^n_i}\biggl( \al {T^n_{i-1}}+   \iint_{t^n_{i-1}}^{t+T} \ga(t^n_i,z)\wh\mu(ds,dz)\biggr)\\
				&\quad+iu_{T^n_{i-1}} \int_{t^n_{i-1}}^{t+T} \biggl(\si_{t^n_i}+\iint_{t^n_{i-1}}^{s} \ga^\si(t^n_i,z)\wh\mu(dr,dz)\biggr) dW_s \\
				&\quad+ iu_{T^n_{i-1}} \int_{t^n_{i-1}}^{t+T} \int_{t^n_{i-1}}^s (\si^\si_{t^n_i} dW_r+\ov\si^\si_{t^n_i} d\ov W_r) dW_s\biggr)  \biggr]
			\end{align*}
			instead of    $B^{n,i,11}_{t,T}(u)$. Next,  note that by Assumption~2, we have $\int_\R\lvert\ga(t^n_i,z)\rvert\la(dz)<\infty$, so replacing $\wh\mu$ by $\wh\mu^{n,i}$ in the above display only leads to an error of order $O((u^2\Den/T \wedge 1)(\lvert u\rvert(\sqrt{\Den T} +\Den^{1/4}\sqrt{T})))=O((u^2\Den/T \wedge 1)(\lvert u\rvert T\wedge 1))$, which is negligible as we have seen above. Therefore, we may further replace $\wh B^{n,i,11}_{t,T}$ by $\ov B^{n,i,11}_{t,T}$, which  uses
			\begin{align*}
				\ov B^{n,i,11}_{t,T}(u)	&= (\Theta_{t^n_i,\Den}(u_{T^n_{i-1}})-1 )\E^\Q_{t^n_i}\biggl[\exp\biggl(iu_{T^n_i}\biggl( \al {T^n_{i-1}}+   \iint_{t^n_{i-1}}^{t+T} \ga(t^n_i,z)\wh\mu^{n,i}(ds,dz)\biggr)\\
				&\quad+iu_{T^n_{i-1}} \int_{t^n_{i-1}}^{t+T} \biggl(\si_{t^n_i}+\iint_{t^n_{i-1}}^{s} \ga^\si(t^n_i,z)\wh\mu^{n,i}(dr,dz)\biggr) dW_s \\
				&\quad+ iu_{T^n_{i-1}} \int_{t^n_{i-1}}^{t+T} \int_{t^n_{i-1}}^s (\si^\si_{t^n_i} dW_r+\ov\si^\si_{t^n_i} d\ov W_r) dW_s\biggr)  \biggr]
			\end{align*}
			instead of    $\wh B^{n,i,11}_{t,T}(u)$.

			%if we replace $\by(d\bs_k,d\bz_k)$ by $\by^{n,i}(d\bs_k,d\bz_k)$ in the definition of $B^{n,i,11}_{t,T}(u)$, we incur an error of size $O((u^2\Den/T\wedge 1)(\lvert u\rvert \Den^{1/4}\sqrt{T}\wedge1))$, which is negligible %$O((u^2\Den/T \wedge 1)(u\wedge 1))$
			%because
			%
			%by the assumption that $\sqrt{\Den}\asymp T$. Similarly, , which is $O((u^2\Den/T \wedge 1)(\Den^{1/4}\sqrt{T}\wedge 1))$ (because $\sqrt{\Den}\asymp T$) and thus negligible. 
			
			By conditioning on $\wh\mu^{n,i}$, we can use Theorem A.1 of \cite{CT22} to write
			\begin{align*}
				\ov B^{n,i,11}_{t,T}(u)	&= (\Theta_{t^n_i,\Den}(u_{T^n_{i-1}})-1 )\E^\Q_{t^n_i}\biggl[\exp\biggl(iu_{T^n_i}\biggl( \al {T^n_{i-1}}+   \iint_{t^n_{i-1}}^{t+T} \ga(t^n_i,z)\wh\mu^{n,i}(ds,dz)\biggr)\\
				&\quad-\frac12\sum_{j=1}^\infty \biggl(\log(1-2iu_{T^n_{i-1}}\al_j) + 2iu_{T^n_{i-1}}\al_j + \frac{\beta_j^2}{1-2iu_{T^n_{i-1}}\al_j}u^2_{T^n_{i-1}}\biggr)\biggl)\biggr]\\
				&= (\Theta_{t^n_i,\Den}(u_{T^n_{i-1}})-1 )\E^\Q_{t^n_i}\biggl[\exp\biggl(iu_{T^n_i}\biggl( \al {T^n_{i-1}}+   \iint_{t^n_{i-1}}^{t+T} \ga(t^n_i,z)\wh\mu^{n,i}(ds,dz)\biggr)\\
				&\quad-\frac12u^2_{T^n_{i-1}}\sum_{j=1}^\infty \frac{\beta_j^2}{1+4u^2_{T^n_{i-1}}\al^2_j} -iu^3_{T^n_{i-1}}\sum_{j=1}^\infty \frac{\al_j\beta_j^2}{1+4u^2_{T^n_{i-1}}\al^2_j}+ O\biggl(u^2_T\sum_{j=1}^\infty\al_j^2\biggr) \biggl)\biggr],
			\end{align*}
			where the sequences $(\al_j)_{j\geq1}$ and $(\beta_j)_{j\geq1}$  are   random variables (independent of $u$ but dependent on $T$, $n$ and $i$) satisfying
			\begin{align*}
				\sum_{j=1}^\infty\al_j^2	&=\frac12\E^\Q_{t^n_i}\biggl[\biggl(\int_{t^n_{i-1}}^{t+T} \int_{t^n_{i-1}}^s (\si^\si_{t^n_i} dW_r+\ov\si^\si_{t^n_i} d\ov W_r) dW_s\biggr)^2\biggr] = \frac14[(\si^\si_{t^n_i})^2+(\ov\si^\si_{t^n_i})^2](T^n_{i-1})^2, \\
				\sum_{j=1}^\infty\beta_j^2	&=\E^\Q_{t^n_i}\biggl[\biggl(\int_{t^n_{i-1}}^{t+T} \biggl(\si_{t^n_i}+\iint_{t^n_{i-1}}^{s} \ga^\si(t^n_i,z)\wh\mu^{n,i}(dr,dz)\biggr) dW_s \biggr)^2\mathrel{\Big|}\wh\mu^{n,i}\biggr]\\ 
				&= \int_{t^n_{i-1}}^{t+T} \biggl(\si_{t^n_i}+\iint_{t^n_{i-1}}^{s} \ga^\si(t^n_i,z)\wh\mu^{n,i}(dr,dz)\biggr)^2 ds.
			\end{align*}
			Taking absolute values and bounding $\al_j^2 \leq \sum_{j=1}^\infty \al_j^2$, we obtain
			\begin{align*}
				\lvert\ov B^{n,i,11}_{t,T}(u)\rvert	&\leq (1-\Theta_{t^n_i,\Den}(u_{T^n_{i-1}}))\E^\Q_{t^n_i}\biggl[ \exp\biggl(-\frac{u^2}{2(1+u^2[(\si^\si_{t^n_i})^2+(\ov\si^\si_{t^n_i})^2]T^n_{i-1})} \\
				&\quad\times\frac{1}{T^n_{i-1}}\int_{t^n_{i-1}}^{t+T} \biggl(\si_{t^n_i}+\iint_{t^n_{i-1}}^{s} \ga^\si(t^n_i,z)\wh\mu^{n,i}(dr,dz)\biggr)^2 ds + C_{t^n_i}u^2T\biggr)\wedge1\biggr].
			\end{align*}
			Note that
			\begin{align*}
				\E^\Q_{t^n_i}\biggl[\biggl(\sup_{s\in[t^n_{i-1},t+T]}\biggl\lvert\iint_{t^n_{i-1}}^{s} \ga^\si(t^n_i,z)\wh\mu^{n,i}(dr,dz)\biggr\rvert\biggr)^2\biggr]	\leq C_{t^n_i}T
			\end{align*}
			by Doob's inequality. Therefore, distinguishing whether or not we are on the event $A=\{\sup_{s\in[t^n_{i-1},t+T]}\lvert\iint_{t^n_{i-1}}^{s} \ga^\si(t^n_i,z)\wh\mu^{n,i}(dr,dz)\rvert>\si_{t^n_i}/2\}$, we derive the estimate
			\begin{align*}
				\lvert\ov B^{n,i,11}_{t,T}(u)\rvert	&\leq (1-\Theta_{t^n_i,\Den}(u_{T^n_{i-1}}))\Bigl[\Bigl( 4T/\si_{t^n_i}^2 + e^{-\frac14\si_{t^n_i}^2u^2/(1+C_{t^n_i}u^2T) + C_{t^n_i}u^2T}\Bigr)\wedge1\Bigr]\\
				&\leq C_{t^n_i} (u^2\Den/T\wedge1)\Bigl(T+e^{C_{t^n_i}-u^2{\si^2_{t^n_i}}/(4(1+C_{t^n_i}))}\bone_{\{u^2T\leq 1\}} \\
				&\quad+ e^{-\si^2_{t^n_i}/(16T)}\bone_{\{1<u^2T\leq \si^2_{t^n_i}/(16C_{t^n_i}T)\}} + \bone_{\{u^2T> \si^2_{t^n_i}/(16C_{t^n_i}T)\}}\Bigr).
			\end{align*}
			This implies that
			\begin{align*}
				\int_\R \frac{\lvert\ov B^{n,i,11}_{t,T}(u)\rvert}{u^2/\sqrt{T}\vee \lvert u\rvert} du&\leq C_{t^n_i} \biggl(\int_\R \frac{u^2\Den\wedge T}{u^2/\sqrt{T}\vee \lvert u\rvert} du+\int_\R  e^{-c_{t^n_i}(u^2\wedge T^{-1})}\frac{u^2\Den/T\wedge 1}{u^2/\sqrt{T}\vee \lvert u\rvert} du\\
				&\quad + \int_\R   \frac{u^2\Den/T\wedge 1}{u^2/\sqrt{T}\vee \lvert u\rvert}\bone_{\{ \lvert u\rvert > c_{t^n_i}/T \}} du\biggr)\\
				&\leq C_{t^n_i}(\sqrt{\Den}T+\Den/\sqrt{T}+e^{-c_{t^n_i}/T}T + T^{3/2})=O(\sqrt{\Den T}),
			\end{align*}
			where $c$ is a positive $\F$-adapted process that is locally bounded away from $0$. As a result, $\wt B^{n,i,11}_{t,T} = O(\sqrt{\Den T}\log\Den^{-1})$.
			
			Finally, regarding $\wt B^{n,i,12}_{t,T}$, a similar argument shows that we can dismiss terms with $k\geq1$ in the sum $\sum_{k=1}^N$ and afterwards replace $\by(d\bs_k,d\bz_k)$ by $\by^{n,i}(d\bs_k,d\bz_k)$. Upon realizing that the resulting conditional expectation is zero (because the first term is $\calf_{t^n_{i-1}}$-measurable and centered and the second term is independent of $\calf_{t^n_{i-1}}$), we conclude that  $\wt B^{n,i,12}_{t,T} = O(\sqrt{\Den T}\log\Den^{-1})$.
		\end{proof}
		
		\begin{proof}[Proof of Lemma 8.2]
			Let $\cale$ denote the characteristic function of $\ov\eps_{j,t,T}$ (which does not depend on $j$, $t$ or $T$ by Assumption~3) and 
			\begin{equation}\label{eq:beta-2}\begin{split}
					\beta_{j,t,T}(u)&=\frac{2F'(\si^2_{t,T}(u))}{u^2\lvert \call_{t,T}(u)\rvert^2}\biggl[ \Re\bigl(\call_{t,T}(u)\bigr)\Re\biggl(\biggl(\frac{u^2}{T} +i \frac{u}{\sqrt{T}}\biggr)e^{(iu/\sqrt{T}-1)(k_{j-1,t,T}-x_t)}\biggr) \\
					&\quad \qquad+\Im\bigl(\call_{t,T}(u)\bigr)\Im\biggl(\biggl(\frac{u^2}{T} +i \frac{u}{\sqrt{T}}\biggr)e^{(iu/\sqrt{T}-1)(k_{j-1,t,T}-x_t)}\biggr)\biggr]\\
					&\quad\times e^{-x_t}\zeta_{t,1}(k_{j-1,t,T}-x_t)O_{t,T}(k_{j-1,t,T})\delta_{j,t,T}.
			\end{split}\end{equation}
			Then $\eps_{t,T}(u)=\sum_{j=2}^{N_{t,T}} \beta_{j,t,T}(u)\ov\eps_{j-1,t,T}$ and $\E[e^{iU_n\eps_{t,T}(u)}\mid \calf]=\prod_{j=2}^{N_{t,T}} \cale(U_n\beta_{j,t,T}(u))$. Since $\E[\ov\eps_{j,t,T}]=0$ by assumption and $\ov\eps_{j,t,T}$ has finite moments of all orders, we have
			\[ \biggl\lvert\cale(U_n\beta_{j,t,T}(u))-\biggl(1+ \frac{(iU_nm_2\beta_{j,t,T}(u))^2}{2!}\biggr) \biggr\rvert \leq \frac{(U_nm_{3}\beta_{j,t,T}(u))^{3}}{3!}.\]
			By a classical localization argument, we can assume that $x_t$, $\zeta_{t,1}(k)$, $\rho_{t,1}(k)$ (as defined in Assumption~3) and their derivatives with respect to $k$  are uniformly bounded. In this case, using (8.1) and the assumption $\delta/\sqrt{T}\to0$, one has
			\begin{equation}\label{eq:beta-bound} \begin{split}
					\lvert \beta_{j,t,T}(u)\rvert &\leq \begin{cases} C(u)e^{-\lvert k_{j-1,t,T}-x_t\rvert} \delta_{j,t,T} &\text{if } \lvert k_{j-1,t,T}-x_t\rvert>1,\\
						C(u)\delta_{j,t,T}/\lvert k-x_s\rvert &\text{if } \sqrt{T}<\lvert k_{j-1,t,T}-x_t\rvert<1,\\ C(u)	 {\delta_{j,t,T}}/{\sqrt{T}} &\text{if }\lvert k_{j-1,t,T}-x_t\rvert\leq \sqrt{T}\end{cases}\\
					& \leq C(u), \end{split}
			\end{equation}
			where $C(u)$ does not depend on $j$, $t$ or $T$ and is uniformly bounded on compacts in $u$. Since $\lvert\prod_{j=1}^n a_j-\prod_{j=1}^n b_j\rvert \leq (\bigvee_{j=1}^n (\lvert a_j\rvert \vee\lvert b_j\rvert))^{n-1} \sum_{j=1}^n \lvert a_j-b_j\rvert$, we obtain
			\begin{equation}\label{eq:prod}\begin{split}
					&\E[e^{iU_n\eps_{t,T}(u)}\mid \calf]-\prod_{j=2}^{N_{t,T}}\biggl(1+  \frac{(iU_nm_2\beta_{j,t,T}(u))^2}{2}\biggr) \\
					&\quad 	= O^\uc\Biggl(\sum_{j=2}^{N_{t,T}}\frac{(U_nm_{3}\beta_{j,t,T}(u))^{3}}{3!}\Biggr)= O^\uc(\Den^{-3/2}(\delta/\sqrt{T})^{2})=o^\uc(\sqrt{\Den})
			\end{split}\end{equation}
			by (5.33). We can rewrite this as %Since $k_n\lvert \calt\rvert = O(1/\Den)$ by assumption,
			\[ \E[e^{iU_n\eps_{t,T}(u)}\mid \calf]=\prod_{j=2}^{N_{t,T}} \biggl(1+  \frac{(iU_nm_2\beta_{j,t,T}(u))^2}{2}\biggr) + o^\uc(\sqrt{\Den}).\]
			Taking complex logarithm on both sides and applying  a first-order Taylor expansion (both of which are permitted because $\beta_{j,t,T}(u)\to0$ uniformly), we arrive at
			\begin{align*}
				\Log\E[e^{iU_n\eps_{t,T}(u)}\mid \calf]&=\sum_{j=2}^{N_{t,T}}\Log \biggl(1+ \frac{(iU_nm_2\beta_{j,t,T}(u))^2}{2}\biggr) + o^\uc(\sqrt{\Den}) \\
				& =\sum_{j=2}^{N_{t,T}} \frac{(iU_nm_2\beta_{j,t,T}(u))^2}{2}  + o^\uc(\sqrt{\Den})\\
				& =-\frac12m_2^2 U_n^2  \sum_{j=2}^{N_{t,T}} (\beta_{j,t,T}(u))^{2}+ o^\uc(\sqrt{\Den}).
			\end{align*}
			Define $ \beta_{2,j,t,T}(u)$ in the same way as $\beta_{2,t,T}(u)_k$ from (8.4) but with $k$ replaced by $k_{j-1,t,T}$.
			Then, 
			\begin{equation*}
				\sum_{j=2}^{N_{t,T}} (\beta_{j,t,T}(u))^{2}-\int_\R \beta_{2,t,T}(u)_kdk	= I^1_{t,T}(u)+I^2_{t,T}(u) + I^3_{t,T}(u),
			\end{equation*}
			where
			\begin{align*}
				I^1_{t,T}(u)	&=\sum_{j=2}^{N_{t,T}} \int_{k_{j-1,t,T}}^{k_{j,t,T}} [\beta_{2,j,t,T}(u)-\beta_{2,t,T}(u)_k] dk, \\
				I^2_{t,T}(u)	&= -\int_{-\infty}^{\un k_{t,T}} (\beta_{2,t,T}(u)_k dk -\int_{\ov k_{t,T}}^\infty \beta_{2,t,T}(u)_k dk,  \\
				I^3_{t,T}(u)&=\sum_{j=2}^{N_{t,T}} [(\beta_{j,t,T}(u))^{2}-\beta_{2,j,t,T}(u)\delta_{j,t,T}].
			\end{align*}
			Using (5.11), (5.12), (5.33) and (8.1), one can show that $U_n^2I^3_{t,T}(u)= O^\uc(\delta(\delta/\sqrt{T})/\Den)=o^\uc(\sqrt{\Den})$ and $U_n^2I^2_{t,T}(u)=O^\uc(e^{-(\lvert \un k_{t,T}\rvert\wedge \lvert \ov k_{t,T}\rvert)}/\Den)=O^\uc(e^{-(\delta/\sqrt{T})^\iota}/\Den)=o^\uc(\sqrt{\Den})$. Similarly, the previous bounds in conjunction with (8.2) and the differentiability properties of $\rho_{t,1}(k)$ and $\zeta_{t,1}(k)$ with respect to $k$  yield $U_n^2I^1_{t,T}(u)=O^\uc((\delta/\sqrt{T})^2/\Den)=o^\uc(\sqrt{\Den})$, which completes the proof of (8.5).
		\end{proof}
		
		%\begin{proof}[Proof of Corollary~\ref{cor:finact}]
		%	The proof is the same as for Corollary~\ref{cor:incr}, except that we no longer bound $\Delta^n_i \Psi_{t,T}(u)$ and $\Delta^n_i \Phi_{t,T}(t^n_i,u)$ by $O^\uc(\sqrt{\Den T})$. Instead, we write
		%	\begin{align*}
		%	\Delta^n_i \Phi_{t,T}(t^n_i,u)	&=\frac{2T^n_{i-1}}{u^2}  \int_\R \bigl[\cos(u_{T^n_i}z)-\cos(u_{T^n_{i-1}} z)\bigr] \nu_{t^n_{i-1}}(z)dz\\
		%	&\quad+\frac{2T^n_{i-1}}{u^2} \int_\R (1-\cos(u_{T^n_i} z))(\nu_{t^n_{i-1}}(z)-\nu_{t^n_i}(z)) dz- \frac{2\Den}{u^2} \int_\R (1-\cos(u_{T^n_i}z))\nu_{t^n_i}(z)dz.
		%	\end{align*}
		%By the mean-value theorem, the first  term is $O^\uc(\Den/\sqrt{T})$. The second one is $O^\uc(\sqrt{\Den}T)$, while the last one is $O^\uc(\Den)$ by Theorem 2 of \cite{todorov2021bias}. In conclusion, 
		%\begin{align*}\Delta^n_i \Phi_{t,T}(t^n_i,u)&= \frac{2T^n_{i-1}}{u^2} \int_\R (1-\cos(u_{T^n_i} z))(\nu_{t^n_{i-1}}(z)-\nu_{t^n_i}(z)) dz +  O^\uc(\Den/\sqrt{T}) \\
		%	&= O^\uc(\sqrt{\Den}T+\Den/\sqrt{T}).\end{align*}
		%\end{proof}
		
	\end{appendix}
	
\end{document}